%% file: main.tex
\pdfoutput=1
\documentclass[conference]{IEEEtran}
\input{preamble}

\usepackage{paracol}
\usepackage{needspace}
\usepackage{capt-of}
\makeatletter\globalcounter{minted@pygmentizecounter}\makeatother
\hypersetup{pdftitle={FVQKD: A Formal Verification of Quantum Key Distribution Security in Rocq},pdfauthor={Ruixuan Deng, Zhengfeng Ji, Jiayu Zhang, Li Zhou}}

\begin{document}

\title{FVQKD: A Formal Verification of\\ Quantum Key Distribution Security in Rocq}

\author{
\IEEEauthorblockN{Ruixuan Deng\IEEEauthorrefmark{1}\IEEEauthorrefmark{2},
  Zhengfeng Ji\IEEEauthorrefmark{1}\IEEEauthorrefmark{2},
  Jiayu Zhang\IEEEauthorrefmark{2}, and
  Li Zhou\IEEEauthorrefmark{3}}
\IEEEauthorblockA{\IEEEauthorrefmark{1}Tsinghua University}
\IEEEauthorblockA{\IEEEauthorrefmark{2}Zhongguancun Laboratory}
\IEEEauthorblockA{\IEEEauthorrefmark{3}Institute of Software, Chinese Academy of Sciences}
}

\maketitle
\pagestyle{plain}

\begin{abstract}
Quantum key distribution (QKD), whose security is derived from quantum mechanics rather than computational assumptions, is moving into standardization and commercial deployment.
Yet its security guarantees rest on technically delicate proofs checked by hand.
Formal verification closes this gap, establishing a two-factor assurance: unconditional security from physical principles, and machine-checked proofs against reasoning errors.
We present \FVQKD{}, a \Rocq{} library for machine-checking QKD security proofs.
Its first result is a complete, end-to-end machine-checked proof of composable, finite-key security for BB84.
The mechanization reconstructs the complete proof infrastructure from first principles, including the smooth entropy framework, the entropic uncertainty relation, and the Quantum Leftover Hash Lemma.
It uncovers necessary side conditions that the published argument leaves unstated, one of which reveals a missing proof branch, and it grounds the information-theoretic bound in a formal semantics of protocol execution, which the pen-and-paper analysis lacks.
The proof is entirely admit-free, relying only on the \Rocq{} kernel and standard foundational axioms.
To our knowledge, this is the first end-to-end formal verification of a QKD security proof.
Built on \MathComp{}, \MathCompAnalysis{}, and \CoqQ{}, the development provides quantum-information infrastructure designed for reuse across broader QKD variants and quantum cryptography applications.
\end{abstract}

\input{sections/introduction}
\input{sections/qkd}
\input{sections/security}
\input{sections/protocol}
\input{sections/architecture}
\input{sections/engineering}
\input{sections/findings}
\input{sections/discussion}
\input{sections/conclusion}

\section*{Acknowledgment}
We thank Marco Tomamichel for pointing out to us that the side conditions discussed in \Cref{sec:findings:eur} are already addressed in the later analysis of Tomamichel and Leverrier~\cite{tomamichel2017largelyself}, which offers a more formal pen-and-paper treatment of~\cite{tomamichel2012tight}.

Generative AI tools were employed under human direction to assist with \Rocq{} proof engineering, manuscript improvement and grammar enhancement, and the preparation of figures, charts, and \Cref{app:semantics-alignment}.
All formal artifacts were mechanically verified by the \Rocq{} type checker, and the authors retain sole responsibility for the work.

\bibliographystyle{IEEEtran}
\bibliography{IEEEabrv,./main}

\clearpage
\appendices
\crefalias{section}{appendix}

\input{sections/app-qkd-variants}
\input{sections/app-bb84-sig}
\input{sections/app-semantics-alignment}

\end{document}

%% file: preamble.tex
\usepackage{cite}
\usepackage[T1]{fontenc}
\usepackage[utf8]{inputenc}
\usepackage{float}
\usepackage{url}
\usepackage{xcolor}
\usepackage[hidelinks]{hyperref}

\usepackage{amsmath}
\usepackage{amssymb}
\usepackage{stmaryrd}
\usepackage{alltt}
\usepackage{graphicx}
\DeclareUnicodeCharacter{2297}{\ensuremath{\otimes}}
\DeclareUnicodeCharacter{2291}{\ensuremath{\sqsubseteq}}
\usepackage{array}
\usepackage{tikz}
\usetikzlibrary{arrows.meta,calc,fit,backgrounds}
\usepackage[frozencache,cachedir=minted-cache]{minted}
\definecolor{coqbg}{RGB}{248,250,253}
\usepackage{cleveref}
\newmintinline[coqi]{coq}{fontsize=\small}
\newminted[coqcode]{coq}{
  fontsize=\footnotesize,
  breaklines=true,
  breakafter=_.,
  autogobble,
  bgcolor=coqbg,
  frame=single,
  framesep=2mm,
  rulecolor=\color{gray!30}
}
\newcommand{\coqkw}[1]{\texttt{\small\textcolor[rgb]{0.00,0.50,0.00}{\textbf{#1}}}}

\newtheorem{theorem}{Theorem}

\newtheorem{definition}[theorem]{Definition}

\newcommand{\Tr}{\operatorname{Tr}}
\newcommand{\Hmin}{H_{\min}}
\newcommand{\Hmax}{H_{\max}}
\newcommand{\Hmineps}[1][\eps]{H^{#1}_{\min}}
\newcommand{\Hmaxeps}[1][\eps]{H^{#1}_{\max}}
\newcommand{\eps}{\varepsilon}
\newcommand{\eb}{\bar{\eps}}

\newcommand{\R}{\mathbb{R}}

\newcommand{\id}{\mathbf{1}}
\newcommand{\ket}[1]{|#1\rangle}
\newcommand{\bra}[1]{\langle#1|}
\newcommand{\given}{\,|\,}
\newcommand{\norm}[1]{\lVert #1 \rVert}
\newcommand{\Qtol}{Q_{\mathrm{tol}}}
\newcommand{\leakec}{\mathrm{leak}_{\mathrm{EC}}}
\newcommand{\Rocq}{Rocq\xspace}
\newcommand{\MathComp}{\textsc{MathComp}\xspace}
\newcommand{\MathCompAnalysis}{\textsc{MathComp-Analysis}\xspace}
\newcommand{\CoqQ}{\textsc{CoqQ}\xspace}
\newcommand{\FVQKD}{\textsc{FVQKD}\xspace}
\newcommand{\code}[1]{\texttt{\small #1}}
\newcommand{\sem}[1]{\llbracket #1 \rrbracket}
\newcommand{\auxsem}[1]{\sem{#1}^{\mathrm{aux}}}
\usepackage{xspace}
\newcommand{\TotalLoCk}{151}
\newcommand{\InfrastructureLoC}{102\,302}
\newcommand{\TotalFiles}{67}

%% file: sections/introduction.tex
\section{Introduction}\label{sec:intro}

Quantum key distribution (QKD) offers a distinctive approach to secure key
establishment for the quantum era.
Unlike conventional public-key exchange protocols such as Diffie--Hellman, QKD
derives the secrecy of a shared key from physical measurements governed by
quantum mechanics rather than from computational hardness.
Under suitable device model and authentication assumptions, it provides
information-theoretic security against adversaries with unbounded computational
resources, including a full-scale quantum computer.
BB84~\cite{bennett1984proc}, proposed in 1984 by Bennett and Brassard, is the
founding QKD protocol. Its inventors received the 2025 ACM A.M.\ Turing Award
for their role in establishing the foundations of quantum information
science~\cite{2025turing}.
QKD is also technologically distinctive: it uses quantum mechanics but requires
only the ability to prepare, transmit, and measure single-qubit states, making
it one of the earliest commercially mature quantum-information technologies.
In the past four decades since its invention, QKD has advanced from a theoretical
proposal to satellite-to-ground demonstrations~\cite{satellite} and active
standardization~\cite{ETSIQKD005,isoiec2023qkd23837one}.

Although BB84 has a compelling security intuition, a rigorous proof that an
eavesdropper cannot learn the key without causing detectable disturbance took
more than a decade to establish.
Lo and Chau~\cite{lo1999unconditional} and Mayers~\cite{mayers2001unconditional}
established the unconditional security of QKD in the asymptotic limit of
infinitely many transmitted qubits.
Shor and Preskill~\cite{peter2000simple} simplified the argument by reducing
it to the theory of quantum error correction.
Subsequent work strengthened QKD security proofs in two complementary
directions.
Finite-key analysis makes the guarantee quantitative for a single protocol run
at realistic block sizes that can be as few as $10^{4}$ quantum signals.
Composable security specifies that the generated key remains safe when used as
a component of a larger protocol.
For the proof line formalized in this paper, these newer advances rely on the
smooth-entropy framework of Renner~\cite{renner2008security} and the tight
finite-key analysis of Tomamichel et al.~\cite{tomamichel2012tight}.

These security proofs combine deep results from quantum information theory
and cryptography, including an entropic uncertainty relation that bounds an
eavesdropper's knowledge, a sampling inequality that bounds the unobserved error
rate, and a hashing lemma that converts residual uncertainty into secret key.
The mathematical sophistication of these ingredients makes the resulting proofs
difficult to understand and verify.
This concern is not merely theoretical.
Important results in quantum information theory have been found to contain
serious gaps years after publication~\cite{Berta2023gapinproofof,erratum}.
The QKD community has also recognized such challenges.
A recent survey identifies gaps and limitations in published decoy-state
BB84 security proofs~\cite{tupkary2026securityproofspracticalqkd}, while
a new pen-and-paper proof explicitly aims for completeness and
certification-grade rigor~\cite{tupkary2026rigorouscompletesecurityproof}.
These developments raise a fundamental question:
\textit{can QKD security proofs be placed on a more trustworthy foundation, one
  that systematically consolidates and verifies the results drawn from quantum
  information theory, cryptography, and physics?}

Formal verification provides a systematic way to answer this question
as it has already proved effective in classical security.
Notable examples include verified cryptographic libraries such as
HACL*~\cite{zinzindohoue2017haclstar} and
EverCrypt~\cite{protzenko2020evercrypt}, verified OS kernels such as
seL4~\cite{klein2009sel4}, and certified compilers such as
CompCert~\cite{leroy2009fv}.
For cryptographic proofs, frameworks such as
SSProve~\cite{haselwarter2023ssprove},
CryptHOL~\cite{basin2020crypthol}, EasyCrypt~\cite{barthe2011computer}, and
CryptoVerif~\cite{blanchet2008cryptoverif} support machine-checked
reasoning about computational-security reductions~\cite{barbosa2021sok}.
Protocol-level frameworks such as Squirrel~\cite{baelde2021interactive} and
DY*~\cite{bhargavan2021dystar,wallez2025treekem} support machine-checked security analyses of
practical protocols.
These cryptographic proof and protocol frameworks work well for classical cryptography, but QKD falls outside their
scope.
Its security is information-theoretic, resting on smooth-entropy bounds against
an unbounded quantum adversary.
It also requires mathematical foundations such as density operators, quantum
channels, and smooth min- and max-entropy that these classical tools do not
provide.

Proof assistants have also been applied to quantum computing and quantum
information.
Several foundational libraries and program verification frameworks have been established: in \Rocq{}, \CoqQ{}~\cite{zhou2023coqq}
provides a finite-dimensional quantum mechanics library and program logics~\cite{ying2016fqp},
QWIRE~\cite{Paykin2017qwire,rand2019reqwire} and VOQC~\cite{hietala2021voqc}
target quantum circuit verification, with a certified end-to-end implementation
of Shor's factoring algorithm~\cite{wu2023shor} as one result,
and VyZX~\cite{lehmann2026vyzx} mechanizes reasoning about the ZX-calculus~\cite{coecke2011interacting,Kissinger2019PyZX}, while ViCAR~\cite{shah2025vicar} supports automated rewriting and visualization for monoidal categories.
Isabelle/HOL hosts a quantum computing library~\cite{bordg2021certified, bordg2020isabelle}
and a mechanized quantum Hoare logic~\cite{liu2019fvqhl};
Qbricks~\cite{chareton2021circuit} embeds a circuit-building quantum language
in Why3 to support high automation based on path-sum representations~\cite{amy2018towards}.
Recent Lean work develops kernel-checked infrastructure for finite-dimensional
quantum information theory and applies it to quantum coding
theorems~\cite{zhu2026leanqit}.
For (post-)quantum cryptographic protocols, qRHL~\cite{unruh2019qrhl} has
machine-checked the post-quantum security of the Fujisaki--Okamoto
transform~\cite{unruh2020post}.
EasyPQC~\cite{barbosa2021easypqc} and a computational post-quantum protocol
logic~\cite{cremers2022logic} verify classical cryptographic schemes and
protocols against quantum adversaries.

Closer to the focus of this paper are several works applying formal methods directly to QKD\@.
Kubota et al.~\cite{kubota2011formal,kubota2016semi} give a
semi-automated verification of the Shor--Preskill-style BB84 proof route,
different from the composable finite-key analysis formalized here.
Boender et al.~\cite{boender2015formalization} develop Coq infrastructure
for modeling quantum protocols, while Kamm\"{u}ller~\cite{kammuller2019qkdisabellebayesian}
formalizes a Bayesian/probabilistic BB84 argument in Isabelle.
The Lean \code{formal-qkd} project~\cite{FormalQKDLean} mechanizes
secret-key-length computation equations and partial post-processing code,
without the underlying composable security proof.
Model-checking tools such as PRISM, QCTL, and QMC~\cite{nagarajan2005automatedanalysissecurityquantum,baltazar2008qctl,gay2008qmc}
and symbolic analyses~\cite{hirschi2019symbolic,prevost2024etsi,malik2026beyond}
analyze protocol-level behavior, intruder models, or ETSI protocol-layer
properties rather than mechanizing the finite-key security proof itself.
Adjacent formalization efforts include the generalized quantum Stein's lemma
in Lean~\cite{meiburg2025formalizationgeneralizedquantumsteins}
and classical Shannon information theory in \Rocq{}~\cite{affeldt2014shannon}.

\begin{table}[t]
  \centering
  \caption{A comparison of \FVQKD{} with prior work.}
  \label{tab:related-comparison}
  \begingroup
  \scriptsize
  \renewcommand{\arraystretch}{1.15}
  \setlength{\tabcolsep}{1.5pt}
    \begin{tabular*}{\columnwidth}{@{}>{\raggedright\arraybackslash}p{2.80cm}
                    >{\raggedright\arraybackslash}p{1.55cm}
                    @{\extracolsep{\fill}}ccccc@{}}
      \hline
      \textbf{Development}
      & \textbf{Framework}
      & \textbf{UQA}
      & \textbf{Comp.}
      & \textbf{FK}
      & \textbf{Model}
      & \textbf{E2E} \\ \hline
      Kubota et al.~\cite{kubota2011formal,kubota2016semi}
      & qCCS\textsuperscript{*}
      & \checkmark & --- & --- & \checkmark & --- \\
      Boender et al./Kamm\"{u}ller%
      \cite{boender2015formalization,kammuller2019qkdisabellebayesian}
      & Coq/Isabelle
      & --- & --- & --- & \checkmark\textsuperscript{\dag} & --- \\
      Lean project~\cite{FormalQKDLean}
      & Lean
      & --- & --- & ---\textsuperscript{\ddag} & --- & --- \\
      CryptHOL constr.\ crypto~\cite{lochbihler2019crypthol,basin2020crypthol}
      & Isabelle/HOL
      & --- & \checkmark & --- & \checkmark\textsuperscript{\S} & --- \\
      CHSH rigidity/Tsirelson%
      \cite{zhao2026formalizingchshrigiditylean,echenim2023chsh}
      & Lean/Isabelle
      & --- & --- & --- & --- & --- \\ \hline
      \FVQKD{} (this work)
      & \Rocq{}
      & \checkmark & \checkmark & \checkmark & \checkmark & \checkmark \\\hline
    \end{tabular*}
  \par\vskip 0.6\baselineskip
  \vspace{0.35\baselineskip}
  \begin{minipage}{0.98\columnwidth}
  UQA: unbounded quantum adversary; Comp.: composable security; FK:
  finite-key bound; Model: explicit protocol model; E2E: end-to-end QKD
  security proof.
  A checkmark marks each criterion met; --- marks an unmet criterion.
  \textsuperscript{*}Custom semi-automated quantum-process-calculus verifier,
  not a proof-assistant kernel.
  \textsuperscript{\dag}Individual components only.
  \textsuperscript{\ddag}Finite-key key-length equations without security
  derivation.
  \textsuperscript{\S}Composable framework for classical, not quantum,
  protocols.
  \end{minipage}
  \endgroup
\end{table}

As \Cref{tab:related-comparison} summarizes, the closest existing formal
developments each satisfy only a subset of these five properties; none achieves all
of them.
The central challenge of placing QKD security proofs on a trusted foundation
therefore remains open.
In this work, we fill this gap by presenting \FVQKD{}, to our knowledge the first
complete end-to-end machine-checked proof of composable finite-key security for a
QKD protocol.
By complementing security from quantum mechanics with machine-checked
verification in \Rocq{}, \FVQKD{} provides an additional assurance layer for
QKD systems.
Built with \Rocq{}~\cite{rocq911}, it uses \MathComp{} for algebra,
\MathCompAnalysis{} for real analysis, and
\CoqQ{}~\cite{zhou2023coqq} for finite-dimensional quantum theory.
As a major case study, \FVQKD{} formalizes the tight composable
finite-key security analysis of BB84 by Tomamichel et al.~\cite{tomamichel2012tight}.
To this end, the mechanization reconciles the conventions of the six
theoretical works underlying the BB84 security proof and supplies
the side conditions that the published argument leaves implicit
(\Cref{sec:findings}).
\Cref{fig:deps} presents the overall structure of the development.

\noindent\textbf{Contributions.}
Our contributions are twofold.
First, \FVQKD{} provides
\emph{a reusable, modular formal-verification framework for QKD security proofs.}
It brings together machine-checked \Rocq{} formalizations of
Uhlmann's fidelity theorem, generalized (sub-normalized)
fidelity, purified distance as a metric with monotonicity under completely
positive trace-nonincreasing (CPTN) maps, the
smooth-entropy calculus (smooth conditional min- and max-entropy, their duality
and classical chain rules), the POVM-general smooth entropic uncertainty
relation, the Quantum Leftover Hash Lemma, Serfling's sampling
bound~\cite{serfling1974probability}, and semidefinite-programming (SDP)
strong duality.
The SDP result establishes the primal direction of Tomamichel's
Theorem~2.6~\cite{tomamichel2013frameworknonasymptoticquantuminformation},
proved from first principles.
The majority of the project is protocol-agnostic, with no dependence on BB84 or
any particular QKD protocol, and serves as a general framework that future QKD
variants can instantiate without re-deriving the underlying mathematical
machinery.

Second, \FVQKD{} establishes
\emph{an end-to-end machine-checked QKD security proof, instantiated for BB84}:
an unbroken, kernel-checked chain from a formal protocol program to a
composable finite-key $\eps$-security statement at the tight bound
of Tomamichel et al.~\cite{tomamichel2012tight}.
The proof contains no incomplete lemmas;
developed from first principles, it relies solely on the \Rocq{} kernel
and three standard axioms inherited from \MathCompAnalysis{}.

It is worth emphasizing that existing formal verification efforts in protocol
security usually focus on protocol-level modeling:
message flows, adversary capabilities, and reachability properties.
For QKD, the mathematical security bound, not the protocol model, carries the
main proof burden.
Finite-key QKD proofs in the smooth-entropy tradition
(BB84~\cite{tomamichel2012tight}, decoy-state
variants~\cite{lim2014concise}, and measurement-device-independent
protocols~\cite{curty2014finite})
use the same chain of entropic uncertainty relation, Quantum Leftover Hash
Lemma, and sampling bounds.
Formalizing this chain once provides a reusable foundation
for their mechanization in the future.

\noindent\textbf{Outline.}
The paper is organized as follows.
\Cref{sec:qkd} introduces the quantum-mechanical foundations, the common
structure of QKD protocols, and BB84.
\Cref{sec:security} fixes the composable security model
and derives the finite-key bound from the chain of entropy inequalities.
\Cref{sec:protocol} describes the protocol model
and connects the mathematical proof to a concrete protocol specification.
\Cref{sec:architecture} details the proof architecture:
the module decomposition and its design for reuse.
\Cref{sec:engineering} discusses proof engineering:
the backend libraries, the trusted computing base, and the representation
layers.
\Cref{sec:findings} presents the verification findings:
a missing proof branch and competing conventions across the literature.
\Cref{sec:discussion} discusses scope and future work.
\Cref{sec:conclusion} concludes.

%% file: sections/qkd.tex
\section{Quantum Key Distribution}\label{sec:qkd}

Quantum key distribution is a family of protocols that let two parties (conventionally Alice and Bob) establish a shared secret key whose security rests on quantum mechanics rather than on any computational hardness assumption.
A sufficiently powerful adversary could break classical key exchange by solving a hard mathematical problem; a QKD key remains secret even against an eavesdropper with unlimited computational resources.
This section introduces the quantum-mechanical foundations, the common structure of QKD protocols, and the BB84 protocol, our running example.
These set the stage for the security analysis (\Cref{sec:security}) and the protocol model (\Cref{sec:protocol}).

\subsection{Quantum Foundations}\label{sec:qkd:foundations}

\noindent\textit{Quantum states.}\
A \emph{qubit} is the quantum analog of a classical bit: a two-level quantum system that, unlike its classical counterpart, can exist in a \emph{superposition} of both values simultaneously.
In Dirac notation, a pure quantum state is written $\ket{\psi}$ (a column vector), with conjugate transpose $\bra{\psi}$; together they form the outer product $\ket{\psi}\!\bra{\psi}$, a rank-one matrix describing the state.
The standard \emph{computational basis} (also called the $Z$-basis) consists of the states $\ket{0}$ and $\ket{1}$, corresponding to the two classical bit values.
The \emph{$X$-basis} (Hadamard basis), whose states $\ket{+} = \tfrac{1}{\sqrt{2}}(\ket{0}+\ket{1})$ and $\ket{-} = \tfrac{1}{\sqrt{2}}(\ket{0}-\ket{1})$ are superpositions of the computational-basis states, is equally important for QKD.
When there is classical uncertainty about the prepared state, or when a qubit is part of a larger entangled system, the state is described by a \emph{density matrix}~$\rho$---a positive-semidefinite matrix with unit trace ($\Tr\rho = 1$, where $\Tr$ denotes the matrix trace), generalizing probability distributions to the quantum setting.
Composite systems of multiple parties are modeled by tensor products: a joint state $\rho_{ABE}$ describes the combined quantum system shared among Alice, Bob, and an eavesdropper~\cite{nc2010qcqi}.

\noindent\textit{Measurement and bases.}\
Extracting classical information from a qubit requires a \emph{measurement},
which produces a definite outcome but, in general, irreversibly disturbs the quantum state.
The $Z$-basis and $X$-basis are \emph{conjugate}: a state prepared in one basis yields a uniformly random outcome when measured in the other.
For example, a qubit prepared as $\ket{0}$ (a definite $Z$-basis state) produces the outcomes $+$ or~$-$ with equal probability when measured in the $X$-basis, and the original state is destroyed in the process.
This conjugacy is the physical mechanism that makes QKD possible.
An eavesdropper who intercepts a qubit and measures it in the wrong basis gains no useful information and causes detectable errors in the legitimate parties' data~\cite{wiesner1983conjugate}.

\noindent\textit{Fidelity and distance.}\
A security proof must give precise meaning to ``the key is approximately secret,'' which requires a notion of closeness between quantum states.
\emph{Fidelity} quantifies the overlap between two states: a fidelity of~$1$ means the states are identical, while~$0$ means they are perfectly distinguishable.
The \emph{purified distance} converts fidelity into a proper metric, and the $\eps$-ball $B^\eps(\rho)$ around a state~$\rho$ collects all states within purified distance~$\eps$.
The smooth entropy measures used in the security proof optimize over this neighborhood~\cite{tomamichel2013frameworknonasymptoticquantuminformation,tomamichel2016qip}.
The mathematical definitions of fidelity, purified distance, and smoothing
appear in Definitions~\ref{def:purified-distance} and~\ref{def:smooth-min} of \Cref{app:semantics-alignment},
alongside their formal definitions in \Rocq{}.

\subsection{Components and Variants}\label{sec:qkd:components}

\noindent\textit{What QKD achieves.}\
A QKD protocol lets Alice and Bob distill a shared secret key from quantum signals exchanged over a public channel, with security derived from quantum mechanics under explicit modeling assumptions (\Cref{sec:security:model}).
The guarantee is \emph{information-theoretic}: no computational assumption is needed, and the key remains secure even against a future quantum computer.
The \emph{no-cloning theorem}~\cite{Wootters1982single} (an unknown quantum state cannot be perfectly copied) prevents an eavesdropper from silently duplicating the quantum signals and forces any information-gaining interception to leave a detectable trace.

\noindent\textit{The common pipeline.}\
Despite their differences in quantum hardware, QKD protocols share a common high-level structure: (i)~source preparation, in which quantum states encoding random key material are generated; (ii)~quantum transmission through an optical channel; (iii)~measurement of the transmitted signals; (iv)~\emph{sifting}, which discards rounds where incompatible settings were used; (v)~\emph{parameter estimation}, which tests a random sample for evidence of eavesdropping; (vi)~\emph{error correction}, which reconciles the remaining raw keys; and (vii)~\emph{privacy amplification}, which compresses the corrected key to remove any information the eavesdropper may have gained.
The differences among protocols lie in what is prepared (single photons, entangled pairs, or coherent states) and how it is measured; the classical post-processing steps are largely similar.

\noindent\textit{Variants.}\
BB84~\cite{bennett1984proc}, the first QKD protocol, uses a prepare-and-measure scheme in which Alice sends individually prepared qubits to Bob.
E91~\cite{ekert1991bell} is an entanglement-based protocol whose security is certified by the violation of a Bell inequality.
BBM92~\cite{bbm92} is an entanglement-based variant of BB84 that shares its post-processing structure.
Subsequent extensions address practical deployment challenges: the six-state protocol~\cite{bruss1998sixstate} uses three mutually unbiased bases to tolerate a higher error rate; decoy-state methods~\cite{lo2005decoy} defeat photon-number-splitting attacks~\cite{brassard2000limitations} in implementations with weak coherent pulses; and measurement-device-independent (MDI) protocols~\cite{lo2012mdi,curty2014finite} and twin-field protocols~\cite{lucamarini2018overcoming} remove the need to trust the detectors.
QKD has moved beyond the laboratory to real-world deployment,
including satellite-to-ground key exchange demonstrated by the Micius satellite~\cite{satellite},
and is undergoing active standardization~\cite{ETSIQKD005,isoiec2023qkd23837one}.

\subsection{BB84}\label{sec:qkd:bb84}

\noindent\textit{Protocol steps.}\
BB84 proceeds in five stages, applied to a block of quantum signals transmitted from Alice to Bob:
\begin{enumerate}
\item \textbf{Preparation and measurement.}\
  Alice encodes each random key bit by preparing a qubit in one of two randomly chosen conjugate bases ($Z$ or~$X$) and sends it to Bob through the quantum channel.
  Bob measures each arriving qubit in a randomly chosen basis.
\item \textbf{Sifting.}\
  Alice and Bob publicly announce their basis choices (but not the bit values) and keep only the detected rounds where both chose the same basis, yielding $n{+}k$ matching positions ($n$ in the $X$~basis and $k$ in the $Z$~basis).
\item \textbf{Parameter estimation.}\
  The $k$ $Z$-basis positions form the test set: Alice and Bob compare outcomes and abort if the observed bit-error rate exceeds a threshold~$\Qtol$.
\item \textbf{Error correction.}\
  On the $n$ $X$-basis positions, Alice sends correction information (at most $\leakec$ bits) so that Bob can reconcile his key with hers; a hash-based tag confirms agreement.
\item \textbf{Privacy amplification.}\
  A two-universal hash function~\cite{carter1979universal}
  compresses the corrected $n$-bit string to a final key of length~$\ell$, removing essentially all residual information the eavesdropper may hold.
\end{enumerate}

\noindent \Cref{fig:bb84budget} traces the key material through these five stages and renders the key-length bound~\eqref{eq:keylen} of \Cref{sec:security} as a bit budget.

\begin{figure*}[t]
\centering
\includegraphics[width=0.8\textwidth]{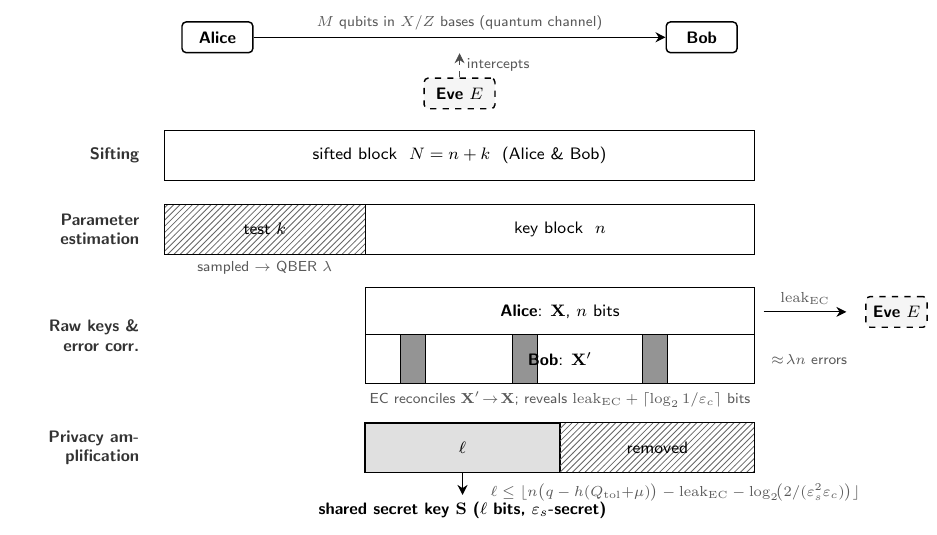}
\caption{BB84 finite-key budget (schematic; bar lengths are illustrative).
Sifting leaves $N=n{+}k$ matching-basis rounds: $k$ $Z$-basis rounds
estimate the QBER~$\lambda$, and $n$ $X$-basis rounds form the raw keys
$\mathbf{X},\mathbf{X}'$.
Error correction reveals $\leakec+\lceil\log_2 1/\eps_c\rceil$ bits;
privacy amplification produces the $\ell$-bit $\eps_s$-secret key~$\mathbf{S}$.
Hatching denotes terms deducted in~\eqref{eq:keylen}, not literal key substrings.}
\label{fig:bb84budget}
\end{figure*}

\noindent\textit{Notation.}\
We denote a protocol run by $\Phi(n, k, \ell, \Qtol, \eps_c, \leakec)$, where $n$ is the raw-key length, $k$ the number of test positions, $\ell$ the final key length, $\Qtol$ the error-rate threshold, $\eps_c$ the correctness parameter, and $\leakec$ the error-correction leakage.
A successful run outputs a pair of keys $(S, \hat{S})$ held by Alice and Bob; an aborted run outputs~$\bot$.
BB84 admits two source models that differ only in the preparation step: in the \emph{prepare-and-measure} (PM) variant, Alice prepares individual qubits in randomly chosen states; in the \emph{entanglement-based} (EB) variant, she prepares entangled pairs and measures her half locally.
The standard security analysis relates these source models, and their subsequent post-processing steps are identical.
Our verified protocol result treats the source-box EB interpretation as a BB84 source choice.
\Cref{sec:security} presents the security criterion and the proof that the final key is composably secure.

%% file: sections/security.tex
\section{Security Analysis}\label{sec:security}

This section develops the mathematical backbone of the formalization.
We trace the proof route from asymptotic unconditional security to the finite-key setting (\Cref{sec:security:routes}), fix a rigorous security model with composable definitions (\Cref{sec:security:model}), and instantiate the finite-key analysis for BB84 (\Cref{sec:security:bound}).

\subsection{Proof Routes}\label{sec:security:routes}

\noindent\textit{Historical development.}\
Unconditional security of QKD was first established in the asymptotic limit by Lo and Chau~\cite{lo1999unconditional} and Mayers~\cite{mayers2001unconditional}, showing that information-theoretic key distribution is possible in principle.
Shor and Preskill~\cite{peter2000simple} simplified the argument by reducing it to quantum error correction, making the asymptotic proof accessible but leaving practical finite-size effects unaddressed.
Renner~\cite{renner2008security} introduced smooth entropies to handle finite block sizes, and Tomamichel et~al.~\cite{tomamichel2012tight} obtained tight composable bounds without asymptotic or de~Finetti penalties.
More recently, entropy accumulation theorems~\cite{dupuis2020ea,metger2023security} have extended finite-key techniques to broader settings, including device-independent ones, where no assumptions are made about the internal workings of Alice's and Bob's devices.

\noindent\textit{Our route.}\
\FVQKD formally proves security along the smooth-entropy route of Tomamichel et~al.~\cite{tomamichel2012tight},
the most established route for prepare-and-measure protocols and one yielding finite-key bounds among the tightest known.
Its modular proof composes three independent, separately machine-checkable components---a smooth entropic uncertainty relation,
a Serfling sampling bound~\cite{serfling1974probability}, and the Quantum Leftover Hash Lemma---into the final key-length theorem (\Cref{sec:security:bound}).

\subsection{Security Model}\label{sec:security:model}

\noindent\textit{Parties.}\
Three parties participate.
Alice (sender) prepares and transmits quantum states.
Bob (receiver) measures them.
Eve (adversary) controls the quantum channel and observes the classical channel.
Alice and Bob are honest: they follow the protocol faithfully and their internal computations are private.

\noindent\textit{Adversary capabilities.}\
Eve is an \emph{unbounded quantum adversary}.
She exercises full control over the quantum channel: she may intercept, store, modify, or replace every quantum signal in transit, so the joint post-sifted state $\rho_{ABE}$ is entirely unrestricted.
She observes all messages on the authenticated classical channel (basis announcements, error-correction data, hash choices) but cannot forge or modify them.
She has unlimited computational power and unlimited quantum memory, and may perform coherent attacks across all $n{+}k$ rounds simultaneously.
No restriction beyond the stated model is placed on her attack strategy; the security guarantee holds for every such strategy.

\noindent\textit{Trust assumptions.}\
The paper's security claims involve six conditions, surfaced at different layers of the machine-checked development:
\begin{enumerate}
\item[\textup{(T1)}] \emph{Side-channel freedom.}\
  Alice's and Bob's devices are free of side channels (e.g., detector blinding~\cite{Lydersen2010hacking}, Trojan-horse attacks~\cite{gisin2006trojan}).
  The bound models the source through the single parameter~$q$; imperfections not captured by the declared source and measurement sets and~$q$ are outside the model.
\item[\textup{(T2)}] \emph{Authenticated classical channel.}\
  An authenticated classical channel is assumed; without it, QKD is vulnerable to man-in-the-middle attacks~\cite{paterson2004quantum}.
  In the information-theoretic setting, this channel can be realized from a pre-shared secret using Wegman--Carter authentication~\cite{wegman1981newhash}.
  The proof treats authentication as a primitive and does not model the key consumed by it;
  Portmann~\cite{portmann2014recycling} bounds the key consumed by authentication in the composable setting,
  and Cederl\"{o}f and Larsson~\cite{cederlof2008security} analyze the security loss when the authentication key is partially known to the adversary.
\item[\textup{(T3)}] \emph{Source characterization.}\
  The source emits states characterized by preparation quality~$q$, which measures basis separation.
  Weak-coherent-pulse sources require additional decoy-state techniques~\cite{lo2005decoy,lim2014concise} and are outside the current scope.
\item[\textup{(T4)}] \emph{Delayed measurement.}\
  The mathematical theorem assumes that Bob's $\mathbb X$-basis measurements can be postponed until after parameter estimation without changing their statistics, as in the memoryless-device model of~\cite{tomamichel2012tight}.
  For the calibrated class in T6, the artifact proves that the state to which the uncertainty relation is applied equals the state that the protocol program produces when parameter estimation passes.
  The min-entropy bound therefore holds for the executed protocol.
\item[\textup{(T5)}] \emph{Basis-independent detection.}\
  The no-click POVM element, and hence the detection probability, is independent of the basis choice~\cite{tomamichel2012tight}; this places the detection-efficiency loophole~\cite{makarov2006efficiency} outside the model.
\item[\textup{(T6)}] \emph{Calibrated projective $\mathbb X/\mathbb Z$ class.}\
  The executable endpoint is restricted to coherent pure $\mathbb X/\mathbb Z$ preparations, a standard source pointer and source $\mathbb X$-measurement, and projective $\mathbb Z$ readout.
\end{enumerate}
\noindent
The source and measurement characterization and~$q$ enter as theorem parameters: no security conclusion can be drawn without supplying them; T1 is the condition under which they faithfully summarize the devices.
T2 is fixed by the channel model, under which Eve reads but cannot modify classical messages.

\noindent\textit{Composable security.}\
A QKD protocol must satisfy two properties: \emph{correctness} and \emph{secrecy}.
``Composable'' means these guarantees survive when the extracted key is used as input to another protocol~\cite{canetti2001uc,portmann2022security}, for instance as a one-time pad for encrypted communication.
Both definitions quantify over all adversaries and attack strategies, regardless of computational resources.

\begin{definition}[$\eps_c$-Correctness]\label{def:correct}
The protocol is \emph{$\eps_c$-correct} if, for every adversary strategy,
\[
\Pr[\hat{\mathbf{S}} \neq \mathbf{S} \;\wedge\; \text{not abort}] \;\le\; \eps_c\,,
\]
where $\mathbf{S}$ and $\hat{\mathbf{S}}$ are Alice's and Bob's final keys.
\end{definition}

\begin{definition}[$\eps_s$-Secrecy]\label{def:secret}
The protocol is \emph{$\eps_s$-secret} if, for every adversary,
\[
(1 - p_\mathrm{abort})\cdot
\min_{\sigma_E}\,
\tfrac{1}{2}\norm{\rho_{\mathbf{S}E}
  - \tau_{\mathbf{S}} \otimes \sigma_E}_1
\;\le\; \eps_s\,,
\]
where $p_\mathrm{abort}$ is the abort probability, $\rho_{\mathbf{S}E}$ is the joint state of the final key and the adversary conditioned on non-abort, $\tau_{\mathbf{S}}$ is the maximally mixed state on the key space, and the minimum ranges over states~$\sigma_E$ of Eve's system.\footnote{If two quantum states have trace distance at most~$\Delta$, no physical experiment can distinguish them with probability better than~$\frac{1}{2}+\frac{\Delta}{2}$.}
\end{definition}

\noindent Here $\tfrac{1}{2}\norm{\cdot}_1$ denotes trace distance.
The protocol is \emph{$\eps$-secure} if it is $\eps_c$-correct and $\eps_s$-secret with $\eps_c + \eps_s \le \eps$~\cite{tomamichel2012tight,ferradini2025definingsecurityquantumkey};
this trace-distance criterion is universally composable~\cite{renner2005ucpa,portmann2022security}.

We write $p_\mathrm{abort} := \Pr[\mathbf{S}=\perp]$ for the abort probability, $p_\mathrm{mismatch} := \Pr[\hat{\mathbf{S}}\ne\mathbf{S}]$ for the mismatch probability, and $p_\mathrm{pass}$ for the probability that parameter estimation (PE) passes; all probabilities are over the protocol's randomness for a fixed adversary strategy.

\noindent\textit{Protocol inputs.}\
The formalization exposes the following as explicit theorem inputs in~\Cref{sec:protocol}:
\begin{itemize}
\item \emph{Certified device:} the source-state family $\mathsf{Src}(a,y)=\rho_{a,y}$, Bob's measurement-operator family~$\mathcal M$, and the preparation-quality bound~$q$.
\item \emph{Sizes:} $n$ (key-generation bits), $k$ (test bits), and $M$ (finite round budget).
\item \emph{Error tolerance:} $\Qtol$ (parameter-estimation threshold).
\item \emph{Error correction:} $\leakec$ (public leakage), $\eps_c$ (failure probability).
\item \emph{Privacy amplification:} $\ell$ (target key length).
\item \emph{Security budget:} $\eps_s$ (secrecy), $\eps_{\mathrm{tot}}$ (total); the smoothing parameter~$\eps$ is optimized inside the key-length bound.
\end{itemize}

\noindent\textit{Verification boundary.}\
\FVQKD{} makes the verification boundary across T1--T6 explicit.
T3 becomes an explicit device-characterization obligation; T1 is the condition under which that characterization describes the physical devices; and T2 is fixed by the modeled read-only classical channel.
At the executable endpoint, T5 and T6 are device obligations, while the entropy comparison corresponding to T4 is proved within that class.

\subsection{Finite-Key Bound}\label{sec:security:bound}

\noindent\textit{Proof setup.}\
The security proof works in the \emph{post-sifted} model: the $n{+}k$ sifted positions are modeled as a single joint quantum state $\rho_{ABE}$ across three registers ($A$: Alice's signal, $B$: Bob's quantum register, $E$: Eve's system).
A uniform random permutation designates $k$ positions as the test block and the remaining $n$ as the key block.
Measurements on the key block produce the $n$-bit raw keys $\mathbf{X},\mathbf{X}'$ for Alice and Bob.

\noindent\textit{Smooth entropies.}\
The conditional min-entropy $\Hmin(A\given E)$ measures randomness in~$A$ hidden from~$E$;
for classical~$A$, $2^{-\Hmin(A\given E)}$ is the adversary's optimal guessing probability~\cite{koenig2009operational}.
The conditional max-entropy $\Hmax(A\given B)$ measures the information needed to reconstruct~$A$ from~$B$~\cite{renesrenner2012oneshot}.
Their smoothed forms $\Hmineps$ and $\Hmaxeps$ optimize over states within purified distance~$\eps$ of~$\rho$ and require $0\le\eps<\sqrt{\Tr\rho}$;
this domain condition surfaces a mechanization finding in~\Cref{sec:findings}.

\noindent\textit{Entropic uncertainty.}\
The central tool linking Alice's raw key to Eve's information is the smooth entropic uncertainty relation.

\begin{theorem}[Entropic Uncertainty Relation]\label{thm:eur}
Let $\rho_{ABE}$ be a tripartite quantum state whose register~$A$ comprises $n$ signal rounds, and let $\{M_x\},\{N_z\}$ be two per-round measurements.
The strings $X$ and $Z$ obtained by measuring each round of~$A$ with $\{M_x\}$ or $\{N_z\}$, respectively, satisfy, for any $0 \le \eps < \sqrt{\Tr\rho}$,
\[
\Hmineps(X\given E)_\rho
+ \Hmaxeps(Z\given B)_\rho \;\ge\; nq\,.
\]
Here $q := -\log_2\max_{x,z}\norm{\sqrt{M_x}\sqrt{N_z}}_\infty^2$ is the preparation quality.
\end{theorem}
\noindent The EUR~\cite{tomamichel2011uncertainty} states that the total uncertainty across two complementary measurements is at least $n$ times the preparation quality.
In operational terms: the more Eve knows about the raw key~$X$ (small $\Hmineps$), the less predictable the conjugate string~$Z$ is given Bob's quantum register (large $\Hmaxeps$), and vice versa.
Parameter estimation bounds $\Hmaxeps(Z\given B)$ via the observed error rate; the EUR then yields a lower bound on $\Hmineps(X\given E)$---the input to privacy amplification.

\noindent\textit{Privacy amplification.}\
The classical Leftover Hash Lemma~\cite{impagliazzo1989recycle,wegman1981newhash} guarantees that hashing a source with sufficient min-entropy through a two-universal family~\cite{carter1979universal} produces a nearly uniform key.
The Quantum LHL extends this guarantee to an adversary holding quantum side information: the smooth min-entropy $\Hmineps(X\given E)$ determines the extractable key length.

\begin{definition}[Two-Universal Hash Family]
\label{def:two-universal}
A family~$\mathcal{F}$ of functions from~$\mathcal{X}$ to~$\{0,1\}^\ell$ is \emph{two-universal} if, for any two distinct inputs $x\ne x'$ and a uniformly chosen $f\in\mathcal{F}$,
\[
\Pr_{f\in\mathcal{F}}[f(x)=f(x')] \;\le\; 2^{-\ell}\,.
\]
\end{definition}

\begin{theorem}[Quantum Leftover Hash Lemma]\label{thm:lhl}
Let $\rho_{XE}$ be a classical--quantum state with $\Hmineps(X\given E) \ge m$ for some $\eps \ge 0$, and let~$\mathcal{F}$ be a finite two-universal family with range~$\{0,1\}^\ell$.
For a uniformly chosen $f\in\mathcal{F}$, the extracted key $f(X)$ satisfies
\[
\min_{\sigma_{FE}}\,
\tfrac{1}{2}\norm{\rho_{f(X)FE}
  - \tau_{f(X)} \otimes \sigma_{FE}}_1
\;\le\; \eps + \tfrac{1}{2}\sqrt{2^{\,\ell - m}}\,,
\]
where the register~$F$ records the hash choice and the minimum ranges over states~$\sigma_{FE}$.
\end{theorem}
\noindent The LHL~\cite{tomamichel2011leftover} converts min-entropy into secrecy: any key shorter than $\Hmineps(X\given E)$ (up to logarithmic corrections) is indistinguishable from uniform up to the bound above, even against a quantum adversary who learns which hash was chosen.

\noindent\textit{Preparation quality.}\
The preparation quality $q\in[0,1]$ measures the overlap between the two bases,
as defined for general POVMs in \Cref{thm:eur}; for pure states it reduces to
$q=-\log_2\max_{x,z}|\langle\psi_x|\psi_z\rangle|^2$.
The mathematical theorem remains parametric in~$q$, whereas the executable endpoint specializes to the ideal-overlap projective $\mathbb X/\mathbb Z$ class of T6.

\begin{table}[t]
\centering
\caption{Notation for the BB84 Finite-Key Security Bound}
\label{tab:notation}
\begingroup
\renewcommand{\arraystretch}{1.00}
\setlength{\tabcolsep}{1pt}
\begin{tabular}{@{}>{\centering\arraybackslash}p{1.85cm}
                    >{\raggedright\arraybackslash}p{6.25cm}@{}}
\hline
\textbf{Symbol} & \textbf{Description} \\
\hline
$\Phi$ & protocol $\Phi(n,k,\ell,\Qtol,\eps_c,\leakec)$ \\
$n,k,N$ & raw-key (key-round) and test (PE-round) sizes; total sifted size $N{:=}n{+}k$ \\
$\ell,\Qtol$ & final-key length; tolerated bit-error rate \\
$\eps,\eps_c,\eps_s$ & smoothing-ball radius; correctness and secrecy parameters \\
$\leakec,q$ & error-correction leakage (bits); preparation quality ($q{=}1$ for ideal BB84) \\
\hline
$\mathbf{X},\mathbf{X}'$ & Alice's/Bob's $n$-bit raw-key strings ($X$-basis key rounds); boldface denotes strings \\
$\mathbf{Z},\mathbf{Z}'$ & conjugate-basis ($Z$-basis) outcomes on key rounds \\
$E;\ \mathbf{S},\hat{\mathbf{S}}$ & adversary's quantum side information; final keys ($\perp$ on abort) \\
$\rho_{ABE}$ & post-sifted Alice--Bob--Eve joint state \\
\hline
$\mu,h$ & Serfling sampling correction; truncated binary entropy \\
\hline
\end{tabular}
\endgroup\par
\end{table}

\noindent\textit{Main theorem.}\
Combining the EUR, the Serfling sampling bound, and the LHL yields the composable finite-key security bound for BB84; \Cref{tab:notation} collects the notation.

\begin{theorem}[BB84 Finite-Key Bound]\label{thm:bb84}
Let $\Phi(n,k,\ell,\Qtol,\eps_c,\leakec)$ be the BB84 protocol of \Cref{sec:qkd:bb84} with preparation quality~$q$.
Assume $n,k>0$, $0<\eps_c\le1$, $0<\eps_s\le1$, and $\leakec\ge0$.
\begin{enumerate}
\item[\textup{(i)}] $\Phi$ is $\eps_c$-correct.
\item[\textup{(ii)}] $\Phi$ is $\eps_s$-secret if
\begin{equation}\label{eq:keylen}
\ell \;\le\; n\bigl(q - h(\Qtol + \mu)\bigr)
       - \leakec
       - \log_2\!\tfrac{2}{\eps_s^{\,2}\,\eps_c}\,,
\end{equation}
where $\mu := \sqrt{(n{+}k)(k{+}1)/(nk^2)\cdot\ln(2/\eps_s)}$ is the Serfling correction and $h(x) = -x\log_2 x - (1{-}x)\log_2(1{-}x)$ is the binary entropy, truncated to $h(x) = 1$ for $x>1/2$ (clamping ensures more errors never yield more key).
\end{enumerate}
\end{theorem}
\noindent The bound~\cite{tomamichel2012tight} subtracts sampling-corrected noise $n\,h(\Qtol{+}\mu)$,
error-correction leakage $\leakec$, and privacy-amplification and verification costs from quantum uncertainty $nq$;
\eqref{eq:keylen} specializes the optimized bound in \Cref{thm:protocol-security} to $\eps=\eb=\eps_s/2$.

\noindent\textit{Proof sketch.}\
The EUR gives $\Hmineps(X\given E)+\Hmaxeps(Z\given B)\ge nq$.
Serfling's sampling bound yields $\Hmaxeps(Z\given Z')\le n\,h(\Qtol{+}\mu)$;
data processing gives $\Hmaxeps(Z\given B)\le\Hmaxeps(Z\given Z')$,
hence $\Hmineps(X\given E)\ge n(q-h(\Qtol{+}\mu))$.
The LHL turns this lower bound into secrecy, while the chain rule charges the $\leakec$ disclosed bits and $\lceil\log_2(1/\eps_c)\rceil$ verification bits.
Verification hashing proves correctness because differing corrected keys have matching tags with probability at most~$\eps_c$.
Apart from logarithmic security terms, the bound's only finite-size loss is $\mu\sim1/\sqrt{k}$,
which Tomamichel et~al. argue any finite-sample error-rate estimate must pay~\cite{tomamichel2012tight}.
Working directly with smooth entropies and sampling without replacement avoids asymptotic-equipartition and de~Finetti penalties;
\Cref{sec:protocol} gives the corresponding protocol-level obligations.

%% file: sections/protocol.tex
\section{Protocol Model}\label{sec:protocol}

\Cref{sec:qkd} described QKD as five stages (preparation and measurement, sifting, parameter estimation, error correction, and privacy amplification), and \Cref{sec:security} proved security bounds on the quantum state that passes through them.
But a security proof in isolation certifies only that a quantum state satisfies a secrecy bound; it does not specify which \emph{protocol} produces that state.
This section introduces a programming language that captures the behavior of each phase.
Precise semantics let us state and verify properties (correctness, secrecy, abort) as formal theorems.
The protocol layer connects what a QKD protocol \emph{does} (\Cref{sec:qkd}) to what a security proof \emph{certifies} (\Cref{sec:security}).
The resulting theorem quantifies over every initial state and every transcript-dependent adversary channel.

\subsection{Syntax}\label{sec:protocol:syntax}

We assume a finite quantum variable set $\mathit{qvar} = \{q, q_0, q_1, \ldots\}$ and a finite classical variable set $\mathit{cvar} = \{x, x_0, x_1, \ldots\}$.
Quantum registers $\overline{q}, \overline{q}_1,\ldots$ are lists of distinct quantum variables, so a command can act on a whole register block.
The core language is deliberately small: it separates classical computation, random sampling, quantum initialization, quantum measurement, and the adversarial quantum-channel action.
The syntax can also express the control flow of prepare-and-measure and
entanglement-based protocols.
\Cref{app:qkd-variants} gives paper-level pseudocode schemas inspired by
six-state, B92, BBM92, and E91.

\begin{definition}[Syntax]\label{def:protocol-cmd}
A finite quantum program is generated by the following syntax:
\begin{align}
c \;::=\;
  & \mathbf{skip} \mid x := e
    \mid x \xleftarrow{\$} \mu
    \mid \overline{q} \leftarrow \rho
    \mid x := M[\overline{q}]
    \label{eq:cmd-atomic} \\
    \mid & \mathsf{Attack}[\overline{q}]
    \mid c_1 \mathbin{;;} c_2
    \mid \mathbf{if}\ b\ \mathbf{then}\ c_1\ \mathbf{else}\ c_2.
    \label{eq:cmd-control}
\end{align}
\end{definition}
Atomic constructs include: $\mathbf{skip}$ is the empty command; $x:=e$ is the classical assignment, $x\xleftarrow{\$}\mu$ samples a classical value from the distribution~$\mu$, $\overline q\leftarrow\rho$ initializes the quantum register $\overline{q}$ to the state $\rho$, $x:=M[\overline q]$ performs quantum measurement $M$ on $\overline q$ and stores the outcome in~$x$, and $\mathsf{Attack}[\overline q]$ applies an arbitrary adversarial channel to $\overline q$ together with Eve's memory. Two structural constructs are $c_1 \mathbin{;;} c_2$ for sequential composition, and $\mathbf{if}\ b\ \mathbf{then}\ c_1\ \mathbf{else}\ c_2$ for conditionals. For convenience, we use the following syntactic sugar for fixed-round loops, which is sufficient for describing QKD protocols:
\begin{align*}
\textsc{for}\; i < n,\; c_i
  &\;\triangleq\; \mathbf{skip} \mathbin{;;}
  c_0 \mathbin{;;} c_1 \mathbin{;;} \cdots \mathbin{;;} c_{n-1}.
  \label{eq:for-loop}
\end{align*}
The leading $\mathbf{skip}$ makes the empty loop ($n=0$) a well-formed identity map rather than an empty sequence.
Protocol listings number rounds from one and write
$\textsc{for}\; i=1,\ldots,n,\; c_i$ for the relabeled unrolling
$c_1 \mathbin{;;} \cdots \mathbin{;;} c_n$.

\subsection{Semantics}\label{sec:protocol:semantics}

\noindent\textit{Classical-quantum state (cq-state).}
We follow \cite{feng2021qhlcv, barbosa2021easypqc}
and represent cq-states as functions from classical memories to quantum states.
Let $\mathrm{Val}$ be a finite set of values, and let $\sigma \in \Sigma\triangleq \mathit{cvar}\rightarrow \mathrm{Val}$ be a classical memory, and let $\rho\in\mathcal{D}(\mathcal{H})$ be a partial (i.e., sub-normalized) density operator on $\mathcal{H}\triangleq \bigotimes_{q\in\mathit{qvar}}\mathcal{H}_q$.
A normalized cq-state is a function $\xi:\Sigma\rightarrow\mathcal{D}(\mathcal{H})$ with $\sum_\sigma\Tr(\xi(\sigma))=1$; we write $\Xi$ for the set of these states.
We write $(\sigma,\rho)$ for the possibly sub-normalized singleton cq-state concentrated at~$\sigma$, i.e., $(\sigma,\rho)(\sigma) = \rho$ and $(\sigma,\rho)(\sigma') = 0$ if $\sigma\neq\sigma'$. Any cq-state can be decomposed as a sum of singleton states: $\xi = \sum_\sigma(\sigma,\xi(\sigma))$.

\noindent\textit{Semantics of the adversary.}
We assume a designated classical variable $P$ that collects the public transcript, and a sufficiently large quantum system $q_{A}$ that models both the environment and the adversary's local quantum memory.
Adversary action is modeled by a function $\mathcal{E}_A: \mathrm{Val} \rightarrow \mathcal{QC}$, where $\mathcal{QC}$ is the set of quantum channels, that is, completely positive trace-preserving maps, on $\mathcal{H}_{\overline{q}\cup q_A}$ for the attacked registers~$\overline{q}$.
For every $\mathsf{Attack}[\overline{q}]$ command, the adversary applies $\mathcal{E}_A(P)$ to $\overline{q}\cup q_A$: the adversary may choose its action depending on the public transcript~$P$, but the action cannot depend on any other classical variable or act on any quantum variable outside $\overline{q}$ and $q_A$.
The adversary reads but cannot rewrite~$P$, which is T2 in operational form.
The security of the protocol is then established over all adversary actions $\mathcal{E}_A$.

\begin{definition}[Denotational Semantics]\label{def:protocol-semantics}
Given an adversary action $\mathcal{E}_A : \mathrm{Val} \rightarrow \mathcal{QC}$,
the denotational semantics of a well-formed finite program is $\sem{c}:\Xi\to\Xi$.
The semantics extends linearly to sub-normalized cq-states.
Its auxiliary action on a singleton state is
$\auxsem{c}:\Sigma\times\mathcal{D}(\mathcal{H})\to(\Sigma\to\mathcal{D}(\mathcal{H}))$
and returns a sub-normalized cq-state of total mass $\Tr\rho$:
\begin{itemize}
    \item $\auxsem{\mathbf{skip}}(\sigma,\rho)=(\sigma,\rho)$;
    \item $\auxsem{x := e}(\sigma,\rho)= (\sigma[x\mapsto \sem{e}_\sigma], \rho)$;
    \item $\auxsem{x \xleftarrow{\$} \mu}(\sigma,\rho)=\sum_v(\sigma[x\mapsto v],\sem{\mu}_\sigma(v)\cdot\rho)$;
    \item $\auxsem{\overline{q} \leftarrow \rho_{\mathrm{init}}}(\sigma,\rho) = (\sigma, \Tr_{\overline{q}}(\rho)\otimes \sem{\rho_{\mathrm{init}}}_\sigma[\overline{q}])$;
    \item $\auxsem{x := M[\overline{q}]}(\sigma,\rho) = \sum_v(\sigma[x\mapsto v], M_v[\overline{q}]\cdot\rho\cdot M_v^\dagger[\overline{q}])$ where $\{M_v\}_v := \sem{M}_\sigma$ is a quantum measurement, i.e., a family of measurement operators with $\sum_v M_v^\dagger M_v = \id$;
    \item $\auxsem{\mathsf{Attack}[\overline{q}]}(\sigma,\rho) = (\sigma,\mathcal{E}_A(\sem{P}_\sigma)[\overline{q},q_A](\rho))$;
    \item $\auxsem{\mathbf{if}\ b\ \mathbf{then}\ c_1\ \mathbf{else}\ c_2}(\sigma,\rho) = \left\{\begin{array}{ll}
        \auxsem{c_1}(\sigma,\rho) & \text{if}\ \sem{b}_\sigma = 1 \\
        \auxsem{c_2}(\sigma,\rho) & \text{if}\ \sem{b}_\sigma = 0
    \end{array}\right.$;
    \item $\sem{c}(\xi) = \sum_\sigma \auxsem{c}(\sigma,\xi(\sigma))$;
    \item $\auxsem{c_1 \mathbin{;;} c_2}(\sigma,\rho) = \sem{c_2}(\auxsem{c_1}(\sigma,\rho))$, so that $\sem{c_1 \mathbin{;;} c_2} =\sem{c_2}\circ\sem{c_1}$.
\end{itemize}
Distribution, state, and measurement arguments are expressions, evaluated in the current classical memory~$\sigma$.
Classical value expressions are functions $\Sigma\to\mathrm{Val}$, and guards are functions $\Sigma\to\{0,1\}$.
A program is well-formed when, in every classical memory, each sampling distribution $\sem{\mu}_\sigma$ is a probability distribution on $\mathrm{Val}$, each initial state $\sem{\rho_{\mathrm{init}}}_\sigma$ is a density operator of trace one, each measurement satisfies $\sum_v M_v^\dagger M_v = \id$, and every initialized, measured, or attacked register set is disjoint from $q_A$; for well-formed programs, $\sem{c}$ maps $\Xi$ into $\Xi$.
For every well-formed command and positive cq-state, the semantics preserves
total mass, so $\sem{c}$ maps well-formed cq-states to well-formed cq-states.
\end{definition}

\subsection{QKD and BB84}\label{sec:protocol:qkd-bb84}

We now specialize the QKD phases of \Cref{sec:qkd} to the \cite{tomamichel2012tight} BB84
protocol and write each phase as syntactic sugar using the primitive language in Definition~\ref{def:protocol-cmd}.
Let
$p_{\mathbb X}:=(1+\sqrt{k/n})^{-1}$,
$p_{\mathbb Z}:=1-p_{\mathbb X}$, let $\mu_{\mathbb B}$ be the
distribution on $\{\mathbb X,\mathbb Z\}$ with these probabilities, and let
$\mu_2$ be uniform on $\{0,1\}$.
$\mathsf{Unif}(S)$ denotes the uniform distribution on a nonempty finite
set~$S$, so $\mu_2 = \mathsf{Unif}(\{0,1\})$.
Write $\mathcal K_n(\mathcal X):=\{K\subseteq\mathcal X:|K|=n\}$.
The protocol's symbol $M$ is the realized stopping counter, defined as the
least round after which $|\mathcal X|\ge n$ and
$|\mathcal Z|\ge k$.
The formal protocol language represents iteration by finite unrolling, so in the
program below $M$ denotes a fixed round budget (an explicit model parameter);
a run whose budget ends with $|\mathcal X|<n$ or $|\mathcal Z|<k$ aborts at
the parameter-estimation test.
Because the budget is fixed, the loop needs no per-round announcements
to decide termination, so the program publishes all sifting announcements in one
$\mathsf{Announce}$ step after the loop; announcing bases while quantum
rounds are in flight is also known to invalidate the random-sampling
step of finite-key proofs~\cite{pfister2016sifting}.

\begin{definition}[QKD component sugar]\label{def:qkd-component-sugar}
For round~$i$, write $d_i := (y'_i\ne\emptyset)$ for Bob's detection
flag.
We abbreviate parameter estimation, error correction, and privacy
amplification by $\mathsf{ParamEst}$, $\mathsf{ErrCorr}$, and
$\mathsf{PrivAmp}_\ell$.
The per-round components expand as follows:
\begin{align*}
\mathsf{Abort}
  \;\triangleq\;&
  \mathsf{abort}:=\mathsf{true}
  \mathbin{;;} S:=\perp
  \mathbin{;;} \hat S:=\perp,\\[2pt]
\mathsf{Prep}_i(\rho)
  \;\triangleq\;&
  a_i\xleftarrow{\$}\mu_{\mathbb B}
  \mathbin{;;} y_i\xleftarrow{\$}\mu_2
  \mathbin{;;} q_i\leftarrow\rho_{a_i,y_i},\\[2pt]
\mathsf{Dist}_i
  \;\triangleq\;&
  \mathsf{Attack}[q_i],\\[2pt]
\mathsf{Meas}_i(\mathcal M)
  \;\triangleq\;&
  b_i\xleftarrow{\$}\mu_{\mathbb B}
  \mathbin{;;} y'_i:=\mathcal M_{b_i}[q_i],\\[2pt]
\mathsf{Sift}_i
  \;\triangleq\;&
  d_i := (y'_i\ne\emptyset)\mathbin{;;}\\
  &
  \mathcal X:=\{j\le i \mid a_j=b_j=\mathbb X,\ d_j\}\mathbin{;;}\\
  &
  \mathcal Z:=\{j\le i \mid a_j=b_j=\mathbb Z,\ d_j\},\\[2pt]
\mathsf{Round}_i(\rho,\mathcal M)
  \;\triangleq\;&
  \mathsf{Prep}_i(\rho)
  \mathbin{;;} \mathsf{Dist}_i
  \mathbin{;;} \mathsf{Meas}_i(\mathcal M)
  \mathbin{;;} \mathsf{Sift}_i.
\end{align*}
After the loop terminates, the announcement and classical post-processing
components are
\begin{align*}
\mathsf{Announce}
  \;\triangleq\;&
  P:=P\mathbin{\|}\bigl((a_i,b_i,d_i)\bigr)_{i=1}^{M},\\[2pt]
\mathsf{ParamEst}
  \;\triangleq\;&
  K\xleftarrow{\$}
  \begin{cases}
    \mathsf{Unif}(\mathcal K_n(\mathcal X)) & \text{if } |\mathcal X|\ge n\\
    \delta_\perp & \text{otherwise}
  \end{cases}
  \mathbin{;;}\\
  &
  X_A:=(y_j)_{j\in K}
  \mathbin{;;} X_B:=(y'_j)_{j\in K}\mathbin{;;}\\
  &
  Z_A:=(y_j)_{j\in\mathcal Z}
  \mathbin{;;} Z_B:=(y'_j)_{j\in\mathcal Z}\mathbin{;;}\\
  &
  P:=P\mathbin{\|}(K,Z_A,Z_B)
  \mathbin{;;} \\
  & \lambda:=\bigl(\textstyle\sum_{j\in\mathcal Z}(y_j\oplus y'_j),\,
       |\mathcal Z|\bigr),\\[2pt]
\mathsf{ErrCorr}
  \;\triangleq\;&
  (p_{\rm EC},\widehat X):=\mathsf{IR}_{\leakec}(X_A,X_B)\mathbin{;;}\\
  &
  g\xleftarrow{\$}\mathsf{Unif}(\mathcal G_{\eps_c})
  \mathbin{;;} t_A:=g(X_A)
  \mathbin{;;} t_B:=g(\widehat X)\mathbin{;;}\\
  &
  P:=P\mathbin{\|}(p_{\rm EC},g,t_A),\\[2pt]
\mathsf{PrivAmp}_\ell
  \;\triangleq\;&
  f\xleftarrow{\$}\mathsf{Unif}(\mathcal F_\ell)\mathbin{;;}\\
  &S:=f(X_A)
  \mathbin{;;} \hat S:=f(\widehat X)\mathbin{;;}\\
  &P:=P\mathbin{\|}f .
\end{align*}
\end{definition}
Here $\rho_{a_i,y_i}$ is the BB84 source expression, and
$\mathcal M_{b_i}$ is Bob's measurement expression.
In $\mathsf{ParamEst}$, $\delta_\perp$ is the point distribution at
$\perp$: when the budget leaves $|\mathcal X|<n$, the raw-key sample
defaults to $K=\perp$, making the sampling command total on every branch.
The estimate $\lambda$ is carried as the numerator--denominator pair
shown, and the test $\lambda>\Qtol$ in \Cref{fig:bb84-program} denotes
the cross-multiplied comparison
$\sum_{j\in\mathcal Z}(y_j\oplus y'_j)>\Qtol\,|\mathcal Z|$, well
defined even when $\mathcal Z=\emptyset$.
The expression $\mathsf{IR}_{\leakec}$ denotes the fixed
information-reconciliation scheme, chosen before the protocol run.
Applied to the realized raw-key pair, it produces a public transcript
$p_{\rm EC}$ of length at most~$\leakec$ and Bob's estimate~$\widehat X$.
The family $\mathcal G_{\eps_c}$ is the two-universal verification
hash family with output length $\lceil\log_2(1/\eps_c)\rceil$, and
$\mathcal F_\ell$ is the two-universal privacy-amplification family.
The transcript records the authenticated public data used by this
abstraction: basis and detection announcements needed for sifting, the
raw-key subset and test bits used for parameter estimation,
error-correction data, the verification hash seed and Alice's tag, and
the privacy-amplification seed.
The bit values in the raw-key subset are never appended to~$P$.

\begin{figure}[H]
\centering
\begingroup
\setlength{\tabcolsep}{0pt}
\renewcommand{\arraystretch}{1.05}
\begin{tabular}{@{}r@{\hspace{0.45em}}l@{}}
\multicolumn{2}{@{}l@{}}{$\theta=(n,k,\ell,\Qtol,\eps_c,\leakec)$}\\
\multicolumn{2}{@{}l@{}}{$\mathsf{BB84}(\theta) \;:= $}\\
1 & $P:=\emptyset
     \mathbin{;;}\ \mathcal X:=\emptyset
     \mathbin{;;}\ \mathcal Z:=\emptyset
     \mathbin{;;}\ \mathsf{abort}:=\mathsf{false}$\\
2 & $\textsc{for}\ i=1,\ldots,M,$\\
3 & $\quad \mathsf{Round}_i(\rho,\mathcal M)$\\
4 & $\mathsf{Announce}$\\
5 & $\mathsf{ParamEst}$\\
6 & $\mathbf{if}\ |\mathcal X|<n \vee |\mathcal Z|<k \vee \lambda>\Qtol$\\
  & $\quad \mathbf{then}\ \mathsf{Abort}\ \mathbf{else}$\\
7 & $\quad \mathsf{ErrCorr}$\\
8 & $\quad \mathbf{if}\ t_A\ne t_B\ \mathbf{then}\ \mathsf{Abort}\ \mathbf{else}$\\
9 & $\quad\quad \mathsf{PrivAmp}_\ell$
\end{tabular}
\endgroup
\caption{BB84 written with the QKD component sugar of
Definition~\ref{def:qkd-component-sugar}.
The loop is the finite unrolling through the round budget~$M$.
The abort test on line~6 subsumes the sifting condition of~\cite{tomamichel2012tight}:
a run whose budget ends with $|\mathcal X|<n$ or $|\mathcal Z|<k$ aborts,
so post-processing continues only when that condition holds.
The raw key comes from a uniformly chosen size-$n$ subset of
$\mathcal X$, while the parameter-estimation error rate~$\lambda$ is
computed from all detected matching $\mathbb Z$-basis rounds in
$\mathcal Z$, as in the original protocol definition.}
\label{fig:bb84-program}
\end{figure}

\subsection{Security Theorem}\label{sec:protocol:security}

\noindent
Let $\mathsf{Src}(a,y)=\rho_{a,y}$ be the program source expression, and
let $\mathcal M$ be Bob's measurement expression.
We write $\mathsf{BB84}_{\mathsf{Src},\mathcal M}(n,k,\ell,\Qtol,\eps_c,\leakec)$
for the program of \Cref{fig:bb84-program} with these expressions, and
subscript the semantics of Definition~\ref{def:protocol-semantics} with the
adversary strategy it is taken under.
For any initial cq-state $\xi_{\mathrm{init}}\in\Xi$ and adversary
strategy~$\mathcal E_A$, let
\[
  \xi_{\mathrm{out}}
  :=
  \sem{\mathsf{BB84}_{\mathsf{Src},\mathcal M}
    (n,k,\ell,\Qtol,\eps_c,\leakec)}_{\mathcal E_A}
  (\xi_{\mathrm{init}})
\]
be the terminal state of \Cref{fig:bb84-program}, and
$p_{\mathrm{abort}}$ its abort probability.
When $p_{\mathrm{abort}}<1$, let
$\rho_{S F_{\mathrm{PA}} E^+}$ be the normalized
verification-accepted output state.
Here $F_{\mathrm{PA}}$ is the public privacy-amplification seed, shown
separately, and $E^+$ comprises Eve's final quantum memory and the remaining
public transcript.
The abort probability is $p_{\mathrm{abort}}=\sum_{\sigma:\,\sigma(\mathsf{abort})=\mathsf{true}}\Tr\,\xi_{\mathrm{out}}(\sigma)$.
The state $\rho_{S F_{\mathrm{PA}} E^+}$ retains the key, seed, Eve's quantum memory, and public transcript, conditioned on $\mathsf{abort}=\mathsf{false}$.

\begin{theorem}[Protocol Security]\label{thm:protocol-security}
Assume $\mathsf{Src}$ and~$\mathcal M$ are certified with preparation
quality~$q$.
The executable endpoint is restricted to the calibrated projective
$\mathbb X/\mathbb Z$ device class T6 and assumes basis-independent
detection T5.
Within this class, the artifact proves the entropy comparison corresponding
to T4.
Assume $n,k,M>0$, $0<\eps_c\le1$, $0<\eps_s\le1$, $0\le\leakec$, and
$\eps_c+\eps_s\le\eps_{\mathrm{tot}}$.
Assume error correction leaks at most~$\leakec$ bits, verification is
two-universal with tag length $\lceil\log_2(1/\eps_c)\rceil$, and privacy
amplification is two-universal.
For the constants of Tomamichel et al.~\cite{tomamichel2012tight}, if
\begin{multline*}
\ell \;\le\;
  \sup_{\substack{0<\eps<1,\ \eb>0 \\[1pt] \eps+\eb\,\le\,\eps_s}}
  \bigl\lfloor
    nq - n\,h(\Qtol{+}\mu(\eps)) \\
    {}- 2\log_2\tfrac{1}{2\eb}
    - \leakec - \log_2\tfrac{2}{\eps_c}
  \bigr\rfloor,
\end{multline*}
where
\[
  \mu(\eps)
  :=
  \sqrt{\frac{n+k}{nk}\cdot\frac{k+1}{k}\cdot\ln\frac{1}{\eps}},
\]
and $h$ is the truncated binary entropy,
then for every initial cq-state $\xi_{\mathrm{init}}\in\Xi$ and every adversary
strategy~$\mathcal E_A$:
\begin{enumerate}
\item the output is either on the abort branch
  ($\mathsf{abort}=\mathsf{true}$ and $S=\hat S=\perp$) or on the
  amplified-key branch ($\mathsf{abort}=\mathsf{false}$ and
  $S,\hat S\in\{0,1\}^{\ell}$);
\item the protocol is $\eps_c$-correct,
  \[
    \Pr_{\xi_{\mathrm{out}}}
      [S\ne\hat S \wedge \mathsf{abort}=\mathsf{false}]
    \le \eps_c;
  \]
\item if $p_{\mathrm{abort}}<1$, the normalized verification-accepted output is
  $\eps_s$-secret:
  \[
    (1-p_{\mathrm{abort}})
    \min_{\omega_{F_{\mathrm{PA}} E^+}}
    \tfrac12
    \norm{\rho_{S F_{\mathrm{PA}} E^+}
      - \tau_S\otimes\omega_{F_{\mathrm{PA}} E^+}}_1
    \le \eps_s,
  \]
  If $p_{\mathrm{abort}}=1$, this secrecy clause holds because its
  $(1-p_{\mathrm{abort}})$ prefactor vanishes.
\end{enumerate}
Consequently
$\mathsf{BB84}_{\mathsf{Src},\mathcal M}
(n,k,\ell,\Qtol,\eps_c,\leakec)$ is
$\eps_{\mathrm{tot}}$-secure in the sense of
Definitions~\ref{def:correct} and~\ref{def:secret}.
\end{theorem}

\noindent
Sifting, parameter estimation, and verification are program branches, not
external theorem premises: a run that ends under-sifted
($|\mathcal X|<n$ or $|\mathcal Z|<k$), fails the tolerance test, or
fails verification produces $S=\hat S=\perp$.

\begin{theorem}[BB84 Instance]\label{thm:bb84-instance}
For the concrete calibrated projective $\mathbb X/\mathbb Z$ instance with the
ideal Hadamard source ($q{=}1$), round budget $M=100(n{+}k)$, $n=k=10^4$,
$\Qtol=1\%$, $\leakec=1024$, $\eps=\eb=\eps_c=2^{-34}$, and
$\eps_s=2^{-33}$, this BB84 instance is $\eps_\mathrm{tot}$-secure with
$\eps_\mathrm{tot}=3\cdot 2^{-34}$ and key length $\ell=2000$.
\end{theorem}

\noindent
In this instance, the key length is positive, all numeric side
conditions are concrete, and an identity-adversary run has positive success
probability in either source mode.
At these parameters, the artifact proves that the right-hand side of the
key-length condition in \Cref{thm:protocol-security} is at least $2625$,
which certifies the chosen key length $\ell=2000$.

%% file: sections/architecture.tex
\section{Proof Architecture}\label{sec:architecture}\label{sec:arch}

The security analysis of \Cref{sec:security} decomposes the finite-key bound into three independent proof ingredients (the smooth entropic uncertainty relation, the Serfling sampling bound, and the Quantum Leftover Hash Lemma), composed by a single algebraic chain.
The module structure of \FVQKD{} formalizes this decomposition: each proof ingredient becomes a self-contained, verified component, and the top-level module composes them into the security theorems.

\subsection{Overview}\label{sec:architecture:overview}

Five modules sit in dependency order (\Cref{fig:deps}; size counts in \Cref{tab:loc}):
\[
  \code{utils}
  \;<\;
  \code{entropy}
  \;<\;
  \{\code{lhl},\;\code{eur}\}
  \;<\;
  \code{qkd}.
\]
\code{utils} supplies quantum primitives missing from the foundation libraries: matrix and operator-norm extensions, the set-labeled \CoqQ{} layer, bit-string and probability infrastructure, purification and Stinespring dilation~\cite{steinspring1955positive}, fidelity, purified distance, and the Serfling sampling bound.
\code{entropy} formalizes the smooth-entropy calculus of~\cite{tomamichel2013frameworknonasymptoticquantuminformation} on top of these primitives, including the smooth conditional entropies $\Hmineps$ and $\Hmaxeps$ and their duality; \code{eur} formally proves an uncertainty relation linking $\Hmineps$ and $\Hmaxeps$, while \code{lhl} bounds the key's distance from uniform in terms of $\Hmineps(X\given E)$.
\code{entropy} also contains a self-contained proof of SDP strong duality (the primal direction of Tomamichel's Theorem~2.6: value equality plus primal-optimizer attainment), built from first principles and used in the LHL proof.
To our knowledge, it is the first explicit machine-checked strong-duality theorem specialized to semidefinite programs.
\code{eur} and \code{lhl} formalize the two main theorems of the proof sketch: the smooth entropic uncertainty relation~\cite{tomamichel2011uncertainty} and the Quantum Leftover Hash Lemma~\cite{tomamichel2011leftover}, respectively.
\code{qkd} composes them with the Serfling bound from \code{utils} to close the algebraic chain into the main theorems for correctness, secrecy, and combined $\eps$-security.

The first four modules are protocol-independent:
no definition or theorem in them assumes BB84 or any specific protocol.
The four lower modules are reusable by other QKD variants as well as other applications in quantum information and cryptography.

\begin{figure}[t]
\centering
\begin{tikzpicture}[
  every node/.style={font=\sffamily\normalsize},
  module/.style={draw, rounded corners=2pt,
                 minimum width=2.15cm, minimum height=0.58cm,
                 fill=white, draw=black,
                 inner xsep=5pt, inner ysep=2pt},
  ext/.style={draw, dashed, rounded corners=2pt,
              minimum width=2cm, minimum height=0.58cm,
              fill=white,
              font=\sffamily\normalsize\itshape,
              inner xsep=5pt, inner ysep=2pt},
  arr/.style={-{Stealth[length=1.7mm,width=1.7mm]}, semithick,
              draw=black},
]
\node[ext]                          (mc)    at (-2.4,0)   {MathComp};
\node[ext]                          (an)    at ( 0  ,0)   {Analysis};
\node[ext]                          (cq)    at ( 2.4,0)   {CoqQ};
\begin{scope}[on background layer]
\node[draw=black!55, dashed, rounded corners=4pt,
      fill=white, inner sep=5pt,
      fit=(mc)(an)(cq),
      label={[font=\sffamily\normalsize,text=black!70]below:Foundation libraries}] {};
\end{scope}
\node[module]                       (utils) at ( 0  ,0.95){\texttt{utils}};
\node[module]                       (ent)   at ( 0  ,1.90){\texttt{entropy}};
\node[module] (lhl) at (-1.6,2.85){\texttt{lhl}};
\node[module] (eur) at ( 1.6,2.85){\texttt{eur}};
\node[module] (qkd-protocol) at (-1.6,3.95) {\texttt{protocol}};
\node[module] (qkd-bb84) at ( 1.6,3.95) {\texttt{bb84}};
\node[module] (qkd-example) at (0,4.90) {\texttt{example}};
\begin{scope}[on background layer]
\node[draw=black!55, rounded corners=4pt,
      fill=white, inner sep=5pt,
      fit=(qkd-protocol)(qkd-bb84)(qkd-example),
      label={[font=\sffamily\normalsize,text=black]above:\texttt{qkd}}] {};
\end{scope}

\draw[arr] (mc.north)  -- ($(utils.south) + (-0.45,0)$);
\draw[arr] (an.north)  -- (utils.south);
\draw[arr] (cq.north)  -- ($(utils.south) + ( 0.45,0)$);
\draw[arr] (utils.north) -- (ent.south);
\draw[arr] ($(ent.north) + (-0.45,0)$) -- (lhl.south);
\draw[arr] ($(ent.north) + ( 0.45,0)$) -- (eur.south);
\coordinate (qkd-merge) at ($(qkd-bb84.south) + (0,-0.32)$);
\draw[semithick, draw=black] (lhl.north) |- (qkd-merge);
\draw[arr] (eur.north) -- (qkd-bb84.south);
\draw[arr] (qkd-protocol.north) -- ($(qkd-example.south) + (-0.45,0)$);
\draw[arr] (qkd-bb84.north) -- ($(qkd-example.south) + ( 0.45,0)$);
\end{tikzpicture}
\caption{Module dependency graph.
The structure mirrors the proof sketch of \Cref{sec:security:bound}: \code{eur} and \code{lhl} are mutually independent; inside \code{qkd}, \code{bb84} depends on the lower proof modules, and \code{example} depends on both the protocol language in \code{protocol} and the BB84 theorem layer in \code{bb84}.}
\label{fig:deps}
\end{figure}

\emph{The protocol-specific top layer.}\label{sec:arch:bb84-split}
The \code{qkd} module has three submodules, two independent and one that ties them together.
\code{bb84/} formalizes the main theorem of Tomamichel et~al.~\cite{tomamichel2012tight} by combining the results from the four lower modules into the security theorems (\Cref{fig:proof-structure} in \Cref{app:bb84-sig}).
\code{protocol/} defines the language for QKD protocols.
\code{example/} instantiates the protocol language for BB84 and formally proves its security using the theorems in \code{bb84/}, closing the chain from protocol definition to machine-checked security proof.
It further instantiates the protocol with concrete, physically meaningful parameters, demonstrating that the security theorem yields a positive key length at realistic finite block sizes.

\subsection{Design Features}\label{sec:architecture:features}

Two features matter for the rest of the development.
First, \emph{modular reusable theorem interfaces}\label{sec:library} keep the BB84 layer from absorbing the protocol-independent mathematics below it.
Each lower module exports a small set of statement-matched theorems for smooth min-/max-entropy and purified-distance reasoning, the smooth EUR, the Quantum LHL, and the Serfling finite-sampling bound, with side conditions stated in the formal signatures rather than left to context.
The \code{qkd} layer imports these results as sealed theorems; its composition proof uses them as established facts rather than unfolding their derivations.
Because the lower layers quantify over abstract states, systems, measurements, channels, hash families, and sampling distributions, they contain no BB84-specific theorem content.
They form verified infrastructure that QKD variants following the same smooth-entropy proof route can instantiate rather than re-derive.

Decoy-state BB84~\cite{lo2005decoy}, for instance, is widely used in deployed weak-coherent-pulse systems.
Wiesemann et al.~\cite{Wiesemann2026consolidated} give a finite-size one-/two-decoy proof in the smooth-EUR framework of Tomamichel and Renner, so \FVQKD{}'s smooth-entropy, purified-distance, EUR, and LHL interfaces should carry over.
What is not present is the decoy layer itself: weak-coherent-source and intensity semantics, photon-number and detection statistics, multi-intensity parameter estimation, and the resulting key-length calculation.
We expect closely related variants (six-state~\cite{bruss1998sixstate}, qudit BB84) to reuse the entropy chain; SARG04~\cite{scarani2004nsa} would additionally require modified sifting statistics, and MDI-QKD~\cite{lo2012mdi,curty2014finite} admits a smooth-EUR-based proof but brings additional structural requirements.

Second, \emph{certification-ready protocol extraction} turns the informal protocol specification in the published analysis into a checked formal specification.
Finite-key QKD analyses often specify the protocol at the level needed for the proof, not at the level needed to certify an implementation.
Notation can blur distinct roles: for instance, the symbol~$M$ in the BB84 paper~\cite{tomamichel2012tight} reads like a protocol parameter but is actually a loop counter determined at run time.
\FVQKD{} instead gives BB84 a command language with separate classical and quantum variables and the semantics of \Cref{sec:protocol:semantics}, specifying each stage's inputs, outputs, and continuation checks.
The protocol-security theorem takes a certified device and model assumptions that expose the source model, measurement family, preparation quality, leakage bound, and verification and privacy-amplification hash requirements.
Certifying that a device satisfies these hypotheses suffices to instantiate the composable security guarantee for that device (\Cref{thm:protocol-security}).

%% file: sections/engineering.tex
\section{Proof Engineering}\label{sec:engineering}

The previous section described what the mechanization proves; this section turns to how it is built.

\subsection{Backend and Trusted Computing Base}
\label{sec:engineering:backend}

\FVQKD{} is built on three \Rocq{} libraries that together supply the mathematical and quantum-mechanical substrate.
\MathComp{}~\cite{mahboubi2022mathcomp} provides the algebraic backbone: finite types, big operators, rings, fields, matrices, and the SSReflect proof style.
\MathCompAnalysis{}~\cite{affeldt2020competing} extends it with the analytic machinery the entropy arguments require: real and extended-real numbers, suprema and infima, topology, limits, continuity, exponentials, logarithms, and convexity.
\CoqQ{}~\cite{zhou2023coqq} contributes the finite-dimensional quantum layer: Hilbert spaces, tensor products, quantum operators, density and positivity predicates, channels, traces, norms, spectral decompositions, the partial trace, and labeled Dirac notation for multi-system management.
All three libraries use canonical structures for notation overloading and hierarchy navigation, a pattern \FVQKD{} inherits throughout.

\FVQKD{} develops the smooth-entropy, purified-distance, smooth-EUR, and
Quantum-LHL infrastructure needed by the finite-key BB84 proof.
\FVQKD{} builds this entire layer from scratch: \InfrastructureLoC{} lines of new, fully verified code.

\emph{Trusted computing base.}
Our TCB comprises solely the \Rocq{} kernel and three standard \MathCompAnalysis{} axioms (propositional extensionality, functional extensionality, and constructive indefinite description).
Continuous integration rebuilds the \TotalLoCk{}K-line codebase and mechanically verifies via \texttt{Print Assumptions} that no additional axioms or admitted obligations enter the development.

\subsection{Representation Choices}\label{sec:engineering:repr}

The mechanization operates at three representation levels: bare matrices (full computational access), typed linear maps (operator algebra with domain/codomain tracking), and label-indexed operators (subsystem selection by name via \CoqQ{}'s Dirac notation).
Each successive level is closer to textbook notation but further from direct computation.
Coercion lemmas, proved once in \code{utils} and reused across the development, translate between levels so that each proof step can work at the most convenient abstraction.

\subsection{Artifact Summary}\label{sec:engineering:scale}

The development builds on \Rocq{}~9.1.1, \MathComp{}~2.5.0, \MathCompAnalysis{}~1.16.0, and \CoqQ{}; \CoqQ{} is pinned to the in-tree vendored snapshot.
\Cref{tab:loc} gives per-module file, source-line, and result counts.

\begin{table}[t]
\centering
\caption{Development Size and Trusted Base}
\label{tab:loc}
\begingroup
\footnotesize
\setlength{\tabcolsep}{8pt}
\renewcommand{\arraystretch}{1.1}
\begin{tabular}{@{}l r r r@{}}
\hline
\textbf{Module} & \textbf{Files} & \textbf{LoC} & \textbf{Results}\\
\hline
\code{utils}   & 21 & 24\,719 &   542 \\
\code{entropy} & 17 & 50\,427 &   559 \\
\code{eur}     &  5 &  8\,461 &    93 \\
\code{lhl}     &  8 & 18\,695 &   214 \\
\hline
\code{qkd}     & 16 & 48\,448 & 1\,085 \\
\hline
\textbf{Total} & \textbf{\TotalFiles{}} & \textbf{150\,750} & \textbf{2\,493} \\
\hline
\end{tabular}
\endgroup\par

\vspace{0.35\baselineskip}
\raggedright
LoC counts physical source lines; Results counts non-local top-level
\coqkw{Theorem}, \coqkw{Lemma}, and \coqkw{Corollary} declarations.
The four upper rows (approximately 68\% of the development by LoC) are
protocol-independent; the \code{qkd} row is the protocol layer (language,
BB84 security theorems, and examples).
No \coqkw{Admitted} occurs anywhere in the tree; the trusted base is the
\Rocq{} kernel plus three \MathCompAnalysis{} axioms
(\Cref{sec:engineering:backend}).
\par
\end{table}

The artifact builds from source with \texttt{make};
\texttt{make full-check} additionally runs a \texttt{Print Assumptions} check
confirming that the security theorems depend on exactly the three axioms
listed above.
The artifact is available at \url{https://github.com/fvqkd/FVQKD}.

%% file: sections/findings.tex
\section{Verification Findings}\label{sec:findings}

Mechanization forces every side condition to be stated and discharged; published proofs sometimes leave conditions implicit when the intended reading is clear to domain experts.
This section catalogues the conditions the formalization exposed, organized into two categories: missing proof branches (\Cref{sec:findings:eur}) and competing conventions across the literature (\Cref{sec:findings:conventions}).

\subsection{Missing Proof Branches}\label{sec:findings:eur}

The smooth entropies $\Hmineps$ and $\Hmaxeps$ are defined by optimizing over quantum states within purified distance~$\eps$ of the input state~$\rho$.
This optimization is well-defined only when $\eps < \sqrt{\Tr\rho}$~\cite{tomamichel2013frameworknonasymptoticquantuminformation}; for a normalized state the condition reduces to $\eps < 1$.

The BB84 security proof applies the entropic uncertainty relation (EUR) after conditioning on the parameter-estimation test passing.
Conditioning on the pass event renormalizes the state and rescales the smoothing parameter to $\eps' := \eps / \sqrt{p_{\mathrm{pass}}}$.
The published analysis~\cite{tomamichel2012tight} states ``$\eps' \ge 0$'' without checking the upper bound $\eps' < 1$.
An active adversary that forces frequent aborts can push $p_{\mathrm{pass}}$ to $10^{-12}$ or below.
The rescaled parameter then exceeds unity: with $p_{\mathrm{pass}} = 10^{-12}$ and $\eps = 10^{-5}$, one obtains $\eps' = 10^{-5} / 10^{-6} = 10 \gg 1$, and the EUR is applied outside its domain.

This does not compromise the security guarantee: in the offending regime the protocol aborts with overwhelming probability, so the $(1-p_\mathrm{abort})$ prefactor of \Cref{def:secret} already bounds the secrecy term.
The issue is that the published argument invokes a theorem in a regime where its hypotheses are not satisfied.
The formalization resolves this by introducing an explicit domain predicate and splitting the secrecy proof into two branches: one where the EUR applies and one where the high abort probability makes the bound trivial.
This case split is absent from the pen-and-paper argument; formal verification catches it because every side condition must be discharged before the proof is accepted.

Separately, the published proof~\cite{tomamichel2012tight} derives its entropy
bound after parameter estimation passes, whereas the final key is conditioned
on error verification also accepting.
\FVQKD{} makes this second conditioning explicit: it transfers the entropy bound to the accepted branch, or, when acceptance is too unlikely, closes the secrecy argument directly from the protocol's abort probability.
Both cases retain the published key-length formula.

\subsection{Competing Conventions}\label{sec:findings:conventions}\label{app:definitions}

The BB84 security proof chains theorems from six theoretical treatments, each of which makes its own definitional choices for smooth entropies.
\Cref{tab:conventions} compares these choices along two axes: the smoothing-ball metric and the witness domain for conditional entropies.

\input{sections/app-definitions}

\noindent\textit{Smoothing-ball metric.}
Five of the six treatments define smoothing balls using purified distance; Renner~\cite{renner2008security} instead uses a relative trace-distance ball on sub-normalized operators.
Because the balls differ, using a result from one convention in the other requires a separate comparison argument.
The formalization follows Tomamichel's purified-distance convention, matching
the BB84 paper~\cite{tomamichel2012tight}.
The exact smooth min/max duality used here,
\[ \Hmaxeps(A\given B)_\rho = -\Hmineps(A\given C)_\rho \qquad(\rho_{ABC}\ \text{pure}),
\]
relies on the purification-lifting property of purified-distance balls, which
Renner's trace-distance convention does not satisfy.
This duality relates Eve's uncertainty to Bob's error rate.

\noindent\textit{Optimization domain.}
Conditional min-entropy is sensitive to the domain chosen for the conditioning witness~$\sigma_B$; \FVQKD{} uses the equivalent SDP witness form.
The EUR~\cite{tomamichel2011uncertainty} and chain rules~\cite{vitanov2013chain} allow sub-normalized~$\sigma_B$, and \FVQKD{} corresponds to this convention through its SDP form; the remaining four restrict~$\sigma_B$ to normalized states.
The entropy value is the same either way---the optimum over the larger set is attained at a normalized witness---but the distinction matters for formal composition: a lemma proved for normalized~$\sigma_B$ does not apply directly to a statement that quantifies over sub-normalized ones.
The mechanization therefore fixes one definition per concept and bridges convention differences with coercion lemmas at theorem boundaries.
The checked definitions and their mathematical readings appear in \Cref{tab:alignment-entropy,tab:alignment-conditional-max-entropy,tab:alignment-smooth-definitions} of \Cref{app:semantics-alignment}.
The literature has responded to this definitional scatter: a monograph consolidates the conventions~\cite{tomamichel2016qip}, and a largely self-contained proof assembles the QKD security argument in one place~\cite{tomamichel2017largelyself}; both remain pen-and-paper.

%% file: sections/app-definitions.tex
\begin{table}[t]
\centering
\caption{Smooth-Entropy Definitions Across the BB84 Proof Chain}
\label{tab:conventions}
\begingroup
\footnotesize
\renewcommand{\arraystretch}{1.15}
\setlength{\tabcolsep}{2.0pt}
\begin{tabular}{@{}>{\raggedright\arraybackslash}p{2.15cm}
                *{3}{>{\centering\arraybackslash}p{0.92cm}}
                *{2}{>{\centering\arraybackslash}p{1.18cm}}@{}}
\hline
\textbf{Paper} & $\sigma_{\min}$ & $\sigma_{\max}$
  & \textbf{Ball} & \textbf{Entropy}
  & \textbf{Smooth} \\
\hline
Finite-key~\cite{tomamichel2012tight}\textsuperscript{\dag} & $=$   & $=$   & P & eq/eq & eq/eq \\
EUR~\cite{tomamichel2011uncertainty} & $\leq$ & $\leq$ & P & id/eq & id/eq \\
LHL~\cite{tomamichel2011leftover} & $=$   & $=$   & P & eq/eq & eq/eq \\
Duality~\cite{tomamichel2010duality} & $=$   & $=$   & P & eq/eq & eq/eq \\
Renner~\cite{renner2008security} & $=$   & $=$   & D & eq/neq & eq/neq \\
Chain rules~\cite{vitanov2013chain} & $\leq$ & $\leq$ & P & id/eq & id/eq \\
\hline
\end{tabular}
\endgroup\par

\vspace{0.35\baselineskip}
\raggedright
$=$ normalized $\sigma$ only; $\leq$ sub-normalized $\sigma$ allowed; P purified-distance ball; D generalized trace-distance ball.
\textsuperscript{\dag}Finite-key imports definitions by citing Duality.
Entropy and smoothing entries give min-/max-entropy verdicts: ``id'' = identical to the \FVQKD{} definition; ``eq'' = equivalent but not identical; ``neq'' = not equivalent.
\par
\end{table}

%% file: sections/discussion.tex
\section{Discussion}
\label{sec:discussion}

\noindent\textbf{Scope.}
The proof certifies the information-theoretic security argument: the chain from smooth entropies through the EUR and the Quantum Leftover Hash Lemma to a composable finite-key bound.
Device-level vulnerabilities (detector blinding, photon-number splitting, Trojan-horse attacks) and implementation-level side channels are out of scope under the published model~\cite{tomamichel2012tight}; the theorem applies only once a device model justifies the supplied source and measurement parameters (\Cref{sec:security:model}).
The entropy infrastructure is protocol-independent, and BB84 is verified end to end; proving security for another protocol on the same smooth-entropy route requires a new top layer over the existing core.

\noindent\textbf{Future work.}
Three concrete extensions lie beyond the current scope, each requiring substantial new libraries.
Entropy accumulation~\cite{dupuis2020ea} (the standard framework for device-independent proofs, also underlying recent prepare-and-measure proofs~\cite{kamin2025finite}) would reuse \FVQKD{}'s LHL but needs a R\'{e}nyi chain-rule library.
Numerical security frameworks~\cite{Winick2018reliablenumerical,george2021numerical} need a verified convex-optimization layer,
of which \FVQKD{}'s SDP strong duality theorem provides a core ingredient.
Continuous-variable QKD~\cite{leverrier2015composable} operates in infinite-dimensional Hilbert spaces and requires an extension with no analog in the current development.
Beyond these, connecting \FVQKD{}'s explicit source and measurement obligations (\Cref{sec:security:model}) to a verified device model would narrow the remaining gap between the information-theoretic bound and a deployed system.

%% file: sections/conclusion.tex
\section{Conclusion}\label{sec:conclusion}

\FVQKD{} provides, to our knowledge, the first end-to-end, machine-checked QKD security proof: composable finite-key BB84 security at Tomamichel et al.'s tight bound~\cite{tomamichel2012tight}, with zero incomplete proofs and only three standard axioms inherited from \MathCompAnalysis{} for classical logic.
The result connects the quantum-information-theoretic bound to a formal model of protocol execution through a complete, kernel-checked proof chain.
Establishing it required reconstructing the proof chain from first principles and supplying side conditions that the published argument leaves implicit, one of which revealed a missing proof branch in an analysis on which later finite-key proofs build.
The resulting verified, reusable smooth-entropy framework, EUR, Quantum LHL, and Serfling bound expose every source and measurement assumption as an explicit theorem parameter and provide the foundation for decoy-state, measurement-device-independent, and other QKD protocols that follow the same proof route.
Physics-based security and machine-checked proofs together strengthen the basis for QKD deployment amid standardization.

%% file: sections/app-qkd-variants.tex
\begingroup
\raggedbottom
\section{QKD Variant Program Schemas}\label{app:qkd-variants}

To illustrate the protocol language beyond BB84, we give pseudocode schemas
inspired by
six-state~\cite{bruss1998sixstate}, B92~\cite{bennett1992b92},
BBM92~\cite{bbm92}, and E91~\cite{ekert1991bell}.
They show how prepare-and-measure and entanglement-based control patterns can
be written using the same command forms.
The protocol-specific states, measurements, and classical routines remain
parameters.

The first three schemas share the following control flow; only their round
bodies, announcements, sifting rules, and optional extra tests vary.
E91 replaces parameter estimation with the Bell-test tail shown below.
\par\medskip
\noindent\begin{minipage}{\columnwidth}
\begin{minted}[fontsize=\footnotesize,breaklines]{text}
Variant(round, announce, sift, extra_fail := false):
  P := {}; X := {}; Z := {}; abort := false
  for i = 1,...,M: round(i)
  announce; sift; ParamEst
  fail := |X| < n or |Z| < k or lambda > Q_tol
  if fail or extra_fail then Abort
  else:
    ErrCorr
    if t_A != t_B then Abort else PrivAmp_l
\end{minted}
\end{minipage}\par\medskip

\noindent\textbf{Six-state-style schema.}
The three-setting state and measurement families are parameters.
\begin{minted}[fontsize=\footnotesize,breaklines]{text}
round(i):
  a_i <-$ Unif({0,1,2}); y_i <-$ Unif({0,1})
  q_i <- source[a_i,y_i]; Attack[q_i]
  b_i <-$ Unif({0,1,2}); y'_i := meas[b_i][q_i]
  d_i := (y'_i != empty)
announce := Announce((a_i,b_i,d_i)_i)
sift:
  W := {j in {1,...,M} | d_j and a_j = b_j}
  Z <-$ Unif({T subset W | |T| = min(k,|W|)})
  X := W \ Z; Announce(Z)
\end{minted}

\noindent\textbf{B92-style schema.}
The two distinct nonorthogonal source states and a three-outcome
unambiguous-discrimination interface are parameters.
\begin{minted}[fontsize=\footnotesize,breaklines]{text}
round(i):
  r_i <-$ Unif({key,test}); y_i <-$ Unif({0,1})
  q_i <- source[y_i]; Attack[q_i]
  y'_i := meas[q_i]
  d_i := (y'_i != inconclusive)
announce := Announce((r_i,d_i)_i)
sift:
  X := {j in {1,...,M} | d_j and r_j = key}
  Z := {j in {1,...,M} | d_j and r_j = test}
\end{minted}

\noindent\textbf{BBM92-style schema.}
The pair state and the common two-setting measurement family are parameters.
\begin{minted}[fontsize=\footnotesize,breaklines]{text}
round(i):
  (q_i^A,q_i^B) <- pair[i]; Attack[q_i^A,q_i^B]
  a_i <-$ Unif({0,1}); b_i <-$ Unif({0,1})
  y_i := meas[a_i][q_i^A]; y'_i := meas[b_i][q_i^B]
  d_i := (y_i != empty and y'_i != empty)
announce := Announce((a_i,b_i,d_i)_i)
sift:
  W := {j in {1,...,M} | d_j and a_j = b_j}
  Z <-$ Unif({T subset W | |T| = min(k,|W|)})
  X := W \ Z; Announce(Z)
\end{minted}

\noindent\textbf{E91-style schema.}
The pair state, the two three-setting measurement families, and a threshold
$2<\beta_{\mathrm{CHSH}}\le 2\sqrt{2}$ are parameters.
Round indices are one-based, but setting labels remain zero-based:
$a=0,1,2$ denotes Alice's angles $0^\circ,45^\circ,90^\circ$, whereas
$b=0,1,2$ denotes Bob's angles $45^\circ,90^\circ,135^\circ$.
Detected outcomes are encoded as bits in $\{0,1\}$, with
\texttt{empty} denoting nondetection.
\begin{minted}[fontsize=\footnotesize,breaklines]{text}
E91:
  P := {}; X := {}; R := {}; abort := false
  for i = 1,...,M:
    (q_i^A,q_i^B) <- pair[i]; Attack[q_i^B]
    a_i <-$ Unif({0,1,2}); b_i <-$ Unif({0,1,2})
    y_i := meas_A[a_i][q_i^A]; y'_i := meas_B[b_i][q_i^B]
    d_i := (y_i != empty and y'_i != empty)
    X := {j <= i | d_j and (a_j,b_j) in {(1,0),(2,1)}}
    R := {j <= i | d_j and (a_j,b_j) not in {(1,0),(2,1)}}
  Announce((a_i,b_i,d_i)_i,
           (y_j,y'_j)_{j in R})
  derive A_uw,N_uw from the revealed R rounds for
    (u,w) in {(0,0),(0,2),(2,0),(2,2)}
  if |X| < n or some N_uw = 0 then Abort
  else:
    E_uw := (2*A_uw-N_uw)/N_uw for each CHSH class
    S_CHSH := E_00 - E_02 + E_20 + E_22
    if |S_CHSH| < beta_CHSH then Abort
    else:
      K <-$ Unif({T subset X | |T| = n}); Announce(K)
      X_A := (y_j)_{j in K}; X_B := (1-y'_j)_{j in K}
      ErrCorr
      if t_A != t_B then Abort else PrivAmp_l
\end{minted}
After the nonempty-class guard,
$E_{uw}=(2A_{uw}-N_{uw})/N_{uw}$ for the agreement count~$A_{uw}$
and detected total~$N_{uw}$.
The sketch omits a finite-statistics correction for this test.
\endgroup

%% file: sections/app-bb84-sig.tex
\section{Mechanization of the BB84 security criterion}\label{app:bb84-sig}

This appendix records the proof-structure map, its correspondence with the
supplementary notes of~\cite{tomamichel2012tight}, and the paper-to-\Rocq{}
correspondence table.
The smooth-entropy convention comparison appears in \Cref{sec:findings:conventions}.

\begin{figure*}[t]
\centering
\includegraphics[width=0.86\textwidth]{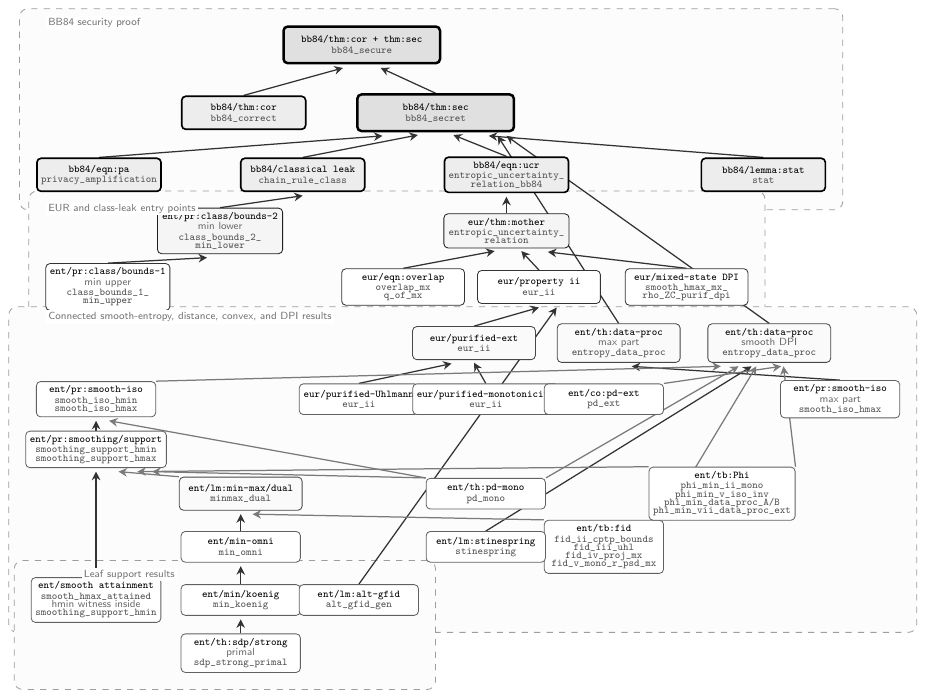}
\caption{Proof-dependency graph for the BB84 security proof.
Each box names a source-paper theorem, lemma, or proof topic together with the corresponding \Rocq{} declaration name(s).
Node prefixes identify the source literature: \texttt{bb84} is Tomamichel--Lim--Gisin--Renner's finite-key BB84 proof~\cite{tomamichel2012tight}, \texttt{eur} is Tomamichel--Renner's smooth entropic uncertainty relation~\cite{tomamichel2011uncertainty}, and \texttt{ent} is Tomamichel's entropy framework/thesis~\cite{tomamichel2013frameworknonasymptoticquantuminformation}; the source-paper prefix and TeX label are separated by a slash.
An arrow $A{\to}B$ means $A$ supports the proof of~$B$.
Every arrow is a dependency in the \Rocq{} proofs; lighter arrows cross between branches of the proof tree.
The top layer identifies \texttt{bb84\_secure} as the combined theorem of Tomamichel et al. and shows its reusable correctness and full secrecy boundaries; the lower layers expand the EUR, smooth-entropy, purified-distance/fidelity, data-processing-inequality (DPI), and convex-analysis support used by \texttt{bb84/thm:sec}.}
\label{fig:proof-structure}
\end{figure*}

\subsection{Proof structure and correspondence}\label{app:bb84-sig:proof}

\coqi{bb84_secret} establishes Supplementary Theorem~2:
it derives the smooth min-entropy bound from the source model's POVM and overlap fields, the EUR, delayed measurement, and the aggregate sampling witness, then proves secrecy of the normalized accepted output through leakage, postselection, and privacy amplification.
\coqi{bb84_secure} combines this result with $\eps_c$-correctness from \coqi{bb84_correct} and applies $\eps_c + \eps_s \le \eps_\mathrm{tot}$.
The secrecy proof follows Supplementary Note~1 of~\cite{tomamichel2012tight}, presented here in eight steps:

\begin{center}
\centering
\footnotesize
\renewcommand{\arraystretch}{1.15}
\begin{tabular}{@{}r@{\hspace{4pt}}l@{\hspace{5pt}}l@{}}
\hline
\textbf{\#} & \textbf{Claim} & \textbf{By} \\
\hline
1 & $\eps' := \eps/\sqrt{p_\mathrm{pass}}$
  & definition\\
2 & $\Hmineps[\eps'](\mathbf{X}{\mid}E) {+} \Hmaxeps[\eps'](\mathbf{Z}{\mid}B_q) \ge n q$
  & EUR\\
3 & $\Hmaxeps[\eps'](\mathbf{Z}{\mid}B_q) \le \Hmaxeps[\eps'](\mathbf{Z}{\mid}\mathbf{Z}')$
  & smooth DPI \\
4 & $\Hmaxeps[\eps'](\mathbf{Z}{\mid}\mathbf{Z}') \le n\,h(\Qtol{+}\mu(\eps))$
  & statistical\\
5 & $\Hmineps[\eps'](\mathbf{X}{\mid}E') \ge nq {-} nh {-} \leakec {-} \log_2\!\frac{2}{\eps_c}$
  & 2--4, chain rule\\
6 & $\Delta := \eps' {+} \tfrac{1}{2}\sqrt{2^{\,\ell - \Hmineps[\eps'](\mathbf{X}{\mid}E')}}$
  & PA\\
7 & $\Delta \le \eps' + \bar\eps$
  & 6, key length (Thm.~\ref{thm:protocol-security})\\
8 & $(1{-}p_\mathrm{abort})\,\Delta \le \eps {+} \bar\eps \le \eps_s$
  & rescale\\
\hline
\end{tabular}
\end{center}

\vspace{4pt}
\noindent\textbf{Secrecy derivation inside \texttt{bb84\_secret} (Supplementary Note~1).}
$\mathbf{X}$: Alice's raw key; $\mathbf{Z},\mathbf{Z}'$: hypothetical $\mathbf{Z}_{\mathrm{key}},\mathbf{Z}'_{\mathrm{key}}$ (EUR conjugate basis); $B_q$: Bob's quantum register on key rounds; $E$: adversary pre-EC; $E'$: adversary post-EC.
Row~5 abbreviates $h:=h(\Qtol{+}\mu(\eps))$.
In \Rocq{}, step~2 is \coqi{entropic_uncertainty_relation_bb84}, and step~6 is \coqi{privacy_amplification}.

\noindent The remaining components correspond as follows:
Supplementary Theorem~1 (correctness) is the reusable \coqi{bb84_correct}; Supplementary Lemma~3 (statistical) is \coqi{stat} in \code{qkd/bb84/prop.v}; and the classical-leakage chain rule is \coqi{chain_rule_class}.
Supplementary Theorem~2 corresponds to \coqi{bb84_secret}; \coqi{bb84_secure} combines it with \coqi{bb84_correct} under $\eps_c{+}\eps_s{\le}\eps_\mathrm{tot}$, and \coqi{bb84_secure_protocol} is the independent protocol-layer culmination.
The culminating theorem depends only on the standard \MathCompAnalysis{} foundations (\Cref{sec:engineering:backend}).

\subsection{Paper-to-Rocq correspondence}\label{app:correspondence}

\Cref{tab:correspondence} maps each numbered security and protocol statement in the body to the Rocq declaration that establishes it or to the model premise that represents it.

\begin{table}[t]
\centering
\caption{Paper-to-Rocq Correspondence}
\label{tab:correspondence}
\footnotesize
\renewcommand{\arraystretch}{1.0}
\newcommand{\cbreak}{\hspace{0pt}}%
\begin{tabular}{@{}l@{\;}>{\hfuzz=8pt}p{2.1cm}@{\;}>{\raggedright\arraybackslash\hfuzz=8pt}p{4.0cm}@{}}
\hline
\textbf{Paper} & \textbf{Statement} & \textbf{Rocq name} \\
\hline
Def.~\ref{def:correct}  & $\eps_c$-correctness & \code{eps\_correct} \\
Def.~\ref{def:secret}   & $\eps_s$-secrecy     & \code{eps\_secret} \\
Defs.~\ref{def:correct},~\ref{def:secret}
        & $\eps$-security & \code{eps\_secure} \\
Thm.~\ref{thm:eur}      & Smooth EUR
        & \code{entropic\_\cbreak{}uncertainty\_\cbreak{}relation} \\
Def.~\ref{def:two-universal} & Two-universal hash & \code{two\_univ} \\
Thm.~\ref{thm:lhl}      & Quantum LHL
        & \code{two\_universal\_hashing} \\
Thm.~\ref{thm:bb84}     & BB84 security        & \code{bb84\_secure} \\[2pt]
	Def.~\ref{def:protocol-cmd} & Protocol cmds & \code{qkd\_cmd} \\
	Fig.~\ref{fig:bb84-program} & BB84 program & \code{bb84\_program} \\
	Def.~\ref{def:protocol-semantics} & Protocol sem. & \code{bb84\_eval} \\
	Thm.~\ref{thm:protocol-security} & Protocol security
	        & \code{bb84\_secure\_protocol} \\
Thm.~\ref{thm:bb84-instance} & Non-vacuous instance
        & \code{bb84\_secure\_example} \\
\hline
\end{tabular}
\end{table}

%% file: sections/app-semantics-alignment.tex
\onecolumn
\raggedbottom
\section{Trusted Specification}
\label{app:semantics-alignment}

\definecolor{alignmentteal}{RGB}{8,126,139}
\definecolor{alignmentmathbg}{RGB}{232,244,245}
\definecolor{alignmentcodebg}{RGB}{247,249,251}
\definecolor{alignmentbandbg}{RGB}{219,236,237}
\definecolor{alignmentmuted}{RGB}{82,94,106}
\newcommand{\AlignmentBoundaryLabel}{%
  \par\smallskip\noindent
  \textcolor{alignmentteal}{\sffamily\bfseries Remark.}\enspace}
\newcommand{\AlignmentParagraph}[1]{%
  \par\smallskip\noindent\textbf{#1}\par
  \nopagebreak[4]}

\setlength{\columnsep}{8pt}
\columnratio{0.615}
\newlength{\AlignLeftW}\setlength{\AlignLeftW}{\dimexpr0.615\textwidth-0.615\columnsep\relax}
\newlength{\AlignRightW}\setlength{\AlignRightW}{\dimexpr0.385\textwidth-0.385\columnsep\relax}
\newcommand{\AlignmentBlockSetup}{\footnotesize\raggedright\setlength{\parindent}{0pt}\setlength{\parskip}{0pt}%
  \setlength{\topsep}{1pt}\setlength{\partopsep}{0pt}}
\newcommand{\AlignmentBlockHead}[1]{%
  \par\Needspace*{9\baselineskip}\addvspace{5pt}%
  {\setlength{\fboxsep}{2pt}\offinterlineskip
  \noindent\rule{\textwidth}{0.4pt}\par
  \noindent\colorbox{alignmentbandbg}{\parbox{\dimexpr\textwidth-2\fboxsep\relax}{%
    \footnotesize\sffamily\bfseries\strut\textcolor{alignmentteal}{#1}\strut}}\par
  \noindent\colorbox{alignmentcodebg}{\parbox{\dimexpr\AlignLeftW-2\fboxsep\relax}{%
    \footnotesize\sffamily\bfseries\strut Formal statement\strut}}%
  \hspace{\columnsep}%
  \colorbox{alignmentmathbg}{\parbox{\dimexpr\AlignRightW-2\fboxsep\relax}{%
    \footnotesize\sffamily\bfseries\strut Mathematical reading\strut}}\par
  \noindent\rule{\textwidth}{0.4pt}\par}\nobreak\vspace{2pt}}
\newcommand{\AlignmentBlockEnd}{\par\nobreak\vspace{1pt}\noindent\rule{\textwidth}{0.4pt}\par\addvspace{5pt}}
\newcounter{alignmentfragment}
\newcommand{\AlignmentContinuationLabel}[1]{%
  \stepcounter{alignmentfragment}%
  \addtocounter{table}{-1}%
  \begingroup
    \renewcommand*{\theHtable}{semantics-alignment.\arabic{alignmentfragment}}%
    \refstepcounter{table}%
    \label{#1}%
  \endgroup}
\setminted[coq]{%
  frame=none,
  framesep=1.2mm,
  tabsize=2,
  fontsize=\footnotesize,
  breaklines=true,
  breakafter={_.,}}

Because the \Rocq{} kernel checks every proof, a reader needs to trust only the definitions and theorem statements on which the security results depend.
This appendix lists them, each as a declaration from the artifact beside its mathematical reading.

\subsection{Main Theorems}
\label{app:semantics-alignment-main}

We begin with \coqi{bb84_secure} and \coqi{bb84_secure_protocol}, followed by
every project-defined object used transitively in their statements.

\subsubsection{Endpoints}

\AlignmentParagraph{Mathematical Security.}

\captionof{table}{Trusted specification of the security results.}
\label{tab:alignment-matrix}
\label{tab:alignment-bb84-secure-inputs}
\AlignmentBlockHead{Main theorem: combined finite-key composition interface.}
\begingroup\AlignmentBlockSetup
\begin{paracol}{2}
\code{src/qkd/bb84/bb84.v}: \coqi{bb84_secure}
\begin{minted}{coq}
Variable R : realType.
Variable (L : finType) (HL : L -> chsType).
Variable (S E Zkey Bq Zp Fl Zl : {set L}).
Hypothesis disSE : [disjoint S & E].
Hypothesis disSF : [disjoint S & Fl].
Hypothesis disSZl : [disjoint S & Zl].
Hypothesis disZBq : [disjoint Zkey & Bq].
Hypothesis disZZp : [disjoint Zkey & Zp].
Hypothesis disFZl : [disjoint Fl & Zl].
Variable Aux : @bb84_aux_regs R L HL Zkey Bq Zp.
Theorem bb84_secure
    (P : bb84_protocol) (q eps_s eps_tot : R)
    (Gsrc : bb84_zsource P)
    (src : bb84_source_model P q)
    (dim_S : bb84_src_nx src = dim ('H[HL]_S))
    (dim_E : bb84_src_nb src = dim ('H[HL]_E))
    (dim_Zk : bb84_src_nz src = dim ('H[HL]_Zkey))
    (dim_Bq : bb84_src_nc src = dim ('H[HL]_Bq))
    (dim_Zp_key : (2 ^ bb84_n P)%N = dim ('H[HL]_Zp))
    (BobZ : 'QC[HL]_(Bq, Zp))
    (p_abort mismatch_prob : R)
    (seed_count tag_count : nat)
    (verify_hash : 'I_seed_count -> 'I_(2 ^ bb84_n P) -> 'I_tag_count)
    (corr_prob : 'I_(2 ^ bb84_n P) -> 'I_(2 ^ bb84_n P) -> R)
    (p_mismatchE :
      mismatch_prob =
        @p_mismatch R (2 ^ bb84_n P) seed_count tag_count verify_hash
          corr_prob)
    (verify_tag_count_pos : (0 < tag_count)%N)
    (verify_two_univ : @two_univ R _ _ _ verify_hash)
    (corr_nonneg : forall x x', 0 <= corr_prob x x')
    (corr_mass : \sum_x \sum_(x' | x != x') corr_prob x x' <= 1)
    (verify_tag_inv_le_epsc :
      ((tag_count%:R : R)^-1 <= bb84_eps_c P)%R)
    (Fh : 'Idx[HL]_Fl -> 'Idx[HL]_S -> 'Idx[HL]_Zl)
    (Fh_two_univ : @two_univ R _ _ _ Fh)
    (Fh_key_dim : dim ('H[HL]_Zl) = (2 ^ bb84_l P)%N)
\end{minted}
The source, two hash-family predicates, mismatch probability, and register
dimensions are expanded in the statement-definition rows below.
The auxiliary-register record is expanded in \Cref{tab:alignment-bb84-aux-regs}.
\switchcolumn
\textbf{Main finite-key theorem (\Cref{thm:bb84}): interface.}
The statement fixes a full-string source $G_{\rm src}$ and a source model
$(\rho_{ABC},M,N)$ with
\[
\begin{gathered}
|X|=\dim S,\qquad \dim B=\dim E,\\
|Z|=\dim Z_{\rm key},\qquad \dim C=\dim B_q,\\
\dim Z'_{\rm key}=2^n,\\
S\perp E,\qquad Z_{\rm key}\perp B_q,\\
Z_{\rm key}\perp Z'_{\rm key},\qquad F_\ell\perp Z_\ell,\\
S\perp F_\ell,\qquad S\perp Z_\ell.
\end{gathered}
\]
It takes the auxiliary registers of \Cref{tab:alignment-bb84-aux-regs}, Bob's delayed $\mathbb Z$ channel, nonnegative correlation weights
with total off-diagonal mass at most one, a nonempty two-universal verification
family satisfying $|T|^{-1}\le\eps_c$, and a nonempty two-universal
privacy-amplification family
$F:F_\ell\times S\to Z_\ell$ with $|Z_\ell|=2^\ell$.
\end{paracol}
\endgroup
\AlignmentBlockEnd

\AlignmentContinuationLabel{tab:alignment-bb84-secure}
\AlignmentBlockHead{Main theorem: accepted state, premises, and conclusion.}
\begingroup\AlignmentBlockSetup
\begin{paracol}{2}
\code{src/qkd/bb84/bb84.v}: \coqi{bb84_secure} (continued)
\begin{minted}{coq}
    (Epub : {set L})
    (pub : bb84_public_leakage Epub
      (eur_rho_XB disSE (bb84_src_rho src) (@bb84_src_M R P q src)
        (bb84_register_pair_dim disSE dim_S dim_E))
      (bb84_leak_EC P + log2 (2%:R / bb84_eps_c P)))
    (disSEpub : [disjoint S & E :|: Epub])
    (disEpubF : [disjoint E :|: Epub & Fl])
    (disEpubZl : [disjoint E :|: Epub & Zl])
    (disZl_FEpub : [disjoint Zl & (Fl :|: (E :|: Epub))])
    (accept reject : 'F[HL]_(S :|: (E :|: Epub)))
    (rho_out : 'F[HL]_(Zl :|: (Fl :|: (E :|: Epub))))
    (Haccept_psd : accept \is psdlf)
    (Hreject_psd : reject \is psdlf)
    (Haccept_cq : @is_cq_state R L HL S (E :|: Epub) accept)
    (Hsplit :
      accept + reject =
        castlf (bb84_public_assoc S E Epub, bb84_public_assoc S E Epub)
          (bb84_public_state pub))
    (Hout_psd : rho_out \is psdlf)
    (Hout_tr : @complex.Re R (\Tr rho_out) <= 1)
    (Hout :
      0 < @complex.Re R (\Tr accept) ->
      rho_out =
        ((@complex.Re R (\Tr accept))^-1)%:C *:
          post_hash_state disSEpub Fh accept)
    (Hnum : bb84_numeric_side P) :
  (0 < bb84_p_pass_tot Gsrc ->
    bb84_delayed_measurement_condition BobZ
      (eur_rho_ZC disZBq (bb84_src_rho src) (@bb84_src_N R P q src)
        (bb84_register_pair_dim disZBq dim_Zk dim_Bq))
      (@bb84_rhoZZ_pass_mixH R L HL Zkey Zp disZZp
        P Gsrc
        (bb84_source_zkey_dim dim_Zk)
        dim_Zp_key)) ->
  bb84_pass_prob_le p_abort
    (bb84_p_pass_tot Gsrc * @complex.Re R (\Tr accept)) ->
  (0 < bb84_n P)%N -> (0 < bb84_k P)%N ->
  bb84_eps_c P + eps_s <= eps_tot ->
  (0 < eps_s) && (eps_s <= 1) ->
  (bb84_l P)%:R <=
    sup (fun m : R =>
      exists eps ebar : R,
        0 < eps /\ eps < 1 /\ 0 < ebar /\ eps + ebar <= eps_s /\
        m = ((floor (bb84_key_length_ub (bb84_n P) (bb84_k P) q
          (bb84_Qtol P) eps ebar (bb84_leak_EC P)
          (bb84_eps_c P)))%:~R : R)) ->
  eps_secure mismatch_prob (bb84_eps_c P) rho_out eps_s p_abort eps_tot.
\end{minted}
\AlignmentBoundaryLabel
The theorem is stated directly in terms of the verification distribution and
the secrecy state. The component theorems appear in the next subsection.
\switchcolumn
\textbf{Main finite-key theorem (\Cref{thm:bb84}): conclusion.}
Let $\rho_{PSE}$ be the certified public-leakage state and let
$\rho_{\rm acc},\rho_{\rm rej}$ be its accepted and rejected blocks.
The output registers satisfy
\[
\begin{gathered}
S\perp(E\cup E_{\rm pub}),\\
(E\cup E_{\rm pub})\perp F_\ell,\\
(E\cup E_{\rm pub})\perp Z_\ell,\\
Z_\ell\perp(F_\ell\cup E\cup E_{\rm pub}).
\end{gathered}
\]
The state premises are
\[
\begin{gathered}
\rho_{\rm acc},\rho_{\rm rej}\succeq0,
\quad\rho_{\rm acc}+\rho_{\rm rej}=\rho_{PSE},\\
\rho_{\rm acc}\text{ is cq on }S,
\quad\rho_{\rm out}\succeq0,\quad\Tr\rho_{\rm out}\le1,
\end{gathered}
\]
and, whenever $\Tr\rho_{\rm acc}>0$,
\[
\rho_{\rm out}=(\Tr\rho_{\rm acc})^{-1}
  \mathsf{PostHash}_{F}(\rho_{\rm acc}).
\]
If positive source-pass mass implies the delayed-measurement condition and
\[
0\le p_{\rm abort}\le1,\qquad
1-p_{\rm abort}\le
p_{\rm pass}^{\rm tot}\Tr\rho_{\rm acc},
\]
then the numeric, error-budget, and optimized key-length premises yield
$\mathsf{Secure}_{\eps_{\rm tot}}
(p_{\rm mismatch},\eps_c,\rho_{\rm out},\eps_s,p_{\rm abort})$.
\end{paracol}
\endgroup
\AlignmentBlockEnd

\AlignmentParagraph{Protocol Security.}

\AlignmentContinuationLabel{tab:alignment-end-to-end-context}
\AlignmentBlockHead{Protocol endpoint: ambient model and device context.}
\begingroup\AlignmentBlockSetup
\begin{paracol}{2}
\code{src/qkd/example/example.v}: \coqi{BB84SecureProtocol} section context
\begin{minted}{coq}
Variable R : realType.
Variable P : @bb84_protocol R.
Variable q : R.
Variable dev : bb84_certified_device P q.
Variable LE : finType.
Variable HE : LE -> chsType.
Notation M := (bb84_round_limit P).
Notation L := (bb84_label M LE bb84_register_index).
Notation HL := (@bb84_label_space R M LE
  bb84_register_index HE (@bb84_register_space R P q dev)).
Notation V := (bb84_val M (bb84_n P) (bb84_l P)
  (cd_ir_bits dev) (cd_vseed_count dev)
  (cd_tag_count dev) (cd_paseed_count dev)).
Notation Eve := (@bb84_Eve R P LE).
Notation Ereg := (@bb84_E R P LE).
Notation Zl := (@bb84_Zl R P LE).
Notation Fl := (@bb84_Fl R P LE).
Notation Ppub := (@bb84_Ppub R P LE).
Hypothesis Hclass : @bb84_projective_xz_device R P q dev.
Hypothesis Hn : (0 < bb84_n P)%N.
Hypothesis Hk : (0 < bb84_k P)%N.
Hypothesis Hround_pos : (0 < M)%N.
\end{minted}
Every notation above is expanded in the statement-definition rows below.
\switchcolumn
\textbf{Protocol Security: ambient context.}
The endpoint fixes a protocol $P$, overlap parameter $q$, certified device,
external register space, and round budget
$M=\mathsf{round\_limit}(P)$.
Its concrete classical values and labelled quantum registers are determined
from these objects, and the device satisfies the projective
$\mathbb X/\mathbb Z$ class with $n,k,M>0$.
\end{paracol}
\endgroup
\AlignmentBlockEnd

\AlignmentContinuationLabel{tab:alignment-end-to-end}
\AlignmentBlockHead{Protocol endpoint: conditional end-to-end security.}
\begingroup\AlignmentBlockSetup
\begin{paracol}{2}
\code{src/qkd/example/example.v}: \coqi{bb84_secure_protocol}
\begin{minted}{coq}
Corollary bb84_secure_protocol (entangled : bool) (eps_s eps_tot : R)
    (Hec_pos : 0 < bb84_eps_c P) (Hec_le1 : bb84_eps_c P <= 1) :
  bb84_eps_c P + eps_s <= eps_tot ->
  (0 < eps_s) && (eps_s <= 1) ->
  (bb84_l P)%:R <=
    sup (fun m : R =>
      exists eps ebar : R,
        0 < eps /\ eps < 1 /\ 0 < ebar /\ eps + ebar <= eps_s /\
        m = ((floor (bb84_key_length_ub (bb84_n P) (bb84_k P) q
          (bb84_Qtol P) eps ebar (bb84_leak_EC P)
          (bb84_eps_c P)))%:~R : R)) ->
  forall (EA : @adversary R L HL V Eve)
    (xi_init : @cq_state R L HL (bb84_cvar M) V),
  cq_wf xi_init ->
  let xi_out := bb84_eval_mode entangled EA xi_init in
  let p_abort :=
    @run_abort_prob R L HL (bb84_cvar M) V (@BVabort M)
      bb84_bool_val xi_out in
  let p_pass :=
    @run_pass_prob R L HL (bb84_cvar M) V (@BVabort M)
      bb84_bool_val xi_out in
  let p_mismatch :=
    @run_output_mismatch_prob R L HL (bb84_cvar M) V
      (@BVS M) (@BVSh M) (@BVabort M) bb84_bool_val xi_out in
  [/\ (forall (h : @run_history V M) (s : @store (bb84_cvar M) V),
         bb84_branch_shape h s),
      p_abort + p_pass = 1
    & @eps_secure R L HL Zl (Fl :|: (Ereg :|: Ppub))
        p_mismatch (bb84_eps_c P)
        (rhoZFEpub_eval_mode entangled EA xi_init)
        eps_s p_abort eps_tot].
\end{minted}
\AlignmentBoundaryLabel
The conclusion packages branch shape, abort/pass partition, and \coqi{eps_secure}.
The certified-device fields \coqi{cd_tag_le} and \coqi{cd_tag_bits}, expanded in \Cref{tab:alignment-device-certificate}, discharge both tag inequalities internally at the canonical target \coqi{bb84_eps_c P}.
Its device-class premise, evaluator, terminal-store predicate, run probabilities, and normalized output are expanded in
\Cref{tab:alignment-device-class,tab:alignment-program-modes,tab:alignment-run-objects,tab:alignment-output-reification};
the key-length functional is expanded in \Cref{tab:alignment-bb84}.
\switchcolumn
\textbf{Protocol Security (\Cref{thm:protocol-security}).}
The formal endpoint assumes
\[
\begin{gathered}
\mathsf{PXZ}(\mathsf{dev}),\qquad n,k,M>0,\\
0<\eps_c\le1,\qquad 0<\eps_s\le1,\qquad
\eps_c+\eps_s\le\eps_{\rm tot},\\
\ell\le
\sup_{\substack{0<\eps<1,\ \bar\eps>0\\
                  \eps+\bar\eps\le\eps_s}}
\lfloor L(\eps,\bar\eps)\rfloor,\qquad
\mathsf{wf}(\xi_{\rm init}).
\end{gathered}
\]
Under these premises, every adversary and well-formed run satisfies
\[
p_{\rm abort}+p_{\rm pass}=1,
\qquad
\Pr[S\ne\hat S\wedge\neg\mathsf{abort}]\le\eps_c,
\]
\[
(1-p_{\rm abort})d_{\rm u}
  (\rho_{S F_{\rm PA}E^+})\le\eps_s,
\qquad
\eps_c+\eps_s\le\eps_{\rm tot},
\]
and every terminal store has either the abort shape
$S=\hat S=\perp$ or the non-abort equalities recorded by
$\mathsf{branch\_shape}$.
\end{paracol}
\endgroup
\AlignmentBlockEnd

\subsubsection{Mathematical Closure}

\AlignmentParagraph{Auxiliary Registers.}

\AlignmentContinuationLabel{tab:alignment-bb84-aux-regs}
\AlignmentBlockHead{Auxiliary registers of the mathematical theorem.}
\begingroup\AlignmentBlockSetup
\begin{paracol}{2}
\code{src/qkd/bb84/def.v}: \coqi{bb84_aux_regs}
\begin{minted}{coq}
Record bb84_aux_regs := BB84AuxRegs {
  bb84_aux_app : {set L};
  bb84_aux_bpp : {set L};
  bb84_aux_ciso : {set L};
  bb84_aux_cp_src : {set L};
  bb84_aux_cpp_src : {set L};
  bb84_aux_cp_tgt : {set L};
  bb84_aux_cpp_tgt : {set L};
  bb84_aux_disZ_app : [disjoint Zkey & bb84_aux_app];
  bb84_aux_disZp_bpp : [disjoint Zp & bb84_aux_bpp];
  bb84_aux_disZ_bpp : [disjoint Zkey & bb84_aux_bpp];
  bb84_aux_disApp_Zp : [disjoint bb84_aux_app & Zp];
  bb84_aux_disApp_bpp : [disjoint bb84_aux_app & bb84_aux_bpp];
  bb84_aux_dim_app :
    (dim 'H[HL]_Zkey * dim 'H[HL]_Zkey <= dim 'H[HL]_bb84_aux_app)%N;
  bb84_aux_dim_bpp :
    (dim 'H[HL]_Bq * dim 'H[HL]_Zp <= dim 'H[HL]_bb84_aux_bpp)%N;
  bb84_aux_disZ_ciso : [disjoint Zkey & bb84_aux_ciso];
  bb84_aux_disBq_ciso : [disjoint Bq & bb84_aux_ciso];
  bb84_aux_disZApp_ciso :
    [disjoint Zkey :|: bb84_aux_app & bb84_aux_ciso];
  bb84_aux_disZpBpp_ciso :
    [disjoint Zp :|: bb84_aux_bpp & bb84_aux_ciso];
  bb84_aux_dim_ZB_ciso :
    (dim 'H[HL]_(Zkey :|: Bq) <= dim 'H[HL]_bb84_aux_ciso)%N;
  bb84_aux_dim_ZAppBBpp_ciso :
    (dim 'H[HL]_((Zkey :|: bb84_aux_app) :|:
                   (Zp :|: bb84_aux_bpp)) <=
     dim 'H[HL]_bb84_aux_ciso)%N;
\end{minted}
\switchcolumn
\textbf{Auxiliary registers.}
Seven registers
$A'',B'',C_{\rm iso},C'_{\rm s},C''_{\rm s},C'_{\rm t},C''_{\rm t}$
serve the Stinespring dilation, the purification, and the copies of the
data-processing steps.
They satisfy
\[
\begin{gathered}
Z_{\rm key}\perp A'',\qquad Z'_{\rm key}\perp B'',\\
Z_{\rm key}\perp B'',\qquad A''\perp Z'_{\rm key},\\
A''\perp B'',\\
(\dim Z_{\rm key})^2\le\dim A'',\\
\dim B_q\,\dim Z'_{\rm key}\le\dim B'',\\
Z_{\rm key}\perp C_{\rm iso},\qquad B_q\perp C_{\rm iso},\\
Z_{\rm key}\cup A''\perp C_{\rm iso},\\
Z'_{\rm key}\cup B''\perp C_{\rm iso},\\
\dim(Z_{\rm key}\cup B_q)\le\dim C_{\rm iso},\\
\dim(Z_{\rm key}\cup A''\cup Z'_{\rm key}\cup B'')\le\dim C_{\rm iso}.
\end{gathered}
\]
\end{paracol}
\endgroup
\AlignmentBlockEnd

\AlignmentContinuationLabel{tab:alignment-bb84-aux-copies}
\AlignmentBlockHead{Auxiliary registers: source and target copies.}
\begingroup\AlignmentBlockSetup
\begin{paracol}{2}
\code{src/qkd/bb84/def.v}: \coqi{bb84_aux_regs}
\begin{minted}{coq}
  bb84_aux_disZ_cp_src : [disjoint Zkey & bb84_aux_cp_src];
  bb84_aux_disZ_cpp_src : [disjoint Zkey & bb84_aux_cpp_src];
  bb84_aux_disBq_cp_src : [disjoint Bq & bb84_aux_cp_src];
  bb84_aux_disBq_cpp_src : [disjoint Bq & bb84_aux_cpp_src];
  bb84_aux_disCiso_cp_src :
    [disjoint bb84_aux_ciso & bb84_aux_cp_src];
  bb84_aux_disCiso_cpp_src :
    [disjoint bb84_aux_ciso & bb84_aux_cpp_src];
  bb84_aux_disCp_cpp_src :
    [disjoint bb84_aux_cp_src & bb84_aux_cpp_src];
  bb84_aux_dim_Z_cp_src :
    dim 'H[HL]_Zkey = dim 'H[HL]_bb84_aux_cp_src;
  bb84_aux_dim_CisoCp_cpp_src :
    (dim 'H[HL]_bb84_aux_ciso * dim 'H[HL]_bb84_aux_cp_src <=
     dim 'H[HL]_bb84_aux_cpp_src)%N;
  bb84_aux_dim_Bq_cpp_src :
    (dim 'H[HL]_Bq <= dim 'H[HL]_bb84_aux_cpp_src)%N;
  bb84_aux_disZApp_cp_tgt :
    [disjoint Zkey :|: bb84_aux_app & bb84_aux_cp_tgt];
  bb84_aux_disZApp_cpp_tgt :
    [disjoint Zkey :|: bb84_aux_app & bb84_aux_cpp_tgt];
  bb84_aux_disZpBpp_cp_tgt :
    [disjoint Zp :|: bb84_aux_bpp & bb84_aux_cp_tgt];
  bb84_aux_disZpBpp_cpp_tgt :
    [disjoint Zp :|: bb84_aux_bpp & bb84_aux_cpp_tgt];
  bb84_aux_disCiso_cp_tgt :
    [disjoint bb84_aux_ciso & bb84_aux_cp_tgt];
  bb84_aux_disCiso_cpp_tgt :
    [disjoint bb84_aux_ciso & bb84_aux_cpp_tgt];
  bb84_aux_disCp_cpp_tgt :
    [disjoint bb84_aux_cp_tgt & bb84_aux_cpp_tgt];
  bb84_aux_dim_ZApp_cp_tgt :
    dim 'H[HL]_(Zkey :|: bb84_aux_app) =
    dim 'H[HL]_bb84_aux_cp_tgt;
  bb84_aux_dim_CisoCp_cpp_tgt :
    (dim 'H[HL]_bb84_aux_ciso * dim 'H[HL]_bb84_aux_cp_tgt <=
     dim 'H[HL]_bb84_aux_cpp_tgt)%N;
  bb84_aux_dim_ZpBpp_cpp_tgt :
    (dim 'H[HL]_(Zp :|: bb84_aux_bpp) <=
     dim 'H[HL]_bb84_aux_cpp_tgt)%N
}.
\end{minted}
\switchcolumn
\textbf{Copy registers.}
The source copies $C'_{\rm s}$ and $C''_{\rm s}$ are disjoint from
$Z_{\rm key}$, $B_q$, $C_{\rm iso}$, and each other, with
\[
\begin{gathered}
\dim C'_{\rm s}=\dim Z_{\rm key},\\
\dim C_{\rm iso}\,\dim C'_{\rm s}\le\dim C''_{\rm s},\\
\dim B_q\le\dim C''_{\rm s}.
\end{gathered}
\]
The target copies $C'_{\rm t}$ and $C''_{\rm t}$ are disjoint from
$Z_{\rm key}\cup A''$, $Z'_{\rm key}\cup B''$, $C_{\rm iso}$, and each
other, with
\[
\begin{gathered}
\dim C'_{\rm t}=\dim(Z_{\rm key}\cup A''),\\
\dim C_{\rm iso}\,\dim C'_{\rm t}\le\dim C''_{\rm t},\\
\dim(Z'_{\rm key}\cup B'')\le\dim C''_{\rm t}.
\end{gathered}
\]
\end{paracol}
\endgroup
\AlignmentBlockEnd

\AlignmentParagraph{Security.}

\AlignmentContinuationLabel{tab:alignment-bb84}
\AlignmentBlockHead{Protocol and key-length function.}
\begingroup\AlignmentBlockSetup
\begin{paracol}{2}
\code{src/utils/arith.v}: \coqi{h_bin};
\code{src/qkd/bb84/def.v}: \coqi{bb84_protocol}, \coqi{bb84_mu}, \coqi{bb84_h_trunc}, \coqi{bb84_key_length_ub}
\begin{minted}{coq}
Record bb84_protocol := BB84Protocol {
  bb84_n : nat; bb84_k : nat; bb84_round_limit : nat;
  bb84_l : nat; bb84_Qtol : R;
  bb84_eps_c : R; bb84_leak_EC : R
}.
Definition h_bin (p : R) : R :=
  if p <= 0 then 0
  else if p >= 1 then 0
  else -(p * log2 p) - ((1 - p) * log2 (1 - p)).
Definition bb84_mu (n k : nat) (eps : R) : R :=
  Num.sqrt ((((n + k)%:R / (n * k)%:R) * ((k.+1)%:R / k%:R)) *
    ln (eps^-1)).
Definition bb84_h_trunc (x : R) : R :=
  if x <= 0 then 0
  else if 1 / 2 <= x then 1
  else h_bin x.
Definition bb84_key_length_ub (n k : nat)
    (q Qtol eps ebar leak_EC eps_c : R) : R :=
  n%:R * (q - bb84_h_trunc (Qtol + bb84_mu n k eps)) -
  2%:R * log2 ((2%:R * ebar)^-1) -
  leak_EC -
  log2 (2%:R / eps_c).
\end{minted}
\switchcolumn
\textbf{Statistical correction and key-length function (\Cref{eq:keylen}).}
\[
\mu(n,k,\eps)=
\sqrt{\frac{n+k}{nk}\frac{k+1}{k}\ln\frac1\eps}.
\]
The entropy function is
\[
h_{\rm trunc}(x)=
\begin{cases}
0,&x\le0,\\
h_{\rm bin}(x),&0<x<\tfrac12,\\
1,&x\ge\tfrac12.
\end{cases}
\]
Writing
\[
\begin{aligned}
L(\eps,\bar\eps)={}&
n\bigl(q-h_{\rm trunc}(\Qtol+\mu(n,k,\eps))\bigr)\\
&-2\log_2(2\bar\eps)^{-1}-\leakec-\log_2(2/\eps_c),
\end{aligned}
\]
the theorem assumes
\[
\ell\le
\sup_{\substack{0<\eps<1,\ \bar\eps>0\\
                  \eps+\bar\eps\le\eps_s}}
\lfloor L(\eps,\bar\eps)\rfloor.
\]
\end{paracol}
\endgroup
\AlignmentBlockEnd

\AlignmentContinuationLabel{tab:alignment-security}
\AlignmentBlockHead{Composable security predicates.}
\begingroup\AlignmentBlockSetup
\begin{paracol}{2}
\code{src/qkd/bb84/def.v}: \coqi{eps_correct}
\begin{minted}{coq}
Definition eps_correct (p_mismatch eps_c : R) : Prop :=
  p_mismatch <= eps_c.
\end{minted}
\switchcolumn
\textbf{Correctness (Definition~\ref{def:correct}).}
For every adversary,
$\Pr[\hat{\mathbf S}\ne\mathbf S\wedge\text{not abort}]\le\eps_c$.
\switchcolumn*
\code{src/utils/labelled/def.v}: \coqi{trnormH};
\code{src/lhl/def.v}: \coqi{dist_uniform_feasible}, \coqi{dist_from_uniform};
\code{src/qkd/bb84/def.v}: \coqi{delta_secret}, \coqi{eps_secret}
\begin{minted}{coq}
Definition trnormH (f : 'F[H]_S) : R :=
  @complex.Re R (trnorm (h2mx f)).
Definition dist_uniform_feasible (rhoAB : 'F[H]_(A :|: B)) (v : R) : Prop :=
  exists sigmaB : 'F[H]_B,
    sigmaB \is psdlf /\
    \Tr sigmaB = \Tr (\Tr_(A) rhoAB) /\
    v = (2%:R)^-1 * trnormH
      (rhoAB - (dim 'H[H]_A)%:R^-1 *: ((\1 : 'F_A) \O sigmaB)).
Definition dist_from_uniform (rhoAB : 'F[H]_(A :|: B)) : R :=
  inf [set v : R | dist_uniform_feasible rhoAB v].
Definition delta_secret (rhoAB : 'F[HL]_S_AB) (Delta : R) : Prop :=
  dist_from_uniform rhoAB <= Delta.
Definition eps_secret (rhoAB : 'F[HL]_S_AB) (eps_s p_abort : R) : Prop :=
  exists Delta : R,
    delta_secret rhoAB Delta /\ (1 - p_abort) * Delta <= eps_s.
\end{minted}
\switchcolumn
\textbf{Secrecy and trace distance (Definition~\ref{def:secret}).}
The optimized distance is
\[
\min_{\substack{\sigma_E\succeq0\\
                  \Tr\sigma_E=\Tr\rho_E}}
\frac12\lVert\rho_{\mathbf S E}-\tau_{\mathbf S}\otimes\sigma_E\rVert_1,
\]
and secrecy multiplies it by $1-p_{\rm abort}$.
\switchcolumn*
\code{src/qkd/bb84/def.v}: \coqi{eps_secure}
\begin{minted}{coq}
Definition eps_secure (p_mismatch eps_c : R) (rhoAB : 'F[HL]_S_AB)
    (eps_s p_abort eps_tot : R) : Prop :=
  eps_correct p_mismatch eps_c /\
  eps_secret rhoAB eps_s p_abort /\
  eps_c + eps_s <= eps_tot.
\end{minted}
These three conjuncts define the combined criterion;
\coqi{eps_correct} and \coqi{eps_secret} are expanded immediately above in this table.
\switchcolumn
\textbf{Combined security (following Definitions~\ref{def:correct} and~\ref{def:secret}).}
The protocol is $\eps_{\rm tot}$-secure when it is $\eps_c$-correct and $\eps_s$-secret with $\eps_c+\eps_s\le\eps_{\rm tot}$.
\end{paracol}
\endgroup
\AlignmentBlockEnd

\AlignmentContinuationLabel{tab:alignment-arithmetic}
\AlignmentBlockHead{Arithmetic closure.}
\begingroup\AlignmentBlockSetup
\begin{paracol}{2}
\code{src/utils/arith.v}: \coqi{log2}
\begin{minted}{coq}
Definition log2 (x : R) : R := ln x / ln 2%:R.
\end{minted}
Every use of the project-local \coqi{log2} refers to this expansion;
\coqi{ln} is the MathComp real logarithm.
\switchcolumn
\textbf{Base-two logarithm.}
Throughout the appendix,
\[
\log_2x=\frac{\ln x}{\ln2}.
\]
\end{paracol}
\endgroup
\AlignmentBlockEnd

\AlignmentContinuationLabel{tab:alignment-ingredients}
\AlignmentBlockHead{Two-universal hashing.}
\begingroup\AlignmentBlockSetup
\begin{paracol}{2}
\code{src/lhl/def.v}: \coqi{collision_probability}, \coqi{delta_almost_two_univ}, \coqi{two_univ}
\begin{minted}{coq}
Definition collision_probability (F : Fidx -> X -> Z) (x x' : X) : R :=
  (#|[set f : Fidx | F f x == F f x']|%:R / #|Fidx|%:R)%R.
Definition delta_almost_two_univ (F : Fidx -> X -> Z) (delta : R) : Prop :=
  (0 < #|Fidx|)%N /\
  forall x x' : X, x != x' -> collision_probability F x x' <= delta.
Definition two_univ (F : Fidx -> X -> Z) : Prop :=
  delta_almost_two_univ F (#|Z|%:R^-1)%R.
\end{minted}
For a finite range of cardinality $2^\ell$, $|Z|^{-1}=2^{-\ell}$.
\switchcolumn
\textbf{Two-universal hashing (Definition~\ref{def:two-universal}).}
For a nonempty seed family and $x\ne x'$, a uniform seed satisfies
$\Pr[f(x)=f(x')]\le2^{-\ell}$.
\end{paracol}
\endgroup
\AlignmentBlockEnd

\AlignmentContinuationLabel{tab:alignment-main-side-conditions}
\AlignmentBlockHead{Mismatch and numeric side conditions.}
\begingroup\AlignmentBlockSetup
\begin{paracol}{2}
\code{src/qkd/bb84/def.v}: \coqi{verification_collision_prob},
\coqi{p_mismatch}
\begin{minted}{coq}
Definition verification_collision_prob (x x' : 'I_key_count) : R :=
  #|[set f : 'I_seed_count | hash f x == hash f x']|%:R /
    seed_count%:R.
Definition p_mismatch : R :=
  \sum_x \sum_(x' | x != x') corr_prob x x' *
    verification_collision_prob x x'.
\end{minted}
The two-universal bound on $v(x,x')$ is the predicate expanded in
\Cref{tab:alignment-ingredients}.
\switchcolumn
\textbf{Verification-weighted mismatch probability.}
For a uniform verification seed,
\[
\begin{aligned}
v(x,x')&=\Pr_f[h_f(x)=h_f(x')],\\
p_{\rm mismatch}
  &=\sum_x\sum_{x'\ne x}p(x,x')\,v(x,x').
\end{aligned}
\]
\switchcolumn*
\code{src/qkd/bb84/model.v}: \coqi{bb84_pass_prob_le},
\coqi{bb84_numeric_side}
\begin{minted}{coq}
Definition bb84_pass_prob_le (p_abort p_pass : R) : Prop :=
  0 <= p_abort /\ p_abort <= 1 /\ 1 - p_abort <= p_pass.
Record bb84_numeric_side (P : bb84_protocol) := BB84NumericSide {
  bb84_epsc_pos : 0 < bb84_eps_c P;
  bb84_epsc_le1 : bb84_eps_c P <= 1;
  bb84_leak_nonneg : 0 <= bb84_leak_EC P
}.
\end{minted}
\switchcolumn
\textbf{Abort/pass relation and numeric domain.}
\[
0\le p_{\rm abort}\le1,\qquad
1-p_{\rm abort}\le p_{\rm pass},
\]
with $0<\eps_c\le1$ and $\leakec\ge0$.
\end{paracol}
\endgroup
\AlignmentBlockEnd

\AlignmentContinuationLabel{tab:alignment-post-hash-definitions}
\AlignmentBlockHead{CQ and post-hash states.}
\begingroup\AlignmentBlockSetup
\begin{paracol}{2}
\code{src/lhl/def.v}: \coqi{is_cq_state}, \coqi{cq_block}, \coqi{post_hash_state}
\begin{minted}{coq}
Definition is_cq_state (rhoXE : 'F[HL]_(X :|: E)) : bool :=
  [forall x, forall x', (x' != x) ==>
    ((([> deltav x ; deltav x <] \O (\1 : 'F_E)) \o rhoXE \o
      ([> deltav x' ; deltav x' <] \O (\1 : 'F_E))) == 0)].
Definition cq_block x (rhoXE : 'F[HL]_(X :|: E)) : 'F[HL]_E :=
  ptrace_castl disXE
    (([> deltav x ; deltav x <] \O (\1 : 'F_E)) \o rhoXE \o
     ([> deltav x ; deltav x <] \O (\1 : 'F_E))).
Definition post_hash_state Fh rhoXE :
    'F[HL]_(Zl :|: (Fl :|: E)) :=
  let eqZFE := esym (finset.setUA Zl Fl E) in
  castlf (eqZFE, eqZFE) ((dim 'H[HL]_Fl)%:R^-1 *:
    \sum_f \sum_z
      [> deltav z ; deltav z <] \O [> deltav f ; deltav f <] \O
      (\sum_(x | Fh f x == z) cq_block disXE x rhoXE)).
\end{minted}
These definitions supply the state constructors used by every LHL and
privacy-amplification row below.
\switchcolumn
\textbf{Cq input and extracted joint state.}
The cq condition and its blocks are
\[
\begin{aligned}
\rho_{XE}
  &=\sum_x|x\rangle\!\langle x|\otimes\rho_E^{[x]},\\
\rho_E^{[x]}
  &=\Tr_X\!\left[
    (|x\rangle\!\langle x|\otimes I_E)\rho_{XE}\right.\\
  &\hspace{4.3em}\left.
    (|x\rangle\!\langle x|\otimes I_E)\right].
\end{aligned}
\]
For a uniform seed $f$,
\[
\begin{aligned}
\rho_{ZFE}
  &=\frac1{|\mathcal F|}\sum_{f,z}
    |z\rangle\!\langle z|\otimes|f\rangle\!\langle f|\\
  &\hspace{4.3em}\otimes
    \sum_{x:f(x)=z}\rho_E^{[x]}.
\end{aligned}
\]
\end{paracol}
\endgroup
\AlignmentBlockEnd

\AlignmentContinuationLabel{tab:alignment-public-leakage}
\AlignmentBlockHead{Public leakage certificate.}
\begingroup\AlignmentBlockSetup
\begin{paracol}{2}
\code{src/qkd/bb84/def.v}: \coqi{bb84_public_assoc}, \coqi{bb84_public_leakage}, \coqi{bb84_public_state}, \coqi{bb84_public_leak_bound}
\begin{minted}{coq}
Definition bb84_public_assoc (Ppub : {set L}) :
  Ppub :|: (S :|: E) = S :|: (E :|: Ppub).
Record bb84_public_leakage
    (Ppub : {set L}) (rhoXE : 'F[HL]_(S :|: E)) (leak : R) :=
  BB84PublicLeakage {
  bb84_public_disSP : [disjoint S & Ppub];
  bb84_public_disEP : [disjoint E & Ppub];
  bb84_public_disP_SE : [disjoint Ppub & (S :|: E)];
  bb84_public_state : 'F[HL]_(Ppub :|: (S :|: E));
  bb84_public_state_psd : bb84_public_state \is psdlf;
  bb84_public_state_tr :
    @complex.Re R (\Tr bb84_public_state) <= 1;
  bb84_public_state_classical : forall x x', x' != x ->
    (([> deltav x ; deltav x <] \O (\1 : 'F_(S :|: E))) \o
      bb84_public_state \o
      ([> deltav x' ; deltav x' <] \O (\1 : 'F_(S :|: E)))) = 0;
  bb84_public_state_marginal :
    ptrace_castl bb84_public_disP_SE bb84_public_state = rhoXE;
  bb84_public_state_pa_cq :
    let eq_pub_assoc := bb84_public_assoc Ppub in
    @is_cq_state R L HL S (E :|: Ppub)
      (castlf (eq_pub_assoc, eq_pub_assoc) bb84_public_state);
  bb84_public_leak_bound :
    log2 (dim 'H[HL]_Ppub)%:R <= leak
}.
\end{minted}
The cq predicate used by this record is defined in \Cref{tab:alignment-post-hash-definitions}; this row closes the semantic dependency used by \coqi{bb84_secret} and the leakage-chain lemmas below.
\switchcolumn
\textbf{Public classical leakage.}
A leakage certificate packages a state $\rho_{PSE}$ such that
\[
\begin{aligned}
\rho_{PSE}&\succeq0, & \Tr\rho_{PSE}&\le1,\\
\Tr_P\rho_{PSE}&=\rho_{SE}, &
\rho_{PSE}&\text{ is classical on }P,
\end{aligned}
\]
\[
\begin{aligned}
\rho_{S(EP)}&\text{ is cq on }S,\\
\log_2 d_P&\le\mathsf{leak}.
\end{aligned}
\]
\end{paracol}
\endgroup
\AlignmentBlockEnd

\AlignmentParagraph{Source.}

\AlignmentContinuationLabel{tab:alignment-source-certificate}
\AlignmentBlockHead{Source certificate.}
\begingroup\AlignmentBlockSetup
\begin{paracol}{2}
\code{src/eur/def.v}: \coqi{is_povm_mx};
\code{src/qkd/bb84/model.v}: \coqi{bb84_source_model}
\begin{minted}{coq}
Definition is_povm_mx (M : 'I_nx -> 'M[C]_na) : bool :=
  [forall x, M x \is psdmx] && (\sum_(x < nx) M x == 1%:M).
Record bb84_source_model (P : bb84_protocol) (q : R) :=
    BB84SourceModel {
  bb84_src_na : nat;
  bb84_src_nb : nat;
  bb84_src_nc : nat;
  bb84_src_nx : nat;
  bb84_src_nz : nat;
  bb84_src_rho : 'M[C]_(bb84_src_na * bb84_src_nb * bb84_src_nc);
  bb84_src_M : 'I_bb84_src_nx -> 'M[C]_bb84_src_na;
  bb84_src_N : 'I_bb84_src_nz -> 'M[C]_bb84_src_na;
  bb84_src_den1 : bb84_src_rho \is den1mx;
  bb84_src_M_povm : is_povm_mx bb84_src_M;
  bb84_src_N_povm : is_povm_mx bb84_src_N;
  bb84_src_overlap_quality :
    (bb84_n P)%:R * q <= q_of_mx bb84_src_M bb84_src_N;
  bb84_src_dim_Z : bb84_src_nz = (2 ^ bb84_n P)%N;
  bb84_src_dim_X : bb84_src_nx = (2 ^ bb84_n P)%N
}.
\end{minted}
\coqi{q_of_mx} is expanded in \Cref{tab:alignment-main-statement-bridges}, and
\coqi{bb84_protocol} in \Cref{tab:alignment-bb84}.
\switchcolumn
\textbf{Finite-dimensional gedanken source.}
A source certificate supplies
\[
\begin{gathered}
\rho_{ABC}\succeq0,\quad\Tr\rho_{ABC}=1,\\
\{M_x\}_x,\{N_z\}_z\text{ POVMs},
\end{gathered}
\]
\[
\begin{gathered}
nq\le-\log_2
 \max_{x,z}\|\sqrt{M_x}\sqrt{N_z}\|_\infty^2,\\
|X|=|Z|=2^n.
\end{gathered}
\]
The side-information dimensions are arbitrary.
This source certificate supplies the EUR inputs.
\end{paracol}
\endgroup
\AlignmentBlockEnd

\AlignmentContinuationLabel{tab:alignment-sampling-objects}
\AlignmentBlockHead{Finite sampling objects.}
\begingroup\AlignmentBlockSetup
\begin{paracol}{2}
\code{src/utils/bitstr.v}: \coqi{bitstr}, \coqi{wt_ham}, \coqi{xor_str}, \coqi{d_ham}, \coqi{ord_to_bitstring}, \coqi{bitstring_to_ord};
\code{src/utils/prob.v}: \coqi{perm_pair_sample}, \coqi{perm_pair_holdout}, \coqi{perm_pair_sample_count}, \coqi{perm_pair_sample_err}, \coqi{perm_pair_sample_le_event}
\begin{minted}{coq}
Definition bitstr (n : nat) := 'I_n -> bool.
Definition wt_ham {n : nat} (x : bitstr n) : nat :=
  #|[set i : 'I_n | x i]|.
Definition xor_str {n : nat} (x y : bitstr n) : bitstr n :=
  fun i => x i (+) y i.
Definition d_ham {n : nat} (x y : bitstr n) : nat :=
  wt_ham (xor_str x y).
Local Definition bitstring_to_tuple {n} (s : bitstr n) : n.-tuple bool :=
  [tuple s i | i < n].
Local Definition tuple_to_bitstring {n} (t : n.-tuple bool) : bitstr n :=
  fun i => tnth t i.
Definition ord_to_bitstring {n} (x : 'I_(2 ^ n)) : bitstr n :=
  tuple_to_bitstring (ord2bseq x).
Definition bitstring_to_ord {n} (s : bitstr n) : 'I_(2 ^ n) :=
  bseq2ord (bitstring_to_tuple s).
Definition perm_pair_sample (z : bitstr N) (pi : 'S_N) : bitstr k :=
  fun i => z (pi (@lshift k n i)).
Definition perm_pair_holdout (z : bitstr N) (pi : 'S_N) : bitstr n :=
  fun i => z (pi (@rshift k n i)).
Definition perm_pair_sample_count (z z' : bitstr N) (pi : 'S_N) : nat :=
  d_ham (perm_pair_sample z pi) (perm_pair_sample z' pi).
Definition perm_pair_sample_err (z z' : bitstr N) (pi : 'S_N) : R :=
  (perm_pair_sample_count z z' pi)%:R / k%:R.
Definition perm_pair_sample_le_event (z z' : bitstr N) (Qtol : R) :=
  [set pi : 'S_N | perm_pair_sample_err z z' pi <= Qtol].
\end{minted}
The dependencies are MathComp finite ordinals, tuples, permutations, shifts,
and its binary-sequence codec.
\switchcolumn
\textbf{Hamming sampling under a random permutation.}
For $z,z'\in\{0,1\}^{k+n}$ and $\pi\in S_{k+n}$, define
\[
\begin{aligned}
z_{\rm smp}(i)&=z_{\pi(i)}, & i&<k,\\
z_{\rm hold}(j)&=z_{\pi(k+j)}, & j&<n,\\
\widehat e_\pi(z,z')
  &=k^{-1}d_H(z_{\rm smp},z'_{\rm smp}),\\
\mathcal A_{\Qtol}(z,z')
  &=\{\pi:\widehat e_\pi(z,z')\le\Qtol\}.
\end{aligned}
\]
The ordinal/bitstring bijection is the finite basis-index encoding used by the state construction.
\end{paracol}
\endgroup
\AlignmentBlockEnd

\AlignmentContinuationLabel{tab:alignment-bb84-pass-probability}
\AlignmentBlockHead{Full-string source and pass probability.}
\begingroup\AlignmentBlockSetup
\begin{paracol}{2}
\code{src/utils/prob.v}: \coqi{Pr_event}, \coqi{subset_uniform_mu}, \coqi{subset_rep_tuple}, \coqi{subset_sorted_rep}, \coqi{perm_pair_sample_le_event};
\code{src/qkd/bb84/model.v}: \coqi{bb84_W}, \coqi{bb84_p_pass_pair}, \coqi{bb84_p_pass_tot}
\begin{minted}{coq}
Definition bb84_total_rounds (P : bb84_protocol) : nat :=
  bb84_k P + bb84_n P.
Definition bb84_full_string_count (P : bb84_protocol) : nat :=
  2 ^ bb84_total_rounds P.
Record bb84_zsource (P : bb84_protocol) := BB84ZSource {
  bb84_zsrc_nb : nat;
  bb84_zsrc_rho :
    'M[C]_(bb84_full_string_count P * bb84_zsrc_nb *
           bb84_full_string_count P);
  bb84_zsrc_den1 : bb84_zsrc_rho \is den1mx
}.
Definition Pr_event (mu : T -> R) (A : {set T}) : R :=
  \sum_(t in A) mu t.
Definition subset_uniform_mu (S : {set 'I_N}) : R :=
  (#|S| == k)%:R * (('C(N, k))%:R)^-1.
Definition subset_rep_tuple (S : {set 'I_N}) : N.-tuple 'I_N :=
  Tuple (subset_rep_seq_size S).
Definition subset_sorted_rep (S : {set 'I_N}) : 'S_N :=
  perm (@subset_rep_tuple_tnth_inj S).
Definition perm_pair_sample_le_event (z z' : bitstr N) (Qtol : R) :=
  [set pi : 'S_N | perm_pair_sample_err z z' pi <= Qtol].
Definition bb84_W (P : bb84_protocol) (G : bb84_zsource P) z z' : R :=
  let Pzz := (delta_mx z z *t 1%:M *t delta_mx z' z') in
  @complex.Re R (\tr (Pzz *m bb84_zsrc_rho G)).
Definition bb84_p_pass_pair (P : bb84_protocol) z z' : R :=
  Pr_event (fun S0 => subset_uniform_mu S0)
    [set S0 | subset_sorted_rep S0 \in
      perm_pair_sample_le_event (ord_to_bitstring z)
        (ord_to_bitstring z') (bb84_Qtol P)].
Definition bb84_p_pass_tot (P : bb84_protocol) (G : bb84_zsource P) : R :=
  \sum_z \sum_z' bb84_W G z z' * bb84_p_pass_pair z z'.
\end{minted}
The complete definitions of \coqi{perm_pair_sample_err}, its Hamming-distance
chain, and the ordinal/bitstring encodings are expanded in
\Cref{tab:alignment-sampling-objects}.
The \coqi{bb84_protocol} record is expanded in
\Cref{tab:alignment-bb84}.
\switchcolumn
\textbf{Probability objects consumed by \coqi{stat}.}
For a full-string source $G$,
\[
\begin{aligned}
W_G(z,z')
  &=\Tr\!\bigl[
    (|z\rangle\!\langle z|\otimes I\\
  &\hspace{5.2em}\otimes|z'\rangle\!\langle z'|)\rho_G
    \bigr],
\end{aligned}
\]
\[
\begin{aligned}
p_{\rm pass}(z,z')
  &=\sum_{S:|S|=k}\binom{n+k}{k}^{-1}
    \mathbf 1[\widehat e_S(z,z')\le\Qtol],\\
p_{\rm pass}^{\rm tot}
  &=\sum_{z,z'}W_G(z,z')p_{\rm pass}(z,z').
\end{aligned}
\]
For a subset $S\subseteq\{0,\ldots,N-1\}$, the canonical permutation lists
$S$ in increasing order followed by its complement:
\[
t_S=\operatorname{enum}(S)\mathbin{+\!+}\operatorname{enum}(S^c),
\qquad \pi_S=\operatorname{perm}(t_S).
\]
\end{paracol}
\endgroup
\AlignmentBlockEnd

\AlignmentContinuationLabel{tab:alignment-bb84-pass-state}
\AlignmentBlockHead{Conditional full-string pass state.}
\begingroup\AlignmentBlockSetup
\begin{paracol}{2}
\code{src/qkd/bb84/model.v}: \coqi{bb84_pair_projector},
\coqi{bb84_holdout_pair_mx}, \coqi{bb84_rhoZZ_pass_mix},
\coqi{bb84_rhoZZ_pass_mixH}
\begin{minted}{coq}
Definition bb84_pair_projector (n : nat) (x x' : 'I_(2 ^ n)) :=
  (delta_mx x x : 'M[C]_(2 ^ n)) *t
  (delta_mx x' x' : 'M[C]_(2 ^ n)).
Definition bb84_holdout_pair_mx (n k : nat)
    (z z' : bitstr (k + n)) (pi : 'S_(k + n)) :=
  bb84_pair_projector
    (bitstring_to_ord (perm_pair_holdout z pi))
    (bitstring_to_ord (perm_pair_holdout z' pi)).
Definition bb84_rhoZZ_pass_mix (P : bb84_protocol) (G : bb84_zsource P) :=
  \sum_z \sum_z'
    \sum_(S0 | subset_sorted_rep S0 \in
      perm_pair_sample_le_event (ord_to_bitstring z)
        (ord_to_bitstring z') (bb84_Qtol P))
      (((bb84_W G z z' * subset_uniform_mu S0 /
        bb84_p_pass_tot G)%:C) *:
        bb84_holdout_pair_mx (ord_to_bitstring z)
          (ord_to_bitstring z') (subset_sorted_rep S0)).
Definition bb84_rhoZZ_pass_mixH P G dim_Zk dim_Zp :=
  biprod_of_mx disZZp
    (castmx (f_equal2 muln dim_Zk dim_Zp,
             f_equal2 muln dim_Zk dim_Zp)
      (bb84_rhoZZ_pass_mix G)).
\end{minted}
The weights, canonical subset sampler, pass event, and normalization denominator
are expanded in \Cref{tab:alignment-bb84-pass-probability};
\coqi{perm_pair_holdout} and the ordinal/bitstring encodings in
\Cref{tab:alignment-sampling-objects}.
\switchcolumn
\textbf{Normalized holdout-pair mixture.}
The conditional state is
\[
\begin{aligned}
\rho_{ZZ'|\rm pass}
  &=\frac1{p_{\rm pass}^{\rm tot}}
    \sum_{\substack{z,z',S:\\ \widehat e_S\le\Qtol}}
      W_G(z,z')\binom{n+k}{k}^{-1}\\
  &\hspace{2em}\cdot
    |z_{\bar S},z'_{\bar S}\rangle
    \!\langle z_{\bar S},z'_{\bar S}|.
\end{aligned}
\]
\end{paracol}
\endgroup
\AlignmentBlockEnd

\AlignmentContinuationLabel{tab:alignment-main-statement-bridges}
\AlignmentBlockHead{Overlap and measured states.}
\begingroup\AlignmentBlockSetup
\begin{paracol}{2}
\code{src/utils/mx/def.v}: \coqi{sqrtmx};
\code{src/eur/def.v}: \coqi{overlap_mx}, \coqi{q_of_mx},
\coqi{rho_XB}, \coqi{rho_ZC}, \coqi{eur_rho_XB}, \coqi{eur_rho_ZC}
\begin{minted}{coq}
Definition sqrtmx {n : nat} (A : 'M[C]_n) : 'M[C]_n :=
  let U := spectralmx A in
  let eigenvals := spectral_diag A in
  let sqrt_eigenvals :=
    \row_i (Num.sqrt (@complex.Re R (eigenvals 0 i)))%:C in
  U^*t *m diag_mx sqrt_eigenvals *m U.
Definition overlap_mx
    (M : 'I_nx -> 'M[C]_na)
    (N : 'I_nz -> 'M[C]_na) : R :=
  \big[maxr/0]_(x < nx)
    \big[maxr/0]_(z < nz)
      (@complex.Re R
        (i2norm (sqrtmx (M x) *m sqrtmx (N z)))) ^ 2.
Definition q_of_mx
    (M : 'I_nx -> 'M[C]_na)
    (N : 'I_nz -> 'M[C]_na) : R :=
  - log2 (overlap_mx M N).
Definition rho_XB
    (M : 'I_nx -> 'M[C]_na)
    (rhoABC : 'M[C]_(na * nb * nc)) : 'M[C]_(nx * nb) :=
  \sum_(x < nx) (delta_mx x x *t
    ptrace1 (ptrace2
      ((sqrtmx (M x) *t 1%:M *t 1%:M) *m rhoABC *m
       (sqrtmx (M x) *t 1%:M *t 1%:M)))).
Definition rho_ZC
    (N : 'I_nz -> 'M[C]_na)
    (rhoABC : 'M[C]_(na * nb * nc)) : 'M[C]_(nz * nc) :=
  \sum_(z < nz) (delta_mx z z *t
    ptrace1
      ((sqrtmx (N z) *t 1%:M *t 1%:M) *m rhoABC *m
       (sqrtmx (N z) *t 1%:M *t 1%:M))).
Definition eur_rho_XB
    (X B : {set L}) (disXB : [disjoint X & B])
    (na nb nc nx : nat)
    (rhoSrc : 'M[C]_(na * nb * nc))
    (M : 'I_nx -> 'M[C]_na)
    (dim_eq : nx * nb = dim ('H[HL]_(X :|: B)))
    : 'F[HL]_(X :|: B) :=
  let Q := castmx (tenv_dim disXB, erefl)
    (@tenv_mxU R L HL X B) in
  let eq_prod := etrans dim_eq (tenv_dim disXB) in
  mx2h (castmx
    (esym (tenv_dim disXB), esym (tenv_dim disXB))
    (Q *m castmx (eq_prod, eq_prod) (rho_XB M rhoSrc)
       *m Q^*t)).
Definition eur_rho_ZC
    (Z C0 : {set L}) (disZC : [disjoint Z & C0])
    (na nb nc nz : nat)
    (rhoSrc : 'M[C]_(na * nb * nc))
    (N : 'I_nz -> 'M[C]_na)
    (dim_eq : nz * nc = dim ('H[HL]_(Z :|: C0)))
    : 'F[HL]_(Z :|: C0) :=
  let Q := castmx (tenv_dim disZC, erefl)
    (@tenv_mxU R L HL Z C0) in
  let eq_prod := etrans dim_eq (tenv_dim disZC) in
  mx2h (castmx
    (esym (tenv_dim disZC), esym (tenv_dim disZC))
    (Q *m castmx (eq_prod, eq_prod) (rho_ZC N rhoSrc)
       *m Q^*t)).
\end{minted}
\switchcolumn
\textbf{Overlap and post-measurement states.}
For POVMs $M=\{M_x\}$ and $N=\{N_z\}$,
\[
\begin{aligned}
\sqrt A&=U^\dagger
  \operatorname{diag}\!\bigl(\sqrt{\operatorname{Re}\lambda_i}\bigr)U,
  \quad A=U^\dagger\operatorname{diag}(\lambda_i)U,\\
c(M,N)&=\max_{x,z}\lVert\sqrt{M_x}\sqrt{N_z}\rVert_\infty^2,\\
q(M,N)&=-\log_2 c(M,N).
\end{aligned}
\]
Measuring $A$ in $\rho_{ABC}$ gives
\[
\rho_{XB}=\sum_x |x\rangle\!\langle x|\otimes
 \Tr_{AC}(\sqrt{M_x}\rho_{ABC}\sqrt{M_x}),
\]
and analogously $\rho_{ZC}$ for $N$.
The labelled wrappers combine register casts with the matrix states expanded
immediately above.
\end{paracol}
\endgroup
\AlignmentBlockEnd

\AlignmentContinuationLabel{tab:alignment-main-dimension-transports}
\AlignmentBlockHead{Register-dimension transports.}
\begingroup\AlignmentBlockSetup
\begin{paracol}{2}
\code{src/qkd/bb84/bb84.v}: \coqi{bb84_register_pair_dim},
\coqi{bb84_source_zkey_dim}
\begin{minted}{coq}
Definition bb84_register_pair_dim (A B : {set L})
    (disAB : [disjoint A & B]) (na nb : nat)
    (dim_A : na = dim ('H[HL]_A))
    (dim_B : nb = dim ('H[HL]_B)) :
  na * nb = dim ('H[HL]_(A :|: B)).
Definition bb84_source_zkey_dim (P : bb84_protocol) (q : R)
    (src : bb84_source_model P q) (Z : {set L})
    (dim_Z : @bb84_src_nz R P q src = dim ('H[HL]_Z)) :
  (2 ^ bb84_n P)%N = dim ('H[HL]_Z).
\end{minted}
These equalities supply the two project-level arguments passed to
\coqi{eur_rho_XB}, \coqi{eur_rho_ZC}, and
\coqi{bb84_rhoZZ_pass_mixH}.
\switchcolumn
\textbf{Dimension transports used in measured-state arguments.}
For disjoint registers $A,B$,
\[
\begin{gathered}
n_A=\dim\mathcal H_A,\quad n_B=\dim\mathcal H_B\\
\Longrightarrow n_An_B=\dim\mathcal H_{A\cup B}.
\end{gathered}
\]
The source certificate also transports $n_Z=2^n$ and
$n_Z=\dim\mathcal H_Z$ to $2^n=\dim\mathcal H_Z$.
\end{paracol}
\endgroup
\AlignmentBlockEnd

\AlignmentContinuationLabel{tab:alignment-main-delayed-definition}
\AlignmentBlockHead{Delayed measurement.}
\begingroup\AlignmentBlockSetup
\begin{paracol}{2}
\code{src/qkd/bb84/def.v}: \coqi{bb84_delayed_measurement_condition}
\begin{minted}{coq}
Record bb84_delayed_measurement_condition
    (BobZ : 'QC[HL]_(Bq, Zp))
    (rhoZB : 'F[HL]_(Zkey :|: Bq))
    (rhoZZp : 'F[HL]_(Zkey :|: Zp)) := {
  bb84_dmc_full : 'F[HL]_(Zkey :|: Bq :|: finset.set0);
  bb84_dmc_full_den1 : bb84_dmc_full \is den1lf;
  bb84_dmc_source :
    ptrace_cast3C (disjointX0 Zkey) (disjointX0 Bq)
      bb84_dmc_full = rhoZB;
  bb84_dmc_target :
    ptrace_cast3C (disjointX0 Zkey) (disjointX0 Zp)
      ((((\:1 : 'SO[HL]_Zkey) :⊗ (BobZ : 'SO[HL]_(Bq, Zp)))
        :⊗ (\:1 : 'SO[HL]_finset.set0)) bb84_dmc_full) = rhoZZp
}.
\end{minted}
The corresponding $H_{\max}$ theorem appears in
\Cref{tab:alignment-bb84-bridges}.
\switchcolumn
\textbf{Delayed-measurement condition.}
A delayed channel $\mathcal B_Z:B\to Z'$ is certified by a normalized common
extension $\omega_{ZB}$ such that
\[
 \Tr_{\varnothing}\omega_{ZB}=\rho_{ZB},\qquad
 \Tr_{\varnothing}\bigl((\id_Z\otimes\mathcal B_Z)\omega_{ZB}\bigr)
   =\rho_{ZZ'}.
\]
\end{paracol}
\endgroup
\AlignmentBlockEnd

\AlignmentContinuationLabel{tab:alignment-statement-transports}
\AlignmentBlockHead{Labelled partial trace and matrix transport.}
\begingroup\AlignmentBlockSetup
\begin{paracol}{2}
\code{src/utils/labelled/def.v}: \coqi{ptrace_cast_gen},
\coqi{ptrace_castl}, \coqi{ptrace_cast3C}, \coqi{biprod_of_mx}
\begin{minted}{coq}
Definition ptrace_cast_gen (S T B : {set L})
    (eq : S :\: B = T) (f : 'F[H]_S) : 'F[H]_T :=
  castlf (eq, eq) (\Tr_(B) f).
Definition ptrace_castl (A B : {set L}) (disAB : [disjoint A & B])
    (f : 'F[H]_(A :|: B)) : 'F[H]_B :=
  ptrace_cast_gen (setDU_l disAB) f.
Definition ptrace_cast3C (A B C : {set L})
    (disAC : [disjoint A & C]) (disBC : [disjoint B & C])
    (f : 'F[H]_(A :|: B :|: C)) : 'F[H]_(A :|: B) :=
  ptrace_cast_gen (setDU3_C disAC disBC) f.
Definition biprod_of_mx
    (A B : {set L}) (disAB : [disjoint A & B])
    (rhoAB : 'M[C]_(dim ('H[H]_A) * dim ('H[H]_B)))
    : 'F[H]_(A :|: B) :=
  let Q : 'M[C]_(dim 'H[H]_A * dim 'H[H]_B) :=
    castmx (tenv_dim disAB, erefl) (@tenv_mxU R L H A B) in
  mx2h (castmx (esym (tenv_dim disAB), esym (tenv_dim disAB))
    (Q *m rhoAB *m Q^*t)).
\end{minted}
\switchcolumn
\textbf{Label-preserving transports.}
For $S\setminus B=T$,
\[
  \operatorname{ptrace\_cast}_{S\to T}^{B}(\rho)
  =\operatorname{cast}_{S\setminus B=T}(\Tr_B\rho).
\]
The two convenience forms retain respectively the right component of
$A\cup B$ and the first two components of $A\cup B\cup C$.
The inverse bipartite transport maps a matrix on
$\mathcal H_A\otimes\mathcal H_B$ to the labelled space
$\mathcal H_{A\cup B}$.
\end{paracol}
\endgroup
\AlignmentBlockEnd

\subsubsection{Protocol Closure}

\AlignmentParagraph{Program.}

\AlignmentContinuationLabel{tab:alignment-semantics-core}
\AlignmentBlockHead{Command syntax.}
\begingroup\AlignmentBlockSetup
\begin{paracol}{2}
\code{src/qkd/protocol/syntax.v}: \coqi{store}, \coqi{store_set},
\coqi{expr}, \coqi{bexpr}, \coqi{dexpr}, \coqi{qexpr}, \coqi{mexpr},
\coqi{qkd_cmd}, \coqi{seq_cmd}, \coqi{for_cmd}
\begin{minted}{coq}
Definition store := {ffun Cv -> V}.
Definition store_set (s : store) (x : Cv) (v : V) : store :=
  [ffun y => if y == x then v else s y].
Definition expr := store -> V.
Definition bexpr := store -> bool.
Definition dexpr := store -> V -> R.
Definition qexpr (Q : {set L}) := store -> 'F[HL]_Q.
Definition mexpr (Q : {set L}) := store -> V -> 'F[HL]_Q.
Inductive qkd_cmd : Type :=
  | Cskip
  | Cassign (x : Cv) (e : expr)
  | Csample (x : Cv) (mu : dexpr)
  | Cinit (Q : {set L}) (rho : qexpr Q)
  | Cmeasure (x : Cv) (Q : {set L}) (M : mexpr Q)
  | Cattack (Q : {set L})
  | Cif (b : bexpr) (c1 c2 : qkd_cmd)
  | Cseq (c1 c2 : qkd_cmd).
Fixpoint seq_cmd (c1 c2 : qkd_cmd) : qkd_cmd := match c1 with
  | Cskip => c2
  | Cseq c11 c12 => Cseq c11 (seq_cmd c12 c2)
  | _ => if c2 is Cskip then c1 else Cseq c1 c2
  end.
Definition for_cmd {limit : nat} (body : 'I_limit -> qkd_cmd) : qkd_cmd :=
  \big[seq_cmd/Cskip]_(i <- enum 'I_limit) body i.
\end{minted}
\switchcolumn
\textbf{Command syntax (Definition~\ref{def:protocol-cmd}).}
The grammar and fixed-loop expansion are
\[
\begin{aligned}
c::={}&\mathsf{skip}\mid x:=e\mid x\gets_{\$}\mu\\
&\mid Q\leftarrow\rho\mid x\leftarrow M[Q]\\
&\mid\mathsf{attack}[Q]\\
&\mid\mathsf{if}\ b\ \mathsf{then}\ c_1
  \ \mathsf{else}\ c_2\\
&\mid c_1;c_2,\\
\mathsf{for}_{i<m}c_i
  & =\mathop{;}_{i\in\{0,\ldots,m-1\}}c_i.
\end{aligned}
\]
The classical store is a finite map $\Sigma=\mathsf{Cv}\to\mathsf V$;
expressions, distributions, states, and measurements are store-indexed functions of the types shown at right.
\end{paracol}
\endgroup
\AlignmentBlockEnd

\AlignmentParagraph{Registers.}

\AlignmentContinuationLabel{tab:alignment-program-schema}
\AlignmentBlockHead{Concrete program variables.}
\begingroup\AlignmentBlockSetup
\begin{paracol}{2}
\code{src/qkd/example/program.v}: \coqi{bb84_round_cvar},
\coqi{bb84_cvar}, and all \coqi{BV*} locations
\begin{minted}{coq}
Definition bb84_round_cvar (M : nat) :=
  ('I_M + ('I_M + ('I_M + ('I_M + 'I_M))))%type.
Definition bb84_cvar (M : nat) := (bb84_round_cvar M + 'I_18)%type.
Definition BVa {M} (i : 'I_M) : bb84_cvar M := inl (inl i).
Definition BVy {M} (i : 'I_M) : bb84_cvar M := inl (inr (inl i)).
Definition BVb {M} (i : 'I_M) : bb84_cvar M :=
  inl (inr (inr (inl i))).
Definition BVy' {M} (i : 'I_M) : bb84_cvar M :=
  inl (inr (inr (inr (inl i)))).
Definition BVd {M} (i : 'I_M) : bb84_cvar M :=
  inl (inr (inr (inr (inr i)))).
Definition BVP {M} : bb84_cvar M := inr (inord 0).
Definition BVX {M} : bb84_cvar M := inr (inord 1).
Definition BVZ {M} : bb84_cvar M := inr (inord 2).
Definition BVK {M} : bb84_cvar M := inr (inord 3).
Definition BVXA {M} : bb84_cvar M := inr (inord 4).
Definition BVXB {M} : bb84_cvar M := inr (inord 5).
Definition BVZA {M} : bb84_cvar M := inr (inord 6).
Definition BVZB {M} : bb84_cvar M := inr (inord 7).
Definition BVlam {M} : bb84_cvar M := inr (inord 8).
Definition BVpec {M} : bb84_cvar M := inr (inord 9).
Definition BVXhat {M} : bb84_cvar M := inr (inord 10).
Definition BVg {M} : bb84_cvar M := inr (inord 11).
Definition BVtA {M} : bb84_cvar M := inr (inord 12).
Definition BVtB {M} : bb84_cvar M := inr (inord 13).
Definition BVf {M} : bb84_cvar M := inr (inord 14).
Definition BVS {M} : bb84_cvar M := inr (inord 15).
Definition BVSh {M} : bb84_cvar M := inr (inord 16).
Definition BVabort {M} : bb84_cvar M := inr (inord 17).
\end{minted}
The five-way and eighteen-way sums make all locations distinct by
construction.
\switchcolumn
\textbf{Finite classical-variable layout.}
Each round has locations for Alice's basis and bit, Bob's basis and outcome,
and the detection flag.  Eighteen global locations hold
\[
\begin{gathered}
P,X,Z,K,X_A,X_B,Z_A,Z_B,\lambda,\\
p_{\rm EC},\widehat X,
g,t_A,t_B,f,S,\widehat S,\mathsf{abort}.
\end{gathered}
\]
\end{paracol}
\endgroup
\AlignmentBlockEnd

\AlignmentContinuationLabel{tab:alignment-value-universe}
\AlignmentBlockHead{Finite value universe and constructors.}
\begingroup\AlignmentBlockSetup
\begin{paracol}{2}
\code{src/qkd/example/program.v}: \coqi{bb84_val} and its value constructors
\begin{minted}{coq}
Definition bb84_partial_bits (M : nat) := {ffun 'I_M -> option bool}.
Definition bb84_round_transcript (M : nat) :=
  {ffun 'I_M -> option (bool * bool * bool)}.
Definition bb84_transcript
    (M ec_bits vseed_count tag_count paseed_count : nat) :=
  (bb84_round_transcript M *
   option ({set 'I_M} * bb84_partial_bits M * bb84_partial_bits M) *
   option (option 'I_(2 ^ ec_bits) * option 'I_vseed_count *
           option 'I_tag_count) * option (option 'I_paseed_count))%type.
Definition bb84_eur_public (M vseed_count : nat) :=
  (bb84_round_transcript M * {set 'I_M} * bb84_partial_bits M *
   bb84_partial_bits M * option 'I_vseed_count)%type.
Definition bb84_error_value (M : nat) := ('I_M.+1 * 'I_M.+1)%type.
Definition bb84_val (M n l ec_bits vseed_count tag_count paseed_count : nat) :=
  (bool + (bool + (bool + (option bool + ({set 'I_M} +
  (bb84_partial_bits M + (bb84_error_value M +
  (bb84_transcript M ec_bits vseed_count tag_count paseed_count +
  (option 'I_(2 ^ ec_bits) + (option 'I_vseed_count +
  (option 'I_tag_count + (option 'I_paseed_count +
  (option 'I_(2 ^ n) + (option 'I_(2 ^ l) + unit))))))))))))))%type.
Context {M n l ec_bits vseed_count tag_count paseed_count : nat}.
Notation V := (bb84_val M n l ec_bits vseed_count tag_count paseed_count).
Definition bb84_bool_val (b : bool) : V := inl b.
Definition bb84_basis_bool_val (b : bool) : V := inr (inl b).
Definition bb84_bit_val (b : bool) : V := inr (inr (inl b)).
Definition bb84_outcome_val (o : option bool) : V :=
  inr (inr (inr (inl o))).
Definition bb84_set_val (X : {set 'I_M}) : V :=
  inr (inr (inr (inr (inl X)))).
Definition bb84_partial_bits_val (bits : bb84_partial_bits M) : V :=
  inr (inr (inr (inr (inr (inl bits))))).
Definition bb84_error_val (err total : nat) : V :=
  inr (inr (inr (inr (inr (inr (inl (inord err, inord total))))))).
Definition bb84_transcript_val
    (tr : bb84_transcript M ec_bits vseed_count tag_count paseed_count) : V :=
  inr (inr (inr (inr (inr (inr (inr (inl tr))))))).
Definition bb84_ec_val (x : option 'I_(2 ^ ec_bits)) : V :=
  inr (inr (inr (inr (inr (inr (inr (inr (inl x)))))))).
Definition bb84_vseed_val (x : option 'I_vseed_count) : V :=
  inr (inr (inr (inr (inr (inr (inr (inr (inr (inl x))))))))).
Definition bb84_tag_val (x : option 'I_tag_count) : V :=
  inr (inr (inr (inr (inr (inr (inr (inr (inr (inr (inl x)))))))))).
Definition bb84_paseed_val (x : option 'I_paseed_count) : V :=
  inr (inr (inr (inr (inr (inr (inr (inr (inr (inr (inr (inl x))))))))))).
Definition bb84_raw_key_val (x : option 'I_(2 ^ n)) : V :=
  inr (inr (inr (inr (inr (inr (inr (inr (inr (inr (inr (inr
    (inl x)))))))))))).
Definition bb84_final_key_val (x : option 'I_(2 ^ l)) : V :=
  inr (inr (inr (inr (inr (inr (inr (inr (inr (inr (inr (inr
    (inr (inl x))))))))))))).
Definition bb84_bot_val : V :=
  inr (inr (inr (inr (inr (inr (inr (inr (inr (inr (inr (inr
    (inr (inr tt)))))))))))).
\end{minted}
\switchcolumn
\textbf{Tagged finite values.}
The store value type is the disjoint sum of abort booleans, bases, bits,
detector outcomes, round sets, partial bit strings, error fractions,
transcripts, reconciliation data, verification and PA seeds, tags, raw and
final keys, and $\perp$.
The injections at right distinguish values even when their payload types
coincide.
\end{paracol}
\endgroup
\AlignmentBlockEnd

\AlignmentContinuationLabel{tab:alignment-value-decoders}
\AlignmentBlockHead{Value and transcript decoders.}
\begingroup\AlignmentBlockSetup
\begin{paracol}{2}
\code{src/qkd/example/program.v}: basic \coqi{bb84_*_of} decoders and
\coqi{bb84_transcript_or_empty}
\begin{minted}{coq}
Context {M n l ec_bits vseed_count tag_count paseed_count : nat}.
Notation V := (bb84_val M n l ec_bits vseed_count tag_count paseed_count).
Definition bb84_bit_of (v : V) : option bool :=
  match v with inr (inr (inl b)) => Some b | _ => None end.
Definition bb84_basis_bool_of (v : V) : option bool :=
  match v with inr (inl b) => Some b | _ => None end.
Definition bb84_outcome_of (v : V) : option (option bool) :=
  match v with inr (inr (inr (inl o))) => Some o | _ => None end.
Definition bb84_read_bit (v : V) : option bool :=
  match bb84_bit_of v with
  | Some b => Some b
  | None =>
      match bb84_outcome_of v with
      | Some (Some b) => Some b | _ => None
      end
  end.
Definition bb84_set_of (v : V) : {set 'I_M} :=
  match v with
  | inr (inr (inr (inr (inl X)))) => X | _ => set0
  end.
Definition bb84_partial_bits_of (v : V) : option (bb84_partial_bits M) :=
  match v with
  | inr (inr (inr (inr (inr (inl bits))))) => Some bits
  | _ => None
  end.
Definition bb84_error_of (v : V) : nat :=
  match v with
  | inr (inr (inr (inr (inr (inr (inl e)))))) => nat_of_ord e.1
  | _ => 0%N
  end.
Definition bb84_transcript_of (v : V) :
    option (bb84_transcript M ec_bits vseed_count tag_count paseed_count) :=
  match v with
  | inr (inr (inr (inr (inr (inr (inr (inl tr))))))) => Some tr
  | _ => None
  end.
Definition bb84_empty_transcript :
    bb84_transcript M ec_bits vseed_count tag_count paseed_count :=
  ([ffun _ => None], None, None, None).
Definition bb84_empty_transcript_val : V :=
  bb84_transcript_val bb84_empty_transcript.
Definition bb84_transcript_or_empty (v : V) :
    bb84_transcript M ec_bits vseed_count tag_count paseed_count :=
  if bb84_transcript_of v is Some tr then tr else bb84_empty_transcript.
\end{minted}
\switchcolumn
\textbf{Tagged-value decoding.}
The decoders implement the displayed tagged projections.
\end{paracol}
\endgroup
\AlignmentBlockEnd

\AlignmentContinuationLabel{tab:alignment-key-decoders}
\AlignmentBlockHead{Seed, transcript, and key decoders.}
\begingroup\AlignmentBlockSetup
\begin{paracol}{2}
\code{src/qkd/example/program.v}: typed \coqi{bb84_*_of} decoders
\begin{minted}{coq}
Context {M n l ec_bits vseed_count tag_count paseed_count : nat}.
Notation V := (bb84_val M n l ec_bits vseed_count tag_count paseed_count).
Definition bb84_vseed_of (v : V) : option 'I_vseed_count :=
  match v with
  | inr (inr (inr (inr (inr (inr (inr (inr (inr
      (inl x))))))))) => x
  | _ => None
  end.
Definition bb84_ec_of (v : V) : option 'I_(2 ^ ec_bits) :=
  match v with
  | inr (inr (inr (inr (inr (inr (inr (inr
      (inl x)))))))) => x
  | _ => None
  end.
Definition bb84_tag_of (v : V) : option 'I_tag_count :=
  match v with
  | inr (inr (inr (inr (inr (inr (inr (inr (inr (inr
      (inl x)))))))))) => x
  | _ => None
  end.
Definition bb84_paseed_of (v : V) : option 'I_paseed_count :=
  match v with
  | inr (inr (inr (inr (inr (inr (inr (inr (inr (inr (inr
      (inl x))))))))))) => x
  | _ => None
  end.
Definition bb84_raw_key_of (v : V) : option 'I_(2 ^ n) :=
  match v with
  | inr (inr (inr (inr (inr (inr (inr (inr (inr (inr (inr (inr
      (inl x)))))))))))) => x
  | _ => None
  end.
Definition bb84_final_key_of (v : V) : option 'I_(2 ^ l) :=
  match v with
  | inr (inr (inr (inr (inr (inr (inr (inr (inr (inr (inr (inr
      (inr (inl x))))))))))))) => x
  | _ => None
  end.
\end{minted}
The displayed sum tags are the exact inverse positions of the constructors in
\Cref{tab:alignment-value-universe}.
\switchcolumn
\textbf{Typed projections.}
Typed value decoders map verification seeds, reconciliation messages,
verification tags, PA seeds, raw keys, and final keys to their payload types:
\[
\mathsf{decode}_T:V\longrightarrow\operatorname{option}(T).
\]
These decoders are used by the transcript adapters in
\Cref{tab:alignment-program-transcript-adapters}.
\end{paracol}
\endgroup
\AlignmentBlockEnd

\AlignmentContinuationLabel{tab:alignment-label-layout}
\AlignmentBlockHead{Labelled register layout.}
\begingroup\AlignmentBlockSetup
\begin{paracol}{2}
\code{src/qkd/example/program.v}: \coqi{bb84_label}, label/register families,
\coqi{bb84_label_space}, and matrix transports
\begin{minted}{coq}
Definition bb84_label (M : nat) (LE LB : finType) :=
  ('I_M + ('I_M + (LE + LB)))%type.
Definition bb84_signal_label {M LE LB} (i : 'I_M) : bb84_label M LE LB :=
  inl i.
Definition bb84_source_label {M LE LB} (i : 'I_M) : bb84_label M LE LB :=
  inr (inl i).
Definition bb84_bridge_label {M LE LB} (b : LB) : bb84_label M LE LB :=
  inr (inr (inr b)).
Definition bb84_signal_reg {M LE LB} (i : 'I_M) :=
  [set bb84_signal_label i].
Definition bb84_signal_regs {M LE LB} : 'I_M -> {set bb84_label M LE LB} :=
  bb84_signal_reg.
Definition bb84_source_reg {M LE LB} (i : 'I_M) :=
  [set bb84_source_label i].
Definition bb84_source_regs {M LE LB} : 'I_M -> {set bb84_label M LE LB} :=
  bb84_source_reg.
Definition bb84_eve_reg (M : nat) (LE LB : finType) :=
  [set x : bb84_label M LE LB |
    if x is inr (inr (inl _)) then true else false].
Definition bb84_bridge_reg {M LE LB} (b : LB) :=
  [set bb84_bridge_label b].
Definition bb84_label_space (x : bb84_label M LE LB) : chsType :=
  match x with
  | inl _ => @ihb_chsType R 'I_2
  | inr (inl _) => @ihb_chsType R 'I_4
  | inr (inr (inl e)) => HE e
  | inr (inr (inr b)) => HB b
  end.
Definition bb84_signal_lfun (i : 'I_M) (rho : 'M[C]_2) :=
  mx2h (castmx
    (esym (bb84_signal_reg_dim i), esym (bb84_signal_reg_dim i)) rho).
Definition bb84_source_lfun (i : 'I_M) (rho : 'M[C]_4) :=
  mx2h (castmx
    (esym (bb84_source_reg_dim i), esym (bb84_source_reg_dim i)) rho).
Definition bb84_source_signal_lfun (i : 'I_M) (rho : 'M[C]_(4 * 2)) :=
  let d := f_equal2 muln (bb84_source_reg_dim i)
    (bb84_signal_reg_dim i) in
  biprod_of_mx (bb84_source_signal_disjoint i)
    (castmx (esym d, esym d) rho).
\end{minted}
The underlying matrix transport \coqi{biprod_of_mx} is expanded in
\Cref{tab:alignment-statement-transports}.
\switchcolumn
\textbf{Quantum labels and their Hilbert spaces.}
\[
L=\mathbb I_M^{\rm sig}+\mathbb I_M^{\rm src}+L_E+L_B,
\]
\[
\mathcal H_x=
\begin{cases}
\mathbb C^2,&x\in L^{\rm sig},\\
\mathbb C^4,&x\in L^{\rm src},\\
\mathcal H_E(x),&x\in L_E,\\
\mathcal H_B(x),&x\in L_B.
\end{cases}
\]
Singleton signal and source labels form the round register families; the
matrix transports embed $2$-, $4$-, and $8$-dimensional operands into those
labelled spaces.
\end{paracol}
\endgroup
\AlignmentBlockEnd

\AlignmentContinuationLabel{tab:alignment-bridge-registers}
\AlignmentBlockHead{Bridge registers and dimensions.}
\begingroup\AlignmentBlockSetup
\begin{paracol}{2}
\code{src/qkd/example/readout.v}: \coqi{bb84_register_index},
\coqi{bb84_register_dim}, \coqi{bb84_register_space}, and bridge sets
\begin{minted}{coq}
Definition fin_hilbert_of_dim (d : nat) (Hd : (0 < d)%N) : ihbFinType :=
  FinInhabited.Pack (FinInhabited.Class (Choice.on 'I_d)
    (Countable.on 'I_d) (Finite.on 'I_d)
    (isPointed.Build 'I_d (Ordinal Hd))).
Definition bb84_register_index := 'I_8.
Definition bb84_reg_S : bb84_register_index :=
  Ordinal (isT : (0 < 8)%N).
Definition bb84_reg_E : bb84_register_index :=
  Ordinal (isT : (1 < 8)%N).
Definition bb84_reg_Zkey : bb84_register_index :=
  Ordinal (isT : (2 < 8)%N).
Definition bb84_reg_Bq : bb84_register_index :=
  Ordinal (isT : (3 < 8)%N).
Definition bb84_reg_Zp : bb84_register_index :=
  Ordinal (isT : (4 < 8)%N).
Definition bb84_reg_Fl : bb84_register_index :=
  Ordinal (isT : (5 < 8)%N).
Definition bb84_reg_Zl : bb84_register_index :=
  Ordinal (isT : (6 < 8)%N).
Definition bb84_reg_Ppub : bb84_register_index :=
  Ordinal (isT : (7 < 8)%N).
Definition bb84_register_dim (r : bb84_register_index) : nat :=
  match val r with
  | 0 => 2 ^ bb84_n P
  | 1 => #|{: bb84_eur_public (bb84_round_limit P) (cd_vseed_count dev)}|
  | 2 => 2 ^ bb84_n P | 3 => 2 ^ bb84_n P | 4 => 2 ^ bb84_n P
  | 5 => cd_paseed_count dev | 6 => 2 ^ bb84_l P
  | 7 => (2 ^ cd_ir_bits dev) * cd_tag_count dev | _ => 1
  end.
Definition bb84_register_space (r : bb84_register_index) : chsType :=
  'Hs (fin_hilbert_of_dim (bb84_register_dim_pos r)).
Definition bb84_reg (r : bb84_register_index) : {set L} :=
  bb84_bridge_reg r.
Definition bb84_S : {set L} := bb84_reg bb84_reg_S.
Definition bb84_Eve : {set L} := bb84_eve_reg M LE bb84_register_index.
Definition bb84_Epub : {set L} := bb84_reg bb84_reg_E.
Definition bb84_E : {set L} := bb84_Eve :|: bb84_Epub.
Definition bb84_Zkey : {set L} := bb84_reg bb84_reg_Zkey.
Definition bb84_Bq : {set L} := bb84_reg bb84_reg_Bq.
Definition bb84_Zp : {set L} := bb84_reg bb84_reg_Zp.
Definition bb84_Fl : {set L} := bb84_reg bb84_reg_Fl.
Definition bb84_Zl : {set L} := bb84_reg bb84_reg_Zl.
Definition bb84_Ppub : {set L} := bb84_reg bb84_reg_Ppub.
\end{minted}
\switchcolumn
\textbf{Eight finite bridge registers.}
The bridge coordinates are
\[
(S,E_{\rm pub},Z_{\rm key},B_q,Z',F_{\rm PA},Z_\ell,P_{\rm pub}),
\]
with dimensions respectively
\[
2^n, |E_{\rm pub}|, 2^n, 2^n, 2^n,
|\mathcal F|, 2^\ell, 2^{d_{\rm EC}}|T|.
\]
Eve's register is the external-label region and $E=E_{\rm Eve}\cup E_{\rm pub}$.
\end{paracol}
\endgroup
\AlignmentBlockEnd

\AlignmentContinuationLabel{tab:alignment-device-roots}
\AlignmentBlockHead{Basis and source-pointer roots.}
\begingroup\AlignmentBlockSetup
\begin{paracol}{2}
\code{src/qkd/bb84/model.v}: \coqi{bb84_basis};
\code{src/qkd/example/stat.v}: \coqi{bb84_basis_value},
\coqi{bb84_basis_of_value}, \coqi{bb84_prep_index};
\code{src/qkd/example/deferred.v}: \coqi{bb84_src_ptr_slot}
\begin{minted}{coq}
Inductive bb84_basis := BasisX | BasisZ.
Definition bb84_basis_value (basis : bb84_basis) : V :=
  match basis with
  | BasisX => bb84_basis_bool_val false
  | BasisZ => bb84_basis_bool_val true
  end.
Definition bb84_basis_of_value (v : V) : option bb84_basis :=
  match bb84_basis_bool_of v with
  | Some false => Some BasisX
  | Some true => Some BasisZ
  | None => None
  end.
Definition bb84_prep_index (basis : bb84_basis) (bit : bool) : 'I_4 :=
  inord (match basis, bit with
         | BasisX, false => 0 | BasisX, true => 1
         | BasisZ, false => 2 | BasisZ, true => 3
         end).
Definition bb84_src_ptr_slot : 'I_2 := Ordinal (isT : (1 < 2)%N).
\end{minted}
These are executable finite encodings of the mathematical basis labels.
\switchcolumn
\textbf{Basis and pointer encodings.}
$\mathbb X$ and $\mathbb Z$ are distinct basis choices.  Their store values
are distinct Boolean tags; a source preparation $(b,x)$ occupies one of four
basis states, and the low source-qubit slot carries its coherent pointer.
\end{paracol}
\endgroup
\AlignmentBlockEnd

\AlignmentParagraph{Execution.}

\AlignmentContinuationLabel{tab:alignment-denotation}
\AlignmentBlockHead{CQ-state denotational semantics.}
\begingroup\AlignmentBlockSetup
\begin{paracol}{2}
\code{src/qkd/protocol/semantics.v}: \coqi{cq_state}, \coqi{cq_mass},
\coqi{cq_wf}, \coqi{appendso_fun}, \coqi{appendso}, \coqi{replaceso},
\coqi{adversary}, \coqi{denote}
\begin{minted}{coq}
Definition cq_state := {ffun store -> 'F[HL]_setT}.
Definition cq_mass (xi : cq_state) : R :=
  \sum_s @complex.Re R (\Tr (xi s)).
Definition cq_wf (xi : cq_state) : Prop :=
  (forall s, xi s \is psdlf) /\ cq_mass xi = 1.
Definition appendso_fun (Q : {set L}) (rho : 'F[HL]_Q)
    (x : 'F[HL]_(setT :\: Q)) : 'F[HL]_setT :=
  castlf (setUD_sub (subsetT Q), setUD_sub (subsetT Q)) (rho \⊗ x).
Definition appendso (Q : {set L}) (rho : 'F[HL]_Q) :=
  Superop (linfun (appendso_fun rho)).
Definition replaceso (Q : {set L}) (rho : 'F[HL]_Q) : 'SO[HL]_setT :=
  appendso rho :o ptraceso Q.
Definition adversary :=
  forall Q : {set L}, V -> 'QC('H[HL]_(Q :|: qA)).
Fixpoint denote (c : qkd_cmd) (xi : cq_state) : cq_state := match c with
  | Cskip => xi
  | Cassign x e => [ffun t => \sum_(s | store_set s x (e s) == t) xi s]
  | Csample x mu => [ffun t => \sum_s
      \sum_(v | store_set s x v == t) (mu s v)%:C *: xi s]
  | Cinit Q rho => [ffun s => replaceso (rho s) (xi s)]
  | Cmeasure x Q M => [ffun t => \sum_s
      \sum_(v | store_set s x v == t)
        formso (liftf_lf (M s v)) (xi s)]
  | Cattack Q => [ffun s => liftfso (EA Q (s Pvar)) (xi s)]
  | Cif b c1 c2 =>
      denote c1 [ffun s => if b s then xi s else 0] +
      denote c2 [ffun s => if b s then 0 else xi s]
  | Cseq c1 c2 => denote c2 (denote c1 xi)
  end.
\end{minted}
\AlignmentBoundaryLabel
\coqi{cq_wf} packages positivity and unit total mass, matching the paper's
state invariant.
The adversary family is indexed by the attacked register set $Q$, and its
channel acts on $Q\cup q_A$.
The initialization dependency \coqi{replaceso} supplies the displayed
append/partial-trace construction.
The dependencies \coqi{store}, \coqi{store_set}, and \coqi{qkd_cmd} are expanded in \Cref{tab:alignment-semantics-core}.
\switchcolumn
\textbf{Cq-states and denotational semantics (Definition~\ref{def:protocol-semantics}).}
\[
\begin{aligned}
\Xi&=\Sigma\to\mathcal L(\mathcal H),\\
\mathsf{wf}(\xi)&\iff
  (\forall s,\xi(s)\succeq0)\\
  &\quad\land\sum_s\Tr\xi(s)=1.
\end{aligned}
\]
For a singleton $(s,\rho)$, the paper clauses are
\[
\begin{aligned}
\auxsem{\mathsf{skip}}(s,\rho)&=(s,\rho),\\
\auxsem{x:=e}(s,\rho)&=(s[x\mapsto e(s)],\rho),\\
\auxsem{x\gets_\$\mu}(s,\rho)&=\sum_v(s[x\mapsto v],\mu_s(v)\rho),\\
\auxsem{Q\leftarrow\rho_0}(s,\rho)&=(s,\Tr_Q\rho\otimes\rho_0(s)),\\
\auxsem{x\leftarrow M[Q]}(s,\rho)&=
 \sum_v(s[x\mapsto v],M_v\rho M_v^\dagger),\\
\auxsem{\mathsf{attack}[Q]}(s,\rho)&=(s,\mathcal E_A(s(P))(\rho)).
\end{aligned}
\]
The conditional selects its branch, denotation sums over singleton inputs, and
$\sem{c_1;c_2}=\sem{c_2}\circ\sem{c_1}$.
\end{paracol}
\endgroup
\AlignmentBlockEnd

\AlignmentContinuationLabel{tab:alignment-language}
\AlignmentBlockHead{QKD round components.}
\begingroup\AlignmentBlockSetup
\begin{paracol}{2}
\code{src/qkd/protocol/qkd\_cmd.v}: \coqi{pX}, \coqi{mu_basis}, \coqi{mu_bit}, \coqi{detected}, \coqi{sift_X}, \coqi{sift_Z}, \coqi{rounds_through}, \coqi{Abort_cmd}, \coqi{PrepPM}, \coqi{PrepEB}, \coqi{Prep}, \coqi{Dist}, \coqi{Meas}, \coqi{Sift}, \coqi{Round}, \coqi{Announce}
\begin{minted}{coq}
Definition pX : R := (1 + Num.sqrt (k%:R / n%:R))^-1.
Definition mu_basis : dexpr := fun _ v =>
  if v == basisX then pX else if v == basisZ then 1 - pX else 0.
Definition mu_bit : dexpr := fun _ v =>
  if v == vbit false then (2%:R)^-1
  else if v == vbit true then (2%:R)^-1 else 0.
Definition detected (s : store) (i : 'I_M) : bool :=
  s (y'_ i) != outEmpty.
Definition sift_X (s : store) (i : 'I_M) : bool :=
  (s (a_ i) == basisX) && (s (b_ i) == basisX) &&
  (s (d_ i) == vbool true).
Definition sift_Z (s : store) (i : 'I_M) : bool :=
  (s (a_ i) == basisZ) && (s (b_ i) == basisZ) &&
  (s (d_ i) == vbool true).
Definition rounds_through (i : 'I_M) (p : 'I_M -> bool) :=
  [set j : 'I_M | (nat_of_ord j <= nat_of_ord i)%N && p j].
Definition Abort_cmd : @qkd_cmd R L HL Cv V :=
  (abortv <-c (fun _ => vbool true)) ;;
  (Sv <-c (fun _ => bot)) ;; (Shv <-c (fun _ => bot)).
Definition PrepPM (i : 'I_M) : @qkd_cmd R L HL Cv V :=
  (a_ i <-^$ mu_basis) ;; (y_ i <-^$ mu_bit) ;;
  (Q i <-q (fun s => dev_rho i (s (a_ i)) (s (y_ i)))).
Definition PrepEB (i : 'I_M) : @qkd_cmd R L HL Cv V :=
  ((Qsrc i :|: Q i) <-q (fun _ => dev_pair i)) ;;
  (a_ i <-^$ mu_basis) ;;
  (y_ i :=M[Qsrc i] (fun s v => dev_ptr i (s (a_ i)) v)).
Definition Prep (i : 'I_M) : @qkd_cmd R L HL Cv V :=
  if entangled then PrepEB i else PrepPM i.
Definition Dist (i : 'I_M) : @qkd_cmd R L HL Cv V := Attack[Q i].
Definition Meas (i : 'I_M) : @qkd_cmd R L HL Cv V :=
  (b_ i <-^$ mu_basis) ;;
  (y'_ i :=M[Q i] (fun s v => dev_meas i (s (b_ i)) v)).
Definition Sift (i : 'I_M) : @qkd_cmd R L HL Cv V :=
  (d_ i <-c (fun s => vbool (detected s i))) ;;
  (Xv <-c (fun s => vset (rounds_through i (sift_X s)))) ;;
  (Zv <-c (fun s => vset (rounds_through i (sift_Z s)))).
Definition Round (i : 'I_M) : @qkd_cmd R L HL Cv V :=
  Prep i ;; Dist i ;; Meas i ;; Sift i.
Definition Announce : @qkd_cmd R L HL Cv V :=
  (Pvar <-c (fun s =>
    tr_app_rounds (fun i => (s (a_ i), s (b_ i), s (d_ i))) (s Pvar))).
\end{minted}
The paper's $\mathsf{Abort}$ is named \coqi{Abort_cmd}, and \coqi{Prep}
selects the prepare-and-measure or entanglement-based source mode.
The semantic domains \coqi{store}, \coqi{dexpr}, and \coqi{qkd_cmd} are expanded in \Cref{tab:alignment-semantics-core}.
\switchcolumn
\textbf{QKD component sugar (Definition~\ref{def:qkd-component-sugar}).}
With $p_{\mathbb X}=(1+\sqrt{k/n})^{-1}$,
$\mu_{\mathbb B}(\mathbb X)=p_{\mathbb X}$, and uniform $\mu_2$,
the preparation component has both Box~1 modes:
\[
\begin{aligned}
\mathsf{PrepPM}_i&:=a_i\gets_{\$}\mu_{\mathbb B};\\
  &\quad y_i\gets_{\$}\mu_2;Q_i\leftarrow\rho_{a_i,y_i},\\
\mathsf{PrepEB}_i&:=Q_i^{\rm src}Q_i\leftarrow\rho_{\rm pair};\\
  &\quad a_i\gets_{\$}\mu_{\mathbb B};
  y_i\leftarrow M^{\rm ptr}_{a_i}[Q_i^{\rm src}],\\
\mathsf{Prep}_i(b)&:=
  \begin{cases}\mathsf{PrepEB}_i,&b=\mathsf{true},\\
                \mathsf{PrepPM}_i,&b=\mathsf{false},\end{cases}\\
\mathsf{Round}_i&:=\mathsf{Prep}_i;\mathsf{Dist}_i;\\
  &\quad\mathsf{Meas}_i;\mathsf{Sift}_i.
\end{aligned}
\]
This row also expands $\mathsf{Abort}$, $\mathsf{Meas}$,
$\mathsf{Sift}$, and $\mathsf{Announce}$; the three classical post-processing components follow in the next row.
\end{paracol}
\endgroup
\AlignmentBlockEnd

\AlignmentContinuationLabel{tab:alignment-qkd-classical-components}
\AlignmentBlockHead{Classical QKD components.}
\begingroup\AlignmentBlockSetup
\begin{paracol}{2}
\code{src/qkd/protocol/qkd\_cmd.v}: \coqi{has_key_sample}, \coqi{has_test_sample}, \coqi{sift_ok}, \coqi{paramest_K_mu}, \coqi{alice_bits}, \coqi{bob_bits}, \coqi{ParamEst}, \coqi{pe_error_high}, \coqi{pe_fail}, \coqi{ErrCorr}, \coqi{PrivAmp}
\begin{minted}{coq}
Definition has_key_sample (s : store) : bool :=
  (n <= #|vset_of (s Xv)|)%N.
Definition has_test_sample (s : store) : bool :=
  (k <= #|vset_of (s Zv)|)%N.
Definition sift_ok : bexpr := fun s =>
  has_key_sample s && has_test_sample s.
Definition paramest_K_mu : dexpr := fun s v =>
  if has_key_sample s then unif_nsubset_mass n (vset_of (s Xv)) v
  else (v == bot)%:R.
Definition alice_bits (s : store) : 'I_M -> V := fun i => s (y_ i).
Definition bob_bits (s : store) : 'I_M -> V := fun i => s (y'_ i).
Definition ParamEst : @qkd_cmd R L HL Cv V :=
  (Kv <-^$ paramest_K_mu) ;;
  (XA <-c (fun s => key_bits_on (vset_of (s Kv)) (alice_bits s))) ;;
  (XB <-c (fun s => key_bits_on (vset_of (s Kv)) (bob_bits s))) ;;
  (ZA <-c (fun s => test_bits_on (vset_of (s Zv)) (alice_bits s))) ;;
  (ZB <-c (fun s => test_bits_on (vset_of (s Zv)) (bob_bits s))) ;;
  (Pvar <-c (fun s => tr_app_pe (s Kv) (s ZA) (s ZB) (s Pvar))) ;;
  (lamv <-c (fun s => vrat (err_count (s ZA) (s ZB))
    #|vset_of (s Zv)|)).
Definition pe_error_high : bexpr := fun s =>
  Qtol * (#|vset_of (s Zv)|)%:R < (err_of (s lamv))%:R.
Definition pe_fail : bexpr := fun s => ~~ sift_ok s || pe_error_high s.
Definition ErrCorr : @qkd_cmd R L HL Cv V :=
  (pECv <-c (fun s => ir_pub leakEC (s XA) (s XB))) ;;
  (Xhatv <-c (fun s => ir_dec (s pECv) (s XB))) ;;
  (gv <-^$ unif_vseed) ;;
  (tAv <-c (fun s => verify_hash (s gv) (s XA))) ;;
  (tBv <-c (fun s => verify_hash (s gv) (s Xhatv))) ;;
  (Pvar <-c (fun s => tr_app_ec (s pECv) (s gv) (s tAv) (s Pvar))).
Definition PrivAmp : @qkd_cmd R L HL Cv V :=
  (fv <-^$ unif_pseed) ;;
  (Sv <-c (fun s => privacy_hash l (s fv) (s XA))) ;;
  (Shv <-c (fun s => privacy_hash l (s fv) (s Xhatv))) ;;
  (Pvar <-c (fun s => tr_app_pa (s fv) (s Pvar))).
\end{minted}
These are the immediate dependencies of the whole-program branch displayed next; protocol-specific encoders and transcript appenders are parameters of the generic component layer.
The semantic domains \coqi{store}, \coqi{dexpr}, \coqi{bexpr}, and \coqi{qkd_cmd} are expanded in \Cref{tab:alignment-semantics-core}.
\switchcolumn
\textbf{Parameter estimation, reconciliation, and privacy amplification.}
\[
\begin{aligned}
\mathsf{ParamEst}:&\ K\gets_\$\mathsf{Unif}(\mathcal K_n(\mathcal X));\\
 &\quad X_A,X_B,Z_A,Z_B;P;\lambda,\\
\mathsf{ErrCorr}:&\ p_{\rm EC},\widehat X;g\gets_\$\mathcal G;\\
 &\quad t_A=g(X_A);t_B=g(\widehat X);P,\\
\mathsf{PrivAmp}:&\ f\gets_\$\mathcal F_\ell;\\
 &\quad S=f(X_A);\widehat S=f(\widehat X);P.
\end{aligned}
\]
The program aborts if either sample is too small or the observed error exceeds
$\Qtol$.
\end{paracol}
\endgroup
\AlignmentBlockEnd

\AlignmentContinuationLabel{tab:alignment-program-physical-adapters}
\AlignmentBlockHead{Source and measurement adapters.}
\begingroup\AlignmentBlockSetup
\begin{paracol}{2}
\code{src/qkd/example/program.v}: \coqi{bb84_program_basis_value}, \coqi{bb84_program_basis_of_value}, \coqi{bb84_program_device_rho}, \coqi{bb84_program_device_meas}, \coqi{bb84_program_device_pair}, \coqi{bb84_program_device_ptr}
\begin{minted}{coq}
Definition bb84_program_basis_value (basis : bb84_basis) : V :=
  match basis with
  | BasisX => bb84_basis_bool_val false
  | BasisZ => bb84_basis_bool_val true
  end.
Definition bb84_program_basis_of_value (v : V) : option bb84_basis :=
  match bb84_basis_bool_of v with
  | Some false => Some BasisX
  | Some true => Some BasisZ
  | None => None
  end.
Definition bb84_program_device_rho (i : 'I_M) (basis bit : V)
    : 'F[HL]_(bb84_signal_reg i) :=
  match bb84_program_basis_of_value basis, bb84_bit_of bit with
  | Some BasisX, Some b => bb84_signal_lfun HE HB i (cd_rho dev BasisX b)
  | Some BasisZ, Some b => bb84_signal_lfun HE HB i (cd_rho dev BasisZ b)
  | _, _ => bb84_signal_lfun HE HB i (cd_rho dev BasisX false)
  end.
Definition bb84_program_device_meas (i : 'I_M) (basis outcome : V)
    : 'F[HL]_(bb84_signal_reg i) :=
  match bb84_program_basis_of_value basis, bb84_outcome_of outcome with
  | Some BasisX, Some o => bb84_signal_lfun HE HB i (cd_meas dev BasisX o)
  | Some BasisZ, Some o => bb84_signal_lfun HE HB i (cd_meas dev BasisZ o)
  | None, Some o => bb84_signal_lfun HE HB i (cd_meas dev BasisX o)
  | _, None => 0
  end.
Definition bb84_program_device_pair (i : 'I_M) :=
  bb84_source_signal_lfun HE HB i (cd_pair dev).
Definition bb84_program_device_ptr (i : 'I_M) (basis bit : V) :=
  match bb84_program_basis_of_value basis, bb84_bit_of bit with
  | Some BasisX, Some b => bb84_source_lfun HE HB i (cd_ptr dev BasisX b)
  | Some BasisZ, Some b => bb84_source_lfun HE HB i (cd_ptr dev BasisZ b)
  | None, Some b => bb84_source_lfun HE HB i (cd_ptr dev BasisX b)
  | _, None => 0
  end.
\end{minted}
The fields \coqi{cd_rho}, \coqi{cd_meas}, \coqi{cd_pair}, and \coqi{cd_ptr} are expanded in \Cref{tab:alignment-device-certificate}.
The labelled-register transports and value decoders are expanded in
\Cref{tab:alignment-label-layout,tab:alignment-value-decoders,%
tab:alignment-key-decoders}.
\switchcolumn
\textbf{Embedding the certified device in each protocol round.}
The concrete adapters decode the store's basis/bit/outcome values and transport
\[
\rho_{b,y},\quad K_{b,o},\quad \rho_{\rm pair},
\quad M^{\rm ptr}_{b,y}
\]
into the labelled signal and source registers for round $i$.
\end{paracol}
\endgroup
\AlignmentBlockEnd

\AlignmentContinuationLabel{tab:alignment-program-classical-adapters}
\AlignmentBlockHead{Sampling and bit extraction.}
\begingroup\AlignmentBlockSetup
\begin{paracol}{2}
\code{src/qkd/example/program.v}: \coqi{bb84_selected_index}, \coqi{bb84_bits_on_val}, \coqi{bb84_partial_bits_on_val}, \coqi{bb84_partial_error_count}, \coqi{bb84_err_count_val}, \coqi{bb84_nsubset_count}, \coqi{bb84_unif_nsubset_mass}, \coqi{bb84_unif_vseed_mass}, \coqi{bb84_unif_pseed_mass}
\begin{minted}{coq}
Definition bb84_selected_index {M n : nat}
    (K : {set 'I_M}) (j : 'I_n) : option 'I_M :=
  nth None [seq Some i | i <- enum K] j.
Definition bb84_bits_on_val (K : {set 'I_M}) (bits : 'I_M -> V) : V :=
  if #|K| == n then
    bb84_raw_key_val
      (Some (bitstring_to_ord (fun j =>
        match bb84_selected_index K j with
        | Some i => if bb84_read_bit (bits i) is Some b then b else false
        | None => false
        end)))
  else bb84_raw_key_val None.
Definition bb84_partial_bits_on_val
    (K : {set 'I_M}) (bits : 'I_M -> V) : V :=
  bb84_partial_bits_val
    [ffun i => if i \in K then bb84_read_bit (bits i) else None].
Definition bb84_partial_error_count (M : nat)
    (a b : bb84_partial_bits M) : nat :=
  #|[set i : 'I_M |
    if a i is Some x then
      if b i is Some y then x != y else false
    else false]|.
Definition bb84_err_count_val (a b : V) : nat :=
  match bb84_partial_bits_of a, bb84_partial_bits_of b with
  | Some aa, Some bb => bb84_partial_error_count aa bb
  | _, _ =>
      match bb84_raw_key_of a, bb84_raw_key_of b with
      | Some x, Some y => d_ham (ord_to_bitstring x) (ord_to_bitstring y)
      | _, _ => 0%N
      end
  end.
Definition bb84_nsubset_count (M m : nat) (X : {set 'I_M}) : nat :=
  #|[set K : {set 'I_M} | (#|K| == m) && (K \subset X)]|.
Definition bb84_unif_nsubset_mass (m : nat) (X : {set 'I_M}) (v : V) : R :=
  \sum_(K : {set 'I_M} | (#|K| == m) && (K \subset X))
    if v == bb84_set_val K
    then (bb84_nsubset_count m X)%:R^-1 else 0.
Definition bb84_unif_vseed_mass : dexpr := fun _ v =>
  \sum_g if v == bb84_vseed_val (Some g)
    then (vseed_count%:R)^-1 else 0.
Definition bb84_unif_pseed_mass : dexpr := fun _ v =>
  \sum_f if v == bb84_paseed_val (Some f)
    then (paseed_count%:R)^-1 else 0.
\end{minted}
The Hamming-distance and ordinal/bitstring dependencies are expanded in
\Cref{tab:alignment-sampling-objects}; all value constructors and decoders are
expanded in \Cref{tab:alignment-value-universe,tab:alignment-value-decoders,%
tab:alignment-key-decoders}.
\switchcolumn
\textbf{Finite sample, key/test strings, and seed distributions.}
For $K\subseteq\{0,\ldots,M-1\}$, the adapter enumerates $K$ in canonical order,
forms Alice/Bob's $n$-bit key strings when $|K|=n$, retains partial test strings,
and counts disagreements.
The sample law is uniform over $n$-subsets of $\mathcal X$; verification and PA seeds are uniform on their finite seed spaces.
\end{paracol}
\endgroup
\AlignmentBlockEnd

\AlignmentContinuationLabel{tab:alignment-program-transcript-adapters}
\AlignmentBlockHead{Transcript and hash adapters.}
\begingroup\AlignmentBlockSetup
\begin{paracol}{2}
\code{src/qkd/example/program.v}: \coqi{bb84_verify_hash_val}, \coqi{bb84_ir_pub_val}, \coqi{bb84_ir_dec_val}, \coqi{bb84_privacy_hash_val}, \coqi{bb84_tr_app_rounds}, \coqi{bb84_tr_app_pe}, \coqi{bb84_tr_app_ec}, \coqi{bb84_tr_app_pa}, \coqi{bb84_eur_public_of_transcript}, \coqi{bb84_eur_public_of_value}, \coqi{bb84_eur_public_of_store}
\begin{minted}{coq}
Definition bb84_verify_hash_val (vh : 'I_vseed_count -> 'I_(2 ^ n) ->
    'I_tag_count) (seed raw : V) : V :=
  match bb84_vseed_of seed, bb84_raw_key_of raw with
  | Some g, Some x => bb84_tag_val (Some (vh g x))
  | _, _ => bb84_tag_val None
  end.
Definition bb84_ir_pub_val
    (ir_pub : 'I_(2 ^ n) -> 'I_(2 ^ n) -> 'I_(2 ^ ec_bits))
    (_leakEC : nat) (alice bob : V) : V :=
  match bb84_raw_key_of alice, bb84_raw_key_of bob with
  | Some x, Some y => bb84_ec_val (Some (ir_pub x y))
  | _, _ => bb84_ec_val None
  end.
Definition bb84_ir_dec_val
    (ir_dec : 'I_(2 ^ ec_bits) -> 'I_(2 ^ n) -> 'I_(2 ^ n))
    (transcript bob : V) : V :=
  match bb84_ec_of transcript, bb84_raw_key_of bob with
  | Some p, Some y => bb84_raw_key_val (Some (ir_dec p y))
  | _, _ => bb84_raw_key_val None
  end.
Definition bb84_privacy_hash_val
    (pa : 'I_paseed_count -> 'I_(2 ^ n) -> 'I_(2 ^ l))
    (_l : nat) (seed raw : V) : V :=
  match bb84_paseed_of seed, bb84_raw_key_of raw with
  | Some f, Some x => bb84_final_key_val (Some (pa f x))
  | _, _ => bb84_final_key_val None
  end.
Definition bb84_tr_app_rounds (rounds : 'I_M -> V * V * V) (trv : V) : V :=
  let '(round_old, pe, ec, pa) := bb84_transcript_or_empty trv in
  let round_new := [ffun i =>
    let '(a, b, d) := rounds i in
    match bb84_basis_bool_of a, bb84_basis_bool_of b with
    | Some aa, Some bb =>
        match d with inl detected => Some (aa, bb, detected) | _ => None end
    | _, _ => None
    end] in
  bb84_transcript_val (round_new, pe, ec, pa).
Definition bb84_tr_app_pe (K ZA ZB trv : V) : V :=
  let '(rounds, _pe, ec, pa) := bb84_transcript_or_empty trv in
  let empty_bits : bb84_partial_bits M := [ffun _ => None] in
  let ZA_bits := if bb84_partial_bits_of ZA is Some bits then bits
                 else empty_bits in
  let ZB_bits := if bb84_partial_bits_of ZB is Some bits then bits
                 else empty_bits in
  bb84_transcript_val
    (rounds, Some (bb84_set_of K, ZA_bits, ZB_bits), ec, pa).
Definition bb84_tr_app_ec (pEC g tA trv : V) : V :=
  let '(rounds, pe, _ec, pa) := bb84_transcript_or_empty trv in
  bb84_transcript_val
    (rounds, pe, Some (bb84_ec_of pEC, bb84_vseed_of g,
      bb84_tag_of tA), pa).
\end{minted}
\switchcolumn
\textbf{Reconciliation, verification, PA, and public-transcript projection.}
The value adapters realize
\[
p_{\rm EC}=\mathsf{IR}_{\rm pub}(X_A,X_B),\quad
\widehat X=\mathsf{IR}_{\rm dec}(p_{\rm EC},X_B),
\]
\[
t=g(X),\qquad S=f(X),
\]
and append round, parameter-estimation, reconciliation, and PA data to the transcript.
The EUR-public projection keeps exactly $(\mathsf{rounds},K,Z_A,Z_B,g)$.
\end{paracol}
\endgroup
\AlignmentBlockEnd

\AlignmentContinuationLabel{tab:alignment-program-transcript-adapters-cont}
\AlignmentBlockHead{Transcript and hash adapters (continued).}
\begingroup\AlignmentBlockSetup
\begin{paracol}{2}
\code{src/qkd/example/program.v}: \coqi{bb84_tr_app_pa},
\coqi{bb84_eur_public_of_transcript}, \coqi{bb84_eur_public_of_value},
\coqi{bb84_eur_public_of_store}
\begin{minted}{coq}
Definition bb84_tr_app_pa (f trv : V) : V :=
  let '(rounds, pe, ec, _pa) := bb84_transcript_or_empty trv in
  bb84_transcript_val (rounds, pe, ec, Some (bb84_paseed_of f)).
Definition bb84_eur_public_of_transcript (tr : bb84_transcript M
    ec_bits vseed_count tag_count paseed_count) :=
  let empty_bits : bb84_partial_bits M := [ffun _ => None] in
  let empty_set : {set 'I_M} := [set i : 'I_M | false] in
  let '(rounds, pe, ec, _pa) := tr in
  let '(K, ZA, ZB) := if pe is Some d then d
                      else (empty_set, empty_bits, empty_bits) in
  let g := if ec is Some (_, g, _) then g else None in
  (rounds, K, ZA, ZB, g).
Definition bb84_eur_public_of_value (v : V) :=
  bb84_eur_public_of_transcript (bb84_transcript_or_empty v).
Definition bb84_eur_public_of_store (s : store) :=
  bb84_eur_public_of_value (s BVP).
\end{minted}
The underlying transcript type, tagged-value constructors, decoders, and
transcript initialization are expanded in
\Cref{tab:alignment-value-universe,tab:alignment-value-decoders,%
tab:alignment-key-decoders}.
\switchcolumn
\textit{Continued.}
\end{paracol}
\endgroup
\AlignmentBlockEnd

\AlignmentContinuationLabel{tab:alignment-program-modes}
\AlignmentBlockHead{Source-mode evaluator.}
\begingroup\AlignmentBlockSetup
\begin{paracol}{2}
\code{src/qkd/example/example.v}: \coqi{bb84_program_mode}
\begin{minted}{coq}
Definition bb84_program_mode (entangled : bool)
    : @qkd_cmd R L HL (bb84_cvar M) V :=
  (BVP <-c (fun _ : @store (bb84_cvar M) V => bb84_empty_transcript_val)) ;;
  (BVX <-c (fun _ : @store (bb84_cvar M) V => bb84_set_val set0)) ;;
  (BVZ <-c (fun _ : @store (bb84_cvar M) V => bb84_set_val set0)) ;;
  (BVabort <-c (fun _ : @store (bb84_cvar M) V => bb84_bool_val false)) ;;
  (FOR i < M,
    @Round R L HL (bb84_cvar M) V (bb84_n P) (bb84_k P) M
      (@BVa M) (@BVy M) (@BVb M) (@BVy' M) (@BVd M)
      (@bb84_signal_regs M LE LB) (@BVX M) (@BVZ M)
      (@bb84_program_basis_value R P q dev BasisX)
      (@bb84_program_basis_value R P q dev BasisZ)
      (bb84_outcome_val None)
      (@bb84_bool_val M (bb84_n P) (bb84_l P) (cd_ir_bits dev)
        (cd_vseed_count dev) (cd_tag_count dev) (cd_paseed_count dev))
      (@bb84_bit_val M (bb84_n P) (bb84_l P) (cd_ir_bits dev)
        (cd_vseed_count dev) (cd_tag_count dev) (cd_paseed_count dev))
      (@bb84_set_val M (bb84_n P) (bb84_l P) (cd_ir_bits dev)
        (cd_vseed_count dev) (cd_tag_count dev) (cd_paseed_count dev))
      (@bb84_program_device_rho R P q LE HE bb84_register_index HB dev)
      (@bb84_program_device_meas R P q LE HE bb84_register_index HB dev)
      entangled (@bb84_source_regs M LE LB)
      (@bb84_program_device_pair R P q LE HE bb84_register_index HB dev)
      (@bb84_program_device_ptr R P q LE HE bb84_register_index HB dev) i) ;;
\end{minted}
\switchcolumn
\textbf{Initialization and rounds (\Cref{fig:bb84-program,sec:protocol:security}).}
For $b=\mathsf{false}$ (PM) or $b=\mathsf{true}$ (EB),
\[
\begin{gathered}
P\gets\varnothing_{\rm tr},\qquad X\gets\varnothing,\\
Z\gets\varnothing,\qquad \mathsf{abort}\gets\mathsf{false},\\
\textbf{for }i=0,\ldots,M-1:\ \mathsf{Round}_{b}(i).
\end{gathered}
\]
\end{paracol}
\endgroup
\AlignmentBlockEnd

\AlignmentContinuationLabel{tab:alignment-program-modes-announce-pe}
\AlignmentBlockHead{Source-mode evaluator (continued).}
\begingroup\AlignmentBlockSetup
\begin{paracol}{2}
\code{src/qkd/example/example.v}: \coqi{bb84_program_mode} (continued)
\begin{minted}{coq}
  @Announce R L HL (bb84_cvar M) V M (@BVa M) (@BVb M) (@BVd M) (@BVP M)
    (@bb84_tr_app_rounds M (bb84_n P) (bb84_l P) (cd_ir_bits dev)
      (cd_vseed_count dev) (cd_tag_count dev) (cd_paseed_count dev)) ;;
  @ParamEst R L HL (bb84_cvar M) V (bb84_n P) M (@BVy M) (@BVy' M)
    (@BVX M) (@BVZ M) (@BVK M) (@BVXA M) (@BVXB M) (@BVZA M)
    (@BVZB M) (@BVlam M) (@BVP M)
    (@bb84_bot_val M (bb84_n P) (bb84_l P) (cd_ir_bits dev)
      (cd_vseed_count dev) (cd_tag_count dev) (cd_paseed_count dev))
    (@bb84_set_of M (bb84_n P) (bb84_l P) (cd_ir_bits dev)
      (cd_vseed_count dev) (cd_tag_count dev) (cd_paseed_count dev))
    (@bb84_error_val M (bb84_n P) (bb84_l P) (cd_ir_bits dev)
      (cd_vseed_count dev) (cd_tag_count dev) (cd_paseed_count dev))
    (@bb84_tr_app_pe M (bb84_n P) (bb84_l P) (cd_ir_bits dev)
      (cd_vseed_count dev) (cd_tag_count dev) (cd_paseed_count dev))
    (@bb84_bits_on_val M (bb84_n P) (bb84_l P) (cd_ir_bits dev)
      (cd_vseed_count dev) (cd_tag_count dev) (cd_paseed_count dev))
    (@bb84_partial_bits_on_val M (bb84_n P) (bb84_l P) (cd_ir_bits dev)
      (cd_vseed_count dev) (cd_tag_count dev) (cd_paseed_count dev))
    (@bb84_err_count_val M (bb84_n P) (bb84_l P) (cd_ir_bits dev)
      (cd_vseed_count dev) (cd_tag_count dev) (cd_paseed_count dev))
    (@bb84_unif_nsubset_mass R M (bb84_n P) (bb84_l P) (cd_ir_bits dev)
      (cd_vseed_count dev) (cd_tag_count dev) (cd_paseed_count dev)) ;;
\end{minted}
\switchcolumn
\textbf{Announcement and parameter estimation.}
\[
\begin{gathered}
P\gets\mathsf{Announce}(a,b,d;P),\\
(K,X_A,X_B,Z_A,Z_B,\lambda,P)\\
\quad\gets\mathsf{ParamEst}(y,y',X,Z;P).
\end{gathered}
\]
\end{paracol}
\endgroup
\AlignmentBlockEnd

\AlignmentContinuationLabel{tab:alignment-program-modes-pe-branch}
\AlignmentBlockHead{Source-mode evaluator (continued).}
\begingroup\AlignmentBlockSetup
\begin{paracol}{2}
\code{src/qkd/example/example.v}: \coqi{bb84_program_mode} (continued)
\begin{minted}{coq}
  IF @pe_fail R (bb84_cvar M) V (bb84_n P) (bb84_k P) M (bb84_Qtol P)
      (@BVX M) (@BVZ M) (@BVlam M)
      (@bb84_set_of M (bb84_n P) (bb84_l P) (cd_ir_bits dev)
        (cd_vseed_count dev) (cd_tag_count dev) (cd_paseed_count dev))
      (@bb84_error_of M (bb84_n P) (bb84_l P) (cd_ir_bits dev)
        (cd_vseed_count dev) (cd_tag_count dev) (cd_paseed_count dev)) THEN
    @Abort_cmd R L HL (bb84_cvar M) V (@BVS M) (@BVSh M) (@BVabort M)
      (@bb84_bot_val M (bb84_n P) (bb84_l P) (cd_ir_bits dev)
        (cd_vseed_count dev) (cd_tag_count dev) (cd_paseed_count dev))
      (@bb84_bool_val M (bb84_n P) (bb84_l P) (cd_ir_bits dev)
        (cd_vseed_count dev) (cd_tag_count dev) (cd_paseed_count dev))
  ELSE
\end{minted}
\switchcolumn
\textbf{Parameter-estimation branch.}
\[
\mathsf{PEFail}(X,Z,\lambda)\Longrightarrow
(S,\hat S,\mathsf{abort})\gets(\perp,\perp,\mathsf{true});
\]
otherwise execution enters reconciliation.
\end{paracol}
\endgroup
\AlignmentBlockEnd

\AlignmentContinuationLabel{tab:alignment-program-modes-ec}
\AlignmentBlockHead{Source-mode evaluator (continued).}
\begingroup\AlignmentBlockSetup
\begin{paracol}{2}
\code{src/qkd/example/example.v}: \coqi{bb84_program_mode} (continued)
\begin{minted}{coq}
    (@ErrCorr R L HL (bb84_cvar M) V (cd_ir_bits dev) (@BVXA M) (@BVXB M)
       (@BVpec M) (@BVXhat M) (@BVg M) (@BVtA M) (@BVtB M)
       (@BVP M)
       (@bb84_tr_app_ec M (bb84_n P) (bb84_l P) (cd_ir_bits dev)
        (cd_vseed_count dev) (cd_tag_count dev) (cd_paseed_count dev))
       (@bb84_ir_pub_val M (bb84_n P) (bb84_l P) (cd_ir_bits dev)
        (cd_vseed_count dev) (cd_tag_count dev) (cd_paseed_count dev)
        (@cd_ir_pub R P q dev))
       (@bb84_ir_dec_val M (bb84_n P) (bb84_l P) (cd_ir_bits dev)
        (cd_vseed_count dev) (cd_tag_count dev) (cd_paseed_count dev)
        (@cd_ir_dec R P q dev))
       (@bb84_unif_vseed_mass R M (bb84_n P) (bb84_l P)
        (cd_ir_bits dev) (cd_vseed_count dev) (cd_tag_count dev)
        (cd_paseed_count dev))
       (@bb84_verify_hash_val M (bb84_n P) (bb84_l P) (cd_ir_bits dev)
        (cd_vseed_count dev) (cd_tag_count dev) (cd_paseed_count dev)
        (@cd_vh R P q dev)) ;;
\end{minted}
\switchcolumn
\textbf{Error correction and verification tags.}
\[
(p_{\rm EC},\hat X,g,t_A,t_B,P)
\gets\mathsf{ErrCorr}(X_A,X_B;P).
\]
The public reconciliation map, decoder, uniform verification seed, and
two-universal verification hash are the certified-device operations.
\end{paracol}
\endgroup
\AlignmentBlockEnd

\AlignmentContinuationLabel{tab:alignment-program-modes-tag-pa-eval}
\AlignmentBlockHead{Source-mode evaluator (continued).}
\begingroup\AlignmentBlockSetup
\begin{paracol}{2}
\code{src/qkd/example/example.v}: \coqi{bb84_program_mode},
\coqi{bb84_eval_mode} (continued)
\begin{minted}{coq}
     IF (fun s : @store (bb84_cvar M) V =>
           s (@BVtA M) != s (@BVtB M)) THEN
       @Abort_cmd R L HL (bb84_cvar M) V (@BVS M) (@BVSh M) (@BVabort M)
        (@bb84_bot_val M (bb84_n P) (bb84_l P) (cd_ir_bits dev)
          (cd_vseed_count dev) (cd_tag_count dev) (cd_paseed_count dev))
        (@bb84_bool_val M (bb84_n P) (bb84_l P) (cd_ir_bits dev)
          (cd_vseed_count dev) (cd_tag_count dev) (cd_paseed_count dev))
     ELSE
       @PrivAmp R L HL (bb84_cvar M) V (bb84_l P) (@BVXA M) (@BVXhat M)
         (@BVf M) (@BVS M) (@BVSh M) (@BVP M)
         (@bb84_tr_app_pa M (bb84_n P) (bb84_l P) (cd_ir_bits dev)
          (cd_vseed_count dev) (cd_tag_count dev) (cd_paseed_count dev))
         (@bb84_unif_pseed_mass R M (bb84_n P) (bb84_l P)
          (cd_ir_bits dev) (cd_vseed_count dev) (cd_tag_count dev)
          (cd_paseed_count dev))
         (@bb84_privacy_hash_val M (bb84_n P) (bb84_l P)
          (cd_ir_bits dev) (cd_vseed_count dev) (cd_tag_count dev)
          (cd_paseed_count dev) (@cd_pa R P q dev))).

Definition bb84_eval_mode (entangled : bool)
    (EA : @adversary R L HL V Eve)
    (xi : @cq_state R L HL (bb84_cvar M) V)
    : @cq_state R L HL (bb84_cvar M) V :=
  @denote R L HL (bb84_cvar M) V Eve (@BVP M) EA
    (bb84_program_mode entangled) xi.
\end{minted}
\AlignmentBoundaryLabel
The definitions of every command named in \coqi{bb84_program_mode} are expanded in
\Cref{tab:alignment-language,tab:alignment-qkd-classical-components}.
The concrete operands are expanded in
\Cref{tab:alignment-program-physical-adapters,tab:alignment-program-classical-adapters,tab:alignment-program-transcript-adapters}.
Its denotational dependencies \coqi{adversary}, \coqi{cq_state}, and \coqi{denote} are expanded in \Cref{tab:alignment-denotation}.
The concrete register, value, and labelled-space arguments are expanded in
\Cref{tab:alignment-program-schema,tab:alignment-label-layout}.
\switchcolumn
\textbf{Tag check, privacy amplification, and evaluation.}
\[
t_A\ne t_B\Longrightarrow
(S,\hat S,\mathsf{abort})\gets(\perp,\perp,\mathsf{true});
\]
otherwise
\[
\begin{gathered}
(f,S,\hat S,P)\gets\mathsf{PrivAmp}(X_A,\hat X;P),\\
\xi_{\rm out}^{b}=\sem{\mathsf{BB84}^{b}}_{\mathcal E_A}(\xi_{\rm init}).
\end{gathered}
\]
\end{paracol}
\endgroup
\AlignmentBlockEnd

\AlignmentParagraph{Output.}

\AlignmentContinuationLabel{tab:alignment-output-indices}
\AlignmentBlockHead{Output-register indices.}
\begingroup\AlignmentBlockSetup
\begin{paracol}{2}
\code{src/qkd/example/readout.v}: public/key index decoders
\begin{minted}{coq}
Definition bb84_public_idx
    (p : bb84_eur_public (bb84_round_limit P) (cd_vseed_count dev))
    : 'Idx[HL]_bb84_Epub :=
  n2i (cast_ord (esym bb84_dim_Epub) (enum_rank p)).
Definition bb84_Fl_index (f : 'I_(cd_paseed_count dev))
    : 'Idx[HL]_bb84_Fl := n2i (cast_ord bb84_dim_Fl f).
Definition bb84_pa_output_idx (z : 'I_(2 ^ bb84_l P))
    : 'Idx[HL]_bb84_Zl :=
  n2i (cast_ord (esym bb84_dim_Zl) z).
Definition bb84_Zl_index (z : 'I_(2 ^ bb84_l P))
    : 'Idx[HL]_bb84_Zl := bb84_pa_output_idx z.
Definition bb84_Ppub_index
    (p : ('I_(2 ^ cd_ir_bits dev) * 'I_(cd_tag_count dev))%type)
    : 'Idx[HL]_bb84_Ppub :=
  n2i (cast_ord bb84_dim_Ppub_pair (enum_rank p)).
Definition bb84_ec_default : 'I_(2 ^ cd_ir_bits dev).
Proof. exact: Ordinal (expn_gt0 2 (cd_ir_bits dev)). Defined.
Definition bb84_ec_index (v : V) : 'I_(2 ^ cd_ir_bits dev) :=
  odflt bb84_ec_default (bb84_ec_of v).
Definition bb84_tag_default : 'I_(cd_tag_count dev).
Proof. exact: Ordinal (cd_tag_pos dev). Defined.
Definition bb84_tag_index (v : V) : 'I_(cd_tag_count dev) :=
  odflt bb84_tag_default (bb84_tag_of v).
Definition bb84_paseed_default : 'I_(cd_paseed_count dev).
Proof. exact: Ordinal (cd_paseed_pos dev). Defined.
Definition bb84_paseed_index (v : V) : 'I_(cd_paseed_count dev) :=
  odflt bb84_paseed_default (bb84_paseed_of v).
Definition bb84_final_key_default : 'I_(2 ^ bb84_l P).
Proof. exact: Ordinal (expn_gt0 2 (bb84_l P)). Defined.
Definition bb84_final_key_index (v : V) : 'I_(2 ^ bb84_l P) :=
  odflt bb84_final_key_default (bb84_final_key_of v).
\end{minted}
\switchcolumn
\textbf{Finite output-register indices.}
Classical public data, PA seeds, final keys, reconciliation strings, and tags
are transported to the matching computational-basis index.
The decoders map values to the displayed finite-register indices.
\end{paracol}
\endgroup
\AlignmentBlockEnd

\AlignmentContinuationLabel{tab:alignment-output-reification}
\AlignmentBlockHead{Accepted output reification.}
\begingroup\AlignmentBlockSetup
\begin{paracol}{2}
\code{src/qkd/example/program.v}: \coqi{bb84_eur_public_of_transcript}, \coqi{bb84_eur_public_of_value}, \coqi{bb84_eur_public_of_store};
\code{src/qkd/example/readout.v}: \coqi{bb84_public_idx}, \coqi{bb84_public_projector}, \coqi{bb84_public_register_state}, \coqi{bb84_ec_index}, \coqi{bb84_tag_index}, \coqi{bb84_paseed_index}, \coqi{bb84_final_key_index}, \coqi{bb84_Fl_projector}, \coqi{bb84_Zl_projector}, \coqi{bb84_Ppub_projector}, \coqi{bb84_leakage_register_state}, \coqi{bb84_pa_public_register_state};
\code{src/qkd/example/example.v}: \coqi{rhoZFEpub_eval_mode}
\begin{minted}{coq}
Definition bb84_public_projector p : 'F[HL]_bb84_Epub :=
  [> deltav (bb84_public_idx p); deltav (bb84_public_idx p) <].
Definition bb84_public_register_state (s : store) : 'F[HL]_bb84_Epub :=
  bb84_public_projector (bb84_eur_public_of_store s).
Definition bb84_ec_index (v : V) := odflt bb84_ec_default (bb84_ec_of v).
Definition bb84_tag_index (v : V) := odflt bb84_tag_default (bb84_tag_of v).
Definition bb84_paseed_index (v : V) :=
  odflt bb84_paseed_default (bb84_paseed_of v).
Definition bb84_final_key_index (v : V) :=
  odflt bb84_final_key_default (bb84_final_key_of v).
Definition bb84_Fl_projector f : 'F[HL]_bb84_Fl :=
  [> deltav (bb84_Fl_index f); deltav (bb84_Fl_index f) <].
Definition bb84_Zl_projector z : 'F[HL]_bb84_Zl :=
  [> deltav (bb84_Zl_index z); deltav (bb84_Zl_index z) <].
Definition bb84_Ppub_projector p : 'F[HL]_bb84_Ppub :=
  [> deltav (bb84_Ppub_index p); deltav (bb84_Ppub_index p) <].
Definition bb84_leakage_register_state (s : store) :=
  bb84_Ppub_projector
    (bb84_ec_index (s BVpec), bb84_tag_index (s BVtA)).
Definition bb84_pa_public_register_state (s : store) :=
  bb84_Zl_projector (bb84_final_key_index (s BVS)) \⊗
    (bb84_Fl_projector (bb84_paseed_index (s BVf)) \⊗
      (bb84_public_register_state s \⊗ bb84_leakage_register_state s)).
Definition rhoZFEpub_eval_mode (entangled : bool) EA xi :=
  let xi_out := bb84_eval_mode entangled EA xi in
  let st := \sum_(s | s BVabort == bb84_bool_val false)
    ptrace_cast_gen (bb84_setT_diff_complement
      (Zl :|: (Fl :|: (Ereg :|: Ppub))))
      (replaceso (bb84_pa_public_register_state s) (xi_out s)) in
  ((@complex.Re R (\Tr st))^-1)%:C *: st.
\end{minted}
The replacement first traces the old public registers out of
$\xi_{\rm out}^{b}(s)$ and then appends their classical reification.
The EUR-public transcript projection used by \coqi{bb84_public_register_state} is expanded in \Cref{tab:alignment-program-transcript-adapters}.
The index decoders and their finite-register transports are expanded
in \Cref{tab:alignment-output-indices}.
The replacement channel \coqi{replaceso} is expanded in \Cref{tab:alignment-denotation}, and \coqi{bb84_eval_mode} in \Cref{tab:alignment-program-modes}.
\switchcolumn
\textbf{Classical-store reification and non-abort normalization.}
For the output mode $b$,
\[
\begin{aligned}
\rho_{S F_{\rm PA}E^+}^{b}
  &=\bigl(\Re\Tr\Omega_b\bigr)^{-1}\Omega_b,
  \qquad 0^{-1}=0,\\
Y_s^b
  &:=\operatorname{Replace}_{Q_{\rm pub}}
      (R_{\rm pub}(s),\xi_{\rm out}^{b}(s)),\\
\Omega_b
  &=\sum_{s:\neg\mathsf{abort}(s)}\Tr_{\rm drop}(Y_s^b),\\
\operatorname{Replace}_{Q}(R,x)
  &:=\operatorname{cast}\bigl(R\otimes\Tr_Q x\bigr).
\end{aligned}
\]
Here the final key, privacy-amplification seed, transcript, reconciliation message, and verification tag are reified as computational-basis projectors.
\end{paracol}
\endgroup
\AlignmentBlockEnd

\AlignmentParagraph{Premises.}

\AlignmentContinuationLabel{tab:alignment-device-certificate}
\AlignmentBlockHead{Certified device interface.}
\begingroup\AlignmentBlockSetup
\begin{paracol}{2}
\code{src/qkd/example/device.v}: \coqi{bb84_certified_device}
\begin{minted}{coq}
Record bb84_certified_device (P : bb84_protocol) (q : R)
    := BB84CertifiedDevice {
  cd_src : bb84_source_model P q;
  cd_src_dim_A : bb84_src_na cd_src = (2 ^ bb84_n P)%N;
  cd_src_dim_X : bb84_src_nx cd_src = (2 ^ bb84_n P)%N;
  cd_src_nb_pos : (0 < bb84_src_nb cd_src)%N;
  cd_src_overlap_standard_z :
    (bb84_n P)%:R * q <= q_of_mx
      (bb84_source_x_povm cd_src_dim_A cd_src_dim_X)
      (@std_projective_povm R (2 ^ bb84_n P));

  cd_rho : bb84_basis -> bool -> 'M[C]_2;
  cd_rho_den1 : forall basis bit, cd_rho basis bit \is den1mx;

  cd_meas : bb84_basis -> option bool -> 'M[C]_2;
  cd_meas_complete : forall basis,
    (cd_meas basis None)^*t *m cd_meas basis None +
    (cd_meas basis (Some false))^*t *m cd_meas basis (Some false) +
    (cd_meas basis (Some true))^*t *m cd_meas basis (Some true) = 1%:M;
  cd_detect_basis_indep :
    (cd_meas BasisX None)^*t *m cd_meas BasisX None =
    (cd_meas BasisZ None)^*t *m cd_meas BasisZ None;

  cd_entangled : bool;
  cd_pair : 'M[C]_(4 * 2);
  cd_pair_den1 : cd_pair \is den1mx;
  cd_ptr : bb84_basis -> bool -> 'M[C]_4;
  cd_ptr_complete : forall basis,
    (cd_ptr basis false)^*t *m cd_ptr basis false +
    (cd_ptr basis true)^*t *m cd_ptr basis true = 1%:M;

  cd_ir_bits : nat;
  cd_ir_pub :
    'I_(2 ^ bb84_n P) -> 'I_(2 ^ bb84_n P) -> 'I_(2 ^ cd_ir_bits);
  cd_ir_dec :
    'I_(2 ^ cd_ir_bits) -> 'I_(2 ^ bb84_n P) -> 'I_(2 ^ bb84_n P);
  cd_ir_leak : (cd_ir_bits)%:R <= bb84_leak_EC P;

  cd_vseed_count : nat;
  cd_vseed_pos : (0 < cd_vseed_count)%N;
  cd_tag_count : nat;
  cd_vh :
    'I_cd_vseed_count -> 'I_(2 ^ bb84_n P) -> 'I_cd_tag_count;
  cd_vh_two_univ : @two_univ R _ _ _ cd_vh;
  cd_tag_pos : (0 < cd_tag_count)%N;
  cd_tag_le : ((cd_tag_count%:R : R)^-1 <= bb84_eps_c P)%R;
  cd_tag_bits :
    log2 (cd_tag_count%:R : R) <= log2 (2%:R / bb84_eps_c P);

  cd_paseed_count : nat;
  cd_paseed_pos : (0 < cd_paseed_count)%N;
  cd_pa :
    'I_cd_paseed_count -> 'I_(2 ^ bb84_n P) -> 'I_(2 ^ bb84_l P);
  cd_pa_two_univ : @two_univ R _ _ _ cd_pa
}.
\end{minted}
\coqi{bb84_source_model} is expanded in \Cref{tab:alignment-source-certificate};
\coqi{two_univ} in \Cref{tab:alignment-ingredients}; and the constructed
\coqi{bb84_source_x_povm} in \Cref{tab:alignment-device-primitives}.
The record packages the certified device properties.
Its field \coqi{cd_detect_basis_indep} states the equality of no-click POVM
effects assumed by Tomamichel et al.~\cite{tomamichel2012tight}.
\switchcolumn
\textbf{Certified device assumptions.}
The certificate packages the source overlap, normalized preparation states,
complete three-outcome measurements, basis-independent no-click probability,
an EB pair with complete pointer readout, and classical reconciliation,
verification, and privacy-amplification maps:
\[
\sum_{o\in\{\emptyset,0,1\}}K_{b,o}^{\dagger}K_{b,o}=I,
\qquad
K_{\mathbb X,\emptyset}^{\dagger}K_{\mathbb X,\emptyset}
=K_{\mathbb Z,\emptyset}^{\dagger}K_{\mathbb Z,\emptyset}.
\]
For verification and privacy amplification,
\[
\begin{gathered}
\Pr_g[g(x)=g(x')]\le |\mathcal T|^{-1},\\
\Pr_f[f(x)=f(x')]\le2^{-\ell},
\end{gathered}
\]
with
$d_{\rm IR}\le\leakec$,
$|\mathcal T|^{-1}\le\eps_c$, and
$\log_2|\mathcal T|\le\log_2(2/\eps_c)$.
\end{paracol}
\endgroup
\AlignmentBlockEnd

\AlignmentContinuationLabel{tab:alignment-device-primitives}
\AlignmentBlockHead{Device-class primitives: projective measurement and matrix helpers.}
\begingroup\AlignmentBlockSetup
\begin{paracol}{2}
\code{src/utils/mx/def.v}: \coqi{castmx_sq}, \coqi{bigten};
\code{src/eur/def.v}: \coqi{std_projective_povm}, \coqi{xbasis_sign2},
\coqi{hadamard_mx2}, \coqi{hadamard_mx}, \coqi{x_projective_povm};
\code{src/qkd/example/device.v}: \coqi{bb84_source_x_povm};
\code{src/qkd/example/deferred.v}: \coqi{bit_of_ord}, \coqi{bito},
\coqi{bitprod_mx}
\begin{minted}{coq}
Definition castmx_sq {n n' : nat} (e : n = n')
    (A : 'M[C]_n) : 'M[C]_(n') :=
  castmx (e, e) A.
Fixpoint bigten {n : nat} (A : 'M[C]_n) (k : nat)
    : 'M[C]_(n ^ k) :=
  match k with
  | 0 => 1%:M
  | S k' => castmx (esym (expnSr n k'), esym (expnSr n k'))
      ((bigten A k') *t A)
  end.
Definition bit_of_ord {d : nat} (x : 'I_(2 ^ d)) (j : 'I_d) : bool :=
  ord_to_bitstring x j.
Definition bito (b : bool) : 'I_2 :=
  if b then Ordinal (isT : (1 < 2)%N)
  else Ordinal (isT : (0 < 2)%N).
Variable d : nat.
Definition bitprod_mx (F : 'I_d -> 'M[C]_2) : 'M[C]_(2 ^ d) :=
  \matrix_(z, z')
    \prod_(j < d) F j (bito (bit_of_ord z j)) (bito (bit_of_ord z' j)).
Local Notation pidx b := (bb84_prep_index BasisX b).
Definition std_projective_povm (n : nat) : 'I_n -> 'M[C]_n :=
  fun i => delta_mx i i.
Definition xbasis_sign2 (x z : 'I_2) : C :=
  if (nat_of_ord x == 0%N) || (nat_of_ord z == 0%N) then 1 else -1.
Definition hadamard_mx2 : 'M[C]_2 :=
  \matrix_(i, j) ((((Num.sqrt (2%:R))^-1)%R)%:C * xbasis_sign2 i j).
Definition hadamard_mx (n : nat) : 'M[C]_(2 ^ n) :=
  bigten hadamard_mx2 n.
Definition x_projective_povm (n : nat) :
    'I_(2 ^ n) -> 'M[C]_(2 ^ n) :=
  fun x => (hadamard_mx n)^*t *m delta_mx x x *m hadamard_mx n.
Definition cast_square_povm_mx {n m} (e : n = m)
    (M : 'I_m -> 'M[C]_m) : 'I_n -> 'M[C]_n :=
  fun i => castmx_sq (esym e) (M (cast_ord e i)).
Definition bb84_source_x_povm (P : bb84_protocol) (q : R)
    (src : bb84_source_model P q)
    (dim_A : bb84_src_na src = (2 ^ bb84_n P)%N)
    (dim_X : bb84_src_nx src = (2 ^ bb84_n P)%N) :=
  cast_square_povm_mx (esym dim_X)
    (fun x => castmx
      (etrans dim_A (esym dim_X), etrans dim_A (esym dim_X))
      (@bb84_src_M R P q src x)).
\end{minted}
\switchcolumn
\textbf{Projective $X$ measurement and matrix helpers.}
\[
M_x^{\mathbb X}
  =H^{\otimes n\dagger}|x\rangle\!\langle x|H^{\otimes n}.
\]
The $n$-qubit source POVM is the source model's POVM transported to the standard $2^n$ basis.
The matrix helpers use
\[
\begin{gathered}
A^{\otimes0}=I,\qquad
A^{\otimes(k+1)}=A^{\otimes k}\otimes A,\\
\operatorname{cast}^{e}_{\square}(A)=\operatorname{cast}_{(e,e)}(A),\\
\iota(\mathsf{false})=0,\qquad \iota(\mathsf{true})=1,\\
\operatorname{bit}(z,j)=\operatorname{ordToBits}(z)_j,\\
\operatorname{pidx}(b)=\operatorname{PrepIndex}(\mathbb X,b),\\
\alpha_j=\iota(\operatorname{bit}(z,j)),\quad
\alpha'_j=\iota(\operatorname{bit}(z',j)),\\
\bigl[\operatorname{BitProd}(G)\bigr]_{z,z'}
=\prod_{j<d}(G_j)_{\alpha_j,\alpha'_j}.
\end{gathered}
\]
The source-model and device-field dependencies are expanded in
\Cref{tab:alignment-source-certificate,tab:alignment-device-certificate}.
\end{paracol}
\endgroup
\AlignmentBlockEnd

\AlignmentContinuationLabel{tab:alignment-device-primitives-pair}
\AlignmentBlockHead{Device-class primitives: coherent pair and click operator.}
\begingroup\AlignmentBlockSetup
\begin{paracol}{2}
\code{src/qkd/example/deferred.v}: \coqi{bitlift_mx}, \coqi{bb84_eb_z_sign},
\coqi{eb_pair_ket}, \coqi{bb84_eb_pair_of_pure}, \coqi{bb84_M_click},
\coqi{bb84_M_click_sqrt}
\begin{minted}{coq}
Definition bitlift_mx (j : 'I_d) (A : 'M[C]_2) : 'M[C]_(2 ^ d) :=
  bitprod_mx (fun k => if k == j then A else 1%:M).
Definition bb84_eb_z_sign (zbit : bool) : C :=
  if zbit then -1 else 1.
Local Definition eb_pair_ket (v0 v1 : 'cV[C]_2) : 'cV[C]_(4 * 2) :=
  (sqrtC (2%:R : C))^-1 *:
    (((delta_mx (pidx false) 0 : 'cV[C]_4) *t v0) +
     ((delta_mx (pidx true) 0 : 'cV[C]_4) *t v1)).
Definition bb84_eb_pair_of_pure (v0 v1 : 'cV[C]_2) : 'M[C]_(4 * 2) :=
  eb_pair_ket v0 v1 *m (eb_pair_ket v0 v1)^*t.
Definition bb84_M_click : 'M[C]_2 :=
  1%:M - (cd_meas dev BasisX None)^*t *m cd_meas dev BasisX None.
Definition bb84_M_click_sqrt : 'M[C]_2 :=
  mapf_mx sqrtC bb84_M_click.
\end{minted}
\switchcolumn
\textbf{Coherent pair and click operator.}
\[
M_{\rm click}=I-K_{\mathbb X,\emptyset}^{\dagger}
K_{\mathbb X,\emptyset}.
\]
For normalized qubit kets $v_0,v_1$, the canonical EB source is
\[
\begin{gathered}
|\Psi\rangle=\frac1{\sqrt2}
 (|{\mathbb X},0\rangle\otimes v_0+
  |{\mathbb X},1\rangle\otimes v_1),\\
\rho_{\rm pair}=|\Psi\rangle\!\langle\Psi|.
\end{gathered}
\]
The project definitions used here are expanded in
\Cref{tab:alignment-device-roots,tab:alignment-label-layout}; delta matrices,
tensor products, and casts are library operations.
\end{paracol}
\endgroup
\AlignmentBlockEnd

\AlignmentContinuationLabel{tab:alignment-device-class}
\AlignmentBlockHead{Projective $\mathbb X/\mathbb Z$ device class.}
\begingroup\AlignmentBlockSetup
\begin{paracol}{2}
\code{src/qkd/example/deferred.v}: \coqi{bb84_eb_ptr_standard};
\code{src/qkd/example/projective\_xz.v}: \coqi{bb84_pure_xz_source}, \coqi{bb84_projective_z_readout}, \coqi{bb84_projective_xz_device}
\begin{minted}{coq}
Definition bb84_eb_ptr_standard : Prop :=
  (forall b, cd_ptr dev BasisX b =
    bitlift_mx bb84_src_ptr_slot (delta_mx (bito b) (bito b))) /\
  (forall b, cd_ptr dev BasisZ b =
    bitlift_mx bb84_src_ptr_slot
      (hadamard_mx2 *m delta_mx (bito b) (bito b) *m hadamard_mx2^*t)).
Record bb84_pure_xz_source := {
  cx_ket : bool -> 'cV[C]_2;
  cx_ket_norm : forall b, (cx_ket b)^*t *m cx_ket b = 1%:M;
  cx_rho_X : forall b,
    cd_rho dev BasisX b = cx_ket b *m (cx_ket b)^*t;
  cx_rho_Z : forall z,
    cd_rho dev BasisZ z = ((2%:R)^-1)%:C *:
      ((cx_ket false + bb84_eb_z_sign R z *: cx_ket true) *m
       (cx_ket false + bb84_eb_z_sign R z *: cx_ket true)^*t)
}.
Definition bb84_projective_z_readout : Prop :=
  forall b, cd_meas dev BasisZ (Some b) =
    delta_mx (bito b) (bito b) *m bb84_M_click_sqrt.
Record bb84_projective_xz_device := {
  pxz_source : bb84_pure_xz_source;
  pxz_z_readout : bb84_projective_z_readout;
  pxz_ptr_standard : bb84_eb_ptr_standard;
  pxz_pair : cd_pair dev =
    bb84_eb_pair_of_pure (cx_ket pxz_source false)
      (cx_ket pxz_source true);
  pxz_source_x_povm : forall x,
    bb84_source_x_povm dev x = x_projective_povm (bb84_n P) x
}.
\end{minted}
The table displays every structural equality in the class.
The certified-device record and every helper definition used by these equalities are expanded in
\Cref{tab:alignment-device-certificate,tab:alignment-device-primitives}.
\switchcolumn
\textbf{Structural device class used by Protocol Security.}
There are normalized kets $|\psi_0\rangle,|\psi_1\rangle$ such that
\[
\rho_{\mathbb X,b}=|\psi_b\rangle\!\langle\psi_b|,
\quad
\rho_{\mathbb Z,z}=\tfrac12
(|\psi_0\rangle+(-1)^z|\psi_1\rangle)
(\cdots)^\dagger.
\]
Bob's $\mathbb Z$ readout is projective up to the click operator; the EB pointer is the standard computational/Hadamard readout, the device pair is the coherent purification of the two kets, and the source $\mathbb X$ POVM is standard projective.
\end{paracol}
\endgroup
\AlignmentBlockEnd

\AlignmentParagraph{Run.}

\AlignmentContinuationLabel{tab:alignment-replay-rounds}
\AlignmentBlockHead{Deterministic history replay.}
\begingroup\AlignmentBlockSetup
\begin{paracol}{2}
\code{src/qkd/example/run.v}: \coqi{round_choice}, \coqi{round_trace}, \coqi{run_history}, \coqi{hist_K}, \coqi{hist_g}, \coqi{hist_f}, \coqi{start_store}, \coqi{replay_rounds_through}, \coqi{replay_sift_X}, \coqi{replay_sift_Z}, \coqi{round_prepare_store}, \coqi{round_measure_store}, \coqi{round_sift_store}, \coqi{write_trace_round}, \coqi{rounds_store_seq}, \coqi{rounds_store}
\begin{minted}{coq}
Definition round_choice := ((V * V) * (V * V))%type.
Definition choice_a (c : round_choice) : V := c.1.1.
Definition choice_y (c : round_choice) : V := c.1.2.
Definition choice_b (c : round_choice) : V := c.2.1.
Definition choice_y' (c : round_choice) : V := c.2.2.
Definition round_trace := {ffun 'I_M -> round_choice}.
Definition run_history := (round_trace * ((V * V) * V))%type.
Definition hist_K (h : run_history) : V := h.2.1.1.
Definition hist_g (h : run_history) : V := h.2.1.2.
Definition hist_f (h : run_history) : V := h.2.2.
Definition start_store (s : store) : store :=
  store_set
    (store_set
      (store_set
        (store_set s Pvar empty_transcript)
        Xv (vset set0))
      Zv (vset set0))
    abortv (vbool false).
Definition replay_sift_X (s : store) (i : 'I_M) : bool :=
  (s (a_ i) == basisX) && (s (b_ i) == basisX) &&
  (s (d_ i) == vbool true).
Definition replay_sift_Z (s : store) (i : 'I_M) : bool :=
  (s (a_ i) == basisZ) && (s (b_ i) == basisZ) &&
  (s (d_ i) == vbool true).
Definition replay_rounds_through
    (i : 'I_M) (p : 'I_M -> bool) : {set 'I_M} :=
  [set j : 'I_M | (nat_of_ord j <= nat_of_ord i)%N && p j].
Definition round_prepare_store
    (c : round_choice) (i : 'I_M) (s : store) : store :=
  let s1 := store_set s (a_ i) (choice_a c) in
  store_set s1 (y_ i) (choice_y c).
Definition round_measure_store
    (c : round_choice) (i : 'I_M) (s : store) : store :=
  let s1 := store_set (round_prepare_store c i s) (b_ i) (choice_b c) in
  store_set s1 (y'_ i) (choice_y' c).
Definition round_sift_store
    (c : round_choice) (i : 'I_M) (s : store) : store :=
  let s4 := round_measure_store c i s in
  let s5 := store_set s4 (d_ i) (vbool (choice_y' c != outEmpty)) in
  let s6 := store_set s5 Xv
    (vset (replay_rounds_through i (replay_sift_X s5))) in
  store_set s6 Zv
    (vset (replay_rounds_through i (replay_sift_Z s6))).
Definition write_round_choice
    (c : round_choice) (i : 'I_M) (s : store) : store :=
  round_sift_store c i s.
Definition write_trace_round
    (t : round_trace) (i : 'I_M) (s : store) : store :=
  write_round_choice (t i) i s.
Definition rounds_store_seq
    (h : run_history) (r : seq 'I_M) (s : store) : store :=
  foldl (fun t i => write_trace_round h.1 i t) s r.
Definition rounds_store (h : run_history) (s : store) : store :=
  rounds_store_seq h (enum 'I_M) s.
\end{minted}
\switchcolumn
\textbf{A sampled history determines a classical replay store.}
A history records every round's $(a_i,y_i,b_i,y'_i)$ together with the
parameter-estimation subset $K$, verification seed $g$, and PA seed $f$.
Starting from an empty transcript, empty sift sets, and a false abort flag,
the replay writes each round and recomputes
\[
\begin{aligned}
\mathcal X_i&=\{j\le i:a_j=b_j=\mathbb X,\ y'_j\ne\emptyset\},\\
\mathcal Z_i&=\{j\le i:a_j=b_j=\mathbb Z,\ y'_j\ne\emptyset\}.
\end{aligned}
\]
\coqi{store} and \coqi{store_set} are expanded in
\Cref{tab:alignment-semantics-core}; the concrete variable and value encodings
are expanded in
\Cref{tab:alignment-program-schema,tab:alignment-value-universe,%
tab:alignment-value-decoders,tab:alignment-key-decoders}.
\end{paracol}
\endgroup
\AlignmentBlockEnd

\AlignmentContinuationLabel{tab:alignment-replay-postprocessing}
\AlignmentBlockHead{Replayed post-processing and final store.}
\begingroup\AlignmentBlockSetup
\begin{paracol}{2}
\code{src/qkd/example/run.v}: \coqi{announce_store}, \coqi{paramest_store}, \coqi{errcorr_seed_store}, \coqi{errcorr_store}, \coqi{privamp_store}, \coqi{abort_store}, \coqi{replay_sift_ok}, \coqi{replay_pe_error_high}, \coqi{replay_pe_fail}, \coqi{pe_failed}, \coqi{final_store}
\begin{minted}{coq}
Definition announce_store (s : store) : store :=
  store_set s Pvar
    (tr_app_rounds (fun i => (s (a_ i), s (b_ i), s (d_ i))) (s Pvar)).
Definition paramest_store (h : run_history) (s : store) : store :=
  let s1 := store_set s Kv (hist_K h) in
  let s2 := store_set s1 XA
    (bits_on (vset_of (s1 Kv)) (fun i => s1 (y_ i))) in
  let s3 := store_set s2 XB
    (bits_on (vset_of (s2 Kv)) (fun i => s2 (y'_ i))) in
  let s4 := store_set s3 ZA
    (test_bits_on (vset_of (s3 Zv)) (fun i => s3 (y_ i))) in
  let s5 := store_set s4 ZB
    (test_bits_on (vset_of (s4 Zv)) (fun i => s4 (y'_ i))) in
  let s6 := store_set s5 Pvar
    (tr_app_pe (s5 Kv) (s5 ZA) (s5 ZB) (s5 Pvar)) in
  store_set s6 lamv
    (vrat (err_count (s6 ZA) (s6 ZB)) #|vset_of (s6 Zv)|).
Definition errcorr_seed_store (s : store) : store :=
  let s1 := store_set s pECv (ir_pub leakEC (s XA) (s XB)) in
  store_set s1 Xhatv (ir_dec (s1 pECv) (s1 XB)).
Definition errcorr_store (h : run_history) (s : store) : store :=
  let s1 := errcorr_seed_store s in
  let s2 := store_set s1 gv (hist_g h) in
  let s3 := store_set s2 tAv (verify_hash (s2 gv) (s2 XA)) in
  let s4 := store_set s3 tBv (verify_hash (s3 gv) (s3 Xhatv)) in
  store_set s4 Pvar (tr_app_ec (s4 pECv) (s4 gv) (s4 tAv) (s4 Pvar)).
Definition privamp_store (h : run_history) (s : store) : store :=
  let s1 := store_set s fv (hist_f h) in
  let s2 := store_set s1 Sv (privacy_hash l (s1 fv) (s1 XA)) in
  let s3 := store_set s2 Shv (privacy_hash l (s2 fv) (s2 Xhatv)) in
  store_set s3 Pvar (tr_app_pa (s3 fv) (s3 Pvar)).
Definition abort_store (s : store) : store :=
  store_set (store_set (store_set s abortv (vbool true)) Sv bot) Shv bot.
Definition replay_sift_ok (s : store) : bool :=
  (n <= #|vset_of (s Xv)|)%N && (k <= #|vset_of (s Zv)|)%N.
Definition replay_pe_error_high (s : store) : bool :=
  Qtol * (#|vset_of (s Zv)|)%:R < (err_of (s lamv))%:R.
Definition replay_pe_fail (s : store) : bool :=
  ~~ replay_sift_ok s || replay_pe_error_high s.
Definition pe_failed (s : store) : bool := replay_pe_fail s.
Definition final_store (h : run_history) (s : store) : store :=
  let s0 := start_store s in
  let s1 := rounds_store h s0 in
  let s2 := announce_store s1 in
  let s3 := paramest_store h s2 in
  if pe_failed s3 then abort_store s3 else
  let s4 := errcorr_store h s3 in
  if s4 tAv != s4 tBv then abort_store s4 else privamp_store h s4.
\end{minted}
\coqi{run_history}, \coqi{hist_K}, \coqi{hist_g}, \coqi{hist_f},
\coqi{start_store}, and \coqi{rounds_store} are expanded in
\Cref{tab:alignment-replay-rounds}.
All concrete bit, transcript, reconciliation, verification, and PA adapters are expanded in
\Cref{tab:alignment-program-classical-adapters,tab:alignment-program-transcript-adapters}.
\switchcolumn
\textbf{Announcement, PE, reconciliation, verification, and PA replay.}
The deterministic replay computes the same classical tail as the sampled program:
\[
\begin{aligned}
&K,\ X_A,X_B,Z_A,Z_B,\lambda,\\
&p_{\rm EC},\widehat X,g,t_A,t_B,\\
&f,\ S=f(X_A),\widehat S=f(\widehat X).
\end{aligned}
\]
It aborts on a failed sift/PE test or unequal verification tags, otherwise it writes the two final keys and transcript.
\end{paracol}
\endgroup
\AlignmentBlockEnd

\AlignmentContinuationLabel{tab:alignment-run-objects}
\AlignmentBlockHead{Run probabilities and terminal branches.}
\begingroup\AlignmentBlockSetup
\begin{paracol}{2}
\code{src/qkd/example/run.v}: \coqi{output_mismatch_event}, \coqi{run_event_prob}, \coqi{run_abort_prob}, \coqi{run_pass_prob}, \coqi{run_output_mismatch_prob};
\code{src/qkd/example/readout.v}: \coqi{bb84_branch_shape}
\begin{minted}{coq}
Definition output_mismatch_event (s : store) : bool :=
  (s Sv != s Shv) && (s abortv == vbool false).
Definition run_event_prob (ev : store -> bool) (xi : cq_state) : R :=
  \sum_(s | ev s) @complex.Re R (\Tr (xi s)).
Definition run_abort_prob (xi : cq_state) : R :=
  run_event_prob (fun s => s abortv == vbool true) xi.
Definition run_pass_prob (xi : cq_state) : R :=
  run_event_prob (fun s => s abortv == vbool false) xi.
Definition run_output_mismatch_prob (xi : cq_state) : R :=
  run_event_prob output_mismatch_event xi.
Definition bb84_branch_shape (h : run_history) (s : store) : Prop :=
  let fs := final_store h s in
  let spe :=
    paramest_store h
      (announce_store (rounds_store h (start_store s))) in
  let sec := errcorr_store h spe in
  (fs BVabort = bb84_bool_val true /\
   fs BVS = bb84_bot_val /\ fs BVSh = bb84_bot_val) \/
  (fs BVabort = bb84_bool_val false /\
   fs BVS = bb84_privacy_hash_val (hist_f h) (sec BVXA) /\
   fs BVSh = bb84_privacy_hash_val (hist_f h) (sec BVXhat)).
\end{minted}
The complete replay-store definitions are expanded in
\Cref{tab:alignment-replay-rounds,tab:alignment-replay-postprocessing};
\coqi{bb84_privacy_hash_val} is expanded in \Cref{tab:alignment-program-transcript-adapters}.
The register and value arguments are expanded in
\Cref{tab:alignment-program-schema,tab:alignment-label-layout}.
\switchcolumn
\textbf{Event probabilities and branch shape.}
For output cq-state $\xi$,
\[
\begin{aligned}
\Pr_\xi[A]&=\sum_{s\in A}\Re\Tr\xi(s),\\
p_{\rm abort}&=\Pr[\mathsf{abort}],\\
p_{\rm pass}&=\Pr[\neg\mathsf{abort}],
\end{aligned}
\]
\[
p_{\rm mismatch}=\Pr[S\ne\widehat S\land\neg\mathsf{abort}].
\]
Every terminal store is either
$\mathsf{abort}=1,S=\widehat S=\perp$ or
$\mathsf{abort}=0$ with both keys equal to outputs of the total
privacy-amplification adapter.
\end{paracol}
\endgroup
\AlignmentBlockEnd

\subsection{Component Theorems}
\label{app:semantics-alignment-infrastructure}

The component security theorems precede the reusable foundations and proof
pipeline below.

\subsubsection{Secrecy}

\AlignmentParagraph{Theorem of Tomamichel et al.~\cite{tomamichel2012tight}.}

\AlignmentContinuationLabel{tab:alignment-bb84-secret}
\AlignmentBlockHead{Component theorem: Supplementary Theorem~2 and \coqi{bb84_secret}.}
\begingroup\AlignmentBlockSetup
\begin{paracol}{2}
\code{src/qkd/bb84/bb84.v}: \coqi{bb84_secret}
\begin{minted}{coq}
Theorem bb84_secret
    (P : bb84_protocol) (q eps_s : R)
    (Gsrc : bb84_zsource P)
    (src : bb84_source_model P q)
    (dim_S : bb84_src_nx src = dim ('H[HL]_S))
    (dim_E : bb84_src_nb src = dim ('H[HL]_E))
    (dim_Zk : bb84_src_nz src = dim ('H[HL]_Zkey))
    (dim_Bq : bb84_src_nc src = dim ('H[HL]_Bq))
    (dim_Zp_key : (2 ^ bb84_n P)%N = dim ('H[HL]_Zp))
    (BobZ : 'QC[HL]_(Bq, Zp))
    (p_abort : R)
    (Fh : 'Idx[HL]_Fl -> 'Idx[HL]_S -> 'Idx[HL]_Zl)
    (Fh_two_univ : @two_univ R _ _ _ Fh)
    (Fh_key_dim : dim ('H[HL]_Zl) = (2 ^ bb84_l P)%N)
    (Epub : {set L})
    (pub : bb84_public_leakage Epub
      (eur_rho_XB disSE (bb84_src_rho src) (@bb84_src_M R P q src)
        (bb84_register_pair_dim disSE dim_S dim_E))
      (bb84_leak_EC P + log2 (2%:R / bb84_eps_c P)))
    (disSEpub : [disjoint S & E :|: Epub])
    (disEpubF : [disjoint E :|: Epub & Fl])
    (disEpubZl : [disjoint E :|: Epub & Zl])
    (disZl_FEpub : [disjoint Zl & (Fl :|: (E :|: Epub))])
    (accept reject : 'F[HL]_(S :|: (E :|: Epub)))
    (rho_out : 'F[HL]_(Zl :|: (Fl :|: (E :|: Epub))))
    (Haccept_psd : accept \is psdlf)
    (Hreject_psd : reject \is psdlf)
    (Haccept_cq : @is_cq_state R L HL S (E :|: Epub) accept)
    (Hsplit :
      accept + reject =
        castlf (bb84_public_assoc S E Epub, bb84_public_assoc S E Epub)
          (bb84_public_state pub))
    (Hout_psd : rho_out \is psdlf)
    (Hout_tr : @complex.Re R (\Tr rho_out) <= 1)
    (Hout :
      0 < @complex.Re R (\Tr accept) ->
      rho_out =
        ((@complex.Re R (\Tr accept))^-1)%:C *:
          post_hash_state disSEpub Fh accept)
    (Hnum : bb84_numeric_side P) :
\end{minted}
The implication chain and conclusion continue in
\Cref{tab:alignment-bb84-secret-tail}; every statement object is expanded in
the preceding statement-closure subsection.
\switchcolumn
\textbf{Secrecy component (Supplementary Theorem~2).}
The source state is extended by a public classical register and split into
accepted and rejected blocks satisfying
\[
\begin{gathered}
\rho_{\rm acc},\rho_{\rm rej}\succeq0,\\
\rho_{\rm acc}+\rho_{\rm rej}=\rho_{PSE},\\
\rho_{\rm acc}\text{ is cq on }S,\\
\rho_{\rm out}\succeq0,\qquad \Tr\rho_{\rm out}\le1,\\
\Tr\rho_{\rm acc}>0\Longrightarrow\\
\rho_{\rm out}=(\Tr\rho_{\rm acc})^{-1}
\mathsf{PostHash}_{F}(\rho_{\rm acc}).
\end{gathered}
\]
The protocol is $\eps_s$-secret if
\[
\begin{aligned}
\ell\le
\max_{\substack{\eps>0,\ \bar\eps>0\\
                  \eps+\bar\eps\le\eps_s}}
\Biggl\lfloor{}
 &n\bigl(q-h(\Qtol+\mu(\eps))\bigr)\\
 &-2\log_2\frac1{2\bar\eps}-\leakec\\
 &-\log_2\frac2{\eps_c}
\Biggr\rfloor .
\end{aligned}
\]
where
$\mu(\eps)=\sqrt{\frac{n+k}{nk}\frac{k+1}{k}\ln\frac1\eps}$.
\end{paracol}
\endgroup
\AlignmentBlockEnd

\AlignmentParagraph{Conclusion.}

\AlignmentContinuationLabel{tab:alignment-bb84-secret-tail}
\AlignmentBlockHead{Secrecy theorem: premises and conclusion.}
\begingroup\AlignmentBlockSetup
\begin{paracol}{2}
\code{src/qkd/bb84/bb84.v}: \coqi{bb84_secret} (continued)
\begin{minted}{coq}
  (0 < bb84_p_pass_tot Gsrc ->
    bb84_delayed_measurement_condition BobZ
      (eur_rho_ZC disZBq (bb84_src_rho src) (@bb84_src_N R P q src)
        (bb84_register_pair_dim disZBq dim_Zk dim_Bq))
      (@bb84_rhoZZ_pass_mixH R L HL Zkey Zp disZZp
        P Gsrc
        (bb84_source_zkey_dim dim_Zk)
        dim_Zp_key)) ->
  bb84_pass_prob_le p_abort
    (bb84_p_pass_tot Gsrc * @complex.Re R (\Tr accept)) ->
  (0 < bb84_n P)%N -> (0 < bb84_k P)%N ->
  (0 < eps_s) && (eps_s <= 1) ->
  (bb84_l P)%:R <=
    sup (fun m : R =>
      exists eps ebar : R,
        0 < eps /\ eps < 1 /\ 0 < ebar /\ eps + ebar <= eps_s /\
        m = ((floor (bb84_key_length_ub (bb84_n P) (bb84_k P) q
          (bb84_Qtol P) eps ebar (bb84_leak_EC P)
          (bb84_eps_c P)))%:~R : R)) ->
  eps_secret rho_out eps_s p_abort.
\end{minted}
\AlignmentBoundaryLabel
The declaration states the optimized key-length premise for the integer key
target and the admissible secrecy-error interval.
\switchcolumn
\textbf{Secrecy premises and conclusion.}
Under
\[
\begin{gathered}
0<\eps_c\le1,\quad \leakec\ge0,\\
n,k>0,\quad 0<\eps_s\le1,\\
p_{\rm pass}^{\rm tot}>0\Longrightarrow
\mathsf{DelayedMeasurement}(\rho_{ZB},\rho_{ZZ'}),\\
0\le p_{\rm abort}\le1,\qquad
1-p_{\rm abort}\le p_{\rm pass}^{\rm tot}\Tr\rho_{\rm acc},
\end{gathered}
\]
and the unfurled optimized key-length inequality, the theorem concludes
$\eps_s$-secrecy of $\rho_{\rm out}$.
\end{paracol}
\endgroup
\AlignmentBlockEnd

\subsubsection{Correctness}

\AlignmentParagraph{Verification.}

\AlignmentContinuationLabel{tab:alignment-bb84-correct}
\AlignmentBlockHead{Component theorem: correctness from two-universal verification.}
\begingroup\AlignmentBlockSetup
\begin{paracol}{2}
\code{src/qkd/bb84/bb84.v}: \coqi{bb84_correct}
\begin{minted}{coq}
Theorem bb84_correct (eps_c : R) :
  (0 < tag_count)%N ->
  @two_univ R _ _ _ verify_hash ->
  (forall x x', 0 <= corr_prob x x') ->
  (\sum_x \sum_(x' | x != x') corr_prob x x' <= 1) ->
  (tag_count%:R)^-1 <= eps_c ->
  eps_correct p_mismatch eps_c.
\end{minted}
The transitive statement definitions \coqi{p_mismatch},
\coqi{eps_correct}, and \coqi{two_univ} are expanded in the first subsection.
\switchcolumn
\textbf{BB84 correctness (Supplementary Theorem~1).}
For a nonempty tag set $T$, a two-universal verification hash, and mismatch
weights $p(x,x')$ with
\[
p(x,x')\ge0,\qquad
\sum_{x\ne x'}p(x,x')\le1,\qquad |T|^{-1}\le\eps_c,
\]
two-universal verification gives
$\Pr[S\ne\hat S\wedge\neg\mathsf{abort}]\le\eps_c$.
\end{paracol}
\endgroup
\AlignmentBlockEnd

\subsection{Entropy}
\label{app:semantics-alignment-purified-entropy}

\subsubsection{Definitions}

\AlignmentParagraph{Extended logarithm.}

\AlignmentContinuationLabel{tab:alignment-proof-only-opening}
\AlignmentBlockHead{Extended logarithm.}
\begingroup\AlignmentBlockSetup
\begin{paracol}{2}
\code{src/utils/arith.v}: \coqi{elog2}
\begin{minted}{coq}
Definition elog2 (x : R) : extended R :=
  if (0 < x)%R then EFin (log2 x) else @ENInf R.
\end{minted}
This definition supports the entropy objects used later in the appendix.
\switchcolumn
\textbf{Extended logarithm used by entropy infrastructure.}
\[
\overline{\log}_2x=
\begin{cases}
\log_2x,&x>0,\\
-\infty,&x\le0.
\end{cases}
\]
\end{paracol}
\endgroup
\AlignmentBlockEnd

\AlignmentParagraph{Core quantities.}

\AlignmentContinuationLabel{tab:alignment-entropy}
\AlignmentBlockHead{Distance and conditional min-entropy definitions.}
\begingroup\AlignmentBlockSetup
\begin{paracol}{2}
\code{src/utils/labelled/def.v}: \coqi{sqrtH};
\code{src/utils/dist/fid.v}: \coqi{fidelity}, \coqi{embed_subnorm}, \coqi{gfidelity}, \coqi{gfid}, \coqi{alt_gfid_gen};
\code{src/utils/dist/pd.v}: \coqi{pur_dist}, \coqi{pur_dist_mx}
\begin{minted}{coq}
Definition sqrtH (A : 'F[H]_S) : 'F[H]_S :=
  sumoutp
    (fun i => (Num.sqrt (@complex.Re R (eigenval_all A i)))%:C)
    (eigenvec_all A) (eigenvec_all A).
Definition fidelity (rho tau : 'F[H]_S) : R :=
  @complex.Re R (trnorm (h2mx (sqrtH rho \o sqrtH tau))).
Definition embed_subnorm (rho : 'M[C]_n) : 'M[C]_(n + 1) :=
  block_mx rho 0 0 ((1 - \tr rho) *: 1%:M).
Definition gfidelity (rho tau : 'M[C]_n) : R :=
  @complex.Re R
    (trnorm (sqrtmx (embed_subnorm rho) *m
      sqrtmx (embed_subnorm tau))).
Definition gfid (rho sigma : 'F[H]_S) : R :=
  gfidelity (h2mx rho) (h2mx sigma).
Lemma alt_gfid_gen (rho tau : 'M[C]_n) :
  rho \is psdmx -> tau \is psdmx ->
  gfidelity rho tau =
    @complex.Re R (trnorm (sqrtmx rho *m sqrtmx tau)) +
    Num.sqrt (@complex.Re R (1 - \tr rho)) *
    Num.sqrt (@complex.Re R (1 - \tr tau)).
Definition pur_dist (rho tau : 'F[HL]_S) : R :=
  Num.sqrt (1 - (gfid rho tau)^2).
Definition pur_dist_mx (rho tau : 'M[C]_(na * nb)) : R :=
  Num.sqrt (1 - (gfidelity rho tau)^2).
\end{minted}
\AlignmentBoundaryLabel
Generalized fidelity is implemented by a one-dimensional block embedding.
For PSD inputs with trace at most one, \coqi{alt_gfid_gen} specializes to the
source's additive trace-deficit formula.
\switchcolumn
\begin{definition}[Purified Distance]\label{def:purified-distance}
For sub-normalized states $\rho$ and $\sigma$,
\[
\begin{aligned}
F(\rho,\sigma)
  &:=\lVert\sqrt\rho\sqrt\sigma\rVert_1,\\
\bar F(\rho,\sigma)
  &:=F(\rho,\sigma)
    +\sqrt{(1-\Tr\rho)(1-\Tr\sigma)},\\
P(\rho,\sigma)
  &:=\sqrt{1-\bar F(\rho,\sigma)^2}.
\end{aligned}
\]
The second term in $\bar F$ accounts for the trace deficit.
\end{definition}
\switchcolumn*
\code{src/entropy/pur/def.v}: \coqi{phi_min_feasible}, \coqi{phi_min}, \coqi{hmin}, \coqi{hminE}
\begin{minted}{coq}
Definition phi_min_feasible (rhoS : 'F[H]_S) (v : R) : Prop :=
  exists sigma : 'F[H]_(S :&: T),
    sigma \is psdlf /\
    (castlf (eqST, eqST) ((\1 : 'F_(S :\: T)) \O sigma) - rhoS) \is psdlf /\
    v = @complex.Re R (\Tr sigma).
Definition phi_min (rhoS : 'F[H]_S) : R :=
  inf [set v : R | phi_min_feasible rhoS v].
Definition hmin (rhoS : 'F[H]_S) : R :=
  - log2 (phi_min rhoS).
Definition hminE (rhoS : 'F[H]_S) : extended R :=
  - elog2 (phi_min rhoS).
\end{minted}
\AlignmentBoundaryLabel
The code uses the equivalent semidefinite feasibility form and provides both
real-valued \coqi{hmin} and extended-real \coqi{hminE}.
Its primitive is an infimum; \coqi{phi_min_primal_attained} in
\Cref{tab:alignment-convex-infrastructure} realizes the source minimum.
The intersection/difference register expression makes the definition total
before subset premises are supplied.
The project-local \coqi{log2} dependency is expanded in
\Cref{tab:alignment-arithmetic}; \coqi{elog2} is expanded in
\Cref{tab:alignment-proof-only-opening}.
\switchcolumn
\begin{definition}[Conditional Min-Entropy]\label{def:cond-min}
For sub-normalized $\rho_{AB}$,
\[
\begin{aligned}
\Hmin(A\given B)_\rho
  &:={\max}_{\substack{\sigma_B\ge0\\\Tr\sigma_B\le1}}
    \sup\{\lambda\in\R:\\[-2pt]
  &\hspace{3.5em}
      \rho_{AB}\le2^{-\lambda}\id_A\otimes\sigma_B\}\\
  &=-{\min}_{\substack{\sigma_B\ge0\\\Tr\sigma_B\le1}}\\
  &\hspace{2em}
    D_{\max}(\rho_{AB}\|\id_A\otimes\sigma_B),
\end{aligned}
\]
where
\[
D_{\max}(\rho\|\sigma)
  :=\inf\{\lambda\in\R:\rho\le2^\lambda\sigma\}.
\]
If the feasible set is empty, $D_{\max}(\rho\|\sigma):=+\infty$.
\end{definition}
\end{paracol}
\endgroup
\AlignmentBlockEnd

\AlignmentContinuationLabel{tab:alignment-conditional-max-entropy}
\AlignmentBlockHead{Conditional max-entropy definition.}
\begingroup\AlignmentBlockSetup
\begin{paracol}{2}
\code{src/entropy/pur/def.v}: \coqi{hmax_feasible}, \coqi{hmax}, \coqi{hmaxE}
\begin{minted}{coq}
Definition hmax_feasible (rhoS : 'F[H]_S) (v : R) : Prop :=
  exists sigma : 'F[H]_(S :&: T),
    sigma \is den1lf /\
    elog2 ((dim 'H[H]_(S :\: T))%:R *
      (gfid rhoS
        (castlf (eqST_max, eqST_max)
          ((((dim 'H[H]_(S :\: T))%:R^-1)%:C *:
            ((\1 : 'F_(S :\: T)) \O sigma))))) ^+ 2) = v%:E.
Definition hmax (rhoS : 'F[H]_S) : R :=
  sup [set v : R | hmax_feasible rhoS v].
Definition hmaxE (rhoS : 'F[H]_S) : extended R :=
  ereal_sup [set v : extended R |
    exists sigma : 'F[H]_(S :&: T),
      sigma \is den1lf /\
      v = elog2 ((dim 'H[H]_(S :\: T))%:R *
        (gfid rhoS
          (castlf (eqST_max, eqST_max)
            ((((dim 'H[H]_(S :\: T))%:R^-1)%:C *:
              ((\1 : 'F_(S :\: T)) \O sigma))))) ^+ 2)].
\end{minted}
\AlignmentBoundaryLabel
The labelled-register casts implement the paper's
$d_A\bar F(\rho,\pi_A\otimes\sigma)^2$ expression;
\coqi{hmaxE} preserves its infinite branches.
The generalized-fidelity dependency \coqi{gfid} is expanded in
\Cref{tab:alignment-entropy}; \coqi{log2} is expanded in
\Cref{tab:alignment-arithmetic}, and \coqi{elog2} in
\Cref{tab:alignment-proof-only-opening}.
\switchcolumn
\begin{definition}[Conditional Max-Entropy]\label{def:cond-max}
For sub-normalized $\rho_{AB}$,
\[
\begin{gathered}
\Hmax(A\given B)_\rho
  :=\max_{\substack{\sigma_B\ge0\\\Tr\sigma_B=1}}\\
  \log_2\!\left(
    d_A\bar F(\rho_{AB},\pi_A\otimes\sigma_B)^2
  \right),
\end{gathered}
\]
where $\pi_A:=\id_A/d_A$ and $d_A$ is the dimension of~$A$.
\end{definition}
\end{paracol}
\endgroup
\AlignmentBlockEnd

\AlignmentContinuationLabel{tab:alignment-smooth-definitions}
\AlignmentBlockHead{Smooth entropy and duality.}
\begingroup\AlignmentBlockSetup
\begin{paracol}{2}
\code{src/entropy/pur/def.v}: \coqi{eps_ball}, \coqi{smooth_hmin}, \coqi{smooth_hminE}, \coqi{smooth_hmax}, \coqi{smooth_hmaxE}
\begin{minted}{coq}
Definition eps_ball (eps : R) (rho tau : 'F[H]_S) : bool :=
  ((tau \is psdlf) && (0 < @complex.Re R (\Tr tau)) &&
    (@complex.Re R (\Tr tau) <= 1)) && (pur_dist tau rho <= eps).
Definition smooth_hmin (eps : R) (rhoS : 'F[H]_S) : R :=
  sup [set v : R | exists tauS : 'F[H]_S,
    eps_ball eps rhoS tauS /\ v = hmin T tauS].
Definition smooth_hminE (eps : R) (rhoS : 'F[H]_S) : extended R :=
  ereal_sup [set v : extended R | exists tauS : 'F[H]_S,
    eps_ball eps rhoS tauS /\ v = hminE T tauS].
Definition smooth_hmax (eps : R) (rhoS : 'F[H]_S) : R :=
  inf [set v : R | exists tauS : 'F[H]_S,
    eps_ball eps rhoS tauS /\ v = hmax T tauS].
Definition smooth_hmaxE (eps : R) (rhoS : 'F[H]_S) : extended R :=
  - ereal_sup [set v : extended R | exists tauS : 'F[H]_S,
    eps_ball eps rhoS tauS /\ v = - hmaxE T tauS].
\end{minted}
\AlignmentBoundaryLabel
The witnesses in \Cref{tab:alignment-smooth-infrastructure} realize the
displayed extrema throughout the strict smoothing domain.
The dependencies \coqi{pur_dist}, \coqi{hmin}, and \coqi{hminE} are expanded in
\Cref{tab:alignment-entropy}; \coqi{hmax} and \coqi{hmaxE} are expanded in
\Cref{tab:alignment-conditional-max-entropy}.
\switchcolumn
\begin{definition}[Smooth Min- and Max-Entropy]
\label{def:smooth-min}\label{def:smooth-max}
For sub-normalized $\rho_{AB}$ and $0\le\eps<\sqrt{\Tr\rho}$,
\[
\begin{aligned}
\Hmineps[\eps](A\given B)_\rho
  &:=\max_{\bar\rho\in B^\eps(\rho_{AB})}
    \Hmin(A\given B)_{\bar\rho},\\
\Hmaxeps[\eps](A\given B)_\rho
  &:=\min_{\bar\rho\in B^\eps(\rho_{AB})}
    \Hmax(A\given B)_{\bar\rho},
\end{aligned}
\]
where
\[
B^\eps(\rho)
  :=\{\bar\rho\ge0:\Tr\bar\rho\le1,\ P(\bar\rho,\rho)\le\eps\}.
\]
\end{definition}
\switchcolumn*
\code{src/entropy/pur/smooth/prop.v}: \coqi{smooth_dual}
\begin{minted}{coq}
Lemma smooth_dual (eps : R) (rhoABC : 'F[H]_S_ABC) :
  rhoABC \is denlf ->
  \Rank rhoABC = 1%N ->
  0 <= eps < Num.sqrt (@complex.Re R (\Tr rhoABC)) ->
  smooth_hmaxE B eps
    (@ptrace_cast3C R L H A B C0 disAC disBC rhoABC) =
  oppe (smooth_hminE C0 eps
    (@ptrace_cast3B R L H A B C0 disAB disBC rhoABC)).
\end{minted}
\AlignmentBoundaryLabel
This is the extended-real form; its enclosing section carries explicit disjointness, copy-register, and dimension premises used to represent the two marginals and their purifications.
\switchcolumn
\textbf{Smooth duality (\Cref{sec:findings}).}
For pure $\rho_{ABC}$,
$\Hmaxeps(A\given B)_\rho=-\Hmineps(A\given C)_\rho$.
\end{paracol}
\endgroup
\AlignmentBlockEnd

\AlignmentContinuationLabel{tab:alignment-zero-radius-smoothing}
\AlignmentBlockHead{Zero-radius smoothing laws.}
\begingroup\AlignmentBlockSetup
\begin{paracol}{2}
\code{src/entropy/pur/prop.v}: \coqi{eps_ball_refl};
\code{src/entropy/pur/smooth/prop.v}: \coqi{eps_ball_0_eq},
\coqi{smooth_hminE_0}, \coqi{smooth_hmaxE_0}
\begin{minted}{coq}
Lemma eps_ball_refl (eps : R) (rho : 'F[H]_S) :
  rho \is psdlf -> 0 < @complex.Re R (\Tr rho) ->
  @complex.Re R (\Tr rho) <= 1 -> 0 <= eps ->
  eps_ball eps rho rho.
Lemma eps_ball_0_eq (rho tau : 'F[H]_S) :
  rho \is psdlf -> @complex.Re R (\Tr rho) <= 1 ->
  eps_ball 0 rho tau -> tau = rho.
Lemma smooth_hminE_0 (rho : 'F[H]_S) :
  rho \is psdlf -> @complex.Re R (\Tr rho) <= 1 ->
  0 < @complex.Re R (\Tr rho) ->
  smooth_hminE T 0 rho = hminE T rho.
Lemma smooth_hmaxE_0 (rho : 'F[H]_S) :
  rho \is psdlf -> @complex.Re R (\Tr rho) <= 1 ->
  0 < @complex.Re R (\Tr rho) ->
  smooth_hmaxE T 0 rho = hmaxE T rho.
\end{minted}
These declarations recover the zero-radius smoothing identities.
\switchcolumn
\textbf{Zero-radius ball and entropy recovery.}
For a nonzero subnormalized state,
\[
\begin{aligned}
B^0(\rho)&=\{\rho\},\\
H_{\min}^0(A|B)_\rho&=H_{\min}(A|B)_\rho,\\
H_{\max}^0(A|B)_\rho&=H_{\max}(A|B)_\rho.
\end{aligned}
\]
\end{paracol}
\endgroup
\AlignmentBlockEnd

\subsubsection{State Distance}

\AlignmentParagraph{Metric laws.}

\AlignmentContinuationLabel{tab:alignment-distance-infrastructure}
\AlignmentBlockHead{Reusable state-distance infrastructure.}
\begingroup\AlignmentBlockSetup
\begin{paracol}{2}
\code{src/utils/dist/fid.v}: \coqi{gfid_sym_mx};
\code{src/utils/dist/pd.v}: \coqi{pur_dist_refl}, \coqi{pd_metric_triangle}, \coqi{pur_dist_eq0}
\begin{minted}{coq}
Lemma gfid_sym_mx (rho tau : 'M[C]_n) :
  rho \is psdmx -> tau \is psdmx ->
  gfidelity rho tau = gfidelity tau rho.
Lemma pur_dist_refl (rho : 'F[HL]_S) :
  rho \is psdlf -> @complex.Re R (\Tr rho) <= 1 ->
  pur_dist rho rho = 0.
Lemma pd_metric_triangle (rho tau sigma : 'F[HLtri]_Stri) :
  rho \is psdlf -> tau \is psdlf -> sigma \is psdlf ->
  @complex.Re R (\Tr rho) <= 1 ->
  @complex.Re R (\Tr tau) <= 1 ->
  @complex.Re R (\Tr sigma) <= 1 ->
  pur_dist rho tau <= pur_dist rho sigma + pur_dist sigma tau.
Lemma pur_dist_eq0 (rho tau : 'F[HL]_S) :
  rho \is psdlf -> tau \is psdlf ->
  @complex.Re R (\Tr rho) <= 1 ->
  @complex.Re R (\Tr tau) <= 1 ->
  pur_dist rho tau = 0 -> rho = tau.
\end{minted}
\AlignmentBoundaryLabel
Generalized-fidelity symmetry yields purified-distance symmetry by unfolding
\coqi{pur_dist} and \coqi{gfid}; the other three metric laws are direct
reusable declarations.
\switchcolumn
\textbf{Purified distance is a metric (reusable infrastructure claimed in \Cref{sec:architecture}).}
On sub-normalized states,
\[
\begin{aligned}
P(\rho,\rho)&=0,\\
P(\rho,\tau)&=P(\tau,\rho),\\
P(\rho,\tau)=0&\Rightarrow\rho=\tau,
\end{aligned}
\]
and
$P(\rho,\tau)\le P(\rho,\sigma)+P(\sigma,\tau)$.
\switchcolumn*
\code{src/utils/dist/pd.v}: \coqi{pd_gtd_bar_lb}, \coqi{pd_gtd_ub},
\coqi{pd_gtd_ub'}
\begin{minted}{coq}
Lemma pd_gtd_bar_lb (rho tau : 'F[HL]_S) :
  rho \is psdlf -> tau \is psdlf ->
  @complex.Re R (\Tr rho) <= 1 ->
  @complex.Re R (\Tr tau) <= 1 ->
  trnormH (rho - tau) / 2%:R +
    `|@complex.Re R (\Tr rho) - @complex.Re R (\Tr tau)| / 2%:R
  <= pur_dist rho tau.
Lemma pd_gtd_ub (rho tau : 'F[HL]_S) :
  rho \is psdlf -> tau \is psdlf ->
  @complex.Re R (\Tr rho) <= 1 ->
  @complex.Re R (\Tr tau) <= 1 ->
  pur_dist rho tau <=
  Num.sqrt (2%:R *
    ((trnormH (rho - tau) +
      `|@complex.Re R (\Tr rho) - @complex.Re R (\Tr tau)|) / 2%:R) -
    ((trnormH (rho - tau) +
      `|@complex.Re R (\Tr rho) - @complex.Re R (\Tr tau)|) / 2%:R) ^ 2).
Lemma pd_gtd_ub' (rho tau : 'F[HL]_S) :
  rho \is psdlf -> tau \is psdlf ->
  @complex.Re R (\Tr rho) <= 1 ->
  @complex.Re R (\Tr tau) <= 1 ->
  pur_dist rho tau <=
    Num.sqrt (trnormH (rho - tau) +
      `|@complex.Re R (\Tr rho) - @complex.Re R (\Tr tau)|).
\end{minted}
The declarations inline the displayed expression for $\bar D$.
\switchcolumn
\textbf{Generalized trace-distance bounds.}
Writing $\bar D(\rho,\tau)=\frac12\lVert\rho-\tau\rVert_1+\frac12|\Tr\rho-\Tr\tau|$,
\[
\begin{gathered}
  \bar D(\rho,\tau)\le P(\rho,\tau)\\
  \le\sqrt{2\bar D(\rho,\tau)-\bar D(\rho,\tau)^2}\\
  \le\sqrt{2\bar D(\rho,\tau)}.
\end{gathered}
\]
\end{paracol}
\endgroup
\AlignmentBlockEnd

\AlignmentParagraph{Channels and extensions.}

\AlignmentContinuationLabel{tab:alignment-distance-channels}
\AlignmentBlockHead{Channel dilation and distance monotonicity.}
\begingroup\AlignmentBlockSetup
\begin{paracol}{2}
\code{src/utils/dist/dilation.v}: \coqi{stinespring};
\code{src/utils/dist/pd.v}: \coqi{pd_mono}
\begin{minted}{coq}
Lemma stinespring (F : 'SO[H]_(U, V)) :
  F \is cpmap <->
  exists (z : nat)
      (W : 'M[C]_(dim 'H[H]_V * z, dim 'H[H]_U)),
    [/\ z = dim 'H[H]_U * dim 'H[H]_V,
        forall rho : 'M[C]_(dim 'H[H]_U),
          h2mx (F (mx2h rho)) =
            ptrace2 (W *m rho *m W^*t),
        F \is tpmap -> W^*t *m W = 1%:M
      & F \is tnmap -> W^*t *m W <= 1%:M].
Lemma pd_mono (E : 'SO[HL]_(U, V)) (rho tau : 'F[HL]_U) :
  E \is cptn -> rho \is psdlf -> tau \is psdlf ->
  @complex.Re R (\Tr rho) <= 1 ->
  @complex.Re R (\Tr tau) <= 1 ->
  [/\ gfid rho tau <= gfid (E rho) (E tau) &
      pur_dist (E rho) (E tau) <= pur_dist rho tau].
\end{minted}
Here \coqi{cptn} packages complete positivity and trace non-increase.
\switchcolumn
\textbf{Stinespring dilation and channel monotonicity.}
Every CP map has a finite dilation; TP makes the dilation an isometry and TNI makes it a contraction.
For CPTN $\mathcal E$,
$\bar F(\rho,\tau)\le\bar F(\mathcal E\rho,\mathcal E\tau)$ and
$P(\mathcal E\rho,\mathcal E\tau)\le P(\rho,\tau)$.
\end{paracol}
\endgroup
\AlignmentBlockEnd

\AlignmentContinuationLabel{tab:alignment-distance-extensions}
\AlignmentBlockHead{Purified extensions.}
\begingroup\AlignmentBlockSetup
\begin{paracol}{2}
\code{src/utils/dist/pd.v}: \coqi{pd_uhl}, \coqi{pd_ext}, \coqi{pd_ext_subnorm}
\begin{minted}{coq}
Lemma pd_uhl (rho tau : 'F[HL]_A) (phi : 'F[HL]_S_AB) :
  phi \is denlf -> (\Rank phi <= 1%N)%N ->
  ptrace_cast disAB phi = rho -> tau \is denlf ->
  exists (vartheta : 'F[HL]_S_AB), vartheta \is denlf /\
    (\Rank vartheta <= 1%N)%N /\
    ptrace_cast disAB vartheta = tau /\
    pur_dist phi vartheta = pur_dist rho tau.
Lemma pd_ext (rho tau : 'F[HL]_A) (rhobar : 'F[HL]_S_AB) :
  rhobar \is denlf -> ptrace_cast disAB rhobar = rho -> tau \is denlf ->
  exists taubar : 'F[HL]_S_AB, taubar \is denlf /\
    ptrace_cast disAB taubar = tau /\
    pur_dist rhobar taubar = pur_dist rho tau.
Lemma pd_ext_subnorm (rho tau : 'F[HL]_A)
    (rhobar : 'F[HL]_S_AB) :
  rhobar \is psdlf ->
  @complex.Re R (\Tr rhobar) <= 1 ->
  ptrace_cast disAB rhobar = rho ->
  tau \is psdlf -> @complex.Re R (\Tr tau) <= 1 ->
  exists taubar : 'F[HL]_S_AB,
    taubar \is psdlf /\
    @complex.Re R (\Tr taubar) <= 1 /\
    ptrace_cast disAB taubar = tau /\
    pur_dist rhobar taubar = pur_dist rho tau.
\end{minted}
The equal-dimension and labelled-register disjointness hypotheses encode the paper's auxiliary purification space.
\switchcolumn
\textbf{Purified Uhlmann and extension theorems.}
A purification, or any extension, of $\rho$ can be matched by one of $\tau$ at exactly $P(\rho,\tau)$.
\end{paracol}
\endgroup
\AlignmentBlockEnd

\AlignmentParagraph{Fidelity and purification.}

\AlignmentContinuationLabel{tab:alignment-fidelity-infrastructure}
\AlignmentBlockHead{Reusable fidelity and purification laws.}
\begingroup\AlignmentBlockSetup
\begin{paracol}{2}
\code{src/utils/dist/dilation.v}: \coqi{fid_ii_cptp_bounds};
\code{src/utils/dist/fid.v}: \coqi{fid_i_sym_mx}, \coqi{fid_iii_uhl}, \coqi{fid_iv_proj_mx}, \coqi{fid_v_mono_r_psd_mx}, \coqi{fid_vi_block_diag_mx}
\begin{minted}{coq}
Lemma fid_i_sym_mx (rho tau : 'M[C]_n) :
  rho \is psdmx -> tau \is psdmx ->
  @complex.Re R (trnorm (sqrtmx rho *m sqrtmx tau)) =
  @complex.Re R (trnorm (sqrtmx tau *m sqrtmx rho)).
Lemma fid_ii_cptp_bounds (rho tau : 'F[H]_U)
    (E : 'SO[H]_(U, V)) :
  rho \is psdlf -> tau \is psdlf -> E \is cptp ->
  fidelity rho tau <= fidelity (E rho) (E tau) /\
  fidelity (E rho) (E tau) <=
    Num.sqrt (@complex.Re R (\Tr rho)) *
    Num.sqrt (@complex.Re R (\Tr tau)).
Lemma fid_iii_uhl (phi : 'F[HL]_S_AB) (tau : 'F[HL]_A) :
  phi \is den1lf -> \Rank phi = 1%N ->
  tau \is denlf -> h2mx tau != 0 ->
  (\Rank tau <= dim 'H[HL]_B)%N ->
  exists psi : 'F[HL]_S_AB,
    psi \is denlf /\ \Rank psi = 1%N /\
    ptrace_cast disAB psi = tau /\
    fidelity phi psi = fidelity (ptrace_cast disAB phi) tau.
Lemma fid_iv_proj_mx (P rho tau : 'M[C]_n) :
  P \is projmx -> rho \is psdmx -> tau \is psdmx ->
  let rhoP := P *m rho *m P in
  let tauP := P *m tau *m P in
  [/\ @complex.Re R (trnorm (sqrtmx rhoP *m sqrtmx tau)) =
        @complex.Re R (trnorm (sqrtmx rhoP *m sqrtmx tauP)) &
      @complex.Re R (trnorm (sqrtmx rhoP *m sqrtmx tauP)) =
        @complex.Re R (trnorm (sqrtmx rho *m sqrtmx tauP))].
Lemma fid_v_mono_r_psd_mx (rho sigma tau : 'M[C]_n) :
  rho \is psdmx -> sigma \is psdmx -> tau \is psdmx ->
  sigma <= tau ->
  @complex.Re R (trnorm (sqrtmx rho *m sqrtmx sigma)) <=
  @complex.Re R (trnorm (sqrtmx rho *m sqrtmx tau)).
Lemma fid_vi_block_diag_mx (rho1 tau1 : 'M[C]_m)
    (rho2 tau2 : 'M[C]_k) :
  rho1 \is psdmx -> tau1 \is psdmx ->
  rho2 \is psdmx -> tau2 \is psdmx ->
  @complex.Re R (trnorm
    (sqrtmx (block_mx rho1 0 0 rho2) *m
     sqrtmx (block_mx tau1 0 0 tau2))) =
  @complex.Re R (trnorm (sqrtmx rho1 *m sqrtmx tau1)) +
  @complex.Re R (trnorm (sqrtmx rho2 *m sqrtmx tau2)).
\end{minted}
The matrix-level declarations provide the reusable symmetry, projector,
order, and direct-sum laws; \coqi{fid_iii_uhl} provides the labelled
purification witness used by purified-distance lifting.
\switchcolumn
\textbf{Raw-fidelity laws (the \texttt{ent/tb:fid} dependency node).}
The reusable laws are symmetry; CPTP monotonicity with upper bound
$\sqrt{\Tr\rho}\sqrt{\Tr\tau}$; a finite-dimensional Uhlmann purification
witness;
the projector identity
$F(P\rho P,\tau)=F(P\rho P,P\tau P)=F(\rho,P\tau P)$; Loewner monotonicity; and direct-sum additivity,
$F(\rho_1\oplus\rho_2,\tau_1\oplus\tau_2)=F(\rho_1,\tau_1)+F(\rho_2,\tau_2)$.
\end{paracol}
\endgroup
\AlignmentBlockEnd

\AlignmentContinuationLabel{tab:alignment-purification-infrastructure}
\AlignmentBlockHead{Finite-dimensional purification.}
\begingroup\AlignmentBlockSetup
\begin{paracol}{2}
\code{src/utils/mx/prop.v}: \coqi{purification_exists_in_dim}
\begin{minted}{coq}
Lemma purification_exists_in_dim {a b : nat}
    (sigma_A : 'M[C]_a) :
  sigma_A \is denmx -> sigma_A != 0 ->
  (\rank sigma_A <= b)%N ->
  exists theta0 : 'M[C]_(a * b),
    theta0 \is denmx /\ \rank theta0 = 1%N /\
    ptrace2 theta0 = sigma_A.
\end{minted}
This is the reusable finite-dimensional existence result beneath the labelled purification constructions.
\switchcolumn
\textbf{Purification in a fixed finite environment.}
If $\sigma_A$ is a nonzero subnormalized state with $\operatorname{rank}\sigma_A\le b$, it has a rank-one subnormalized purification on $A\otimes\mathbb C^b$.
\end{paracol}
\endgroup
\AlignmentBlockEnd

\subsubsection{Entropy Calculus}

\AlignmentParagraph{Optimization.}

\AlignmentContinuationLabel{tab:alignment-convex-infrastructure}
\AlignmentBlockHead{Convex optimization and non-smooth entropy infrastructure.}
\begingroup\AlignmentBlockSetup
\begin{paracol}{2}
\code{src/entropy/pur/def.v}: \coqi{sdp_primal_set}, \coqi{sdp_dual_set};
\code{src/entropy/pur/pre.v}: \coqi{sdp_strong_primal}
\begin{minted}{coq}
Definition sdp_primal_set : set R :=
  [set x : R | exists X : 'M[C]_(dim U),
    [/\ X \is psdmx,
        h2mx (Psi (mx2h X)) >= B &
        x = @complex.Re R (\tr (A *m X))]].
Definition sdp_dual_set : set R :=
  [set x : R | exists Y : 'M[C]_(dim V),
    [/\ Y \is psdmx,
        h2mx (Psi^*o (mx2h Y)) <= A &
        x = @complex.Re R (\tr (B *m Y))]].
Lemma sdp_strong_primal (alpha beta : R) :
  A \is hermmx -> B \is hermmx ->
  (forall X, X \is hermmx -> h2mx (Psi (mx2h X)) \is hermmx) ->
  has_inf sdp_primal_set ->
  alpha = ExtNumTopology.etinf sdp_primal_set ->
  beta = ExtNumTopology.etsup sdp_dual_set ->
  (exists Y, Y \is pdmx /\
    (A - h2mx (Psi^*o (mx2h Y))) \is pdmx) ->
  alpha = beta /\
  exists X : 'M[C]_(dim U),
    [/\ X \is psdmx,
        h2mx (Psi (mx2h X)) >= B &
        @complex.Re R (\tr (A *m X)) = alpha].
\end{minted}
\AlignmentBoundaryLabel
This yields equality and primal attainment.
\switchcolumn
\textbf{SDP strong duality (\Cref{sec:architecture}).}
For the primal and dual SDPs induced by $A,B,\Psi$, finite primal value and a strict dual Slater witness imply $\alpha=\beta$ and attainment of the primal optimum.
\switchcolumn*
\code{src/entropy/pur/prop.v}: \coqi{phi_min_sdp_strong}, \coqi{phi_min_dual}, \coqi{phi_min_primal_attained}, \coqi{phi_min_dual_feasible_ub}
\begin{minted}{coq}
Lemma phi_min_sdp_strong (rhoAB : 'F[H]_S_AB) :
  rhoAB \is psdlf ->
  phi_min B rhoAB =
    sup [set v : R | exists Y : 'F[H]_S_AB, [/\ Y \is psdlf,
      ptrace_castl disAB Y <= (\1 : 'F[H]_B) &
      v = @complex.Re R (\Tr (rhoAB \o Y))]] /\
  phi_min_feasible B rhoAB (phi_min B rhoAB).
Lemma phi_min_dual_feasible_ub (rhoAB Y : 'F[H]_S_AB) :
  rhoAB \is psdlf -> Y \is psdlf ->
  ptrace_castl disAB Y <= (\1 : 'F[H]_B) ->
  @complex.Re R (\Tr (rhoAB \o Y)) <= phi_min B rhoAB.
\end{minted}
The declarations give the dual supremum, primal attainment, and pointwise
weak duality.
\switchcolumn
\textbf{Min-entropy SDP specialization and weak duality.}
For
\[
\Phi(\rho_{AB})=\min_{\sigma_B\succeq0}
  \{\Tr\sigma_B:\rho_{AB}\preceq I_A\otimes\sigma_B\},
\]
strong duality gives the dual supremum and primal attainment.
Every dual-feasible $Y_{AB}$ also satisfies
\[
\begin{gathered}
Y_{AB}\succeq0,\quad \Tr_A Y_{AB}\preceq I_B\\
\Longrightarrow
\Tr(\rho_{AB}Y_{AB})\le\Phi(\rho_{AB}).
\end{gathered}
\]
\end{paracol}
\endgroup
\AlignmentBlockEnd

\AlignmentParagraph{Operational laws.}

\AlignmentContinuationLabel{tab:alignment-entropy-algebra}
\AlignmentBlockHead{Entropy algebra and operational duality.}
\begingroup\AlignmentBlockSetup
\begin{paracol}{2}
\code{src/entropy/pur/conv/prop.v}: \coqi{phi_min_i_scale}, \coqi{phi_min_ii_mono}, \coqi{phi_min_iii_subadd}, \coqi{phi_min_iv_convex}, \coqi{phi_min_v_iso_inv}, \coqi{phi_min_data_proc_A}, \coqi{phi_min_data_proc_B}, \coqi{phi_min_vi_data_proc}, \coqi{phi_min_vii_data_proc_ext}, \coqi{phi_min_viii_tens}
\begin{minted}{coq}
Lemma phi_min_i_scale rhoS lambda :
  rhoS \is psdlf -> 0 <= lambda ->
  phi_min T (lambda%:C *: rhoS) = lambda * phi_min T rhoS.
Lemma phi_min_ii_mono rhoS tauS :
  rhoS \is psdlf -> tauS \is psdlf -> rhoS <= tauS ->
  phi_min T rhoS <= phi_min T tauS.
Lemma phi_min_iii_subadd rhoS tauS :
  rhoS \is psdlf -> tauS \is psdlf ->
  phi_min T (rhoS + tauS) <=
    phi_min T rhoS + phi_min T tauS.
Lemma phi_min_iv_convex rhoS tauS lambda :
  rhoS \is psdlf -> tauS \is psdlf ->
  0 <= lambda <= 1 ->
  phi_min T (lambda%:C *: rhoS +
    (1 - lambda)%:C *: tauS) <=
  lambda * phi_min T rhoS +
    (1 - lambda) * phi_min T tauS.
Lemma phi_min_v_iso_inv rhoAB U V :
  rhoAB \is psdlf ->
  phi_min B' ((U \⊗ V) \o rhoAB \o (U \⊗ V)^A) =
    phi_min B rhoAB.
Lemma phi_min_data_proc_B rhoAB (E : 'QO[H]_(B, B')) :
  rhoAB \is psdlf ->
  phi_min B' (((\:1 : 'SO[H]_A) :⊗ E) rhoAB) <=
    phi_min B rhoAB.
Lemma phi_min_data_proc_A rhoAB (E : 'SO[H]_(A, A')) :
  rhoAB \is psdlf -> E \is cpmap -> E^*o \is tnmap ->
  phi_min B ((E :⊗ (\:1 : 'SO[H]_B)) rhoAB) <=
    phi_min B rhoAB.
Lemma phi_min_vi_data_proc rhoAB E_B E_A :
  rhoAB \is psdlf -> E_A \is cpmap ->
  E_A^*o \is tnmap ->
  phi_min B' (((\:1 : 'SO[H]_A) :⊗ E_B) rhoAB) <=
      phi_min B rhoAB /\
  phi_min B ((E_A :⊗ (\:1 : 'SO[H]_B)) rhoAB) <=
      phi_min B rhoAB.
Lemma phi_min_vii_data_proc_ext rhoABC E :
  rhoABC \is psdlf -> E \is cpmap -> E^*o \is tnmap ->
  phi_min C
    ((((\:1 : 'SO[H]_A) :⊗ E) :⊗
      (\:1 : 'SO[H]_C)) rhoABC) <=
    phi_min C rhoABC.
Lemma phi_min_viii_tens rho1 rho2 :
  rho1 \is psdlf -> rho2 \is psdlf ->
  phi_min (B1 :|: B2) (rho1 \⊗ rho2) =
    phi_min B1 rho1 * phi_min B2 rho2.
\end{minted}
\AlignmentBoundaryLabel
The one-sided declarations \coqi{phi_min_data_proc_A} and
\coqi{phi_min_data_proc_B} are packaged by
\coqi{phi_min_vi_data_proc}; \coqi{phi_min_vii_data_proc_ext} states the
three-register data-processing law.
\switchcolumn
\textbf{Algebra of $\Phi$.}
For positive $K,L,M$, the declarations establish
\[
\begin{gathered}
\Phi_T(\lambda K)=\lambda\Phi_T(K) \quad(\lambda\ge0),\\
K\succeq L\Rightarrow\Phi_T(K)\ge\Phi_T(L),\\
\Phi_T(K+L)\le\Phi_T(K)+\Phi_T(L),\\
\Phi_T(\lambda K+(1-\lambda)L)\\
  \quad\le\lambda\Phi_T(K)
  +(1-\lambda)\Phi_T(L).
\end{gathered}
\]
Local isometries preserve $\Phi$, processing either subsystem cannot
increase it under the displayed map hypotheses, and
\[
\Phi_{B_1B_2}(L\otimes K)
=\Phi_{B_1}(L)\Phi_{B_2}(K).
\]
\end{paracol}
\endgroup
\AlignmentBlockEnd

\AlignmentContinuationLabel{tab:alignment-operational-entropy}
\AlignmentBlockHead{Operational entropy characterizations.}
\begingroup\AlignmentBlockSetup
\begin{paracol}{2}
\code{src/utils/pur.v}: \coqi{pur_vec}, \coqi{purification}, \coqi{max_entangled},
\coqi{gamma_state};
\code{src/entropy/pur/conv/minmax.v}: \coqi{gamma_ABp}, \coqi{min_koenig}, \coqi{min_omni};
\code{src/entropy/pur/conv/minmax\_dual.v}: \coqi{minmax_dual}
\begin{minted}{coq}
Local Definition pur_vec (rho : 'F[HL]_A) : 'H[HL]_S_AB :=
  \sum_(i : 'I_(dim 'H[HL]_A))
    tenv ((sqrtH rho) (eb i)) (eb (cast_ord dim_HA_HB i)).
Definition purification (rho : 'F[HL]_A) : 'F[HL]_S_AB :=
  [> pur_vec rho ; pur_vec rho <].
Definition max_entangled : 'F[HL]_S_AB :=
  purification (\1 : 'F[HL]_A).
Definition gamma_state : 'F[HL]_S_AB :=
  ((dim 'H[HL]_A)%:R^-1)%:C *: max_entangled.
Local Definition gamma_ABp : 'F[H]_S_ABp :=
  @gamma_state R L H A Bp dimABp_set.
Lemma min_koenig (rhoAB : 'F[H]_S_AB) :
  rhoAB \is denlf ->
  hminE B rhoAB = oppe (elog2
    ((dim 'H[H]_A)%:R *
      sup [set v : R | exists E : 'SO[H]_(B, Bp),
        [/\ E \is cptp &
          v = (gfid
            (((\:1 : 'SO[H]_A) :⊗ E) rhoAB)
            gamma_ABp)^+2]])).
Lemma min_omni (rhoAB : 'F[H]_S_AB) :
  rhoAB \is denlf ->
  hminE B rhoAB = oppe (elog2
    ((dim 'H[H]_A)%:R *
      sup [set v : R |
        exists V : 'FI[H]_(B, S_BpBpp),
        exists tau : 'F[H]_Bpp,
        [/\ tau \is den1lf &
          v = (gfid
            (((\1 : 'F[H]_A) \⊗ V) \o rhoAB \o
              (((\1 : 'F[H]_A) \⊗ V)^A))
            (castlf (esym eqABpBpp, esym eqABpBpp)
              ((gamma_ABp \O tau) :
                'F[H]_S_ABpBpp)))^+2]])).
Lemma minmax_dual (rhoABC : 'F[H]_S_ABC) :
  rhoABC \is denlf -> \Rank rhoABC = 1%N ->
  hmaxE B
    (@ptrace_cast3C R L H A B C0 disAC disBC rhoABC) =
  oppe (hminE C0
    (@ptrace_cast3B R L H A B C0 disAB disBC rhoABC)).
\end{minted}
These are extended-real statements over labelled channel-output and
purification spaces.
\switchcolumn
\textbf{Operational min-entropy and min/max duality.}
With $\gamma_{AA'}$ maximally entangled,
\[
\begin{gathered}
H_{\min}(A|B)_\rho
=-\log_2\!\left(d_A\max_{\mathcal E:B\to A'}\right.\\
\left.\bar F\!\left(\mathcal E(\rho_{AB}),\gamma_{AA'}\right)^2\right).
\end{gathered}
\]
Equivalently, dilating the channel gives
\[
\begin{aligned}
H_{\min}(A|B)_\rho
  &=-\log_2\!\left(d_A\max_{V,\tau}\Xi(V,\tau)\right),\\
\Xi(V,\tau)
  &:=\bar F\!\left(
       (I_A\otimes V)\rho\right.\\
  &\quad\left.(I_A\otimes V)^\dagger,
       \gamma_{AA'}\otimes\tau_{B''}\right)^2,\\
&\hspace{1em}V:B\hookrightarrow A'B'',\\
&\hspace{1em}\tau_{B''}\in\mathcal S_=.
\end{aligned}
\]
where $\gamma_{AA'}=d_A^{-1}|\Phi\rangle\!\langle\Phi|$.
For pure $\rho_{ABC}$,
$H_{\max}(A|B)_\rho=-H_{\min}(A|C)_\rho$.
\end{paracol}
\endgroup
\AlignmentBlockEnd

\AlignmentParagraph{Smoothing.}

\AlignmentContinuationLabel{tab:alignment-smooth-infrastructure}
\AlignmentBlockHead{Reusable smooth-entropy calculus.}
\begingroup\AlignmentBlockSetup
\begin{paracol}{2}
\code{src/entropy/pur/smooth/sm.v}: \coqi{smoothing_support_hmin}, \coqi{smoothing_support_hmax};
\code{src/entropy/pur/prop.v}: \coqi{smooth_hmax_attained}
\begin{minted}{coq}
Lemma smoothing_support_hmin (eps : R) (rhoAB : 'F[H]_S_AB) :
  rhoAB \is denlf ->
  0 <= eps < Num.sqrt (@complex.Re R (\Tr rhoAB)) ->
  exists tauAB : 'F[H]_S_AB,
    eps_ball eps rhoAB tauAB /\
    smooth_hminE B eps rhoAB = hminE B tauAB /\
    (supplf (ptrace_cast disAB rhoAB) \O
     supplf (ptrace_castl disAB rhoAB)) \o tauAB = tauAB /\
    tauAB \o (supplf (ptrace_cast disAB rhoAB) \O
     supplf (ptrace_castl disAB rhoAB)) = tauAB.
Lemma smoothing_support_hmax (eps : R) (rhoAB : 'F[H]_S_AB) :
  rhoAB \is denlf ->
  0 <= eps < Num.sqrt (@complex.Re R (\Tr rhoAB)) ->
  exists tauAB : 'F[H]_S_AB,
    eps_ball eps rhoAB tauAB /\
    smooth_hmaxE B eps rhoAB = hmaxE B tauAB /\
    (supplf (ptrace_cast disAB rhoAB) \O
     supplf (ptrace_castl disAB rhoAB)) \o tauAB = tauAB /\
    tauAB \o (supplf (ptrace_cast disAB rhoAB) \O
     supplf (ptrace_castl disAB rhoAB)) = tauAB.
Lemma smooth_hmax_attained eps (rhoS : 'F[H]_S) :
  rhoS \is psdlf -> @complex.Re R (\Tr rhoS) <= 1 ->
  0 <= eps < Num.sqrt (@complex.Re R (\Tr rhoS)) ->
  exists tauS : 'F[H]_S,
    eps_ball eps rhoS tauS /\
    smooth_hmax T eps rhoS = hmax T tauS.
\end{minted}
Fresh purifier and copy-register dimensions instantiate the max-entropy
witness construction.
\switchcolumn
\textbf{Smoothing witnesses and attainment.}
For the displayed smoothing radius, optimal min/max witnesses exist and may
be chosen in the product of the two marginal supports.
\end{paracol}
\endgroup
\AlignmentBlockEnd

\AlignmentContinuationLabel{tab:alignment-smooth-infrastructure-cont}
\AlignmentBlockHead{Smooth invariance, data processing, and classical side information.}
\begingroup\AlignmentBlockSetup
\begin{paracol}{2}
\code{src/entropy/pur/smooth/prop.v}: \coqi{smooth_iso_hmin}, \coqi{smooth_iso_hmax}
\begin{minted}{coq}
Lemma smooth_iso_hmin (eps : R) (rhoAB : 'F[H]_S_AB)
    (U : 'FI[H]_(A, A')) (V : 'FI[H]_(B, B')) :
  rhoAB \is denlf ->
  0 <= eps < Num.sqrt (@complex.Re R (\Tr rhoAB)) ->
  smooth_hminE B' eps ((U \⊗ V) \o rhoAB \o (U \⊗ V)^A) =
  smooth_hminE B eps rhoAB.
Lemma smooth_iso_hmax (eps : R) (rhoAB : 'F[H]_S_AB)
    (U : 'FI[H]_(A, A')) (V : 'FI[H]_(B, B')) :
  rhoAB \is denlf ->
  0 <= eps < Num.sqrt (@complex.Re R (\Tr rhoAB)) ->
  smooth_hmaxE B' eps ((U \⊗ V) \o rhoAB \o (U \⊗ V)^A) =
  smooth_hmaxE B eps rhoAB.
\end{minted}
The declaration is parameterized by its purification-space layout.
\switchcolumn
\textbf{Smooth isometry invariance.}
For isometries $U:A\to A'$ and $V:B\to B'$, both smooth conditional entropies are unchanged by $(U\otimes V)\rho(U\otimes V)^\dagger$.
\switchcolumn*
\code{src/entropy/pur/smooth/dpi.v}: \coqi{entropy_data_proc_map}, \coqi{entropy_data_proc_tau}, \coqi{entropy_data_proc}
\begin{minted}{coq}
Local Definition entropy_data_proc_map
    (E : 'SO[H]_(A, A')) (F : 'SO[H]_(B, B'))
    (G : 'SO[H]_(C0, C0')) : 'SO[H]_(S_ABC, S_A'B'C') :=
  (E :⊗ F) :⊗ G.
Local Definition entropy_data_proc_tau (rhoABC : 'F[H]_S_ABC)
    (E : 'SO[H]_(A, A')) (F : 'SO[H]_(B, B'))
    (G : 'SO[H]_(C0, C0')) : 'F[H]_S_A'B'C' :=
  entropy_data_proc_map E F G rhoABC.
Theorem entropy_data_proc (eps : R) (rhoABC : 'F[H]_S_ABC)
    (E : 'SO[H]_(A, A')) (F : 'SO[H]_(B, B'))
    (G : 'SO[H]_(C0, C0')) :
  rhoABC \is denlf -> 0 <= eps ->
  let tauA'B'C' := entropy_data_proc_tau rhoABC E F G in
  (E \is cptn -> E^*o \is tnmap -> F \is cptn -> G \is cptn ->
    eps < Num.sqrt (@complex.Re R
      (\Tr (ptrace_cast3C disA'C' disB'C' tauA'B'C'))) ->
    smooth_hminE B eps (ptrace_cast3C disAC disBC rhoABC) <=
    smooth_hminE B' eps
      (ptrace_cast3C disA'C' disB'C' tauA'B'C')) /\
  (E \is cptp -> E^*o \is tnmap -> F \is cptp -> G \is cptp ->
    eps < Num.sqrt (@complex.Re R
      (\Tr (ptrace_cast3C disAC disBC rhoABC))) ->
    smooth_hmaxE B eps (ptrace_cast3C disAC disBC rhoABC) <=
    smooth_hmaxE B' eps
      (ptrace_cast3C disA'C' disB'C' tauA'B'C')).
\end{minted}
The theorem is parameterized by labelled disjointness, purification-space
dimensions, the state, the maps, the trace, and the smoothing radius.
\switchcolumn
\textbf{Generalized smooth data processing.}
A sub-unital CPTN map on $A$ together with CPTN maps on $B,C$ gives the smooth-min inequality.
For smooth max-entropy, all three maps are CPTP and the map on $A$ is
sub-unital.
\end{paracol}
\endgroup
\AlignmentBlockEnd

\AlignmentContinuationLabel{tab:alignment-classical-side-information}
\AlignmentBlockHead{Classical side information.}
\begingroup\AlignmentBlockSetup
\begin{paracol}{2}
\code{src/entropy/pur/smooth/class.v}: \coqi{class_bounds_1_min_upper}, \coqi{class_bounds_2_min_lower}
\begin{minted}{coq}
Local Notation classical_on_X rho :=
  (forall (x x' : 'Idx[H]_X), x' != x ->
    (([> deltav x ; deltav x <] \O (\1 : 'F_S_AB)) \o rho \o
     ([> deltav x' ; deltav x' <] \O (\1 : 'F_S_AB))) = 0).
Lemma class_bounds_1_min_upper (eps : R)
    (rhoXAB : 'F[H]_S_XAB) :
  rhoXAB \is denlf -> classical_on_X rhoXAB ->
  0 <= eps < Num.sqrt (@complex.Re R (\Tr rhoXAB)) ->
  (smooth_hminE B eps (ptrace_castl disXAB rhoXAB) <=
   smooth_hminE B eps rhoXAB)%E.
Lemma class_bounds_2_min_lower (eps : R)
    (rhoXAB : 'F[H]_S_XAB) :
  rhoXAB \is denlf -> classical_on_X rhoXAB ->
  0 <= eps < Num.sqrt (@complex.Re R (\Tr rhoXAB)) ->
  (smooth_hminE B eps (ptrace_castl disXAB rhoXAB) -
    (log2 (dim 'H[H]_X)%:R)%:E <=
   smooth_hminE (B :|: X) eps rhoXAB)%E.
\end{minted}
\AlignmentBoundaryLabel
These declarations provide the two min-entropy inequalities used by the proof
graph on the displayed smoothing domain.
\switchcolumn
\textbf{Classical side-information min-entropy bounds.}
For a state classical on $X$, meaning
\[
(|x\rangle\!\langle x|\otimes I)\rho
(|x'\rangle\!\langle x'|\otimes I)=0\quad(x\ne x'),
\]
the reusable bounds are
\[
\begin{aligned}
H_{\min}^{\eps}(A|B)&\le H_{\min}^{\eps}(XA|B)
,\\
H_{\min}^{\eps}(A|BX)&\ge
 H_{\min}^{\eps}(A|B)-\log d_X.
\end{aligned}
\]
\end{paracol}
\endgroup
\AlignmentBlockEnd

\AlignmentContinuationLabel{tab:alignment-class-entangled-max}
\AlignmentBlockHead{Coherent-classical max entropy.}
\begingroup\AlignmentBlockSetup
\begin{paracol}{2}
\code{src/entropy/pur/util.v}: \coqi{cc_idx}, \coqi{cc_pair},
\coqi{cc_proj_xx}, \coqi{class_entangled_max}
\begin{minted}{coq}
Let cc_idx (x : 'Idx[H]_X) : 'Idx[H]_(X') :=
  n2i (cast_ord eqXX' (i2n x)).
Let cc_pair (x : 'Idx[H]_X) : 'H[H]_S_XX' :=
  deltav x ⊗v deltav (cc_idx x).
Let cc_proj_xx : 'F[H]_S_XX' :=
  sumoutp (fun=>1) cc_pair cc_pair.
Lemma class_entangled_max (rhoXX'AB : 'F[H]_S_XX'AB) :
  rhoXX'AB \is denlf ->
  0 < @complex.Re R (\Tr rhoXX'AB) ->
  rhoXX'AB =
    lift_lf (@finset.subsetUl _ S_XX' S_AB) cc_proj_xx \o
      rhoXX'AB \o
    lift_lf (@finset.subsetUl _ S_XX' S_AB) cc_proj_xx ->
  (hmaxE S_X'B rhoXX'AB <= hmaxE S_XX'B rhoXX'AB)%E.
\end{minted}
This declaration gives the coherent-copy max-entropy inequality under the
displayed support and trace hypotheses.
\switchcolumn
\textbf{Coherent-classical max-entropy bound.}
For a state supported on the coherent copy subspace of $XX'$,
\[
\Pi_{XX'}=\sum_x|x,x\rangle\!\langle x,x|,
\qquad \rho=\Pi_{XX'}\rho\Pi_{XX'},
\]
and
\[
H_{\max}(XA|X'B)\le H_{\max}(A|XX'B).
\]
\end{paracol}
\endgroup
\AlignmentBlockEnd

\subsubsection{Trace Entropy}

\AlignmentParagraph{Definitions.}

\AlignmentContinuationLabel{tab:alignment-trace-state-entropy}
\AlignmentBlockHead{Trace-distance balls and state-conditioned entropies.}
\begingroup\AlignmentBlockSetup
\begin{paracol}{2}
\code{src/entropy/tr/def.v}: \coqi{ereal_set}, \coqi{etsupR},
\coqi{ball_tr}, \coqi{hmin_state_tr}, \coqi{hmax_state_trE},
\coqi{hmax_state_trR}, \coqi{hmax0_trE}, \coqi{hmax0_trR};
\code{src/entropy/tr/conv.v}: \coqi{hmax_state_trE_zero},
\coqi{hmax0_trE_zero}
\begin{minted}{coq}
Definition ereal_set (S : set R) : set (\bar R) :=
  [set y : \bar R | exists x : R, S x /\ y = x%:E].
Definition etsupR (S : set R) : \bar R := ereal_sup (ereal_set S).
Definition ball_tr {n} (eps : {nonneg R}) (rho : 'M[C]_n) :=
  [set rho' | (rho' \is psdmx) &&
    ((@complex.Re R (\tr rho') <= @complex.Re R (\tr rho)) &&
     (@complex.Re R (\tr| rho' - rho |) <=
       @complex.Re R (\tr rho) * eps%:num))].
Definition hmin_state_tr {m n : nat}
    (rho : 'M[C]_(m * n)) (sigma : 'M[C]_n) : \bar R :=
  etsupR [set lambda : R |
    rho ⊑ (2 `^ (-lambda))%:C *: (1%:M : 'M[C]_m) *t sigma].
Definition hmax_state_trE {m n : nat}
    (rho : 'M[C]_(m * n)) (sigma : 'M[C]_n) : \bar R :=
  elog2 (@complex.Re R (\tr (((1%:M : 'M[C]_m) *t sigma) *m suppmx rho))).
Definition hmax_state_trR {m n : nat}
    (rho : 'M[C]_(m * n)) (sigma : 'M[C]_n) : R :=
  log2 (@complex.Re R (\tr (((1%:M : 'M[C]_m) *t sigma) *m suppmx rho))).
Definition hmax0_trE (n : nat) (rho : 'M[C]_n) : \bar R :=
  elog2 (\rank rho)%:R.
Definition hmax0_trR (n : nat) (rho : 'M[C]_n) : R :=
  log2 (\rank rho)%:R.
Lemma hmax_state_trE_zero {m n : nat} (sigma : 'M[C]_n) :
  Htr_maxE((0 : 'M[C]_(m * n)) | sigma) = -oo%E.
Lemma hmax0_trE_zero {n : nat} :
  Htr_maxE((0 : 'M[C]_n)) = -oo%E.
\end{minted}
The extended-real variants use \coqi{elog2}; the real variants apply
\coqi{log2} on positive support.
\switchcolumn
\textbf{Trace smoothing and state-conditioned entropy.}
Renner's trace ball and state-conditioned quantities are
\[
\begin{aligned}
B_{\rm tr}^{\eps}(\rho)=\{\rho'\succeq0:
&\ \Tr\rho'\le\Tr\rho,\\
&\ \lVert\rho'-\rho\rVert_1\le\eps\Tr\rho\},
\end{aligned}
\]
\[
\begin{aligned}
H_{\min}(\rho|\sigma)
  &=\sup\{\lambda:\rho\preceq2^{-\lambda}I\otimes\sigma\},\\
H_{\max}(\rho|\sigma)
  &=\overline{\log}_2\Tr[(I\otimes\sigma)\Pi_\rho].
\end{aligned}
\]
At zero support, $H_{\max}(0|\sigma)=H_{\max}(0)=-\infty$.
\end{paracol}
\endgroup
\AlignmentBlockEnd

\AlignmentContinuationLabel{tab:alignment-trace-entropy}
\AlignmentBlockHead{Conditional and smooth trace entropies.}
\begingroup\AlignmentBlockSetup
\begin{paracol}{2}
\code{src/entropy/tr/def.v}: \coqi{hmin_tr}, \coqi{hmax_trE},
\coqi{hmax_trR}, \coqi{smooth_hmin_state_tr},
\coqi{smooth_hmax_state_trE}, \coqi{smooth_hmax_state_trR},
\coqi{smooth_hmin_tr}, \coqi{smooth_hmax_trE}, \coqi{smooth_hmax_trR}
\begin{minted}{coq}
Definition hmin_tr {a b : nat} (rho : 'M[C]_(a * b)) : \bar R :=
  etsupR [set lambda : R | exists sigma_B : 'M[C]_b,
    sigma_B \is den1mx /\
    rho <= (2 `^ (-lambda))%:C *: (1%:M : 'M[C]_a) *t sigma_B].
Definition hmax_trE {a b : nat} (rho : 'M[C]_(a * b)) : \bar R :=
  ereal_sup [set h : \bar R | exists sigma_B : 'M[C]_b,
    sigma_B \is den1mx /\ h = hmax_state_trE rho sigma_B].
Definition hmax_trR {a b : nat} (rho : 'M[C]_(a * b)) : R :=
  ExtNumTopology.etsup [set h : R | exists sigma_B : 'M[C]_b,
    sigma_B \is den1mx /\ h = hmax_state_trR rho sigma_B].
Definition smooth_hmin_state_tr {m n : nat} (eps : {nonneg R})
    (rho : 'M[C]_(m * n)) (sigma : 'M[C]_n) : \bar R :=
  ereal_sup [set x | exists rho',
    rho' \in Btr^eps(rho) /\ x = hmin_state_tr rho' sigma].
Definition smooth_hmax_state_trE {m n : nat} (eps : {nonneg R})
    (rho : 'M[C]_(m * n)) (sigma : 'M[C]_n) : \bar R :=
  ereal_inf [set x : \bar R | exists rho',
    rho' \in Btr^eps(rho) /\ x = hmax_state_trE rho' sigma].
Definition smooth_hmax_state_trR {m n : nat} (eps : {nonneg R})
    (rho : 'M[C]_(m * n)) (sigma : 'M[C]_n) : R :=
  ExtNumTopology.etinf [set x : R | exists rho',
    rho' \in Btr^eps(rho) /\ x = hmax_state_trR rho' sigma].
Definition smooth_hmin_tr {a b : nat} (eps : {nonneg R})
    (rho : 'M[C]_(a * b)) : \bar R :=
  ereal_sup [set x | exists sigma_B : 'M[C]_b,
    sigma_B \is den1mx /\ x = smooth_hmin_state_tr eps rho sigma_B].
Definition smooth_hmax_trE {a b : nat} (eps : {nonneg R})
    (rho : 'M[C]_(a * b)) : \bar R :=
  ereal_sup [set x : \bar R | exists sigma_B : 'M[C]_b,
    sigma_B \is den1mx /\ x = smooth_hmax_state_trE eps rho sigma_B].
Definition smooth_hmax_trR {a b : nat} (eps : {nonneg R})
    (rho : 'M[C]_(a * b)) : R :=
  ExtNumTopology.etsup [set x : R | exists sigma_B : 'M[C]_b,
    sigma_B \is den1mx /\ x = smooth_hmax_state_trR eps rho sigma_B].
\end{minted}
The state-conditioned operands and zero-support laws are displayed above.
\switchcolumn
\textbf{Optimization over conditioning states and trace balls.}
\[
\begin{aligned}
H_{\min}(A|B)_\rho&=\sup_{\sigma_B}H_{\min}(\rho_{AB}|\sigma_B),\\
H_{\max}(A|B)_\rho&=\sup_{\sigma_B}H_{\max}(\rho_{AB}|\sigma_B),
\end{aligned}
\]
and the smooth forms optimize additionally over
$\rho'\in B_{\rm tr}^{\eps}(\rho)$.
\end{paracol}
\endgroup
\AlignmentBlockEnd

\AlignmentParagraph{Attainment.}

\AlignmentContinuationLabel{tab:alignment-trace-attainment}
\AlignmentBlockHead{Trace-entropy attainment and superposition.}
\begingroup\AlignmentBlockSetup
\begin{paracol}{2}
\code{src/entropy/tr/smooth.v}: \coqi{epst_card},
\coqi{smooth_hmin_superposition_vs_mixture},
\coqi{smooth_hmin_state_tr_attained}, \coqi{smooth_hmax_state_tr_attained},
\coqi{smooth_hmin_tr_attained}, \coqi{smooth_hmax_tr_attained}
\begin{minted}{coq}
Definition epst_card (eps : {nonneg R}) (x : nat) : {nonneg R} :=
  ((eps%:num ^+ 2 / (6 * x)%:R)%:nng : {nonneg R}).
Lemma smooth_hmin_superposition_vs_mixture {a b e x : nat}
    (eps : {nonneg R}) (psi_x : 'I_x -> 'M[C]_(a * b * e, 1))
    (sigma_BX : 'M[C]_(b * x)) :
  (0 < x)%N -> (a * b <= e)%N ->
  (forall i j, i != j -> (psi_x i)^*t *m psi_x j = 0) ->
  sigma_BX \is psdmx ->
  let psi := \sum_i psi_x i in
  let rho_ABE := psi *m psi^*t in
  let rhot_ABX := \sum_i
    ((ptrace2 (psi_x i *m (psi_x i)^*t)) *t delta_mx i i) in
  let rho_AB := ptrace2 rho_ABE in
  let rhot_X := ptrace1 rhot_ABX in
  let sigma_B := ptrace2 sigma_BX in
  (Htr_min^eps(rho_AB | sigma_B) >=
   Htr_min^(epst_card eps x)
     (castmx (esym (mulnA a b x), esym (mulnA a b x)) rhot_ABX |
      sigma_BX) - Htr_maxE(rhot_X))%E.
Lemma smooth_hmin_state_tr_attained eps rho sigma :
  rho \is psdmx -> sigma \is psdmx ->
  exists2 rhob, rhob \in Btr^eps(rho) &
    Htr_min(rhob | sigma) = Htr_min^eps(rho | sigma).
Lemma smooth_hmax_state_tr_attained eps rho sigma :
  rho \is psdmx -> sigma \is psdmx ->
  exists2 rhob, rhob \in Btr^eps(rho) &
    Htr_maxE(rhob | sigma) = Htr_maxE^eps(rho | sigma).
Lemma smooth_hmin_tr_attained {a b} eps (rho : 'M[C]_(a * b)) :
  (0 < b)%N -> rho \is psdmx ->
  exists2 sigma, sigma \is den1mx &
    Htr_min^eps(rho | sigma) = Htr_min^eps(rho |^ b).
Lemma smooth_hmax_tr_attained {a b} eps (rho : 'M[C]_(a * b)) :
  (0 < b)%N -> rho \is psdmx ->
  exists2 sigma, sigma \is den1mx &
    Htr_maxE^eps(rho | sigma) = Htr_maxE^eps(rho |^ b).
\end{minted}
The purifier premise $ab\le e$ is explicit, and zero summands are permitted.
\switchcolumn
\textbf{Attainment and superposition versus mixture.}
For nonempty $X$, orthogonal $\{|\psi^x\rangle\}_{x\in X}$, a positive
$\sigma_{BX}$, and matrix dimensions $ab\le e$, let
$|\psi\rangle=\sum_x|\psi^x\rangle$ and attach a classical label $X$ to
the incoherent mixture. Then
\[
\begin{aligned}
H_{\min}^{\eps}(AB|\sigma_B)_\rho
  &\ge H_{\min}^{\eps_t}(ABX|\sigma_{BX})_{\widetilde\rho}\\
  &\quad     -H_{\max}(X)_{\widetilde\rho},\\
\eps_t&=\eps^2/(6|X|).
\end{aligned}
\]
All finite-dimensional smooth extrema are attained.
\end{paracol}
\endgroup
\AlignmentBlockEnd

\AlignmentParagraph{Calculus.}

\AlignmentContinuationLabel{tab:alignment-trace-calculus}
\AlignmentBlockHead{Reusable trace-entropy calculus: superadditivity, subadditivity, classical conditioning.}
\begingroup\AlignmentBlockSetup
\begin{paracol}{2}
\code{src/entropy/tr/def.v}: \coqi{addnng_tr};
\code{src/utils/mx/def.v}: \coqi{tensor_prod_split};
\code{src/entropy/tr/smooth.v}:
\coqi{smooth_hmin_conditional_superadditive},
\coqi{smooth_hmin_conditional_strong_subadditive},
\coqi{smooth_hmin_conditional_classical}
\begin{minted}{coq}
Definition addnng_tr (x y : {nonneg R}) : {nonneg R} :=
  ((x%:num + y%:num)%:nng : {nonneg R}).
Definition tensor_prod_split (A : 'M[C]_(m * n))
    (B : 'M[C]_(m' * n')) : 'M[C]_((m * m') * (n * n')) :=
  perm_conj (split_perm m n m' n')
    (castmx_sq (mulnACA m n m' n') (A *t B)).
Lemma smooth_hmin_conditional_superadditive {a b a' b' : nat}
    (eps' : {nonneg R}) (rho_AB : 'M[C]_(a * b))
    (rho_A'B' : 'M[C]_(a' * b')) :
  rho_AB \is psdmx -> rho_A'B' \is psdmx ->
  Htr_min^(addnng_tr eps eps')
    (tensor_prod_split rho_AB rho_A'B' |^ (b * b')) >=
  Htr_min^eps(rho_AB |^ b) + Htr_min^eps'(rho_A'B' |^ b').
Lemma smooth_hmin_conditional_strong_subadditive {a b c : nat}
    (rho_ABC : 'M[C]_(a * b * c)) :
  Htr_min^eps(castmx
    (esym (mulnA a b c), esym (mulnA a b c)) rho_ABC |^ (b * c)) <=
    Htr_min^eps(ptrace2 rho_ABC |^ b).
Lemma smooth_hmin_conditional_classical {a b z : nat}
    (rho_ABz : 'I_z -> 'M[C]_(a * b)) :
  let rho_ABZ := \sum_i (rho_ABz i *t delta_mx i i) in
  rho_ABZ \is den1mx -> (forall i, rho_ABz i \is psdmx) ->
  (forall i, rho_ABz i != 0) ->
  let rhob i := ((@complex.Re R (\tr (rho_ABz i)))^-1)%:C *: rho_ABz i in
  Htr_min^eps(castmx
    (esym (mulnA a b z), esym (mulnA a b z)) rho_ABZ |^ (b * z)) >=
    ereal_inf [set h | exists i, h = Htr_min^eps(rhob i |^ b)].
\end{minted}
\switchcolumn
\textbf{Smooth conditional min-entropy laws.}
The calculus contains tensor superadditivity, strong subadditivity,
classical conditioning, the chain rule, and classical-product data
processing. In particular,
\[
H_{\min}^{\eps+\eps'}(AA'|BB')
\ge H_{\min}^{\eps}(A|B)+H_{\min}^{\eps'}(A'|B'),
\]
\[
H_{\min}^{\eps}(A|BC)
\le H_{\min}^{\eps}(A|B).
\]
\end{paracol}
\endgroup
\AlignmentBlockEnd

\AlignmentContinuationLabel{tab:alignment-trace-calculus-chain}
\AlignmentBlockHead{Reusable trace-entropy calculus: chain rule and data processing.}
\begingroup\AlignmentBlockSetup
\begin{paracol}{2}
\code{src/entropy/tr/smooth.v}:
\coqi{smooth_hmin_conditional_chain},
\coqi{smooth_hmin_conditional_data_processing}, \coqi{smooth_hmin_calculus}
\begin{minted}{coq}
Lemma smooth_hmin_conditional_chain {a b c : nat}
    (rho_ABC : 'M[C]_(a * b * c)) :
  rho_ABC \is psdmx -> rho_ABC != 0 ->
  Htr_min^eps(rho_ABC |^ c) <=
    Htr_min^eps(castmx (esym (mulnA a b c), esym (mulnA a b c))
      rho_ABC |^ (b * c)) + Htr_maxE(ptrace1 (ptrace2 rho_ABC)).
Theorem smooth_hmin_conditional_data_processing {a y c : nat}
    (rho_Ay : 'I_y -> 'M[C]_a)
    (rho_Cy : 'I_y -> 'M[C]_c)
    (HrhoA : forall i, rho_Ay i \is psdmx)
    (HrhoC : forall i, rho_Cy i \is psdmx) :
  let rho_YAC := \sum_i ((delta_mx i i) *t ((rho_Ay i) *t (rho_Cy i))) in
  let rho_YC := \sum_i (\tr (rho_Ay i)) *: (delta_mx i i) *t (rho_Cy i) in
  let rho_AY := \sum_i (\tr (rho_Cy i)) *: (rho_Ay i) *t (delta_mx i i) in
  let rho_Y := \sum_i (\tr (rho_Ay i)) * (\tr (rho_Cy i)) *: (delta_mx i i) in
  (Htr_min^eps(castmx (mulnA y a c, mulnA y a c) rho_YAC |^ c) >=
  Htr_min^eps(rho_YC |^ c) + Htr_min(rho_AY | rho_Y))%E.
\end{minted}
\coqi{smooth_hmin_calculus} conjoins these five reusable statements.
\switchcolumn
\textbf{Chain rule and data processing.}
\[
\begin{aligned}
H_{\min}^{\eps}(AB|C)
&\le H_{\min}^{\eps}(A|BC)\\
&\quad+H_{\max}(B).
\end{aligned}
\]
For the classical-product law, writing
$\rho_{YAC}=\sum_y |y\rangle\!\langle y|\otimes
\rho_A^y\otimes\rho_C^y$, data processing gives
\[
H_{\min}^{\eps}(YAC|C)
\ge H_{\min}^{\eps}(YC|C)+H_{\min}(AY|Y).
\]
\end{paracol}
\endgroup
\AlignmentBlockEnd

\AlignmentParagraph{AEP.}

\AlignmentContinuationLabel{tab:alignment-trace-aep}
\AlignmentBlockHead{Classical finite-block AEP.}
\begingroup\AlignmentBlockSetup
\begin{paracol}{2}
\code{src/entropy/tr/def.v}: \coqi{classical_Y_marginal_tr},
\coqi{classical_cond_prob_tr}, \coqi{classical_cond_surprisal_tr},
\coqi{classical_cond_entropy_tr};
\code{src/entropy/tr/prod.v}: \coqi{classical_aep}
\begin{minted}{coq}
Definition classical_Y_marginal_tr P j := \sum_i P i j.
Definition classical_cond_prob_tr P i j :=
  if P i j == 0 then 1 else P i j / classical_Y_marginal_tr P j.
Definition classical_cond_surprisal_tr P i j :=
  - log2 (classical_cond_prob_tr P i j).
Definition classical_cond_entropy_tr P :=
  \sum_j \sum_i P i j * classical_cond_surprisal_tr P i j.
Lemma classical_aep {x y : nat} (n : nat)
    (P : 'I_x -> 'I_y -> R)
    (HP_nonneg : forall i j, 0 <= P i j)
    (HP_prob : \sum_i \sum_j P i j = 1) :
  (0 < x)%N -> (0 < n)%N ->
  forall eps : {nonneg R}, 0 < eps%:num -> eps%:num <= 1 ->
  let rho_XY_n := \sum_(sy : n.-tuple 'I_y) \sum_(sx : n.-tuple 'I_x)
    ((\prod_(k < n) P (tnth sx k) (tnth sy k))%:C *:
      (delta_mx (enum_rank sx) (enum_rank sx) *t
       delta_mx (enum_rank sy) (enum_rank sy))) in
  let rho_Y_n := \sum_(sy : n.-tuple 'I_y)
    ((\prod_(k < n) \sum_i P i (tnth sy k))%:C *:
      delta_mx (enum_rank sy) (enum_rank sy)) in
  let delta := log2 (x%:R + 3%:R) *
    Num.sqrt (2%:R * log2 (eps%:num^-1) / n%:R) in
  let Hc := classical_cond_entropy_tr P in
  (((n%:R)^-1)%:E * Htr_maxE^eps(rho_XY_n | rho_Y_n) <=
    (Hc + delta)%:E)%E /\
  (((n%:R)^-1)%:E * Htr_min^eps(rho_XY_n | rho_Y_n) >=
    (Hc - delta)%:E)%E.
\end{minted}
Zero-weight pairs contribute zero to the conditional entropy.
\switchcolumn
\textbf{Classical i.i.d. pairs.}
For a probability table $P_{XY}$, $|X|,n>0$, and $0<\eps\le1$,
\[
\begin{aligned}
n^{-1}H_{\min}^{\eps}(X^n|Y^n)&\ge H(X|Y)-\delta,\\
n^{-1}H_{\max}^{\eps}(X^n|Y^n)&\le H(X|Y)+\delta,
\end{aligned}
\]
where $\delta=\log_2(|X|+3)\sqrt{2\log_2(1/\eps)/n}$.
\end{paracol}
\endgroup
\AlignmentBlockEnd

\AlignmentContinuationLabel{tab:alignment-trace-aep-family}
\AlignmentBlockHead{Non-identical classical-family AEP.}
\begingroup\AlignmentBlockSetup
\begin{paracol}{2}
\code{src/entropy/tr/prod.v}: \coqi{classical_aep_family}
\begin{minted}{coq}
Lemma classical_aep_family {x y : nat} (n : nat)
    (P : 'I_n -> 'I_x -> 'I_y -> R) :
  (forall k i j, 0 <= P k i j) ->
  (forall k, \sum_i \sum_j P k i j = 1) ->
  (0 < x)%N -> (0 < n)%N ->
  forall eps : {nonneg R}, 0 < eps%:num -> eps%:num <= 1 ->
  let rho_XY_n := \sum_(sy : n.-tuple 'I_y)
    \sum_(sx : n.-tuple 'I_x)
      ((\prod_(k < n) P k (tnth sx k) (tnth sy k))%:C *:
        (delta_mx (enum_rank sx) (enum_rank sx) *t
         delta_mx (enum_rank sy) (enum_rank sy))) in
  let rho_Y_n := \sum_(sy : n.-tuple 'I_y)
    ((\prod_(k < n) \sum_i P k i (tnth sy k))%:C *:
      delta_mx (enum_rank sy) (enum_rank sy)) in
  let delta := log2 (x%:R + 3%:R) *
    Num.sqrt (2%:R * log2 (eps%:num^-1) / n%:R) in
  let Hcs := \sum_(k < n) classical_cond_entropy_tr (P k) in
  (((n%:R)^-1)%:E * Htr_maxE^eps(rho_XY_n | rho_Y_n) <=
    (n%:R^-1 * Hcs + delta)%:E)%E /\
  (((n%:R)^-1)%:E * Htr_min^eps(rho_XY_n | rho_Y_n) >=
    (n%:R^-1 * Hcs - delta)%:E)%E.
\end{minted}
This extends the finite-block pair bound to independent families.
\switchcolumn
\textbf{Independent non-identical pairs.}
For $P_i(X_i,Y_i)$ under the same finite-block domain,
\[
\begin{aligned}
n^{-1}H_{\min}^{\eps}(X^n|Y^n)
  &\ge n^{-1}\sum_{i<n}H(X_i|Y_i)-\delta,\\
n^{-1}H_{\max}^{\eps}(X^n|Y^n)
  &\le n^{-1}\sum_{i<n}H(X_i|Y_i)+\delta.
\end{aligned}
\]
\end{paracol}
\endgroup
\AlignmentBlockEnd

\AlignmentContinuationLabel{tab:alignment-trace-aep-quantum}
\AlignmentBlockHead{Finite-block quantum product AEP.}
\begingroup\AlignmentBlockSetup
\begin{paracol}{2}
\code{src/utils/mx/def.v}: \coqi{bigten};
\code{src/entropy/tr/prod.v}:
\coqi{smooth_min_entropy_of_product_entropy_rank},
\coqi{smooth_min_entropy_of_product_classical_entropy_rank}
\begin{minted}{coq}
Theorem smooth_min_entropy_of_product_entropy_rank
    (a b n : nat) (rho_AB : 'M[C]_(a * b))
    (sigma_B : 'M[C]_b) (eps : {nonneg R}) :
  (0 < n)%N -> 0 < eps%:num -> eps%:num <= 2%:R ->
  rho_AB \is den1mx -> sigma_B \is psdmx ->
  \tr sigma_B = 1 -> (\rank sigma_B = b)%N ->
  4%:R * log2 (2%:R / eps%:num) <= n%:R ->
  let gamma := (\rank (ptrace2 rho_AB))%:R +
    @complex.Re R (\tr (rho_AB *m rho_AB *m
      invmx ((1%:M : 'M[C]_a) *t sigma_B))) in
  let delta := 2%:R * log2 (gamma + 2) *
    Num.sqrt (log2 (eps%:num^-1) / n%:R + 1) in
  let rho_AB_n := castmx_sq (esym (mul1n ((a * b) ^ n)%N))
    (bigten rho_AB n) in
  let sigma_B_n := bigten ((1%:M : 'M[C]_a) *t sigma_B) n in
  ((n%:R^-1)%:E * Htr_min^eps(rho_AB_n | sigma_B_n) >=
    (H(rho_AB) - H(ptrace1 rho_AB) -
      fine (D(ptrace1 rho_AB || sigma_B)) - delta)%:E)%E.
Corollary smooth_min_entropy_of_product_classical_entropy_rank
    (x b n : nat) (rho_Bx : 'I_x -> 'M[C]_b)
    (eps : {nonneg R}) :
  let rho_XB := \sum_i (delta_mx i i *t rho_Bx i) in
  let rho_B := ptrace1 rho_XB in
  let rho_X := ptrace2 rho_XB in
  (0 < n)%N -> 0 < eps%:num -> eps%:num <= 2%:R ->
  rho_XB \is den1mx -> (forall i, rho_Bx i \is psdmx) ->
  (\rank rho_B = b)%N -> (2 <= \rank rho_X)%N ->
  4%:R * log2 (2%:R / eps%:num) <= n%:R ->
  let delta := (2%:R * Htr_maxR(rho_X) + 3%:R) *
    Num.sqrt (log2 (eps%:num^-1) / n%:R + 1) in
  let rho_XB_n := castmx_sq (esym (mul1n ((x * b) ^ n)%N))
    (bigten rho_XB n) in
  let rho_B_n := bigten ((1%:M : 'M[C]_x) *t rho_B) n in
  ((n%:R^-1)%:E * Htr_min^eps(rho_XB_n | rho_B_n) >=
    (H(rho_XB) - H(rho_B) - delta)%:E)%E.
\end{minted}
The theorem states the finite-block threshold and rank premises explicitly.
\switchcolumn
\textbf{Quantum product states.}
For a density operator $\rho_{AB}$ and a full-rank density operator
$\sigma_B$, assume $n>0$, $0<\eps\le2$, and
$4\log_2(2/\eps)\le n$. Then
\[
\begin{aligned}
n^{-1}&H_{\min}^{\eps}
 (\rho_{AB}^{\otimes n}\\
 &\quad|(I_A\otimes\sigma_B)^{\otimes n})\\
\ge{}&H(\rho_{AB})-H(\rho_B)\\
&-D(\rho_B\Vert\sigma_B)-\delta_q,
\end{aligned}
\]
where
$\delta_q=2\log_2(\gamma+2)
\sqrt{\log_2(1/\eps)/n+1}$ and
$\gamma=\operatorname{rank}\rho_A+
\Tr[\rho_{AB}^2(I_A\otimes\sigma_B)^{-1}]$.
The cq specialization sets $\sigma_B=\rho_B$ and assumes
$\operatorname{rank}\rho_B=d_B$ and $\operatorname{rank}\rho_X\ge2$.
\end{paracol}
\endgroup
\AlignmentBlockEnd

\subsection{Uncertainty}
\label{app:semantics-alignment-eur}

\subsubsection{Measurements}

\AlignmentParagraph{Bases and overlaps.}

\AlignmentContinuationLabel{tab:alignment-eur-objects}
\AlignmentBlockHead{Measurement objects used by the entropic uncertainty relation.}
\begingroup\AlignmentBlockSetup
\begin{paracol}{2}
\code{src/eur/def.v}: \coqi{xbasis_vec2}, \coqi{x_projective_povm2}
\begin{minted}{coq}
Definition xbasis_vec2 (x : 'I_2) : 'cV[C]_2 :=
  \col_z xbasis_sign2 x z.
Definition x_projective_povm2 (x : 'I_2) : 'M[C]_2 :=
  (((2%:R)^-1)%:C) *: (xbasis_vec2 x *m (xbasis_vec2 x)^*t).
\end{minted}
\switchcolumn
\textbf{Two-dimensional $\mathbb X$ basis.}
For $x\in\{0,1\}$,
$v_x=((-1)^{xz})_z$ and
$M_x=\tfrac12v_xv_x^\dagger$.
\switchcolumn*
\code{src/eur/def.v}: \coqi{overlap_mx}, \coqi{q_of_mx}, \coqi{rho_XB}, \coqi{rho_ZC}, \coqi{eur_rho_XB}, \coqi{eur_rho_ZC}
\begin{minted}{coq}
Definition overlap_mx (M : 'I_nx -> 'M[C]_na)
    (N : 'I_nz -> 'M[C]_na) : R :=
  \big[maxr/0]_(x < nx)
    \big[maxr/0]_(z < nz)
      (@complex.Re R
        (i2norm (sqrtmx (M x) *m sqrtmx (N z)))) ^ 2.
Definition q_of_mx (M : 'I_nx -> 'M[C]_na)
    (N : 'I_nz -> 'M[C]_na) : R :=
  - log2 (overlap_mx M N).
Definition rho_XB (M : 'I_nx -> 'M[C]_na)
    (rhoABC : 'M[C]_(na * nb * nc)) : 'M[C]_(nx * nb) :=
  \sum_(x < nx) (delta_mx x x *t
    ptrace1 (ptrace2
      ((sqrtmx (M x) *t 1%:M *t 1%:M) *m rhoABC *m
       (sqrtmx (M x) *t 1%:M *t 1%:M)))).
Definition rho_ZC (N : 'I_nz -> 'M[C]_na)
    (rhoABC : 'M[C]_(na * nb * nc)) : 'M[C]_(nz * nc) :=
  \sum_(z < nz) (delta_mx z z *t
    ptrace1
      ((sqrtmx (N z) *t 1%:M *t 1%:M) *m rhoABC *m
       (sqrtmx (N z) *t 1%:M *t 1%:M))).
Definition eur_rho_XB disXB rhoSrc M dim_eq : 'F[HL]_(X :|: B) :=
  let Q := castmx (tenv_dim disXB, erefl) tenv_mxU in
  let eq_prod := etrans dim_eq (tenv_dim disXB) in
  mx2h (castmx (esym (tenv_dim disXB), esym (tenv_dim disXB))
    (Q *m castmx (eq_prod, eq_prod) (rho_XB M rhoSrc) *m Q^*t)).
Definition eur_rho_ZC disZC rhoSrc N dim_eq : 'F[HL]_(Z :|: C0) :=
  let Q := castmx (tenv_dim disZC, erefl) tenv_mxU in
  let eq_prod := etrans dim_eq (tenv_dim disZC) in
  mx2h (castmx (esym (tenv_dim disZC), esym (tenv_dim disZC))
    (Q *m castmx (eq_prod, eq_prod) (rho_ZC N rhoSrc) *m Q^*t)).
\end{minted}
The \coqi{eur_rho_XB} and \coqi{eur_rho_ZC} adapters transport the
matrix states into labelled registers using the displayed dimension equalities.
\switchcolumn
\textbf{Measurement overlap and post-measurement states.}
For POVMs $M=\{M_x\}$ and $N=\{N_z\}$,
\[
\begin{aligned}
c(M,N)&=\max_{x,z}
  \lVert\sqrt{M_x}\sqrt{N_z}\rVert_\infty^2,\\
q(M,N)&=-\log_2 c(M,N).
\end{aligned}
\]
The cq operands are
\[
\rho_{XB}=\sum_x|x\rangle\!\langle x|\otimes
 \Tr_{AC}\!\left[\sqrt{M_x}\rho_{ABC}\sqrt{M_x}\right],
\]
and analogously $\rho_{ZC}$ with $N_z$ and trace over $AB$.
\end{paracol}
\endgroup
\AlignmentBlockEnd

\subsubsection{Main Relation}

\AlignmentParagraph{Generic theorem.}

\AlignmentContinuationLabel{tab:alignment-eur}
\AlignmentBlockHead{Generic entropic uncertainty relation.}
\begingroup\AlignmentBlockSetup
\begin{paracol}{2}
\code{src/eur/eur.v}: \coqi{entropic_uncertainty_relation}
\begin{minted}{coq}
Theorem entropic_uncertainty_relation
    (L : finType) (HL : L -> chsType)
    (X B Z C0 : {set L})
    (disXB : [disjoint X & B])
    (disZC : [disjoint Z & C0])
    (na nb nc nx nz : nat)
    (eps : R)
    (rhoABC : 'M[C]_(na * nb * nc))
    (M : 'I_nx -> 'M[C]_na)
    (N : 'I_nz -> 'M[C]_na)
    (dimXB : nx * nb = dim ('H[HL]_(X :|: B)))
    (dimZC : nz * nc = dim ('H[HL]_(Z :|: C0)))
    (dim_X : nx = dim ('H[HL]_X))
    (dim_B : nb = dim ('H[HL]_B))
    (dim_Z : nz = dim ('H[HL]_Z))
    (dim_C : nc = dim ('H[HL]_C0)) :
  rhoABC \is denmx -> is_povm_mx M -> is_povm_mx N ->
  0 <= eps -> eps < Num.sqrt (@complex.Re R (\tr rhoABC)) ->
  smooth_hmin B eps (eur_rho_XB disXB rhoABC M dimXB) +
  smooth_hmax C0 eps (eur_rho_ZC disZC rhoABC N dimZC) >=
    q_of_mx M N.
\end{minted}
\AlignmentBoundaryLabel
The declaration makes density, POVM, dimension, and the strict smoothing
domain $0\le\eps<\sqrt{\Tr\rho}$ explicit.
The measurement objects are expanded in \Cref{tab:alignment-eur-objects}.
\switchcolumn
\textbf{Generic entropic uncertainty relation.}
For disjoint output/side-information registers and dimensions satisfying
\[
\begin{gathered}
X\cap B=\varnothing,\qquad Z\cap C=\varnothing,\\
n_Xn_B=\dim\mathcal H_{X\cup B},\\
n_Zn_C=\dim\mathcal H_{Z\cup C}.
\end{gathered}
\]
with $n_X=d_X$, $n_B=d_B$, $n_Z=d_Z$, and $n_C=d_C$, let
$\rho_{ABC}$ be subnormalized, let $M,N$ be POVMs, and require
$0\le\eps<\sqrt{\Tr\rho_{ABC}}$.
Then
\[
  H_{\min}^{\eps}(X|B)_{\rho_{XB}}
  +H_{\max}^{\eps}(Z|C)_{\rho_{ZC}}
  \ge q(M,N).
\]
\end{paracol}
\endgroup
\AlignmentBlockEnd

\AlignmentParagraph{BB84 specialization.}

\AlignmentContinuationLabel{tab:alignment-eur-bb84}
\AlignmentBlockHead{BB84 entropic uncertainty endpoint.}
\begingroup\AlignmentBlockSetup
\begin{paracol}{2}
\code{src/qkd/bb84/model.v}: \coqi{entropic_uncertainty_relation_bb84}
\begin{minted}{coq}
Corollary entropic_uncertainty_relation_bb84 (S0 E0 Zkey B : {set L})
    (disSE0 : [disjoint S0 & E0])
    (disZB : [disjoint Zkey & B])
    (n_key : nat) (eps' q : R)
    (na nb nc nx nz : nat)
    (rhoSrc : 'M[C]_(na * nb * nc))
    (M : 'I_nx -> 'M[C]_na) (N : 'I_nz -> 'M[C]_na)
    (dimSE : nx * nb = dim ('H[HL]_(S0 :|: E0)))
    (dimZB : nz * nc = dim ('H[HL]_(Zkey :|: B)))
    (dim_S : nx = dim ('H[HL]_S0))
    (dim_E : nb = dim ('H[HL]_E0))
    (dim_Zk : nz = dim ('H[HL]_Zkey))
    (dim_B : nc = dim ('H[HL]_B)) :
  rhoSrc \is denmx ->
  0 <= eps' -> eps' < Num.sqrt (@complex.Re R (\tr rhoSrc)) ->
  is_povm_mx M -> is_povm_mx N ->
  n_key%:R * q <= q_of_mx M N ->
  smooth_hmin E0 eps' (eur_rho_XB disSE0 rhoSrc M dimSE) +
  smooth_hmax B eps' (eur_rho_ZC disZB rhoSrc N dimZB)
    >= n_key%:R * q.
\end{minted}
\AlignmentBoundaryLabel
The generic theorem in \Cref{tab:alignment-eur} concludes a one-block
\coqi{q_of_mx} bound; this corollary is the paper's $nq$ form and makes the
overlap lower bound explicit.
The predicate \coqi{is_povm_mx} is expanded in
\Cref{tab:alignment-source-certificate}.
The extension and purification/DPI bridges used by the generic endpoint are expanded in \Cref{tab:alignment-eur-bridges}.
\switchcolumn
\textbf{Entropic uncertainty relation (\Cref{thm:eur}).}
For disjoint register pairs
\[
S\perp E,\qquad Z\perp B,
\]
assume
\[
\begin{gathered}
n_Xn_E=\dim\mathcal H_{S\cup E},\\
n_Zn_B=\dim\mathcal H_{Z\cup B},\\
n_X=\dim\mathcal H_S,\quad n_E=\dim\mathcal H_E,\\
n_Z=\dim\mathcal H_Z,\quad n_B=\dim\mathcal H_B.
\end{gathered}
\]
For a subnormalized $\rho_{ABC}$, POVMs $M,N$, and
$0\le\eps'<\sqrt{\Tr\rho_{ABC}}$, require
$n_{\rm key}q\le q(M,N)$.
Then
$\Hmineps(S\given E)+\Hmaxeps(Z\given B)\ge n_{\rm key}q$.
\end{paracol}
\endgroup
\AlignmentBlockEnd

\subsubsection{Entropy Bridges}

\AlignmentParagraph{Extensions and purification.}

\AlignmentContinuationLabel{tab:alignment-eur-bridges}
\AlignmentBlockHead{Reusable EUR extension and purification bridges.}
\begingroup\AlignmentBlockSetup
\begin{paracol}{2}
\code{src/entropy/pur/def.v}: \coqi{hmax_mx}, \coqi{smooth_hmax_mx};
\code{src/eur/def.v}: \coqi{purif_reshape_eq}, \coqi{purif_reshape};
\code{src/eur/prop.v}: \coqi{eur_ii}, \coqi{smooth_hmax_mx_rho_ZC_purif_dpi}
\begin{minted}{coq}
Definition hmax_mx (rho : 'M[C]_(na * nb)) : R :=
  sup [set v : R | exists sigma : 'M[C]_nb,
    sigma \is den1mx /\
    elog2 (na%:R *
      (gfidelity rho (mmix_mx na *t sigma)) ^+ 2) = v%:E].
Definition smooth_hmax_mx (eps : R)
    (rho : 'M[C]_(na * nb)) : R :=
  inf [set v : R | exists tau : 'M[C]_(na * nb),
    ((tau \is denmx) && (0 < @complex.Re R (\tr tau))) &&
    (pur_dist_mx tau rho <= eps) /\ v = hmax_mx tau].
Definition purif_reshape_eq (na nb nc nd : nat) :
    (na * nb * nc * nd = na * nb * (nc * nd))%N :=
  esym (mulnA (na * nb) nc nd).
Definition purif_reshape rho :=
  castmx (purif_reshape_eq na nb nc nd,
          purif_reshape_eq na nb nc nd) rho.
Lemma eur_ii (eps : R) (rhoA : 'M[C]_na)
    (rhoAB : 'M[C]_(na * nb)) (tauA : 'M[C]_na) :
  rhoAB \is denmx -> rhoAB != 0 ->
  ptrace2 rhoAB = rhoA ->
  (tauA \is denmx) &&
    (Num.sqrt (1 - (gfidelity tauA rhoA)^2) <= eps) ->
  (exists tauAB : 'M[C]_(na * nb),
      tauAB \is denmx /\ ptrace2 tauAB = tauA /\
      pur_dist_mx rhoAB tauAB <= eps) /\
  (\rank rhoAB = 1%N -> (\rank tauA <= nb)%N ->
   tauA != 0 -> exists tauAB : 'M[C]_(na * nb),
     tauAB \is denmx /\ \rank tauAB = 1%N /\
     ptrace2 tauAB = tauA /\ pur_dist_mx rhoAB tauAB <= eps).
Lemma smooth_hmax_mx_rho_ZC_purif_dpi
    (eps : R) (N : 'I_nz -> 'M[C]_na)
    (rhoABC : 'M[C]_(na * nb * nc))
    (rhoABCD : 'M[C]_(na * nb * nc * nd)) :
  rhoABCD \is denmx -> rhoABCD != 0 ->
  ptrace2 rhoABCD = rhoABC -> is_povm_mx N ->
  0 <= eps -> eps < Num.sqrt (@complex.Re R (\tr rhoABCD)) ->
  smooth_hmax_mx eps (rho_ZC N (purif_reshape rhoABCD)) <=
  smooth_hmax_mx eps (rho_ZC N rhoABC).
\end{minted}
Here \coqi{gfidelity} and \coqi{pur_dist_mx} are expanded in
\Cref{tab:alignment-entropy}, while \coqi{rho_ZC} is expanded in
\Cref{tab:alignment-eur-objects} and \coqi{mmix_mx} in
\Cref{tab:alignment-bb84-projector}.
\switchcolumn
\textbf{Close-extension lifting and mixed-state $H_{\max}$ data processing.}
If $\rho_{AB}$ extends $\rho_A$ and
$P(\tau_A,\rho_A)\le\eps$, there is an extension $\tau_{AB}$ with
$P(\rho_{AB},\tau_{AB})\le\eps$.
If $\rho_{AB}$ is pure and
$\operatorname{rank}\tau_A\le\dim B$, the extension may be pure.
For a purification $\rho_{ABCD}$ of $\rho_{ABC}$,
\[
H_{\max}^{\eps}(Z|CD)_{\rho_{ABCD}}
\le H_{\max}^{\eps}(Z|C)_{\rho_{ABC}}.
\]
\end{paracol}
\endgroup
\AlignmentBlockEnd

\subsection{Hashing and Sampling}
\label{app:semantics-alignment-hashing}

\subsubsection{Hash Families}

\AlignmentParagraph{Construction and composition.}

\AlignmentContinuationLabel{tab:alignment-hash-constructions}
\AlignmentBlockHead{Reusable hash-family constructions.}
\begingroup\AlignmentBlockSetup
\begin{paracol}{2}
\code{src/lhl/construction.v}: \coqi{allfun_hash}, \coqi{allfun_collision_cardM}, \coqi{allfun_two_univ}, \coqi{reindex_delta_almost_two_univ}, \coqi{two_univ_transport}
\begin{minted}{coq}
Definition allfun_hash (f : {ffun X -> Z}) (x : X) : Z := f x.
Lemma allfun_collision_cardM (x x' : X) : x != x' ->
  (#|[set f : {ffun X -> Z} | f x == f x']| * #|Z|)%N =
  #|{ffun X -> Z}|.
Theorem allfun_two_univ : (0 < #|Z|)%N ->
  @two_univ R [the finType of {ffun X -> Z}] X Z allfun_hash.
Lemma reindex_delta_almost_two_univ (e : Gidx -> Fidx)
    (F : Fidx -> X -> Z) (d : R) :
  bijective e ->
  delta_almost_two_univ (R := R) F d ->
  delta_almost_two_univ (R := R) (fun i x => F (e i) x) d.
Lemma two_univ_transport (Fidx Fidx' X X' Z Z' : finType)
    (F : Fidx -> X -> Z) (e : Fidx' -> Fidx)
    (hX : X' -> X) (gZ : Z -> Z') :
  bijective e -> injective hX -> bijective gZ ->
  two_univ (R := R) F ->
  two_univ (R := R) (fun f' x' => gZ (F (e f') (hX x'))).
\end{minted}
The all-functions construction is an exact, reusable baseline; its seed is generally exponential in the input description.
\switchcolumn
\textbf{All-functions family and representation transport.}
For nonempty~$Z$, taking the seed to be a function $f:X\to Z$ gives collision probability exactly $|Z|^{-1}$ for distinct inputs.
Bijections of seeds, injective input relabellings, and bijective output relabellings preserve two-universality.
\switchcolumn*
\code{src/lhl/def.v}: \coqi{hash_comp};
\code{src/lhl/prop.v}: \coqi{hash_comp_delta_almost_two_univ}
\begin{minted}{coq}
Definition hash_comp
    (F1 : F1idx -> X -> Y) (F2 : F2idx -> Y -> Z)
    (fg : (F1idx * F2idx)%type) (x : X) : Z :=
  F2 fg.2 (F1 fg.1 x).
Lemma hash_comp_delta_almost_two_univ
    (F1 : F1idx -> X -> Y) (F2 : F2idx -> Y -> Z)
    (delta1 delta2 : R) :
  0 <= delta1 ->
  0 <= delta2 ->
  delta_almost_two_univ F1 delta1 ->
  delta_almost_two_univ F2 delta2 ->
  delta_almost_two_univ (hash_comp F1 F2) (delta1 + delta2).
\end{minted}
This is a generic TCS-facing composition theorem, independent of BB84 and of the quantum leftover-hash proof.
\switchcolumn
\textbf{Composition of almost-two-universal families.}
If $F_1:X\to Y$ is $\delta_1$-almost two-universal and $F_2:Y\to Z$ is
$\delta_2$-almost two-universal with $\delta_1,\delta_2\ge0$, their
independently seeded composition is $(\delta_1+\delta_2)$-almost two-universal.
\end{paracol}
\endgroup
\AlignmentBlockEnd

\subsubsection{Leftover Hash Lemma}

\AlignmentParagraph{Main bound.}

\AlignmentContinuationLabel{tab:alignment-ingredients-cont}
\AlignmentBlockHead{Hashing and sampling endpoints.}
\begingroup\AlignmentBlockSetup
\begin{paracol}{2}
\code{src/lhl/lhl.v}: \coqi{two_universal_hashing}
\begin{minted}{coq}
Theorem two_universal_hashing
    (Fh : 'Idx[HL]_Fl -> 'Idx[HL]_X -> 'Idx[HL]_Zl)
    (eps : R) (rhoXE : 'F[HL]_(X :|: E)) :
  rhoXE \is psdlf ->
  @complex.Re R (\Tr rhoXE) <= 1 ->
  is_cq_state rhoXE ->
  @two_univ R _ _ _ Fh ->
  0 <= eps ->
  let rhoZFE := post_hash_state disXE Fh rhoXE in
  dist_from_uniform rhoZFE <=
  eps + (2%:R)^-1 * Num.sqrt (expR ((log2 (dim 'H[HL]_Zl)%:R -
    smooth_hmin E eps rhoXE) * ln 2%:R)).
\end{minted}
\AlignmentBoundaryLabel
The hypotheses state positivity, sub-normalization, classicality, and
register disjointness; \coqi{expR} expresses the same base-two power.
\switchcolumn
\textbf{Quantum Leftover Hash Lemma (\Cref{thm:lhl}).}
The smoothed endpoint is
\[
d_{\rm u}(f(X)|FE)
\le\eps+\frac12\sqrt{2^{\ell-H_{\min}^{\eps}(X|E)}}.
\]
\end{paracol}
\endgroup
\AlignmentBlockEnd

\subsubsection{Sampling Bounds}

\AlignmentParagraph{Subsets and permutations.}

\AlignmentContinuationLabel{tab:alignment-serfling}
\AlignmentBlockHead{Reusable Serfling objects and bounds.}
\begingroup\AlignmentBlockSetup
\begin{paracol}{2}
\code{src/utils/prob.v}: \coqi{Pr_cond}, \coqi{uniform_perm_mu},
\coqi{perm_pair_holdout_count}, \coqi{perm_pair_holdout_err},
\coqi{perm_pair_deviation_event}, \coqi{perm_pair_total_count},
\coqi{perm_pair_total_err}, \coqi{serfling_bound},
\coqi{serfling_uniform_perm_pair_deviation_tail}
\begin{minted}{coq}
Definition Pr_cond (mu : T -> R) (A B : {set T}) : R :=
  if Pr_event mu B == 0 then 0
  else Pr_event mu (A :&: B) / Pr_event mu B.
Definition uniform_perm_mu (N : nat) (pi : 'S_N) : R :=
  (#|{: 'S_N}|%:R)^-1.
Definition perm_pair_holdout_count
    (z z' : bitstr N) (pi : 'S_N) : nat :=
  d_ham (perm_pair_holdout z pi) (perm_pair_holdout z' pi).
Definition perm_pair_holdout_err
    (z z' : bitstr N) (pi : 'S_N) : R :=
  (perm_pair_holdout_count z z' pi)%:R / n%:R.
Definition perm_pair_deviation_event
    (z z' : bitstr N) (mu : R) : {set 'S_N} :=
  [set pi : 'S_N |
    perm_pair_holdout_err z z' pi >=
      perm_pair_sample_err z z' pi + mu].
Definition perm_pair_total_count (z z' : bitstr N) : nat :=
  d_ham z z'.
Definition perm_pair_total_err
    (z z' : bitstr N) (_ : 'S_N) : R :=
  (perm_pair_total_count z z')%:R / N%:R.
Lemma serfling_bound (z z' : bitstr N) (errtot delta : R) :
  (0 < n)%N -> (0 < k)%N ->
  0 <= delta ->
  Pr_cond
    (fun pi : 'S_(k + n) => @uniform_perm_mu (k + n) pi)
    [set pi : 'S_(k + n) |
      @perm_pair_holdout_err n k z z' pi >= errtot + delta]%SET
    [set pi : 'S_(k + n) |
      @perm_pair_total_err n k z z' pi == errtot]%SET <=
    expR (- (2%:R * n%:R * (n + k)%:R /
      (k + 1)%:R * delta^2)).
Lemma serfling_uniform_perm_pair_deviation_tail
    (z z' : bitstr N) (mu : R) :
  (0 < n)%N -> (0 < k)%N -> 0 <= mu ->
  Pr_event (fun pi : 'S_N => @uniform_perm_mu N pi)
    (perm_pair_deviation_event z z' mu) <=
  expR (- 2%:R *
    (((k * n)%:R / N%:R) *
      (k%:R / k.+1%:R) * mu ^+ 2)).
\end{minted}
The second declaration exposes the unconditional deviation event and bound used by the aggregate BB84 sampling proof.
\switchcolumn
\textbf{Serfling sampling bound (proof step in \Cref{thm:bb84}).}
For the uniform distribution on $S_N$, set
\[
\begin{aligned}
e_{\rm hold}(\pi)&=d_H(z_{\rm hold},z'_{\rm hold})/n,\\
e_{\rm tot}&=d_H(z,z')/N,\\
\mathcal D_\mu&=\{\pi:e_{\rm hold}(\pi)
                 \ge e_{\rm smp}(\pi)+\mu\},\\
\Pr[A\mid B]&=
\begin{cases}\Pr[A\cap B]/\Pr[B],&\Pr[B]\ne0,\\0,&\Pr[B]=0.
\end{cases}
\end{aligned}
\]
For $N=n+k$,
$\Pr[\hat p_{\rm hold}\ge p_{\rm all}+\delta\mid\hat p_{\rm all}=p_{\rm all}]
\le\exp(-2nN\delta^2/(k+1))$.
The unconditional deviation form used downstream is
\[
\begin{gathered}
\Pr_{\pi\leftarrow\mathsf{Unif}(S_N)}
  [\mathsf{dev}(z,z',\pi)\ge\mu]\\
\le
\exp\!\left[-2\frac{kn}{N}\frac{k}{k+1}\mu^2\right].
\end{gathered}
\]
\end{paracol}
\endgroup
\AlignmentBlockEnd

\subsubsection{Extractor Lemmas}

\AlignmentParagraph{Collision bounds.}

\AlignmentContinuationLabel{tab:alignment-lhl-middle}
\AlignmentBlockHead{Reusable middle lemmas in the extractor proof.}
\begingroup\AlignmentBlockSetup
\begin{paracol}{2}
\code{src/lhl/def.v}: \coqi{collision_ent};
\code{src/lhl/util.v}: \coqi{hoelder}, \coqi{coll_from_dmax}, \coqi{min_coll_bound}
\begin{minted}{coq}
Definition collision_ent (rhoAB : 'F[HL]_(A :|: B))
    (sigmaB : 'F[HL]_B) : R :=
  let sigma_inv_half := (\1 : 'F_A) \O pinvlf (sqrtH sigmaB) in
  @complex.Re R
    (\Tr (sigma_inv_half \o rhoAB \o sigma_inv_half \o rhoAB)).
Lemma hoelder {m1 m2 m3 m4 : nat}
    (A : 'M[C]_(m1, m2)) (B : 'M[C]_(m2, m3))
    (C0 : 'M[C]_(m3, m4)) (r s t : R) :
  (0 < r)%R -> (0 < s)%R -> (0 < t)%R ->
  (r^-1 + s^-1 + t^-1 = 1)%R ->
  trnorm (A *m B *m C0) <=
    schnorm r A * schnorm s B * schnorm t C0.
Lemma coll_from_dmax (rhoXE : 'F[HL]_S_XE)
    (sigmaE : 'F[HL]_E) (c : R) :
  rhoXE \is psdlf -> sigmaE \is psdlf -> 0 <= c ->
  rhoXE <= c%:C *: ((\1 : 'F_X) \O sigmaE) ->
  collision_ent rhoXE sigmaE <= c * @complex.Re R (\Tr rhoXE).
Lemma min_coll_bound (rhoXE : 'F[HL]_S_XE) :
  rhoXE \is psdlf -> @complex.Re R (\Tr rhoXE) <= 1 ->
  exists (sigmaE : 'F[HL]_E),
    sigmaE \is den1lf /\
    (supph (ptrace_castl disXE rhoXE) `<=` supph sigmaE)%HS /\
    collision_ent rhoXE sigmaE <=
      expR (- hmin E rhoXE * ln 2%:R).
\end{minted}
\AlignmentBoundaryLabel
\coqi{coll_from_dmax} is the operator implication, and
\coqi{min_coll_bound} supplies its optimized collision-entropy witness.
\switchcolumn
\textbf{Schatten and collision-entropy core.}
For conjugate exponents,
$\lVert ABC\rVert_1\le\lVert A\rVert_r\lVert B\rVert_s\lVert C\rVert_t$.
If $\rho_{XE}\le c(\id_X\otimes\sigma_E)$, then
$\Gamma_C(\rho_{XE}\mid\sigma_E)\le c\,\Tr\rho_{XE}$;
for a subnormalized state $\rho_{XE}$, there exists a normalized
$\sigma_E$ with
$\Gamma_C(\rho_{XE}\mid\sigma_E)\le2^{-H_{\min}(X|E)_\rho}$.
\end{paracol}
\endgroup
\AlignmentBlockEnd

\AlignmentContinuationLabel{tab:alignment-lhl-middle-transport}
\AlignmentBlockHead{Reusable middle lemmas in the extractor proof (continued).}
\begingroup\AlignmentBlockSetup
\begin{paracol}{2}
\code{src/lhl/util.v}: \coqi{cq_smoothing};
\code{src/lhl/post\_hash.v}: \coqi{pd_mono_post_hash};
\code{src/lhl/prop.v}: \coqi{dist_uniform_pd_triangle};
\code{src/utils/dist/pd.v}: \coqi{pd_mono_cptp}
\begin{minted}{coq}
Lemma cq_smoothing (eps : R) (rhoXE : 'F[HL]_S_XE) :
  rhoXE \is psdlf -> @complex.Re R (\Tr rhoXE) <= 1 ->
  0 <= eps -> eps < Num.sqrt (@complex.Re R (\Tr rhoXE)) ->
  is_cq_state rhoXE ->
  exists tauXE, eps_ball eps rhoXE tauXE /\
    is_cq_state tauXE /\
    smooth_hminE E eps rhoXE = hminE E tauXE.
Lemma pd_mono_post_hash Fh (rhoXE tauXE : 'F[HL]_S_XE) :
  rhoXE \is psdlf -> tauXE \is psdlf ->
  @complex.Re R (\Tr rhoXE) <= 1 ->
  @complex.Re R (\Tr tauXE) <= 1 ->
  pur_dist (post_hash_state disXE Fh rhoXE)
           (post_hash_state disXE Fh tauXE)
  <= pur_dist rhoXE tauXE.
Lemma pd_mono_cptp (E : 'SO[HL]_(U, V))
    (rho tau : 'F[HL]_U) :
  E \is cptp -> rho \is psdlf -> tau \is psdlf ->
  @complex.Re R (\Tr rho) <= 1 ->
  @complex.Re R (\Tr tau) <= 1 ->
  pur_dist (E rho) (E tau) <= pur_dist rho tau.
Lemma dist_uniform_pd_triangle
    (rhoAB tauAB : 'F[H]_(A :|: B)) :
  rhoAB \is psdlf -> tauAB \is psdlf ->
  @complex.Re R (\Tr rhoAB) <= 1 ->
  @complex.Re R (\Tr tauAB) <= 1 ->
  dist_from_uniform rhoAB <=
    pur_dist rhoAB tauAB + dist_from_uniform tauAB.
\end{minted}
\coqi{pd_mono_post_hash} instantiates the generic \coqi{pd_mono_cptp}
contraction theorem for post-hashing.
\switchcolumn
\textbf{Cq-preserving smoothing and hashing transport.}
An optimal smooth-min witness for a cq-state can remain cq.
Every CPTP map $\mathcal E$ contracts purified distance,
\[
P(\mathcal E(\rho),\mathcal E(\tau))\le P(\rho,\tau),
\]
so post-hashing contracts purified distance.
Distance from uniform obeys the mixed triangle bound
$d_{\rm u}(\rho)\le P(\rho,\tau)+d_{\rm u}(\tau)$.
\end{paracol}
\endgroup
\AlignmentBlockEnd

\AlignmentParagraph{Privacy amplification.}

\AlignmentContinuationLabel{tab:alignment-lhl-middle-cont}
\AlignmentBlockHead{Collision privacy amplification and the extractor endpoint.}
\begingroup\AlignmentBlockSetup
\begin{paracol}{2}
\code{src/lhl/prop.v}: \coqi{uniform_to_collision}, \coqi{privacy_amp}
\begin{minted}{coq}
Lemma uniform_to_collision (rhoAB : 'F[HL]_S_AB)
    (tauB : 'F[HL]_B) :
  rhoAB \is psdlf -> tauB \is psdlf ->
  @complex.Re R (\Tr tauB) <= 1 ->
  (supph (ptrace_castl disAB rhoAB) `<=` supph tauB)%HS ->
  let rhoB := ptrace_castl disAB rhoAB in
  let tau_inv_half := pinvlf (sqrtH tauB) in
  dist_from_uniform rhoAB <= (2%:R)^-1 * Num.sqrt (
    (dim 'H[HL]_A)%:R * collision_ent rhoAB tauB -
    @complex.Re R
      (\Tr (rhoB \o tau_inv_half \o rhoB \o tau_inv_half))).
Lemma privacy_amp (rhoXE : 'F[HL]_S_XE) (tauE : 'F[HL]_E)
    (Fh : 'Idx[HL]_Fl -> 'Idx[HL]_X -> 'Idx[HL]_Zl)
    (delta : R) :
  rhoXE \is psdlf -> @complex.Re R (\Tr rhoXE) <= 1 ->
  tauE \is psdlf -> is_cq_state rhoXE ->
  0 <= delta -> delta_almost_two_univ Fh delta ->
  let rhoZFE := post_hash_state disXE Fh rhoXE in
  let rhoF_tauE := ((dim 'H[HL]_Fl)%:R^-1)%:C *:
    ((\1 : 'F_Fl) \O tauE) in
  let tau_inv_half := pinvlf (sqrtH tauE) in
  let rhoE := ptrace_castl disXE rhoXE in
  collision_ent rhoZFE rhoF_tauE <=
    collision_ent rhoXE tauE +
    delta * @complex.Re R
      (\Tr (rhoE \o tau_inv_half \o rhoE \o tau_inv_half)).
\end{minted}
These declarations establish the displayed uniform-distance and collision
privacy-amplification bounds.
\switchcolumn
\textbf{Uniform distance and collision privacy amplification.}
\[
 d_{\rm u}(A|B)_\rho
 \le\tfrac12\sqrt{d_A\Gamma_C(\rho|\tau_B)
 -\Tr(\rho_B\tau_B^{-1/2}\rho_B\tau_B^{-1/2})},
\]
and $\delta$-almost-two-universality adds at most
$\delta\Tr(\rho_E\tau_E^{-1/2}\rho_E\tau_E^{-1/2})$ to collision entropy.
\switchcolumn*
\code{src/lhl/lhl.v}: \coqi{almost_two_universal_hashing}
\begin{minted}{coq}
Theorem almost_two_universal_hashing
    (Fh : 'Idx[HL]_Fl -> 'Idx[HL]_X -> 'Idx[HL]_Zl)
    (delta eps eps_bar : R) (rhoXE : 'F[HL]_(X :|: E)) :
  rhoXE \is den1lf ->
  is_cq_state rhoXE ->
  0 <= delta ->
  delta_almost_two_univ Fh delta ->
  0 <= eps ->
  0 < eps_bar ->
  let rhoZFE := post_hash_state disXE Fh rhoXE in
  let printed_bound : R :=
    eps + eps_bar + (2%:R)^-1 * Num.sqrt (
      (dim 'H[HL]_Zl)%:R * delta - 1 +
      expR ((log2 (dim 'H[HL]_Zl)%:R -
        smooth_hmin E eps rhoXE) * ln 2%:R)
        * (2%:R * eps_bar^-2 + 1)) in
    dist_from_uniform rhoZFE <= printed_bound.
\end{minted}
The theorem conditions min-entropy on the adversary register~$E$.
\switchcolumn
\textbf{$\delta$-almost-two-universal LHL.}
For a normalized cq-state,
\[
\begin{aligned}
d_{\rm u}\le{}&\eps+\bar\eps+\frac12\Bigl[
d_Z\delta-1\\
&+2^{\log_2d_Z-H_{\min}^{\eps}(X|E)}
\bigl(2\bar\eps^{-2}+1\bigr)
\Bigr]^{1/2}.
\end{aligned}
\]
\end{paracol}
\endgroup
\AlignmentBlockEnd

\subsubsection{Seed Averaging}

\AlignmentParagraph{Average bounds.}

\AlignmentContinuationLabel{tab:alignment-seed-average-lhl}
\AlignmentBlockHead{Fixed-seed average extractor bounds.}
\begingroup\AlignmentBlockSetup
\begin{paracol}{2}
\code{src/lhl/def.v}: \coqi{post_hash_seed_state},
\coqi{post_hash_seed_average_dist};
\code{src/lhl/lhl.v}: \coqi{two_universal_seed_average_hmin},
\coqi{almost_two_universal_seed_average_hmin}, \coqi{quant}
\begin{minted}{coq}
Definition post_hash_seed_state Fh f rhoXE : 'F[HL]_(Zl :|: E) :=
  \sum_z [> deltav z ; deltav z <] \O
    (\sum_(x | Fh f x == z) cq_block disXE x rhoXE).
Definition post_hash_seed_average_dist Fh rhoXE : R :=
  (dim 'H[HL]_Fl)%:R^-1 *
    \sum_f dist_from_uniform (post_hash_seed_state disXE Fh f rhoXE).
Theorem two_universal_seed_average_hmin Fh rhoXE :
  rhoXE \is den1lf -> is_cq_state rhoXE -> @two_univ R _ _ _ Fh ->
  post_hash_seed_average_dist disXE Fh rhoXE <=
    (2%:R)^-1 * Num.sqrt (expR (
      (log2 (dim 'H[HL]_Zl)%:R - hmin E rhoXE) * ln 2%:R)).
Theorem almost_two_universal_seed_average_hmin
    Fh delta eps_bar rhoXE :
  rhoXE \is den1lf -> is_cq_state rhoXE ->
  0 <= delta -> delta_almost_two_univ Fh delta -> 0 < eps_bar ->
  post_hash_seed_average_dist disXE Fh rhoXE <=
    eps_bar + (2%:R)^-1 * Num.sqrt (
      (dim 'H[HL]_Zl)%:R * delta - 1 +
      expR ((log2 (dim 'H[HL]_Zl)%:R - hmin E rhoXE) * ln 2%:R) *
        (2%:R * eps_bar^-2 + 1)).
Theorem quant Fh delta rhoXE :
  rhoXE \is den1lf -> is_cq_state rhoXE ->
  0 <= delta -> delta_almost_two_univ Fh delta ->
  post_hash_seed_average_dist disXE Fh rhoXE <=
    inf [set Delta : R | exists eps_bar : R,
      0 < eps_bar /\ Delta = eps_bar + (2%:R)^-1 * Num.sqrt (
        (dim 'H[HL]_Zl)%:R * delta - 1 +
        expR ((log2 (dim 'H[HL]_Zl)%:R - hmin E rhoXE) * ln 2%:R) *
          (2%:R * eps_bar^-2 + 1))] /\
  (delta <= ((dim 'H[HL]_Zl)%:R)^-1 ->
    post_hash_seed_average_dist disXE Fh rhoXE <=
      (2%:R)^-1 * Num.sqrt (expR (
        (log2 (dim 'H[HL]_Zl)%:R - hmin E rhoXE) * ln 2%:R))).
\end{minted}
The final conjunct of \coqi{quant} is the two-universal specialization.
\switchcolumn
\textbf{Published-seed states and average LHL.}
For a fixed seed $f$,
\[
\begin{gathered}
\rho_{ZE}^{[f]}=\sum_z|z\rangle\!\langle z|\otimes
  \sum_{x:F(f,x)=z}\rho_E^{[x]},\\
\bar d_{\rm u}=|\mathcal F|^{-1}\sum_f d_{\rm u}(Z|E)_{\rho^{[f]}}.
\end{gathered}
\]
Two-universal hashing yields
\[
\bar d_{\rm u}\le\frac12
\sqrt{2^{\log_2d_Z-H_{\min}(X|E)}}.
\]
The $\delta$-almost-two-universal bound optimizes the corresponding
$\bar\eps$-dependent expression.
\end{paracol}
\endgroup
\AlignmentBlockEnd

\subsubsection{Short Seeds}

\AlignmentParagraph{Construction.}

\AlignmentContinuationLabel{tab:alignment-short-seed-construction}
\AlignmentBlockHead{Polynomial short-seed construction.}
\begingroup\AlignmentBlockSetup
\begin{paracol}{2}
\code{src/lhl/construction.v}: \coqi{binary_field_dim},
\coqi{binary_field_card}, \coqi{binary_field_coord_project},
\coqi{binary_field_coord_output_card},
\coqi{binary_field_coord_project_diff_compat},
\coqi{binary_field_coord_project_surj}, \coqi{field_mult_project_hash},
\coqi{field_mult_project_hash_two_univ}, \coqi{polynomial_hash_poly},
\coqi{polynomial_hash_paper},
\coqi{polynomial_hash_paper_delta_almost_two_univ};
\code{src/lhl/prop.v}: \coqi{hash_comp_delta_almost_two_univ}
\begin{minted}{coq}
Definition binary_field_dim : nat := \dim {: K : vectType 'F_2}.
Definition binary_field_coord_project (ell : nat)
    (Hell : (ell <= binary_field_dim)%N) (x : K) : 'rV['F_2]_ell :=
  \row_i coord (vbasis {: K}) (widen_ord Hell i) x.
Lemma binary_field_card : #|K| = (2 ^ binary_field_dim)%N.
Lemma binary_field_coord_output_card ell :
  #|'rV['F_2]_ell| = (2 ^ ell)%N.
Lemma binary_field_coord_project_diff_compat ell Hell :
  forall u v : K,
    (binary_field_coord_project Hell u ==
      binary_field_coord_project Hell v) =
    (binary_field_coord_project Hell (u - v) ==
      binary_field_coord_project Hell 0).
Lemma binary_field_coord_project_surj ell Hell :
  forall z : 'rV['F_2]_ell,
    exists t : K, binary_field_coord_project Hell t = z.
Definition field_mult_project_hash (Z : finType) (pi : K -> Z)
    (alpha x : K) : Z := pi (x * alpha).
Theorem field_mult_project_hash_two_univ (Z : finType) (pi : K -> Z) :
  (forall u v, (pi u == pi v) = (pi (u - v) == pi 0)) ->
  (forall z, exists t, pi t = z) ->
  two_univ (R := R) (field_mult_project_hash pi).
Definition polynomial_hash_poly (x : r.-tuple K) : {poly K} :=
  \poly_(i < r) nth 0 x i.
Definition polynomial_hash_paper (alpha : K) (x : r.-tuple K) : K :=
  (polynomial_hash_poly [tuple tnth x (rev_ord i) | i < r]).[alpha].
Theorem polynomial_hash_paper_delta_almost_two_univ :
  delta_almost_two_univ (R := R) polynomial_hash_paper
    (r.-1%:R / #|K|%:R).
\end{minted}
The component composition is \coqi{hash_comp_delta_almost_two_univ}.
\switchcolumn
\textbf{Polynomial and finite-field families.}
Over $K=\mathrm{GF}(2^k)$, evaluation of a degree-$(r-1)$ polynomial at a
uniform point has collision probability at most $(r-1)/2^k$.
For $\ell\le k$, the coordinate projection
$\pi_\ell:K\to\mathbb F_2^\ell$ is surjective and
\[
\pi_\ell(u)=\pi_\ell(v)
\Longleftrightarrow \pi_\ell(u-v)=\pi_\ell(0).
\]
The multiply-and-project family is two-universal, and composition gives
$\delta=(r-1)/2^k+2^{-\ell}$.
\end{paracol}
\endgroup
\AlignmentBlockEnd

\AlignmentParagraph{Extractor.}

\AlignmentContinuationLabel{tab:alignment-short-seed}
\AlignmentBlockHead{Polynomial short-seed extractor theorem.}
\begingroup\AlignmentBlockSetup
\begin{paracol}{2}
\code{src/lhl/lhl.v}: \coqi{almost_two_universal}
\begin{minted}{coq}
Theorem almost_two_universal
    (n_bits ell : nat) (eps : R) (rhoXE : 'F[HL]_(X :|: E)) :
  let target := ell%:R + log2 ((n_bits%:R : R) / ell%:R) +
    log2 (eps^-2) - 1 in
  let k_bits := Num.truncn target in
  let r := `|Num.ceil ((n_bits%:R : R) / (k_bits%:R : R))|%N in
  rhoXE \is den1lf -> is_cq_state rhoXE -> 0 < eps ->
  (0 < n_bits)%N -> (0 < ell)%N -> (ell <= k_bits)%N ->
  dim 'H[HL]_Fl = (2 ^ (2 * k_bits))%N ->
  dim 'H[HL]_X = (2 ^ n_bits)%N ->
  dim 'H[HL]_Zl = (2 ^ ell)%N ->
  let delta := r.-1%:R / (2%:R ^+ k_bits) +
    ((2%:R ^+ ell)^-1)%R in
  exists Fh : 'Idx[HL]_Fl -> 'Idx[HL]_X -> 'Idx[HL]_Zl,
    delta_almost_two_univ Fh delta /\
    log2 ((dim 'H[HL]_Fl)%:R : R) = 2%:R * k_bits%:R /\
    log2 ((dim 'H[HL]_Fl)%:R : R) =
      2%:R * ((floor target)%:~R : R) /\
    log2 ((dim 'H[HL]_X)%:R : R) = n_bits%:R /\
    log2 ((dim 'H[HL]_Zl)%:R : R) = ell%:R /\
    delta <= (1 + 4%:R * eps ^+ 2) / (2%:R ^+ ell) /\
    let printed_bound := 3%:R * eps + (2%:R)^-1 * Num.sqrt (
      expR ((log2 (dim 'H[HL]_Zl)%:R - smooth_hmin E eps rhoXE) *
        ln 2%:R) * (2%:R * eps^-2 + 1)) in
    dist_from_uniform (post_hash_state disXE Fh rhoXE) <= printed_bound.
\end{minted}
The theorem exposes positivity, $\ell\le k$, register dimensions, and the
floor convention used by the construction.
\switchcolumn
\textbf{Short-seed endpoint.}
Set
\[
\begin{gathered}
k=\left\lfloor\ell+\log_2\frac n\ell+
  \log_2\frac1{\eps^2}-1\right\rfloor,\\
r=\left\lceil\frac nk\right\rceil,
\quad s=2k.
\end{gathered}
\]
The constructed family satisfies
\[
\begin{gathered}
\delta\le\frac{1+4\eps^2}{2^\ell},\\
d_{\rm u}\le3\eps+\frac12
\sqrt{2^{\ell-H_{\min}^{\eps}(X|E)}(2\eps^{-2}+1)}.
\end{gathered}
\]
\end{paracol}
\endgroup
\AlignmentBlockEnd

\subsection{BB84 Proof}
\label{app:semantics-alignment-bb84-proof}

\subsubsection{Parameter Estimation}

\AlignmentParagraph{Bad-pass mass.}

\AlignmentContinuationLabel{tab:alignment-bb84-bad-pass}
\AlignmentBlockHead{Reusable bad-pass mass and Hamming-ball bounds.}
\begingroup\AlignmentBlockSetup
\begin{paracol}{2}
\code{src/qkd/bb84/model.v}: \coqi{bb84_pass_bad_weight},
\coqi{bb84_pass_bad_weight_le}, \coqi{bb84_pass_bad_weight_le_pass};
\code{src/utils/bitstr.v}: \coqi{hamming_ball_count_max_entropy_bound},
\coqi{hamming_ball_count_max_entropy_bound_total}
\begin{minted}{coq}
Definition bb84_pass_bad_weight
    (P : bb84_protocol) (G : bb84_zsource P) (eps : R) : R :=
  \sum_(z : 'I_(bb84_full_string_count P))
    \sum_(z' : 'I_(bb84_full_string_count P))
      bb84_W G z z' *
        Pr_event
          (fun S0 : {set 'I_(bb84_total_rounds P)} =>
            @subset_uniform_mu R (bb84_n P) (bb84_k P) S0)
          ([set S0 : {set 'I_(bb84_total_rounds P)} |
             @subset_sorted_rep (bb84_n P) (bb84_k P) S0 \in
               @perm_pair_sample_le_event R
                 (bb84_n P) (bb84_k P)
                 (ord_to_bitstring z) (ord_to_bitstring z')
                 (bb84_Qtol P)] :&:
           [set S0 : {set 'I_(bb84_total_rounds P)} |
             @subset_sorted_rep (bb84_n P) (bb84_k P) S0 \in
               @perm_pair_deviation_event R
                 (bb84_n P) (bb84_k P)
                 (ord_to_bitstring z) (ord_to_bitstring z')
                 (bb84_mu (bb84_n P) (bb84_k P) eps)]).
Lemma bb84_pass_bad_weight_le
    (P : bb84_protocol) (G : bb84_zsource P) (eps : R) :
  (0 < bb84_n P)%N -> (0 < bb84_k P)%N -> 0 < eps ->
  bb84_pass_bad_weight G eps <= eps ^+ 2.
Lemma bb84_pass_bad_weight_le_pass
    (P : bb84_protocol) (G : bb84_zsource P) (eps : R) :
  bb84_pass_bad_weight G eps <= bb84_p_pass_tot G.
Lemma hamming_ball_count_max_entropy_bound (n : nat) (delta : R) :
  (0 < n)%N -> 0 < delta -> delta < 1 / 2 ->
  log2 (@relative_hamming_ball_count_max n delta)%:R <=
    n%:R * h_bin delta.
Lemma hamming_ball_count_max_entropy_bound_total
    (n : nat) (delta : R) :
  (0 < n)%N -> delta <= 1 / 2 ->
  log2 (@relative_hamming_ball_count_max n delta)%:R <=
    n%:R * h_bin delta.
\end{minted}
The pass-state weights and canonical subset sampler are expanded in
\Cref{tab:alignment-bb84-pass-probability}; the deviation event is expanded
with the Serfling objects in \Cref{tab:alignment-serfling}.
The total theorem extends the interior Hamming bound to every
$\delta\le\tfrac12$.
\switchcolumn
\textbf{Pass intersected with excessive holdout error.}
Let $\mathcal A_{\Qtol}$ be the pass event and
$\mathcal D_{\mu}$ the event that the holdout error exceeds the sampled error
by at least $\mu$.  The reusable weighted mass is
\[
\begin{gathered}
B_{\rm pass}(\eps)=
 \sum_{z,z'}W_G(z,z')\\
 \Pr_S\!\left[\mathcal A_{\Qtol}(z,z')
               \cap\mathcal D_{\mu(\eps)}(z,z')\right],
\end{gathered}
\]
with the two independent bounds
\[
B_{\rm pass}(\eps)\le\eps^2,
\qquad B_{\rm pass}(\eps)\le p_{\rm pass}^{\rm tot}.
\]
For $n>0$ and $0<\delta<\tfrac12$,
\[
\log_2 B_n(\delta)\le n h_{\rm bin}(\delta),
\]
and the totalized theorem covers every $\delta\le\tfrac12$.
\end{paracol}
\endgroup
\AlignmentBlockEnd

\AlignmentParagraph{Entropy bound.}

\AlignmentContinuationLabel{tab:alignment-bb84-middle}
\AlignmentBlockHead{Finite-key sampling aggregate consumed by BB84.}
\begingroup\AlignmentBlockSetup
\begin{paracol}{2}
\code{src/qkd/bb84/model.v}: \coqi{bb84_total_rounds}, \coqi{bb84_full_string_count}, \coqi{bb84_zsource}, \coqi{bb84_bbkey};
\code{src/qkd/bb84/prop.v}: \coqi{stat}
\begin{minted}{coq}
Definition bb84_total_rounds (P : bb84_protocol) : nat :=
  bb84_k P + bb84_n P.
Definition bb84_full_string_count (P : bb84_protocol) : nat :=
  2 ^ bb84_total_rounds P.
Record bb84_zsource (P : bb84_protocol) := BB84ZSource {
  bb84_zsrc_nb : nat;
  bb84_zsrc_rho : 'M[C]_(bb84_full_string_count P *
    bb84_zsrc_nb * bb84_full_string_count P);
  bb84_zsrc_den1 : bb84_zsrc_rho \is den1mx
}.
Lemma bb84_bbkey (n k : nat) (z z' : bitstr (k + n))
    (Qtol eps : R) :
  (0 < n)%N -> (0 < k)%N -> 0 < eps ->
  let mu := bb84_mu n k eps in
  let pass_event : {set 'S_(k + n)} :=
    [set pi | perm_pair_sample_err z z' pi <= Qtol] in
  let p_pass := Pr_event
    (fun pi : 'S_(k + n) => uniform_perm_mu pi) pass_event in
  0 < p_pass ->
  Pr_cond (fun pi : 'S_(k + n) => uniform_perm_mu pi)
    (perm_pair_deviation_event z z' mu) pass_event <=
      eps ^+ 2 / p_pass.
Lemma stat (Zkey B : {set L})
    (disZB : [disjoint Zkey & B])
    (P : @bb84_protocol R) (G : bb84_zsource P) (eps : R)
    (dim_Zk : (2 ^ bb84_n P)%N = dim ('H[HL]_Zkey))
    (dim_B : (2 ^ bb84_n P)%N = dim ('H[HL]_B)) :
  (0 < bb84_n P)%N -> (0 < bb84_k P)%N -> 0 < eps ->
  0 < bb84_p_pass_tot G ->
  let eps' := eps / Num.sqrt (bb84_p_pass_tot G) in
  eps' < 1 ->
  let h_trunc :=
    bb84_h_trunc (bb84_Qtol P +
      bb84_mu (bb84_n P) (bb84_k P) eps) in
  smooth_hmax B eps'
    (bb84_rhoZZ_pass_mixH disZB G dim_Zk dim_B) <=
    (bb84_n P)%:R * h_trunc.
\end{minted}
\coqi{bb84_bbkey} states the displayed conditional bound, and \coqi{stat}
provides the aggregate smooth max-entropy endpoint.
The full-string aggregate route uses the two bad-pass bounds in
\Cref{tab:alignment-bb84-bad-pass} inside \coqi{stat} to bound
\coqi{bb84_rhoZZ_pass_mixH}; its probability and state definitions are expanded
in \Cref{tab:alignment-bb84-pass-probability,tab:alignment-bb84-pass-state}.
The \coqi{bb84_secret} proof invokes \coqi{stat} after transporting the delayed
measurement through the EUR and derives the smooth-min-entropy bound by
subtracting the resulting smooth-max bound from the EUR.
\switchcolumn
\textbf{Sampling to smooth max-entropy (steps~1 and~4 of the secrecy chain in \Cref{app:bb84-sig}).}
Conditioning the Serfling deviation event on passing gives probability at most $\eps^2/p_{\rm pass}$ and hence
\[
 H_{\max}^{\eps/\sqrt{p_{\rm pass}}}(Z_{\rm key}|Z'_{\rm key})
 \le n\,h_{\rm trunc}(\Qtol+\mu).
\]
\end{paracol}
\endgroup
\AlignmentBlockEnd

\subsubsection{Entropy Transfer}

\AlignmentParagraph{Measurement and leakage.}

\AlignmentContinuationLabel{tab:alignment-bb84-bridges}
\AlignmentBlockHead{Delayed measurement and leakage bridges.}
\begingroup\AlignmentBlockSetup
\begin{paracol}{2}
\code{src/qkd/bb84/def.v}: \coqi{bb84_delayed_measurement_condition};
\code{src/qkd/bb84/bb84.v}: \coqi{bb84_delayed_measurement_hmax};
\code{src/qkd/bb84/prop.v}: \coqi{chain_rule_class}, \coqi{chain_rule_class_leak_bound}
\begin{minted}{coq}
Record bb84_delayed_measurement_condition
    (BobZ : 'QC[HL]_(Bq, Zp))
    (rhoZB : 'F[HL]_(Zkey :|: Bq))
    (rhoZZp : 'F[HL]_(Zkey :|: Zp)) := {
  bb84_dmc_full : 'F[HL]_(Zkey :|: Bq :|: finset.set0);
  bb84_dmc_full_den1 : bb84_dmc_full \is den1lf;
  bb84_dmc_source :
    ptrace_cast3C (disjointX0 Zkey) (disjointX0 Bq)
      bb84_dmc_full = rhoZB;
  bb84_dmc_target :
    ptrace_cast3C (disjointX0 Zkey) (disjointX0 Zp)
      ((((\:1 : 'SO[HL]_Zkey) :⊗ (BobZ : 'SO[HL]_(Bq, Zp))) :⊗
        (\:1 : 'SO[HL]_finset.set0)) bb84_dmc_full) = rhoZZp
}.
Lemma bb84_delayed_measurement_hmax (eps0 : R)
    (BobZ : 'QC[HL]_(Bq, Zp))
    (rhoZB : 'F[HL]_(Zkey :|: Bq))
    (rhoZZp : 'F[HL]_(Zkey :|: Zp)) :
  0 <= eps0 -> eps0 < Num.sqrt (@complex.Re R (\Tr rhoZB)) ->
  bb84_delayed_measurement_condition BobZ rhoZB rhoZZp ->
  smooth_hmax Bq eps0 rhoZB <= smooth_hmax Zp eps0 rhoZZp.
Lemma chain_rule_class (C : {set L})
    (disSC : [disjoint S & C])
    (disEC : [disjoint E & C])
    (disC_SE : [disjoint C & (S :|: E)])
    (eps' : R) (rho : 'F[HL]_(C :|: (S :|: E))) :
  rho \is psdlf -> @complex.Re R (\Tr rho) <= 1 ->
  (forall x x' : 'Idx[HL]_C, x' != x ->
    (([> deltav x ; deltav x <] \O (\1 : 'F_(S :|: E))) \o rho \o
     ([> deltav x' ; deltav x' <] \O (\1 : 'F_(S :|: E)))) = 0) ->
  0 <= eps' < Num.sqrt (@complex.Re R (\Tr rho)) ->
  (smooth_hminE E eps' (ptrace_castl disC_SE rho) -
    (log2 (dim 'H[HL]_C)%:R)%:E <=
   smooth_hminE (E :|: C) eps' rho)%E.
Lemma chain_rule_class_leak_bound (C : {set L})
    (disSC : [disjoint S & C])
    (disEC : [disjoint E & C])
    (disC_SE : [disjoint C & (S :|: E)])
    (eps' leak : R) (rho : 'F[HL]_(C :|: (S :|: E))) :
  rho \is psdlf -> @complex.Re R (\Tr rho) <= 1 ->
  (forall x x' : 'Idx[HL]_C, x' != x ->
    (([> deltav x ; deltav x <] \O (\1 : 'F_(S :|: E))) \o rho \o
     ([> deltav x' ; deltav x' <] \O (\1 : 'F_(S :|: E)))) = 0) ->
  0 <= eps' < Num.sqrt (@complex.Re R (\Tr rho)) ->
  log2 (dim 'H[HL]_C)%:R <= leak ->
  (smooth_hminE E eps' (ptrace_castl disC_SE rho) - leak%:E <=
   smooth_hminE (E :|: C) eps' rho)%E.
\end{minted}
The chain-rule declarations specialize the reusable \coqi{class_bounds_2_min_lower} row in \Cref{tab:alignment-classical-side-information}.
\switchcolumn
\textbf{Delayed measurement and public classical leakage (steps~3 and~5).}
A CPTP delayed measurement gives
$H_{\max}^{\eps}(Z|B_q)\le H_{\max}^{\eps}(Z|Z')$.
For a transcript $C$ classical in the key basis,
$H_{\min}^{\eps}(X|EC)\ge H_{\min}^{\eps}(X|E)-\log d_C$, or the same bound with any certified leakage upper bound.
\end{paracol}
\endgroup
\AlignmentBlockEnd

\AlignmentParagraph{Projectors.}

\AlignmentContinuationLabel{tab:alignment-bb84-projector}
\AlignmentBlockHead{Reusable good-projector lemmas.}
\begingroup\AlignmentBlockSetup
\begin{paracol}{2}
\code{src/entropy/pur/def.v}: \coqi{mmix_mx}, \coqi{hmax_mx}, \coqi{smooth_hmax_mx};
\code{src/utils/bitstr.v}: \coqi{relative_hamming_ball};
\code{src/qkd/bb84/model.v}: \coqi{bb84_pair_projector}, \coqi{bb84_good_projector}, \coqi{bb84_good_projector_pass_mass};
\code{src/qkd/bb84/prop.v}: \coqi{bb84_hmax_good_projector_bound},
\coqi{bb84_smooth_hmax_good_projector_mass}
\begin{minted}{coq}
Definition mmix_mx (n : nat) : 'M[C]_n :=
  (n%:R^-1)%:C *: 1%:M.
Definition hmax_mx (rho : 'M[C]_(na * nb)) : R :=
  sup [set v : R | exists sigma : 'M[C]_nb,
    sigma \is den1mx /\
    elog2 (na%:R *
      (gfidelity rho (mmix_mx na *t sigma)) ^+ 2) = v%:E].
Definition smooth_hmax_mx (eps : R) (rho : 'M[C]_(na * nb)) : R :=
  inf [set v : R | exists tau : 'M[C]_(na * nb),
    ((tau \is denmx) && (0 < @complex.Re R (\tr tau))) &&
    (pur_dist_mx tau rho <= eps) /\ v = hmax_mx tau].
Definition relative_hamming_ball {n : nat} (delta : R)
    (x x' : 'I_(2 ^ n)) : bool :=
  (((d_ham (ord_to_bitstring x) (ord_to_bitstring x'))%:R /
    n%:R) < delta)%R.
Definition bb84_good_projector (n : nat) (delta : R) :=
  \sum_(x : 'I_(2 ^ n))
    \sum_(x' : 'I_(2 ^ n) | relative_hamming_ball delta x x')
      bb84_pair_projector x x'.
Lemma bb84_hmax_good_projector_bound
    (n_key : nat) (delta : R)
    (rho : 'M[C]_((2 ^ n_key) * (2 ^ n_key))) :
  (0 < n_key)%N -> 0 < delta ->
  rho \is psdmx -> @complex.Re R (\tr rho) = 1 ->
  rho *m bb84_good_projector n_key delta = rho ->
  hmax_mx rho <= n_key%:R * bb84_h_trunc delta.
Lemma bb84_good_projector_pass_mass
    (P : bb84_protocol) (G : bb84_zsource P) (eps : R) :
  (0 < bb84_n P)%N -> 0 < bb84_p_pass_tot G ->
  let delta := bb84_Qtol P +
    bb84_mu (bb84_n P) (bb84_k P) eps in
  1 - bb84_pass_bad_weight G eps / bb84_p_pass_tot G <=
    @complex.Re R
      (\tr (bb84_good_projector (bb84_n P) delta *m
        bb84_rhoZZ_pass_mix G *m
        bb84_good_projector (bb84_n P) delta)).
Lemma bb84_smooth_hmax_good_projector_mass
    (n_key : nat) (delta eps0 : R)
    (rho : 'M[C]_((2 ^ n_key) * (2 ^ n_key))) :
  (0 < n_key)%N ->
  rho \is psdmx ->
  @complex.Re R (\tr rho) = 1 ->
  0 <= eps0 -> eps0 < 1 ->
  1 - eps0 ^+ 2 <=
    @complex.Re R
      (\tr (@bb84_good_projector R n_key delta *m rho *m
        @bb84_good_projector R n_key delta)) ->
  smooth_hmax_mx eps0 rho <=
    n_key%:R * bb84_h_trunc delta.
\end{minted}
\switchcolumn
\textbf{Good-subspace mass to smooth max-entropy.}
Let $P_{\mathrm{good}}$ project onto key pairs with relative Hamming distance strictly below $\delta$.
For a normalized state supported on that subspace,
\[
\rho P_{\mathrm{good}}=\rho
\Longrightarrow
H_{\max}(Z|Z')_\rho\le n\,h_{\mathrm{trunc}}(\delta).
\]
For the pass mixture with $p_{\rm pass}^{\rm tot}>0$, parameter estimation
supplies
\[
1-\frac{B_{\rm pass}(\eps)}{p_{\rm pass}^{\rm tot}}
\le
\Tr(P_{\mathrm{good}}\rho_{ZZ'|\rm pass}P_{\mathrm{good}}).
\]
For any normalized state with good-subspace mass at least
$1-\eps_0^2$, $0\le\eps_0<1$, the reusable smoothing step yields
\[
H_{\max}^{\eps_0}(Z|Z')_\rho
\le n\,h_{\mathrm{trunc}}(\delta).
\]
The basis-pair projector used in \coqi{bb84_good_projector} is expanded in
\Cref{tab:alignment-bb84-pass-state}.
The Hamming-distance and ordinal encodings in \coqi{relative_hamming_ball} are
expanded in \Cref{tab:alignment-sampling-objects}; \coqi{gfidelity} and
\coqi{pur_dist_mx}, used by \coqi{hmax_mx} and \coqi{smooth_hmax_mx}, are
expanded in \Cref{tab:alignment-entropy}.
These lemmas isolate the coding-theoretic projector bound, its smoothing
step, and the pass-event probability accounting used by the labelled endpoint.
\end{paracol}
\endgroup
\AlignmentBlockEnd

\AlignmentParagraph{Readout.}

\AlignmentContinuationLabel{tab:alignment-bb84-readout}
\AlignmentBlockHead{Computational-basis readout.}
\begingroup\AlignmentBlockSetup
\begin{paracol}{2}
\code{src/qkd/bb84/model.v}: \coqi{bb84_pair_basis_readout};
\code{src/qkd/bb84/bb84.v}: \coqi{bb84_bob_z_basis_readout_hmax}
\begin{minted}{coq}
Definition bb84_pair_basis_readout (n : nat)
    (rho : 'M[C]_((2 ^ n) * (2 ^ n))) :
    'M[C]_((2 ^ n) * (2 ^ n)) :=
  \sum_(z : 'I_(2 ^ n)) \sum_(b : 'I_(2 ^ n))
    rho (mxtens_index (z, b)) (mxtens_index (z, b)) *:
      @bb84_pair_projector n z b.
Lemma bb84_bob_z_basis_readout_hmax
    (P : bb84_protocol) (ne : nat) (eps0 : R)
    (dim_Zk_key : (2 ^ bb84_n P)%N = dim ('H[HL]_Zkey))
    (dim_Bq_key : (2 ^ bb84_n P)%N = dim ('H[HL]_Bq))
    (dim_Zp_key : (2 ^ bb84_n P)%N = dim ('H[HL]_Zp))
    (rho : 'M[C]_((2 ^ bb84_n P) * ne * (2 ^ bb84_n P))) :
  rho \is den1mx ->
  0 <= eps0 -> eps0 < 1 ->
  let dimZBq :=
    bb84_register_pair_dim disZBq dim_Zk_key dim_Bq_key in
  let rhoZB := eur_rho_ZC disZBq rho
    (@std_projective_povm R (2 ^ bb84_n P)) dimZBq in
  let rhoZZp := biprod_of_mx disZZp
    (castmx (f_equal2 muln dim_Zk_key dim_Zp_key,
             f_equal2 muln dim_Zk_key dim_Zp_key)
      (@bb84_pair_basis_readout R (bb84_n P)
        (rho_ZC (@std_projective_povm R (2 ^ bb84_n P)) rho))) in
  smooth_hmax Bq eps0 rhoZB <= smooth_hmax Zp eps0 rhoZZp.
\end{minted}
The state constructors \coqi{rho_ZC} and \coqi{eur_rho_ZC} are expanded in
\Cref{tab:alignment-main-statement-bridges}; \coqi{bb84_pair_projector} is
expanded in \Cref{tab:alignment-bb84-pass-state}.
This is the dephasing map used by the readout data-processing theorem.
\switchcolumn
\textbf{Dephasing readout and max-entropy data processing.}
The basis-pair readout is
\[
\mathcal R(\rho)=
\sum_{z,b}\rho_{(z,b),(z,b)}
 |z,b\rangle\!\langle z,b|.
\]
For Bob's computational-basis readout,
\[
H_{\max}^{\eps_0}(Z|B_q)\le H_{\max}^{\eps_0}(Z|Z').
\]
\end{paracol}
\endgroup
\AlignmentBlockEnd

\subsubsection{Privacy Amplification}

\AlignmentParagraph{Leakage arithmetic.}

\AlignmentContinuationLabel{tab:alignment-bb84-middle-cont}
\AlignmentBlockHead{Privacy amplification and the optimized witness.}
\begingroup\AlignmentBlockSetup
\begin{paracol}{2}
\code{src/qkd/bb84/model.v}: \coqi{privacy_amplification};
\code{src/qkd/bb84/bb84.v}: \coqi{bb84_public_leakage_pa_hmin_diff}, \coqi{bb84_delta_from_hmin_diff}, \coqi{bb84_delta_from_public_leakage}
\begin{minted}{coq}
Corollary privacy_amplification
    (Fh : 'Idx[HL]_Fl -> 'Idx[HL]_S -> 'Idx[HL]_Zl)
    (eps' : R) (rhoXE : 'F[HL]_(S :|: E)) :
  rhoXE \is psdlf -> @complex.Re R (\Tr rhoXE) <= 1 ->
  is_cq_state rhoXE -> @two_univ R _ _ _ Fh -> 0 <= eps' ->
  dist_from_uniform (post_hash_state disSE Fh rhoXE) <=
    eps' + (2%:R)^-1 * Num.sqrt
      (expR ((log2 (dim 'H[HL]_Zl)%:R -
        smooth_hmin E eps' rhoXE) * ln 2%:R)).
Lemma bb84_public_leakage_pa_hmin_diff
    (Ppub : {set L}) (rhoXE : 'F[HL]_(S :|: E))
    (leak eps' base tail logZ : R)
    (pub : bb84_public_leakage Ppub rhoXE leak) :
  0 <= eps' < Num.sqrt
    (@complex.Re R (\Tr (bb84_public_state pub))) ->
  base <= smooth_hmin E eps' rhoXE ->
  logZ <= base + tail - leak ->
  let eq_pub_assoc := bb84_public_assoc S E Ppub in
  logZ - smooth_hmin (E :|: Ppub) eps'
    (castlf (eq_pub_assoc, eq_pub_assoc)
      (bb84_public_state pub)) <= tail.
Lemma bb84_delta_from_hmin_diff
    (Epub : {set L}) (Fh : 'Idx[HL]_Fl -> 'Idx[HL]_S -> 'Idx[HL]_Zl)
    (eps' ebar : R) (rhoXEpub : 'F[HL]_(S :|: Epub)) :
  rhoXEpub \is psdlf -> @complex.Re R (\Tr rhoXEpub) <= 1 ->
  is_cq_state rhoXEpub -> @two_univ R _ _ _ Fh ->
  0 <= eps' -> 0 < ebar ->
  log2 (dim 'H[HL]_Zl)%:R - smooth_hmin Epub eps' rhoXEpub
    <= 2%:R * log2 (2%:R * ebar) ->
  delta_secret (post_hash_state disSEpub Fh rhoXEpub) (eps' + ebar).
Lemma bb84_delta_from_public_leakage
    (Ppub : {set L}) (rhoXE : 'F[HL]_(S :|: E))
    (leak eps' ebar base : R)
    (pub : bb84_public_leakage Ppub rhoXE leak)
    (Fh : 'Idx[HL]_Fl -> 'Idx[HL]_S -> 'Idx[HL]_Zl) :
  0 <= eps' < Num.sqrt
    (@complex.Re R (\Tr (bb84_public_state pub))) ->
  base <= smooth_hmin E eps' rhoXE ->
  log2 (dim 'H[HL]_Zl)%:R <=
    base + 2%:R * log2 (2%:R * ebar) - leak ->
  @two_univ R _ _ _ Fh -> 0 < ebar ->
  let eq_pub_assoc := bb84_public_assoc S E Ppub in
  delta_secret
    (post_hash_state disSEpub Fh
      (castlf (eq_pub_assoc, eq_pub_assoc)
        (bb84_public_state pub)))
    (eps' + ebar).
\end{minted}
The last declaration packages public leakage and the raw privacy-amplification distance bound.
The EUR-plus-statistics subtraction is derived inside \coqi{bb84_secret}; the same theorem carries that bound through public leakage, accepted-state postselection, and privacy amplification.
\switchcolumn
\textbf{Privacy amplification and leakage arithmetic (steps~6--8).}
The LHL permits taking
$\Delta=\eps'+\frac12\sqrt{2^{\ell-H_{\min}^{\eps'}(X|E')}}$ as an upper bound;
the key-length premise yields $\Delta\le\eps'+\bar\eps$, and the non-abort factor yields secrecy.
The reusable leakage rearrangement is
\[
\begin{gathered}
b\le H_{\min}^{\eps'}(X|E),\qquad
\log |Z|\le b+t-\mathsf{leak}\\
\Longrightarrow\quad
\log |Z|-H_{\min}^{\eps'}(X|EP_{\rm pub})\le t.
\end{gathered}
\]
\end{paracol}
\endgroup
\AlignmentBlockEnd

\subsubsection{Key Length}

\AlignmentParagraph{Optimized witness.}

\AlignmentContinuationLabel{tab:alignment-key-witness}
\AlignmentBlockHead{Optimized key-length witnesses.}
\begingroup\AlignmentBlockSetup
\begin{paracol}{2}
\code{src/qkd/bb84/prop.v}: \coqi{bb84_key_length_has_sup}, \coqi{bb84_key_length_sup_witness}
\begin{minted}{coq}
Lemma bb84_key_length_has_sup
    (P : @bb84_protocol R) (q eps_s : R) :
  0 < eps_s -> eps_s <= 1 ->
  has_sup (fun M : R =>
    exists eps ebar : R,
      0 < eps /\ eps < 1 /\ 0 < ebar /\ eps + ebar <= eps_s /\
      M = ((floor (bb84_key_length_ub
        (bb84_n P) (bb84_k P) q (bb84_Qtol P)
        eps ebar (bb84_leak_EC P) (bb84_eps_c P)))%:~R : R)).
Lemma bb84_key_length_sup_witness
    (P : @bb84_protocol R) (q eps_s target : R) :
  0 < eps_s -> eps_s <= 1 -> target \is a int_num ->
  target <= sup (fun M : R =>
    exists eps ebar : R,
      0 < eps /\ eps < 1 /\ 0 < ebar /\ eps + ebar <= eps_s /\
      M = ((floor (bb84_key_length_ub
        (bb84_n P) (bb84_k P) q (bb84_Qtol P)
        eps ebar (bb84_leak_EC P) (bb84_eps_c P)))%:~R : R)) ->
  exists eps ebar,
    0 < eps /\ eps < 1 /\ 0 < ebar /\ eps + ebar <= eps_s /\
    target <=
      ((floor (bb84_key_length_ub
        (bb84_n P) (bb84_k P) q (bb84_Qtol P)
        eps ebar (bb84_leak_EC P) (bb84_eps_c P)))%:~R : R).
\end{minted}
These declarations provide the real-supremum and integer-rounding witnesses
for the optimized key-length premise.
\switchcolumn
\textbf{Optimized parameter witness.}
For $0<\eps_s\le1$, the feasible floored key-length set is nonempty and
bounded above.
Consequently, an integer key target below its supremum admits concrete
$\eps,\bar\eps$ satisfying positivity, secrecy budget, and the floored
key-length inequality.
\end{paracol}
\endgroup
\AlignmentBlockEnd

\AlignmentParagraph{Pass split.}

\AlignmentContinuationLabel{tab:alignment-pass-split}
\AlignmentBlockHead{The smoothing-domain and low-pass branches.}
\begingroup\AlignmentBlockSetup
\begin{paracol}{2}
\code{src/qkd/bb84/model.v}: \coqi{bb84_entropy_domain}, \coqi{bb84_eps_lt_sqrt_pass};
\code{src/qkd/bb84/bb84.v}: \coqi{bb84_secret}
\begin{minted}{coq}
Definition bb84_entropy_domain (eps p_pass : R) : Prop :=
  0 < eps /\ eps < Num.sqrt p_pass.
Lemma bb84_eps_lt_sqrt_pass eps p_pass :
  bb84_entropy_domain eps p_pass -> eps < Num.sqrt p_pass.
\end{minted}
\coqi{bb84_secret} contains the case split, low-pass arithmetic, and distance
bound; \coqi{bb84_eps_lt_sqrt_pass} provides the smoothing-domain conversion.
\switchcolumn
\textbf{EUR-valid versus high-abort split (\Cref{sec:findings:eur}).}
If $\eps^2<p_{\rm pass}$, then $0<\eps<\sqrt{p_{\rm pass}}$, so the rescaled smoothing radius lies in the EUR domain.
In the complementary branch $p_{\rm pass}\le\eps^2$, the abort/pass relation and $\eps+\bar\eps\le\eps_s$ give $1-p_{\rm abort}\le\eps_s$; the universal distance bound closes secrecy without invoking the EUR.
\end{paracol}
\endgroup
\AlignmentBlockEnd

\subsection{Semantics}
\label{app:semantics-alignment-protocol-semantics}

\subsubsection{Well-Formedness}

\AlignmentParagraph{Command algebra.}

\AlignmentContinuationLabel{tab:alignment-syntax-wf}
\AlignmentBlockHead{Sequencing and command well-formedness.}
\begingroup\AlignmentBlockSetup
\begin{paracol}{2}
\code{src/qkd/protocol/syntax.v}: \coqi{cmd_wf}, \coqi{seq_cmd}, \coqi{for_cmd}, \coqi{cmd_seqA}, \coqi{cmd_seq_skip_l}, \coqi{cmd_seq_skip_r}, \coqi{cmd_wf_seq_cmd}, \coqi{cmd_wf_for_cmd}
\begin{minted}{coq}
Fixpoint cmd_wf (c : qkd_cmd) : Prop := match c with
  | Cskip => True
  | Cassign _ _ => True
  | Csample _ mu =>
      forall s, (forall v, 0 <= mu s v) /\ \sum_v mu s v = 1
  | Cinit Q rho => [disjoint Q & qA] /\
      forall s, rho s \is den1lf
  | Cmeasure _ Q M => [disjoint Q & qA] /\
      forall s, \sum_v ((M s v)^A \o (M s v)) = \1
  | Cattack Q => [disjoint Q & qA]
  | Cif _ c1 c2 => cmd_wf c1 /\ cmd_wf c2
  | Cseq c1 c2 => cmd_wf c1 /\ cmd_wf c2
  end.
Lemma cmd_seqA : associative seq_cmd.
Lemma cmd_seq_skip_l : left_id Cskip seq_cmd.
Lemma cmd_seq_skip_r : right_id Cskip seq_cmd.
Lemma cmd_wf_seq_cmd (c1 c2 : qkd_cmd) :
  cmd_wf c1 -> cmd_wf c2 -> cmd_wf (seq_cmd c1 c2).
Lemma cmd_wf_for_cmd {limit} (body : 'I_limit -> qkd_cmd) :
  (forall i, cmd_wf (body i)) -> cmd_wf (for_cmd body).
\end{minted}
The definitions \coqi{seq_cmd} and \coqi{for_cmd} appear in
\Cref{tab:alignment-semantics-core}; this row gives \coqi{cmd_wf} and the
associated laws.
\switchcolumn
\textbf{Sequencing and well-formedness.}
Sequencing forms a monoid,
$(c_1;c_2);c_3=c_1;(c_2;c_3)$ and
$\mathsf{skip};c=c=c;\mathsf{skip}$.
A sampling distribution is nonnegative and sums to one; initialized states have unit trace; measurements are complete; initialized, measured, and attacked registers are disjoint from the adversary workspace.
Well-formedness is closed under sequencing and finite loops.
\end{paracol}
\endgroup
\AlignmentBlockEnd

\subsubsection{Denotation}

\AlignmentParagraph{Linearity and healthiness.}

\AlignmentContinuationLabel{tab:alignment-denotation-laws}
\AlignmentBlockHead{Reusable denotational metatheory.}
\begingroup\AlignmentBlockSetup
\begin{paracol}{2}
\code{src/qkd/protocol/semantics.v}: \coqi{cq_single},
\coqi{denote_singleton}, \coqi{cq_state_decomp}, \coqi{denote_seq_cmd},
\coqi{denote_scale}, \coqi{denote_sum}, \coqi{denote_singleton_decomp},
\coqi{denote_psd}, \coqi{denote_mass}, \coqi{denote_cq_wf}
\begin{minted}{coq}
Definition cq_single (s : store) (rho : 'F[HL]_setT) : cq_state :=
  [ffun t => if t == s then rho else 0].
Definition denote_singleton (c : qkd_cmd) (s : store)
    (rho : 'F[HL]_setT) : cq_state :=
  denote c (cq_single s rho).
Lemma cq_state_decomp (xi : cq_state) :
  xi = \sum_s cq_single s (xi s).
Lemma denote_seq_cmd (c1 c2 : qkd_cmd) (xi : cq_state) :
  denote (seq_cmd c1 c2) xi = denote c2 (denote c1 xi).
Lemma denote_scale (c : qkd_cmd) (a : C) (xi : cq_state) :
  denote c (a *: xi) = a *: denote c xi.
Lemma denote_sum (I : finType) (c : qkd_cmd) (F : I -> cq_state) :
  denote c (\sum_i F i) = \sum_i denote c (F i).
Lemma denote_singleton_decomp (c : qkd_cmd) (xi : cq_state) :
  denote c xi = \sum_s denote_singleton c s (xi s).
Lemma denote_psd (c : qkd_cmd) (xi : cq_state) :
  cmd_wf c -> (forall s, xi s \is psdlf) ->
  forall s, denote c xi s \is psdlf.
Lemma denote_mass (c : qkd_cmd) (xi : cq_state) :
  cmd_wf c -> (forall s, xi s \is psdlf) ->
  cq_mass (denote c xi) = cq_mass xi.
Lemma denote_cq_wf (c : qkd_cmd) (xi : cq_state) :
  cmd_wf c -> cq_wf xi -> cq_wf (denote c xi).
\end{minted}
All project-defined objects in these laws are expanded in
\Cref{tab:alignment-semantics-core,tab:alignment-syntax-wf,tab:alignment-denotation}.
\switchcolumn
\textbf{Decomposition, linearity, and healthiness.}
Every cq-state decomposes into singleton states, and
\[
\sem c(a\xi)=a\sem c(\xi),\qquad
\sem c\!\left(\sum_i\xi_i\right)=\sum_i\sem c(\xi_i).
\]
Moreover,
\[
\sem{c_1;c_2}=\sem{c_2}\circ\sem{c_1},
\]
and, when $\mathsf{wf}(c)$ and every block $\xi(s)$ is positive
semidefinite,
\[
\operatorname{mass}(\sem c(\xi))=\operatorname{mass}(\xi).
\]
Finally, $\mathsf{wf}(c)\land\mathsf{wf}(\xi)$ implies
$\mathsf{wf}(\sem c(\xi))$.
\end{paracol}
\endgroup
\AlignmentBlockEnd

\subsubsection{Classical Transport}

\AlignmentParagraph{Marginal preservation.}

\AlignmentContinuationLabel{tab:alignment-classical-transport}
\AlignmentBlockHead{Classical-tail marginal transport.}
\begingroup\AlignmentBlockSetup
\begin{paracol}{2}
\code{src/qkd/example/run.v}: \coqi{classical_cmd}, \coqi{denote_classical_ptraceso_eq}
\begin{minted}{coq}
Fixpoint classical_cmd (c : qkd_cmd) : Prop := match c with
  | Cskip => True
  | Cassign _ _ => True
  | Csample _ _ => True
  | Cinit _ _ => False
  | Cmeasure _ _ _ => False
  | Cattack _ => False
  | Cif _ c1 c2 => classical_cmd c1 /\ classical_cmd c2
  | Cseq c1 c2 => classical_cmd c1 /\ classical_cmd c2
  end.
Lemma denote_classical_ptraceso_eq
    (qA : {set L}) (EA : @adversary R L HL V qA)
    (Qsa : {set L}) (c : @qkd_cmd R L HL Cv V)
    (xi1 xi2 : @cq_state R L HL Cv V) :
  classical_cmd c ->
  (forall u, ptraceso Qsa (xi1 u) = ptraceso Qsa (xi2 u)) ->
  forall u,
  ptraceso Qsa (@denote R L HL Cv V qA Pvar EA c xi1 u) =
  ptraceso Qsa (@denote R L HL Cv V qA Pvar EA c xi2 u).
\end{minted}
This generic noninterference-style lemma transports the per-round PM/EB equality through announcement, parameter estimation, error correction, abort branches, and privacy amplification.
\switchcolumn
\textbf{Classical programs preserve quantum-marginal equality.}
For a command with no initialization, measurement, or adversarial quantum action,
\[
\begin{aligned}
&\mathsf{classical}(c)\land
  \bigl(\forall s,\Tr_Q\xi_1(s)=\Tr_Q\xi_2(s)\bigr)\\
&\qquad\Longrightarrow
  \forall s,\Tr_Q\sem c(\xi_1)(s)
  =\Tr_Q\sem c(\xi_2)(s).
\end{aligned}
\]
\end{paracol}
\endgroup
\AlignmentBlockEnd

\subsection{Instance}
\label{app:semantics-alignment-instance}

\subsubsection{Source and Device}

\AlignmentParagraph{Ideal source.}

\AlignmentContinuationLabel{tab:alignment-instance-source-state}
\AlignmentBlockHead{Concrete full-string source state and measurements.}
\begingroup\AlignmentBlockSetup
\begin{paracol}{2}
\code{src/qkd/example/instantiation.v}: \coqi{ex_source_dim}, \coqi{ex_source_zero}, \coqi{ex_source_rho}, \coqi{ex_source_M}, \coqi{ex_source_N}, \coqi{ex_source_quality}
\begin{minted}{coq}
Definition ex_source_dim : nat := ((2 ^ nn) * 1 * (2 ^ nn))%N.
Definition ex_source_zero : 'I_ex_source_dim.
Definition ex_source_rho : 'M[C]_ex_source_dim :=
  delta_mx ex_source_zero ex_source_zero.
Definition ex_source_M : 'I_(2 ^ nn) -> 'M[C]_(2 ^ nn) :=
  @x_projective_povm R nn.
Definition ex_source_N : 'I_(2 ^ nn) -> 'M[C]_(2 ^ nn) :=
  @std_projective_povm R (2 ^ nn).
Definition ex_source_quality :
  (bb84_n Pex)%:R * (1 : R) <= q_of_mx ex_source_M ex_source_N.
\end{minted}
The projective POVMs and overlap functional are expanded in
\Cref{tab:alignment-device-primitives,tab:alignment-eur-objects}.
The shared operands \coqi{nn} and \coqi{Pex} are expanded in
\Cref{tab:alignment-instance}.
\switchcolumn
\textbf{Concrete ideal full-string source state and POVMs.}
The source state and its two POVMs are
\[
\begin{gathered}
\rho_{\rm src}=|0\rangle\!\langle0|,\\
M_x=(H^{\otimes n})^\dagger
  |x\rangle\!\langle x|H^{\otimes n},\\
N_z=|z\rangle\!\langle z|,
\end{gathered}
\]
with the overlap certificate $nq\le q(M,N)$ at $q=1$.
\end{paracol}
\endgroup
\AlignmentBlockEnd

\AlignmentParagraph{Source certificate.}

\AlignmentContinuationLabel{tab:alignment-instance-quantum-components}
\AlignmentBlockHead{Concrete full-string source certificate.}
\begingroup\AlignmentBlockSetup
\begin{paracol}{2}
\code{src/qkd/example/instantiation.v}: \coqi{ex_source_model}
\begin{minted}{coq}
Definition ex_source_model : bb84_source_model Pex (1 : R).
Proof.
  refine (@BB84SourceModel R Pex 1
    (2 ^ nn)%N 1%N (2 ^ nn)%N (2 ^ nn)%N (2 ^ nn)%N
    ex_source_rho ex_source_M ex_source_N _ _ _ _ _ _).
  - apply/den1mxP; split.
    + exact: delta_mx_ii_psdmx.
    + by rewrite tr_delta_mx_ii.
  - exact: x_projective_povm_is_povm.
  - exact: std_projective_povm_is_povm.
  - exact: ex_source_quality.
  - by [].
  - by [].
Defined.
\end{minted}
The target record is expanded in \Cref{tab:alignment-source-certificate};
its state, POVMs, and quality witness are expanded in
\Cref{tab:alignment-instance-source-state}.
The constructor body displays every semantic field and certificate.
\switchcolumn
\textbf{Certified source record.}
The record uses dimensions $2^n\times1\times2^n$, carries the quality
$nq\le q(M,N)$ at $q=1$, and identifies its $X$-outcome and $Z$/Bob
dimensions with $2^n$.
\end{paracol}
\endgroup
\AlignmentBlockEnd

\AlignmentParagraph{Round interface.}

\AlignmentContinuationLabel{tab:alignment-instance-round-components}
\AlignmentBlockHead{Concrete preparation and measurement operands.}
\begingroup\AlignmentBlockSetup
\begin{paracol}{2}
\code{src/qkd/example/instantiation.v}: \coqi{ex_bit_ord}, \coqi{ex_prepare_state}, \coqi{ex_prepare_state_den1}, \coqi{ex_measurement_operator}, \coqi{ex_measurement_complete}, \coqi{ex_detect_basis_indep}, \coqi{ex_source_dim_A}, \coqi{ex_source_dim_X}, \coqi{ex_source_nb_pos}, \coqi{ex_source_overlap_standard_z}
\begin{minted}{coq}
Definition ex_bit_ord (bit : bool) : 'I_2 :=
  if bit then Ordinal (isT : (1 < 2)%N) else ord0.
Definition ex_prepare_state (basis : bb84_basis) (bit : bool) : 'M[C]_2 :=
  match basis with
  | BasisX => @x_projective_povm2 R (ex_bit_ord bit)
  | BasisZ => @std_projective_povm R 2 (ex_bit_ord bit)
  end.
Definition ex_prepare_state_den1 basis bit :
  ex_prepare_state basis bit \is den1mx.
Definition ex_measurement_operator
    (basis : bb84_basis) (outcome : option bool) : 'M[C]_2 :=
  match outcome with
  | None => 0
  | Some bit => ex_prepare_state basis bit
  end.
Definition ex_measurement_complete basis :
  (ex_measurement_operator basis None)^*t *m
      ex_measurement_operator basis None +
  (ex_measurement_operator basis (Some false))^*t *m
      ex_measurement_operator basis (Some false) +
  (ex_measurement_operator basis (Some true))^*t *m
      ex_measurement_operator basis (Some true) = 1%:M.
Definition ex_detect_basis_indep :
  (ex_measurement_operator BasisX None)^*t *m
      ex_measurement_operator BasisX None =
    (ex_measurement_operator BasisZ None)^*t *m
      ex_measurement_operator BasisZ None.

Definition ex_source_dim_A :
  bb84_src_na ex_source_model = (2 ^ bb84_n Pex)%N.
Definition ex_source_dim_X :
  bb84_src_nx ex_source_model = (2 ^ bb84_n Pex)%N.
Definition ex_source_nb_pos : (0 < bb84_src_nb ex_source_model)%N.
Definition ex_source_overlap_standard_z :
  (bb84_n Pex)%:R * (1 : R) <= q_of_mx
    (bb84_source_x_povm ex_source_dim_A ex_source_dim_X)
    (@std_projective_povm R (2 ^ bb84_n Pex)).
\end{minted}
The concrete source and generic device record are expanded in
\Cref{tab:alignment-instance-quantum-components,tab:alignment-device-certificate};
the standard projectors are expanded in \Cref{tab:alignment-device-primitives}.
\switchcolumn
\textbf{Concrete preparation and measurement interface.}
The per-round qubit interface is
\[
\begin{gathered}
\rho_{\mathbb X,b}=|b_{\mathbb X}\rangle\!\langle b_{\mathbb X}|,
\quad
\rho_{\mathbb Z,b}=|b\rangle\!\langle b|,\\
K_{u,\emptyset}=0, \quad K_{u,b}=\rho_{u,b}.
\end{gathered}
\]
The source adapters certify dimensions $d_A=d_X=2^n$, a nonzero Bob
register, and the standard-$\mathbb Z$ overlap lower bound $nq\le q(M,N)$.
\end{paracol}
\endgroup
\AlignmentBlockEnd

\AlignmentParagraph{Entangled source.}

\AlignmentContinuationLabel{tab:alignment-instance-eb-components}
\AlignmentBlockHead{Coherent pair and preparation pointer.}
\begingroup\AlignmentBlockSetup
\begin{paracol}{2}
\code{src/qkd/example/instantiation.v}: \coqi{ex_bb84_x_prep_ket}, \coqi{ex_bb84_x_prep_ket_norm}, \coqi{ex_bb84_x_prep_ket_rho}, \coqi{ex_bb84_x_prep_ket_z_superposition}, \coqi{ex_bb84_x_prep_ket_z_rho}, \coqi{ex_bb84_pure_pair}, \coqi{ex_bb84_pure_pair_den1}, \coqi{ex_prep_ptr}, \coqi{ex_prep_ptr_complete}, \coqi{ex_cd_meas_Z}
\begin{minted}{coq}

Definition ex_bb84_x_prep_ket (bit : bool) : 'cV[C]_2 :=
  (@hadamard_mx2 R)^*t *m
    delta_mx (ex_bit_ord bit) (0 : 'I_1).
Lemma ex_bb84_x_prep_ket_norm (bit : bool) :
  (ex_bb84_x_prep_ket bit)^*t *m ex_bb84_x_prep_ket bit = 1%:M.
Lemma ex_bb84_x_prep_ket_rho (bit : bool) :
  cd_rho ex_device BasisX bit =
    ex_bb84_x_prep_ket bit *m (ex_bb84_x_prep_ket bit)^*t.
Lemma ex_bb84_x_prep_ket_z_superposition (zbit : bool) :
  ((((Num.sqrt (2%:R))^-1)%R)%:C) *:
    (ex_bb84_x_prep_ket false +
      bb84_eb_z_sign R zbit *: ex_bb84_x_prep_ket true) =
  delta_mx (ex_bit_ord zbit) (0 : 'I_1).
Lemma ex_bb84_x_prep_ket_z_rho (zbit : bool) :
  cd_rho ex_device BasisZ zbit = ((2%:R)^-1)%:C *:
    ((ex_bb84_x_prep_ket false +
        bb84_eb_z_sign R zbit *: ex_bb84_x_prep_ket true) *m
     (ex_bb84_x_prep_ket false +
        bb84_eb_z_sign R zbit *: ex_bb84_x_prep_ket true)^*t).
Definition ex_bb84_pure_pair : 'M[C]_(4 * 2) :=
  bb84_eb_pair_of_pure
    (ex_bb84_x_prep_ket false) (ex_bb84_x_prep_ket true).
Lemma ex_bb84_pure_pair_den1 : ex_bb84_pure_pair \is den1mx.
Definition ex_prep_ptr (basis : bb84_basis) (b : bool) : 'M[C]_4 :=
  match basis with
  | BasisX => bitlift_mx bb84_src_ptr_slot
      (delta_mx (bito b) (bito b))
  | BasisZ => bitlift_mx bb84_src_ptr_slot
      (hadamard_mx2 *m delta_mx (bito b) (bito b) *m hadamard_mx2^*t)
  end.
Lemma ex_prep_ptr_complete (basis : bb84_basis) :
  (ex_prep_ptr basis false)^*t *m ex_prep_ptr basis false +
  (ex_prep_ptr basis true)^*t *m ex_prep_ptr basis true = 1%:M.
Lemma ex_cd_meas_Z (b : bool) :
  cd_meas ex_device BasisZ (Some b) =
    delta_mx (ex_bit_ord b) (ex_bit_ord b).
\end{minted}
The standard projectors, coherent-pair constructor, and pointer lift are expanded in
\Cref{tab:alignment-device-primitives}.
\switchcolumn
\textbf{Coherent BB84 pair and Alice's pointer.}
Writing $|b_{\mathbb X}\rangle=H^\dagger|b\rangle$, the EB pair is the
canonical purification of $|0_{\mathbb X}\rangle,|1_{\mathbb X}\rangle$.
Alice's pointer uses computational projectors in the $\mathbb X$ branch and
Hadamard-conjugated projectors in the $\mathbb Z$ branch.
\end{paracol}
\endgroup
\AlignmentBlockEnd

\subsubsection{Classical Components}

\AlignmentParagraph{Field hashing.}

\AlignmentContinuationLabel{tab:alignment-instance-classical-components}
\AlignmentBlockHead{Concrete reconciliation and hash operands.}
\begingroup\AlignmentBlockSetup
\begin{paracol}{2}
\code{src/lhl/construction.v}: \coqi{binary_field_dim}, \coqi{binary_field_coord_project}, \coqi{binary_field_coord_output_card};
\code{src/qkd/example/instantiation.v}: \coqi{ex_ir_bits}, \coqi{ex_ir_pub}, \coqi{ex_ir_dec}, \coqi{ex_ir_leak}, \coqi{ex_hash_field_sig}, \coqi{ex_hash_K}, \coqi{ex_hash_char2}, \coqi{ex_hash_card}, \coqi{ex_hash_dim_nn}, \coqi{ex_field_of_ord}, \coqi{ex_field_of_ord_bij}, \coqi{ex_ord_of_bits}, \coqi{ex_ord_of_bits_bij}, \coqi{ex_field_hash}, \coqi{ex_field_hash_two_univ}, \coqi{ex_verify_seed_count}, \coqi{ex_tag_count}, \coqi{ex_verify_hash}, \coqi{ex_verify_hash_two_univ}, \coqi{ex_pa_seed_count}, \coqi{ex_pa_hash}, \coqi{ex_pa_hash_two_univ}
\begin{minted}{coq}
Definition ex_ir_bits : nat := 0.
Definition ex_ir_pub (_ _ : 'I_(2 ^ nn)) : 'I_(2 ^ ex_ir_bits) := ord0.
Definition ex_ir_dec (_ : 'I_(2 ^ ex_ir_bits)) (xb : 'I_(2 ^ nn)) := xb.
Definition ex_ir_leak : (ex_ir_bits)%:R <= bb84_leak_EC Pex.

Definition binary_field_dim : nat := \dim {: K : vectType 'F_2}.
Definition binary_field_coord_project (ell : nat)
    (Hell : (ell <= binary_field_dim)%N) (x : K) : 'rV['F_2]_ell :=
  \row_i coord (vbasis {: K}) (widen_ord Hell i) x.
Lemma binary_field_coord_output_card (ell : nat) :
  #|'rV['F_2]_ell| = (2 ^ ell)%N.

Definition ex_hash_field_sig :
  {K0 : finFieldType | 2 \in [pchar K0] & #|K0| = (2 ^ nn)%N}.
Definition ex_hash_K : finFieldType := s2val ex_hash_field_sig.
Definition ex_hash_char2 : 2 \in [pchar ex_hash_K] :=
  s2valP ex_hash_field_sig.
Definition ex_hash_card : #|ex_hash_K| = (2 ^ nn)%N :=
  s2valP' ex_hash_field_sig.
Local Notation ex_hash_F := (pPrimeCharType ex_hash_char2).
Lemma ex_hash_dim_nn : binary_field_dim ex_hash_char2 = nn.
Lemma ex_hash_F_card : #|ex_hash_F| = (2 ^ nn)%N.
Definition ex_field_of_ord (x : 'I_(2 ^ nn)) : ex_hash_F :=
  enum_val (cast_ord (esym ex_hash_F_card) x).
Lemma ex_field_of_ord_bij : bijective ex_field_of_ord.
Definition ex_ord_of_bits (l : nat) (z : 'rV['Z_2]_l) : 'I_(2 ^ l) :=
  cast_ord (binary_field_coord_output_card l) (enum_rank z).
Lemma ex_ord_of_bits_bij (l : nat) : bijective (@ex_ord_of_bits l).
Definition ex_field_hash (l : nat)
    (Hl : (l <= binary_field_dim ex_hash_char2)%N)
    (i x : 'I_(2 ^ nn)) : 'I_(2 ^ l) :=
  ex_ord_of_bits
    (field_mult_project_hash (binary_field_coord_project Hl)
      (ex_field_of_ord i) (ex_field_of_ord x)).
Lemma ex_field_hash_two_univ (l : nat)
    (Hl : (l <= binary_field_dim ex_hash_char2)%N) :
  @two_univ R _ _ _ (ex_field_hash Hl).
\end{minted}
\switchcolumn
\textbf{Zero-bit reconciliation and the two concrete field hashes.}
The instance uses
\[
d_{\rm IR}=0,
\quad \mathsf{IR}_{\rm pub}(x,y)=0,
\quad \mathsf{IR}_{\rm dec}(p,y)=y.
\]
Let $K=\mathrm{GF}(2^n)$, let $\phi:[2^n]\simeq K$, and let
$\pi_r$ retain the first $r$ binary coordinates.
For $r\in\{34,2000\}$,
\[
h_i^{(r)}(x)=\operatorname{ord}\!\left(
  \pi_r(\phi(x)\phi(i))\right).
\]
The verification family has $2^n$ seeds and $2^{34}$ tags; the PA family has
$2^n$ seeds and $2^{2000}$ outputs.  Both are two-universal.
\end{paracol}
\endgroup
\AlignmentBlockEnd

\AlignmentContinuationLabel{tab:alignment-instance-classical-components-cont}
\AlignmentBlockHead{Concrete reconciliation and hash operands (continued).}
\begingroup\AlignmentBlockSetup
\begin{paracol}{2}
\code{src/qkd/example/instantiation.v}: \coqi{ex_verify_seed_count},
\coqi{ex_tag_count}, \coqi{ex_verify_hash}, \coqi{ex_verify_hash_two_univ},
\coqi{ex_pa_seed_count}, \coqi{ex_pa_hash}, \coqi{ex_pa_hash_two_univ}
\begin{minted}{coq}
Definition ex_verify_seed_count : nat := (2 ^ nn)%N.
Definition ex_tag_count : nat := 2 ^ ex_tag_bits.
Lemma ex_tag_bits_le_dim :
  (ex_tag_bits <= binary_field_dim ex_hash_char2)%N.
Definition ex_verify_hash :
    'I_ex_verify_seed_count -> 'I_(2 ^ nn) -> 'I_(2 ^ ex_tag_bits) :=
  ex_field_hash ex_tag_bits_le_dim.
Definition ex_verify_seed_pos : (0 < ex_verify_seed_count)%N.
Definition ex_tag_pos : (0 < ex_tag_count)%N.
Definition ex_tag_le : ((ex_tag_count%:R : R)^-1 <= bb84_eps_c Pex)%R.
Definition ex_tag_bits_bound :
  log2 (ex_tag_count%:R : R) <= log2 (2%:R / bb84_eps_c Pex).
Definition ex_verify_hash_two_univ :
  @two_univ R _ _ _ ex_verify_hash.

Definition ex_pa_seed_count : nat := (2 ^ nn)%N.
Lemma exlen_le_dim : (exlen <= binary_field_dim ex_hash_char2)%N.
Definition ex_pa_hash :
    'I_ex_pa_seed_count -> 'I_(2 ^ nn) -> 'I_(2 ^ exlen) :=
  ex_field_hash exlen_le_dim.
Definition ex_pa_seed_pos : (0 < ex_pa_seed_count)%N.
Definition ex_pa_hash_two_univ : @two_univ R _ _ _ ex_pa_hash.
\end{minted}
The multiplication family and transport theorem are expanded in
\Cref{tab:alignment-short-seed}; the generic two-universal predicate in
\Cref{tab:alignment-ingredients}.
The dependent finite-field witness, coordinate basis, ordinal casts, and
their bijectivity proofs are displayed here because the two concrete maps use them.
The numeric operands \coqi{nn}, \coqi{Pex}, \coqi{ex_tag_bits}, and
\coqi{exlen} are expanded in \Cref{tab:alignment-instance}; the shared
\coqi{log2} dependency is expanded in \Cref{tab:alignment-arithmetic}.
\switchcolumn
\textit{Continued.}
\end{paracol}
\endgroup
\AlignmentBlockEnd

\subsubsection{Device Package}

\AlignmentParagraph{Certified package.}

\AlignmentContinuationLabel{tab:alignment-instance-package}
\AlignmentBlockHead{Concrete device packaging and witness adversary.}
\begingroup\AlignmentBlockSetup
\begin{paracol}{2}
\code{src/qkd/example/instantiation.v}: \coqi{ex_device_of}, \coqi{ex_device}, \coqi{ex_bb84_z_readout}, \coqi{ex_bb84_ptr_standard}, \coqi{ex_bb84_pure_xz_source}, \coqi{ex_source_x_povm}, \coqi{ex_device_projective_xz}, \coqi{ex_identity_adversary}
\begin{minted}{coq}
Definition ex_device_of (entangled : bool)
    : bb84_certified_device Pex (1 : R).
Proof.
  refine (@BB84CertifiedDevice R Pex 1
    ex_source_model ex_source_dim_A ex_source_dim_X ex_source_nb_pos
    ex_source_overlap_standard_z
    ex_prepare_state ex_prepare_state_den1
    ex_measurement_operator ex_measurement_complete ex_detect_basis_indep
    entangled ex_bb84_pure_pair ex_bb84_pure_pair_den1
    ex_prep_ptr ex_prep_ptr_complete
    ex_ir_bits ex_ir_pub ex_ir_dec ex_ir_leak
    ex_verify_seed_count ex_verify_seed_pos
    ex_tag_count ex_verify_hash ex_verify_hash_two_univ
    ex_tag_pos ex_tag_le ex_tag_bits_bound
    ex_pa_seed_count ex_pa_seed_pos ex_pa_hash ex_pa_hash_two_univ).
Defined.
Definition ex_device : bb84_certified_device Pex (1 : R) :=
  ex_device_of false.
Lemma ex_bb84_z_readout : @bb84_projective_z_readout R Pex 1 ex_device.
Proof.
  move=> b.
  have Hbo : ex_bit_ord b = bito b by apply/val_inj; case: b.
  by rewrite ex_cd_meas_Z Hbo
    (@bb84_M_click_sqrt_lossless R Pex 1 ex_device (ltac:(by [])))
    mulmx1.
Qed.
Lemma ex_bb84_ptr_standard : @bb84_eb_ptr_standard R Pex 1 ex_device.
Definition ex_bb84_pure_xz_source :
    @bb84_pure_xz_source R Pex 1 ex_device :=
  @BB84PureXZSource R Pex 1 ex_device ex_bb84_x_prep_ket
    ex_bb84_x_prep_ket_norm ex_bb84_x_prep_ket_rho
    ex_bb84_x_prep_ket_z_rho.
Lemma ex_source_x_povm (x : 'I_(2 ^ bb84_n Pex)) :
  @bb84_source_x_povm R Pex 1 (@cd_src R Pex 1 ex_device)
    (@cd_src_dim_A R Pex 1 ex_device)
    (@cd_src_dim_X R Pex 1 ex_device) x =
  @x_projective_povm R (bb84_n Pex) x.
Definition ex_device_projective_xz :
    @bb84_projective_xz_device R Pex 1 ex_device :=
  @BB84ProjectiveXZDevice R Pex 1 ex_device
    ex_bb84_pure_xz_source ex_bb84_z_readout ex_bb84_ptr_standard
    (ltac:(by [])) ex_source_x_povm.
Definition ex_identity_adversary : @adversary R L HL V Eve :=
  fun Q _ =>
    let QE := Q :|: Eve in
    QChannel_Build (@idso_qc R (@hstensor.Hs R L HL QE)).
\end{minted}
Every certified-device constructor operand is expanded in
\Cref{tab:alignment-instance-quantum-components,tab:alignment-instance-round-components,tab:alignment-instance-classical-components};
the target record and structural class in
\Cref{tab:alignment-device-certificate,tab:alignment-device-class}, and the
adversary interface in \Cref{tab:alignment-denotation}.
The protocol operand \coqi{Pex} is expanded in \Cref{tab:alignment-instance}.
\coqi{QChannel_Build}, \coqi{idso_qc}, and labelled-space construction are
CoqQ representation primitives.
\switchcolumn
\textbf{The record instance and positive-pass adversary.}
The shipped device is the prepare-and-measure specialization
\[
\begin{aligned}
\mathsf{dev}&=\mathsf{CertifiedDevice}(\\
  &\quad\rho_{\rm src},M,N,\rho,K,\\
  &\quad\rho_{\rm pair},P,
  \mathsf{IR},h^{(34)},h^{(2000)}),\\
\mathsf{entangled}&=\mathsf{false}.
\end{aligned}
\]
Together with the coherent source and pointer above, it carries the
projective-$\mathbb X/\mathbb Z$ certificate
\[
K_{\mathbb Z,b}=|b\rangle\!\langle b|\,M_{\rm click}^{1/2},
\qquad
M_x^{\rm src}=(H^{\otimes n})^\dagger
  |x\rangle\!\langle x|H^{\otimes n}.
\]
The liveness witness uses the identity attack on the attacked registers and
Eve's register:
\[
\mathcal E_A(Q,P)=\operatorname{id}_{Q\cup E}.
\]
\end{paracol}
\endgroup
\AlignmentBlockEnd

\subsubsection{Run Witness}

\AlignmentParagraph{Concrete history.}

\AlignmentContinuationLabel{tab:alignment-instance-witness}
\AlignmentBlockHead{Concrete device and positive-pass witness.}
\begingroup\AlignmentBlockSetup
\begin{paracol}{2}
\code{src/qkd/example/instantiation.v}: \coqi{ex_initial_store}, \coqi{ex_initial_lf}, \coqi{ex_initial_state}, \coqi{ex_initial_state_cq_wf}, \coqi{ex_key_sample_set}, \coqi{ex_test_set}, \coqi{ex_pass_choice}, \coqi{ex_pass_trace}, \coqi{ex_vseed0}, \coqi{ex_paseed0}, \coqi{ex_key_sample_value}, \coqi{ex_vseed_value}, \coqi{ex_paseed_value}, \coqi{ex_reach_trace}, \coqi{ex_reach_key}, \coqi{ex_reach_vseed}, \coqi{ex_reach_paseed}, \coqi{ex_reach_history}
\begin{minted}{coq}
Definition ex_initial_store : store := [ffun _ => bb84_bot_val].
Definition ex_initial_lf : 'F[HL]_setT :=
  (((dim 'H[HL]_setT)%:R)^-1)%:C *: (\1%VF : 'F[HL]_setT).
Definition ex_initial_state : state :=
  cq_single ex_initial_store ex_initial_lf.
Lemma ex_initial_state_cq_wf : cq_wf ex_initial_state.
Definition ex_key_sample_set : {set 'I_M} :=
  [set i : 'I_M | (nat_of_ord i < nn)%N].
Definition ex_test_set : {set 'I_M} :=
  [set i : 'I_M | (nn <= nat_of_ord i < nn + nn)%N].
Local Definition ex_pass_choice (i : 'I_M) : round_choice :=
  if i \in ex_test_set then
    ((bb84_basis_bool_val true, bb84_bit_val false),
      (bb84_basis_bool_val true, bb84_outcome_val (Some false)))
  else if i \in ex_key_sample_set then
    ((bb84_basis_bool_val false, bb84_bit_val false),
      (bb84_basis_bool_val false, bb84_outcome_val (Some false)))
  else
    ((bb84_basis_bool_val false, bb84_bit_val false),
      (bb84_basis_bool_val true, bb84_outcome_val (Some false))).
Local Definition ex_pass_trace : round_trace :=
  [ffun i => ex_pass_choice i].
Definition ex_vseed0 : 'I_(cd_vseed_count ex_device) :=
  Ordinal (cd_vseed_pos ex_device).
Definition ex_paseed0 : 'I_(cd_paseed_count ex_device) :=
  Ordinal (cd_paseed_pos ex_device).
Definition ex_key_sample_value : V := bb84_set_val ex_key_sample_set.
Definition ex_vseed_value : V := bb84_vseed_val (Some ex_vseed0).
Definition ex_paseed_value : V := bb84_paseed_val (Some ex_paseed0).
Definition ex_reach_trace : round_trace := ex_pass_trace.
Definition ex_reach_key : V := ex_key_sample_value.
Definition ex_reach_vseed : V := ex_vseed_value.
Definition ex_reach_paseed : V := ex_paseed_value.
Definition ex_reach_history : run_history :=
  (ex_reach_trace, ((ex_reach_key, ex_reach_vseed), ex_reach_paseed)).
\end{minted}
The concrete certificate and structural-class package is expanded in
\Cref{tab:alignment-instance-package};
the history shape is expanded in \Cref{tab:alignment-replay-rounds};
\coqi{store} in \Cref{tab:alignment-semantics-core};
\coqi{cq_single}/\coqi{cq_wf} in \Cref{tab:alignment-denotation}; and
\coqi{nn} in \Cref{tab:alignment-instance}.
\switchcolumn
\textbf{Certified ideal device and an error-free detected history.}
The concrete certificate fixes $q=1$ and the prepare-and-measure source flag,
then proves membership in the projective $\mathbb X/\mathbb Z$ class.
For liveness, the first $n$ rounds form the $\mathbb X$ key set, the next $k$ form the $\mathbb Z$ test set, every selected round is detected with bit $0$, and valid verification/PA seeds are chosen.
The singleton initial cq-state is maximally mixed on the quantum registers and is well formed.
\end{paracol}
\endgroup
\AlignmentBlockEnd

\AlignmentParagraph{Reachability.}

\AlignmentContinuationLabel{tab:alignment-run-pass-reachable}
\AlignmentBlockHead{Reusable positive-pass reachability theorem.}
\begingroup\AlignmentBlockSetup
\begin{paracol}{2}
\code{src/qkd/example/run.v}: \coqi{run_pass_prob_reachable}
\begin{minted}{coq}
Lemma run_pass_prob_reachable
    (qA : {set L}) (EA : @adversary R L HL V qA)
    (xi : @cq_state R L HL Cv V)
    (h0 : @run_history V M) (s0 : store) :
  (forall s, xi s \is psdlf) ->
  (forall (i : 'I_M) (a y : V), dev_rho i a y \is psdlf) ->
  (0 < M)%N -> 0 <= Qtol ->
  (forall p y, idec p y = y) ->
  0 < @complex.Re R
    (\Tr (hist_so empty_transcript dev_rho dev_meas EA h0 (xi s0))) ->
  (forall i : 'I_M, choice_y' (h0.1 i) != outEmpty) ->
  (n <= #|@trace_key_set V M basisX outEmpty h0.1|)%N ->
  (k <= #|@trace_test_set V M basisZ outEmpty h0.1|)%N ->
  (forall i : 'I_M, vread (choice_y (h0.1 i)) = Some false) ->
  (forall i : 'I_M, vread (choice_y' (h0.1 i)) = Some false) ->
  (forall (i : 'I_M) (s : store),
    0 < round_choice_weight R n k basisX basisZ vbit
      (h0.1 i) s) ->
  0 < unif_nsubset_mass n
        (@trace_key_set V M basisX outEmpty h0.1) h0.2.1.1 ->
  (forall s : store, 0 < unif_vseed s h0.2.1.2) ->
  (forall s : store, 0 < unif_pseed s h0.2.2) ->
  0 < run_pass_prob abortv vbool
    (@bb84_eval_scheme R L HL Cv V n k l leakEC M Qtol
      a_ y_ b_ y'_ d_ Q Xv Zv Kv XA XB ZA ZB lamv pECv
      Xhatv gv tAv tBv fv Sv Shv abortv Pvar
      basisX basisZ outEmpty bot empty_transcript
      vbool vbit vset vset_of vrat dev_rho dev_meas
      Qsrc dev_pair dev_ptr tr_app_rounds tr_app_pe tr_app_ec tr_app_pa
      bits_on test_bits_on err_count err_of unif_nsubset_mass
      ir_pub ir_dec unif_vseed verify_hash unif_pseed privacy_hash
      qA EA xi).
\end{minted}
The run probability, evaluator, history selectors, and witness history are
expanded in
\Cref{tab:alignment-run-objects,tab:alignment-replay-rounds,tab:alignment-instance-witness}.
\switchcolumn
\textbf{Reachable history implies positive pass mass.}
For a history $h_0$ with positive quantum trace, detected rounds, sufficiently
large key and test sets, zero selected bits, and positive round/subset/seed
weights,
\[
\begin{gathered}
\mathsf{PSD}(\xi,\rho_{\rm dev})\land M>0\land\Qtol\ge0
\land\mathsf{decode}=\id\\
\land\mathsf{reachable\_positive}(h_0)
\quad\Longrightarrow\quad p_{\rm pass}>0.
\end{gathered}
\]
This theorem is generic in the adversary and run operands; the concrete
identity-adversary witness discharges its premises separately.
\end{paracol}
\endgroup
\AlignmentBlockEnd

\subsubsection{Result}

\AlignmentParagraph{Concrete parameters.}

\AlignmentContinuationLabel{tab:alignment-instance}
\AlignmentBlockHead{Concrete numeric key-length witness.}
\begingroup\AlignmentBlockSetup
\begin{paracol}{2}
\code{src/qkd/example/instantiation.v}: \coqi{nn}, \coqi{exlen}, \coqi{ex_round_limit}, \coqi{ex_tag_bits}, \coqi{ex_ec_bits}, \coqi{ex_eps_c}, \coqi{ex_eps}, \coqi{ex_ebar}, \coqi{ex_eps_s}, \coqi{ex_leak_EC}, \coqi{Pex}, \coqi{ex_key_length_ub_ge2625}, \coqi{ex_key_length_sup}
\begin{minted}{coq}
Definition nn : nat := 10 ^ 4.
Definition exlen : nat := 2000.
Definition ex_round_limit : nat := 100 * (nn + nn).
Definition ex_tag_bits : nat := 34.
Definition ex_ec_bits : nat := 1024.
Definition ex_eps_c : R := ((2 ^ ex_tag_bits)%:R)^-1.
Definition ex_eps : R := ex_eps_c.
Definition ex_ebar : R := ex_eps_c.
Definition ex_eps_s : R := ex_eps_c *+ 2.
Definition ex_leak_EC : R := ex_ec_bits%:R.
Definition Pex : @bb84_protocol R :=
  @BB84Protocol R nn nn ex_round_limit exlen
    (100%:R)^-1 ex_eps_c ex_leak_EC.
Lemma ex_key_length_ub_ge2625 :
  (2625%:R : R) <= bb84_key_length_ub
    (bb84_n Pex) (bb84_k Pex) (1 : R) (bb84_Qtol Pex)
    ex_eps ex_ebar (bb84_leak_EC Pex) (bb84_eps_c Pex).
Lemma ex_key_length_sup (ec : R) (Hec_pos : 0 < ec)
    (Htag_le : (((cd_tag_count ex_device)%:R : R)^-1 <= ec)%R) :
  (bb84_l Pex)%:R <=
    sup (fun m : R =>
      exists eps ebar : R,
        0 < eps /\ eps < 1 /\ 0 < ebar /\ eps + ebar <= ex_eps_s /\
        m = ((floor (bb84_key_length_ub (bb84_n Pex) (bb84_k Pex)
          (1 : R) (bb84_Qtol Pex) eps ebar (bb84_leak_EC Pex)
          ec))%:~R : R)).
\end{minted}
The \coqi{bb84_protocol} record used by \coqi{Pex} and the key-length
functional are expanded in \Cref{tab:alignment-bb84}; the concrete
\coqi{ex_device} operand in the optimized witness is expanded in
\Cref{tab:alignment-instance-package}.
\switchcolumn
\textbf{Concrete BB84 parameters and optimized key length.}
For
\[
\begin{gathered}
q=1, \quad n=k=10^4, \quad \Qtol=10^{-2},\\
\leakec=1024,
\quad \eps=\bar\eps=\eps_c=2^{-34},\\
\eps_s=2^{-33},
\end{gathered}
\]
The fixed round budget is $M=100(n+k)=2{,}000{,}000$.
The certified instance has $\ell=2000$.
The checked fixed-parameter expression is at least $2625$, and its optimized
witness proves the $2000$-bit key-length premise for each admissible $e_c$.
\end{paracol}
\endgroup
\AlignmentBlockEnd

\AlignmentParagraph{Security theorem.}

\AlignmentContinuationLabel{tab:alignment-instance-theorem}
\AlignmentBlockHead{Concrete security and liveness export.}
\begingroup\AlignmentBlockSetup
\begin{paracol}{2}
\code{src/qkd/example/instantiation.v}: \coqi{bb84_secure_example}
\begin{minted}{coq}
Theorem bb84_secure_example (entangled : bool) :
  0 <
    @run_pass_prob R L HL (bb84_cvar M) V (@BVabort M) bb84_bool_val
      (@bb84_eval_mode R (Pex R) 1 (ex_device R) LE HE entangled
        (ex_identity_adversary HE) (ex_initial_state HE))
  /\
  (forall (EA : @adversary R L HL V Eve) (xi_init : state),
   cq_wf xi_init ->
   let xi_out := @bb84_eval_mode R (Pex R) 1 (ex_device R) LE HE
     entangled EA xi_init in
   let p_abort :=
     @run_abort_prob R L HL (bb84_cvar M) V (@BVabort M)
       bb84_bool_val xi_out in
   let p_pass :=
     @run_pass_prob R L HL (bb84_cvar M) V (@BVabort M)
       bb84_bool_val xi_out in
   let p_mismatch :=
     @run_output_mismatch_prob R L HL (bb84_cvar M) V
       (@BVS M) (@BVSh M) (@BVabort M) bb84_bool_val xi_out in
   [/\ (forall (h : @run_history V M) (s : store),
         @bb84_branch_shape R (Pex R) 1 (ex_device R) h s),
       p_abort + p_pass = 1
     & @eps_secure R L HL Zl (Fl :|: (Ereg :|: Ppub))
         p_mismatch (bb84_eps_c (Pex R))
         (@rhoZFEpub_eval_mode R (Pex R) 1 (ex_device R) LE HE
           entangled EA xi_init)
         (ex_eps_s R) p_abort
         (bb84_eps_c (Pex R) + ex_eps_s R)]).
\end{minted}
\AlignmentBoundaryLabel
\coqi{bb84_secure_example} combines the fixed positive-pass witness with
security for every adversary and well-formed initial state.
The positive-pass conjunct uses \coqi{ex_identity_adversary} and
\coqi{ex_initial_state}, whose well-formedness is proved by
\coqi{ex_initial_state_cq_wf}.
Its proof constructs \coqi{ex_reach_history} and discharges the generic \coqi{run_pass_prob_reachable} premises inline, while \coqi{ex_key_length_sup} supplies the optimized 2000-bit numeric witness.
The numeric parameters and key-length witness are expanded in
\Cref{tab:alignment-instance}; the device and class certificate in
\Cref{tab:alignment-instance-package}; the concrete initial state and history
construction in \Cref{tab:alignment-instance-witness}; and the generic
positive-pass theorem in \Cref{tab:alignment-run-pass-reachable}.
The theorem uses the protocol's fixed \coqi{bb84_eps_c (Pex R)}, while
\coqi{ex_device} supplies
\coqi{ex_tag_le} and \coqi{ex_tag_bits_bound}.
\switchcolumn
\textbf{Concrete BB84 instance and reachability (\Cref{thm:bb84-instance}).}
For each source mode $b\in\{\mathsf{PM},\mathsf{EB}\}$ and the fixed
parameters in \Cref{tab:alignment-instance},
\[
0<p_{\rm pass}^{b}(\mathcal E_{\rm id},\xi_{\rm ex}),
\]
and, for every $\mathcal E_A$ and well-formed $\xi_{\rm init}$,
\[
\begin{gathered}
\forall h,s,\ \mathsf{BranchShape}(h,s),\\
p_{\rm abort}^{b}+p_{\rm pass}^{b}=1,\\
\mathsf{Secure}\!\left(
  p_{\rm mismatch}^{b},2^{-34},\right.\\
\left.  \rho_{ZFE_{\rm pub}}^{b},
  2^{-33},p_{\rm abort}^{b},3\cdot2^{-34}\right).
\end{gathered}
\]
Here $\rho_{ZFE_{\rm pub}}^{b}$ is the normalized accepted readout from
the mode-$b$ evaluator.
\end{paracol}
\endgroup
\AlignmentBlockEnd

\flushbottom
\twocolumn